\documentclass[12pt]{article}

\usepackage{newtxtext,newtxmath}

\usepackage{graphicx}

\usepackage[letterpaper,margin=1in]{geometry}

\renewenvironment{abstract}
	{\quotation}
	{\endquotation}

\date{}

\makeatletter
\renewcommand{\fnum@figure}{\textbf{Figure \thefigure}}
\renewcommand{\fnum@table}{\textbf{Table \thetable}}
\makeatother

\usepackage{scicite}

\usepackage{url}

\usepackage[colorlinks=true, allcolors=blue]{hyperref}

\makeatletter

\newcounter{theorem}
\renewcommand{\thetheorem}{\arabic{theorem}}

\newcounter{lemma}
\renewcommand{\thelemma}{\arabic{lemma}}

\newcommand{\thm@optname}[1]{%
  \def\thm@tmp{#1}%
  \ifx\thm@tmp\@empty
  \else
    \ (#1)%
  \fi
}

\newenvironment{theorem}[1][]%
{%
  \refstepcounter{theorem}%
  \par\medskip\noindent
  \textbf{Theorem~\thetheorem\thm@optname{#1}.}\ \itshape\ignorespaces
}%
{%
  \par\medskip\ignorespacesafterend
}

\newenvironment{corollary}[1][]%
{%
  \refstepcounter{theorem}%
  \par\medskip\noindent
  \textbf{Corollary~\thetheorem\thm@optname{#1}.}\ \itshape\ignorespaces
}%
{%
  \par\medskip\ignorespacesafterend
}

\newenvironment{lemma}[1][]%
{%
  \refstepcounter{lemma}%
  \par\medskip\noindent
  \textbf{Lemma~\thelemma\thm@optname{#1}.}\ \itshape\ignorespaces
}%
{%
  \par\medskip\ignorespacesafterend
}

\newenvironment{definition}[1][]%
{%
  \refstepcounter{theorem}%
  \par\medskip\noindent
  \textbf{Definition~\thetheorem\thm@optname{#1}.}\ \normalfont\ignorespaces
}%
{%
  \par\medskip\ignorespacesafterend
}

\newenvironment{example}[1][]%
{%
  \refstepcounter{theorem}%
  \par\medskip\noindent
  \textbf{Example~\thetheorem\thm@optname{#1}.}\ \normalfont\ignorespaces
}%
{%
\par\medskip\ignorespacesafterend
}

\newenvironment{remark}[1][]%
{%
  \refstepcounter{theorem}%
  \par\medskip\noindent
  \textit{Remark~\thetheorem\thm@optname{#1}.}\ \normalfont\ignorespaces
}%
{%
  \par\medskip\ignorespacesafterend
}

\newcommand{\qedsymbol}{\rule{0.6em}{0.6em}}
\newcommand{\qed}{\hfill\qedsymbol}

\newenvironment{proof}[1][Proof]%
{%
  \par\medskip\noindent
  \textit{#1.}\ \normalfont\ignorespaces
}%
{%
  \qed\par\medskip\ignorespacesafterend
}

\makeatother

\providecommand{\bm}[1]{\boldsymbol{#1}}
\DeclareMathAlphabet{\mathcal}{OMS}{cmsy}{m}{n}
\DeclareMathAlphabet{\mathbb}{U}{msb}{m}{n}
\DeclareMathAlphabet{\mathfrak}{U}{euf}{m}{n}
\DeclareMathOperator{\Imag}{Im}
\makeatletter
\newcommand{\spx}[1]{%
  \if\relax\detokenize{#1}\relax
    \expandafter\@gobble
  \else
    \expandafter\@firstofone
  \fi
  {^{#1}}%
}
\makeatother

\newcommand{\dif}{\mathop{}\!\mathrm{d}}
\newcommand{\pd}[3][]{\frac{\partial\spx{#1}#2}{\partial#3\spx{#1}}}
\newcommand{\od}[3][]{\frac{\dif\spx{#1}#2}{\dif#3\spx{#1}}}

\def\scititle{Kuramoto meets Koopman: Constants of motion, symmetries, and network motifs}
\title{\bfseries \boldmath \scititle
}

\author{
	Vincent~Thibeault$^{1,2\ast}$,
	Benjamin~Claveau$^{1,2}$,
	Antoine~Allard$^{1,2}$,
    Patrick~Desrosiers$^{1,2,3}$\and
    \small$^{1}$ D\'epartement de physique, de g\'enie physique et d'optique, Universit\'e Laval, Qu\'ebec (Qc), Canada.\and\
    \small$^{2}$Centre interdisciplinaire en mod\'elisation math\'ematique de l'Universit\'e Laval, Qu\'ebec (Qc), Canada.\and\
    \small$^{3}$Centre de recherche CERVO, Qu\'ebec (Qc), Canada.\and
	\small$^\ast$Corresponding author. Email: vincent.thibeault.1@ulaval.ca\and
\small\textbf{Short title} : Kuramoto meets Koopman, \small\textbf{Teaser} : Tools inspired by quantum theory reveal \and\small hidden wiring patterns that preserve phase relationships in complex oscillator networks.
}

\begin{document} 

\let\oldaddcontentsline\addcontentsline
\renewcommand{\addcontentsline}[3]{}

\maketitle

\begin{abstract} \bfseries \boldmath
\vspace{-0.75cm}
Conserved quantities in phase-oscillator dynamics are well established for identically coupled oscillators, or groups thereof, but the explicit connectivity conditions under which more complex networks admit constants of motion remain difficult to identify.
Using Koopman theory, we derive general conditions for the existence of distinct conserved quantities in the Kuramoto model with heterogeneous phase lags on any weighted, directed, and signed graph. To this end, we find Koopman eigenfunctions and continuous Lie symmetries that generate different families of constants of motion.
The derived conditions reveal a broad class of network motifs that support conserved quantities and we detect these motifs in hundreds of complex empirical networks. The results thus point to connectivity patterns that can preserve phase relationships over time and motivate further investigations of Koopman spectral properties for dynamics on complex networks.
\end{abstract}

\centerline{\textit{For the 50th anniversary of the Kuramoto model (1975-2025).}}


\clearpage 

\section*{INTRODUCTION}

\noindent The Kuramoto model is a paradigmatic model of oscillators exhibiting synchronization~\cite{Kuramoto1975, Acebron2005, Rodrigues2016}. In its general form~\cite{Sakaguchi1986}, the model describes the evolution of each oscillator's phase by the system of differential equations
\begin{equation}\label{eq:kuramoto}
    \frac{\mathrm{d}\theta_j}{\mathrm{d}t} = \omega_j + \sum_{k = 1}^{N} W_{jk} \sin(\theta_k - \theta_j - \alpha_{jk})\,,  
\end{equation}
where $j\in\mathcal{V}:=\{1,...,N\}$, $\theta_j(t) \in \mathbb{R}$ is the phase of oscillator $j$ at time $t \in \mathbb{R}$, $\omega_j \in \mathbb{R}$ is the natural frequency of oscillator $j$, $W_{jk} \in \mathbb{R}$ is the $(j,k)$ element of the weight matrix $W$, encoding the strength of the interaction from oscillator $k$ to oscillator $j$, and $-\pi/2 < \alpha_{jk} \leq \pi/2$ is the $(j,k)$ element of the phase-lag matrix $\alpha$. Without loss of generality, we set $W_{jj} = 0$, $\alpha_{jj} = 0$ for all~$j$.

The model has a rich dynamics, giving rise to chaos~\cite{Engelbrecht2014, Bick2018, Karatetskaia2025}, chimeras~\cite{Kuramoto2002, Abrams2004, Abrams2008, Kotwal2017, Forrow2018}, explosive synchronization~\cite{Pazo2005, Gomez-Gardenes2011, Kundu2019, Kuehn2021}, and it has been used, for example, to describe Josephson junctions~\cite{Watanabe1993, Watanabe1994, Wiesenfeld1996}, nanoelectromechanical oscillators~\cite{Matheny2019}, associative memory~\cite{Arenas1994a, Nishikawa2004} in artificial intelligence, and even BOLD signal dynamics from the human cerebral cortex~\cite{Pope2021}. Over the years, it has become a central model for studying complex systems, whose high-dimensional nonlinear dynamics and intricate interactions give rise to emergent collective phenomena.~\cite{Mitchell2009, Ladyman2020, Krakauer2025}.

In this paper, we show that the Kuramoto model finds a natural and insightful formulation within Koopman theory~\cite{Koopman1931, Koopman1932, Carleman1932, Budisic2012, Brunton2022}. Introduced in 1931 by Bernard Koopman~\cite{Koopman1931} and further developed with John von Neumann~\cite{Koopman1932, VonNeumann1932}, Koopman theory was originally motivated by the formal analogy between classical and quantum mechanics: it sought to recast classical nonlinear dynamics in terms of linear operators by focusing on the evolution of observables rather than states~\cite{Joseph2020}.

In recent decades, Koopman theory has been primarily advanced through foundational mathematical works~\cite{Mezic2004, Mezic2005, Mauroy2013, Mauroy2016, Mezic2019} and through data-driven or algorithmic studies—such as dynamic mode decomposition and its various extensions~\cite{Schmid2008, Rowley2009, Schmid2010, Salova2019, Colbrook2023, Klus2024, Nicolaou2026}. It now stands as a unifying language at the intersection of dynamical systems, spectral theory, integrable systems, differential geometry, ergodic theory, and functional analysis. Yet, the extent to which this framework can yield analytical and conceptual insights into high-dimensional dynamics on complex networks remains an open question.

Motivated by this question, we leverage the operator-theoretic advantages of Koopman theory and argue for its conceptual relevance in understanding complex systems. Indeed, the Koopman operator is the time-evolution operator for functions of the system's state---the observables---including those describing emergent collective phenomena. Therefore, the goal of finding informative observables and their time evolution in complex systems is inherently aligned with the Koopmanian way of describing dynamical systems [Fig.~\ref{fig:fig1}]. 

Under Koopman's perspective, the finite-dimensional nonlinear system describing the model is traded for a \textit{linear} differential operator, the generator of the Koopman operator, or simply the Koopman generator. This linearity enables the spectral analysis of the operator, notably through the search for Koopman eigenfunctions and eigenvalues. While the representation of the generator under some basis of observables is typically infinite-dimensional, this will not be a problem in our approach. For the Kuramoto model in Eq.~\eqref{eq:kuramoto} under the change of coordinates $z_j = e^{i\theta_j}$ for all $j$ and $\theta_j\in\mathbb{R}$, it is straightforward to show that the Koopman generator is [Sec.~\ref{sec:intro_kooku}]
\begin{align}\label{eq:kooku}
    \mathcal{K} = \sum_{j,k \in\mathcal{V}}\left(A_{jk}z_k - \bar{A}_{jk}\bar{z}_kz_j^2\right)\partial_j\,,
\end{align}
where $\partial_j$ is the partial derivative with respect to $z_j$ and
\begin{align}\label{eq:complex_matrix}
     A = \frac{1}{2}\left(W \circ e^{-i\alpha} + i\,\mathrm{diag}(\bm{\omega})\right)\,,
\end{align}
with $e^{-i\alpha} =  (e^{-i\alpha_{jk}})_{j,k\in\mathcal{V}}$, $\bm{\omega} = (\omega_1\,\,\,\cdots\,\,\,\omega_N)^\top$, and the Hadamard product $\circ$. The complex weight matrix $A$ encapsulates every parameter of the dynamics and describes a directed, signed and complex-weighted graph~\cite{Bottcher2024}, where the off-diagonal weights are complex due to the non-zero phase lags [Fig.~\ref{fig:fig1}A]. The Koopman generator~\eqref{eq:kooku} will serve as our starting point for extracting constants of motion, whose connection to network motifs will be emphasized throughout the paper. 

Nearly 20 years after the publication of Kuramoto's paper, constants of motion for identical phase oscillators were brought to light in the seminal works of Watanabe and Strogatz~\cite{Watanabe1993, Watanabe1994}. Since then, by shifting the focus away from identical oscillators, there has been a surge of studies on complex networks of phase oscillators and their synchronization, as heterogeneous connections are a key feature of complex systems and substantially influence synchronization patterns~\cite{Arenas2008, Gfeller2008, Golubitsky1999, Dorfler2014, Rodrigues2016, Delabays2019, Fruchart2021}. 

This raises the question: under what network conditions is it still possible to find constants of motion for phase oscillators? Pikovsky and Rosenblum~\cite{Pikovsky2008} recognized that Watanabe-Strogatz (WS) theory is applicable to networks with $m$ all-to-all coupled communities, leading to $N - 3m$ constants of motion~\cite{Pikovsky2008, Hong2011_pre}. Another step towards heterogeneity was to analyze the Kuramoto dynamics on star graphs \cite{Kazanovich2003, Gomez-Gardenes2011, Chen2017, Kundu2019, Thibeault2020}---prevalent motifs of complex networks. WS theory can be applied to these systems, because periphery vertices are identically connected to the core~\cite{Vlasov2015, Vlasov2015a, Vlasov2017, Xu2018, Xu2019, Chen2019a}.

Yet, complex networks feature diverse motifs with potentially important stability and synchronizability properties~\cite{Milo2002, Moreno2004, Lodato2007, DHuys2008, Angulo2015, Schaub2016, Aguiar2018}. Such motifs have been extensively studied for diverse dynamics on networks together with their link to network symmetries, equivariance, pseudosymmetries, and fibrations~\cite{Ashwin1992, Golubitsky1999, Golubitsky2002, Golubitsky2006, Field2007, Nijholt2016, Morone2019a, Morone2020, Golubitsky2023, Makse2025}. Despite this progress, how network motifs support Koopman eigenfunctions and conserved quantities in complex oscillator networks remains, to our knowledge, an open question. We address this problem in the general Kuramoto model on networks~\eqref{eq:kuramoto}, using it as a concrete and analytically tractable setting for Koopman theory.

In Sec.~\ref{sec:constants}, we find general conditions for the existence of monomial eigenfunctions and what we call Vandermonde-ratio eigenfunctions. Combining these eigenfunctions yields families of constants of motion associated with different network motifs. We then reframe WS theory through Koopman theory and derive necessary and sufficient conditions that highlight the network-theoretic mechanisms behind the conserved quantities. In Sec.~\ref{sec:symmetries}, we generate families of conserved quantities from continuous Lie symmetries applied to Vandermonde-ratio eigenfunctions, cross-ratios, and WS integrals. The results of the first two sections are summarized in Table~\ref{tab:eigenfunctions} and their associated network motifs are illustrated in Fig.~\ref{fig:fig2}. Finally, Sec.~\ref{sec:examples} provides concrete examples and applications of these results. Notably, we detect motifs admitting conserved quantities in a dataset of 652 empirical networks including social networks, power grids, and connectomes.

\section*{RESULTS}

\section{Constants of motion and network motifs}
\label{sec:constants}

To begin with, a nonconstant scalar function $C$ of time and $\bm{z} := (z_1,...,z_N)$ is a constant of motion (first integral~\cite{Goriely1996}) of the dynamics with Koopman generator $\mathcal{K}$ if and only if $(\partial_t + \mathcal{K})[C(t, \bm{z})] = 0$, where $\partial_t$ is the partial derivative with respect to time. If $\partial_t C = 0$ and $\mathcal{K}[C]= 0$, the constant of motion $C$ is an eigenfunction of the Koopman generator with null eigenvalue [Fig.~\ref{fig:fig1}B]. One way to obtain constants of motion is to look for an eigenfunction $\psi(\bm{z})$ of the Koopman generator with eigenvalue $\lambda \in \mathbb{C}$, as it implies that $C(t,\bm{z}) = \psi(\bm{z})e^{-\lambda t}$ is conserved. A Koopman eigenfunction is also known as a time-dependent first integral or a second integral~\cite{Goriely2001}. 

If $\psi_1, ... , \psi_m$ are eigenfunctions of $\mathcal{K}$ with eigenvalues $\lambda_1, ... ,\lambda_m$ and $a_1, ..., a_m \in\mathbb{C}$, then the product $\psi_1^{a_1}\,...\,\psi_m^{a_m}$ is also an eigenfunction with eigenvalue $\sum_{j=1}^m a_j\lambda_j$. Therefore, if $\bm{a} = (a_1,...,a_m)$ is chosen to be orthogonal to $\bm{\lambda} =(\lambda_1,...,\lambda_m)$, $\sum_{j=1}^m a_j\lambda_j = 0$ and $\psi_1^{a_1}\,...\,\psi_m^{a_m}$ is conserved~\cite{Koopman1931, Budisic2012}. While this way of finding constants of motion long predates Koopman theory, going back at least to Darboux~\cite{Darboux1878, Goriely2001, Zhang2017}, the Koopman formalism offers a concise and insightful way of framing the problem in spectral terms. Moreover, the eigenvalues $\lambda_1,\ldots,\lambda_m$ also provide insights into the synchronization properties of the dynamics, since global frequency synchronization requires the imaginary parts of the corresponding eigenvalues to be rationally related~\cite{Mezic2019}. Koopman eigenfunctions are thus valuable for uncovering both conserved quantities and synchronization constraints in oscillator networks, and we now turn to identifying explicit eigenfunctions of~$\mathcal{K}$.

\subsection{Monomial eigenfunctions}
\label{subsec:monomials}
Since the vector field is polynomial for the Kuramoto model described in $\bm{z}$, we begin by searching for monomial eigenfunctions $z^{\bm{\mu}} := z_1^{\mu_1}...z_N^{\mu_N}$. In terms of the phases, a monomial eigenfunction corresponds to a complex-valued eigenfunction $\exp(i\bm\mu^\top\bm\theta)$, where $\bm\mu^\top\bm\theta$ is a real-valued linear observable with linear time evolution. The next theorem establishes necessary and sufficient conditions for the existence of these eigenfunctions, which are determined by the presence of specific network motifs [proof outline in Materials and Methods]. 
\begin{theorem}[Monomial eigenfunction]\label{thm:monomials}
    Let $\mathcal{W} \subset \mathcal{V}$ be a non-empty subset of vertices such that $|\alpha_{jk}| < \pi/2$ for all $j,k\in\mathcal{W}$. Let $\bm \mu = (\mu_1\,\,\,\cdots\,\,\,\mu_N)^\top \in \mathbb{R}^N$ satisfy $\mu_j \neq 0$ if and only if $j \in \mathcal{W}$. There exists a $\bm{\mu}$ such that $z^{\bm \mu}$ is an eigenfunction of $\mathcal{K}$ in Eq.~\eqref{eq:kooku} if and only if~: 
    \begingroup
    \renewcommand{\theenumi}{1.\arabic{enumi}}
    \renewcommand{\labelenumi}{\theenumi.}
        
    \begin{enumerate}
        \item \label{itm:1.1} $W_{jk} = 0$ for all $j\in \mathcal{W}$ and $k\in \mathcal{V}\setminus\mathcal{W}$~;
        \item \label{itm:1.2} $W_{jk}\neq 0$ whenever $W_{kj}\neq 0$ for all $j,k\in \mathcal{W}$~;
        \item \label{itm:1.3} $W_{i_1 i_2} ... W_{i_{\eta-1}i_{\eta}}W_{{i_{\eta} i_1}} = W_{{i_1 i_{\eta}}}W_{i_{\eta}i_{\eta-1}} ... W_{i_2 i_1}$ for all sequences $i_1, i_2, ..., i_\eta$ of elements of~$\mathcal{W}$~;
        \item \label{itm:1.4} 
        $\alpha_{jk} = -\alpha_{kj}$ whenever  $j,k\in \mathcal{W}$,  $j\neq k$, $W_{jk}\neq 0$~.
    \end{enumerate} 
    \endgroup
    \noindent If $z^{\bm \mu}$ is an eigenfunction, then its eigenvalue is $i\bm{\mu}^\top\bm{\omega}$.
\end{theorem}
\noindent 
The previous theorem establishes the existence of a monomial eigenfunction, but its exponents are in fact easily constructed. Indeed, in the connected case, start from some $\ell\in\mathcal{W}$, choose an initial $\mu_\ell$ and walk through the subgraph induced by $\mathcal{W}$, setting iteratively \mbox{$\mu_j = (W_{k j}/W_{jk}) \mu_k$} [Sec.~\ref{subsec:mu_procedure}]. The regularity of monomial eigenfunctions is discussed in Materials and Methods. Also, since the eigenvalue of $\mathcal{K}$ associated with $z^{\bm\mu}$ is purely imaginary, the corresponding Koopman-operator eigenvalue $\exp(i\bm{\mu}^\top\bm{\omega}t)$ lies on the unit circle. This is reminiscent of earlier extensions of the Koopman-von Neumann framework, in which eigenfunctions beyond invariants were considered~\cite{Mezic2004}.

As illustrated in Fig.~\ref{fig:fig2}A, condition~\ref{itm:1.1} ensures that the subgraph induced by the vertex set~$\mathcal{W}$ is a source within the whole network. Then, condition~\ref{itm:1.2} constrains reciprocity~: there cannot be a unidirectional edge within the subgraph [Fig.~\ref{fig:fig2}A]. Condition~\ref{itm:1.3} restricts cycles~: the product of the weights in the subgraph when circling clockwise in any cycle must be the same as the product of the weights when circling counterclockwise. In matrix terms, the second and third conditions mean that $(W_{jk})_{j,k\in\mathcal{W}}$ is symmetrizable~[Lem.~\ref{lem:symmetrizable_equivalence}]. To construct the associated (source) motif, one can thus define a symmetric matrix and multiply it by any real, nonzero diagonal matrix. Finally, condition~\ref{itm:1.4} implies that the submatrix $(\alpha_{jk})_{j,k\in\mathcal{W}}$ (with $\alpha_{jk}= 0$ whenever $W_{jk} = 0$) is antisymmetric. Thus, if the interaction from $k$ to $j$ depends on the shifted phase difference $\theta_k-\theta_j-\alpha_{jk}$, then the interaction from $j$ to $k$ involves the opposite shifted phase difference $-(\theta_k-\theta_j-\alpha_{jk})$.

Following Thm.~\ref{thm:monomials}, it is straightforward to show that for a symmetrizable weight matrix $W = (W_{jk})_{j,k\in\mathcal{V}}$, a monomial eigenfunction exists for any phase dynamics of the form 
\begin{align}\label{eq:phase_class_monomials}
    \frac{\mathrm{d}\theta_j}{\mathrm{d}t} = \omega_j +\sum_{k\in\mathcal{V}} W_{jk} H_{jk}(\theta_j, \theta_k)\,,
\end{align}
where $(H_{jk})_{j,k\in\mathcal{V}}$ are periodic functions satisfying $H_{jk}(\theta_j, \theta_k) = - H_{kj}(\theta_k, \theta_j)$ for all \mbox{$j,k \in \mathcal{V}$}. If, instead, $H_{jk}(\theta_j, \theta_k) = H_{kj}(\theta_k, \theta_j)$, then $(W_{jk})_{j,k\in\mathcal{V}}$ must rather be skew-symmetrizable. The generalization to higher-order interactions $\sum_{k_1,....,k_m}W_{jk_1...k_m}H_{jk_1...k_m}(\theta_j, \theta_{k_1},...,\theta_{k_m})$ is also straightforward.

If there are $q$ functionally independent monomial eigenfunctions $z^{\bm\mu_1},\ldots,z^{\bm\mu_q}$ for $\mathcal{K}$ with \mbox{$\bm{\mu}_{\rho}\in\mathbb{R}^N$} and eigenvalues $i\bm\mu_\rho^\top\bm\omega$, then there are $q$ constants of motion having the form $z^{\bm \mu_\rho}\exp(-i\bm\mu_\rho^\top \bm{\omega}\,t)$. If the natural frequencies are such that $\bm\mu_\rho^\top \bm{\omega} = 0$ for some $\rho$, the associated constant of motion is time-independent.  
Otherwise, given the $q$ monomial eigenfunctions, one can always construct $q-1$ functionally independent, time-independent monomial constants of motion $z^{\bm{\nu}_1},\,,\ldots,\, z^{\bm{\nu}_{q-1}}$ whose exponents satisfy the matrix equation $(\bm{\nu}_{1} \; \cdots \; \bm{\nu}_{q-1})=(\bm{\mu}_{1} \; \cdots \; \bm{\mu}_{q})\,O$, where $O\in \mathbb{R}^{q \times (q-1)}$ has linearly independent columns orthogonal to the vector $(i\bm\mu_1^\top \bm{\omega}\; \cdots\; i\bm\mu_q^\top \bm{\omega})^\top$ [Lem.~\ref{lem:monom_constants}].

In other words, the existence of a (time-independent) monomial conserved quantity is guaranteed by the presence of two or more source subgraphs with monomial eigenfunctions. Thm.~\ref{thm:monomials} thus provides a first class of Koopman eigenfunctions to build constants of motion. In the next section, we introduce another type of eigenfunction that can be combined with monomials to generate constants of motion, again related to a specific network structure~: bipartite graphs.


\subsection{Vandermonde-ratio eigenfunctions}
\label{subsec:vandermonde}

While studying symmetry-breaking mechanisms leading to chimera states in two pairs of Kuramoto oscillators, Burylko et al.~\cite{Burylko2022} identified a constant of motion that appears when oscillators are uncoupled within each pair and coupled across pairs, with zero phase lags and identical natural frequencies. If we denote the vertex sets of these pairs by $I =\{1, 2\}$ and $J = \{3, 4\}$, the constant of motion is $\mathrm{cot}\left((\theta_1 - \theta_2 + \theta_3 - \theta_4)/4\right)\tan\left((\theta_1 - \theta_2 - \theta_3 + \theta_4)/4\right)$. 

Motivated by this result, we identify a mechanism underlying this conservation, generalize the conserved quantity, and derive conditions for its existence in networks of nonidentical oscillators. This mechanism relies on the coexistence of a monomial eigenfunction [Thm.~\ref{thm:monomials}] and another Koopman eigenfunction, which exists under the following general conditions [proof outline in Materials and Methods]. 
\begin{theorem}[Vandermonde-ratio eigenfunction]\label{thm:vandermonde}
Let $I$ and $J$ be disjoint subsets of $\mathcal{V}$ of respective sizes $m\geq 2$ and $n\geq 2$ with $m + n \leq N$. Define $S_p:=\sum_{q\in I\setminus\{p\}}\sigma_{pq}$, $T_r:=\sum_{s\in J\setminus\{r\}}\tau_{rs}$, and
\begin{align}\label{eq:vandermonde}
V_{IJ}(\bm z)
&:=
\frac{\textstyle \prod_{p<q\in I}(\bar z_p-\bar z_q)^{\sigma_{pq}}}
     {\textstyle \prod_{r<s\in J}(z_r-z_s)^{\tau_{rs}}}\,,
\end{align}
where $\sigma_{pq} = \sigma_{qp}\in\mathbb{R}\setminus\{0\}$ for all $p,q\in I$ and $\tau_{rs} = \tau_{sr}\in\mathbb{R}\setminus\{0\}$ for all $r,s\in J$. If
\begingroup
    \renewcommand{\theenumi}{2.\arabic{enumi}}
    \renewcommand{\labelenumi}{\theenumi.}
\begin{enumerate}
\item \label{itm:2.1} $A_{jk} =0$ for all $j\in I\cup J$ and $k\in \mathcal V\setminus (I\cup J)$\,,
\item \label{itm:2.2} $A_{pq}=0$ for all distinct $p,q\in I$ and $A_{rs}=0$ for all distinct $r,s\in J$\,,
\item \label{itm:2.3}$A_{pr}=\mathcal A\,T_r$ and $A_{rp}=\bar{\mathcal{A}}\,S_p$ for all $p\in I$, $r\in J$ and $\mathcal{A}\in\mathbb{C}$\,,
\item \label{itm:2.4} $\omega_p = \Omega_I\in \mathbb{R}$ for all $p \in I$ and $\omega_r = \Omega_J\in \mathbb{R}$ for all $r \in J$\,,
\end{enumerate}
\endgroup
\noindent then $V_{IJ}(\bm z)$ is an eigenfunction of $\mathcal{K}$ with eigenvalue
\begin{align}\label{eq:eigenvalue_vandermonde}
    \lambda = -i(\Sigma_I\Omega_I + \Sigma_J\Omega_J)\,,
\end{align}
where $\Sigma_I = \frac{1}{2}\sum_{p\in I}S_p$ and $\Sigma_J = \frac{1}{2}\sum_{r\in J}T_{r}$.
\end{theorem}
We refer to this eigenfunction as a ``Vandermonde-ratio eigenfunction'', since it is the ratio of two Vandermonde determinants with different exponents for each factor. The regularity of these eigenfunctions is discussed in Sec.~\ref{SIsubsec:proof_thm2}.

The first condition states that the subgraph induced by $I\cup J$ is a source subgraph, while the second imposes a bipartite structure between $I$ and $J$. This bipartite graph can be weighted, directed, and signed according to the third condition~: a vertex $k_1 \in J$ (resp. $I$) can send equally-weighted edges to vertices in $I$ (resp. $J$) and another vertex $k_2\in J$ (resp. $I$) can do the same with different equally-weighted edges. 
At the same time, similarly to monomials, the phase lag $\alpha_{pr}$ associated with $A_{pr}$ must take a common value $\mathrm{Arg}(\mathcal{A})$ for all $p\in I$, $r\in J$, and the opposite value for all $p\in J$, $r\in I$.
Finally, the fourth condition allows the two parts $I$ and $J$ to have distinct natural frequencies, while requiring these frequencies to be identical within each part. 

The quantity $\Sigma_I$ (resp. $\Sigma_J$) has a clear network interpretation~: it is the in-strength $s^{\text{in}}_J$ (resp. $s^{\text{in}}_I$) of each vertex in $J$ (resp. $I$) up to a scale factor $2|\mathcal{A}|$. If $\Sigma_I \Omega_I = - \Sigma_J \Omega_J$, or equivalently, $s^{\text{in}}_I/s^{\text{in}}_J = -\Omega_I/\Omega_J$, the eigenvalue vanishes and the Vandermonde-ratio eigenfunction is itself a conserved quantity. 

More generally, for any $\Sigma_I, \Sigma_J, \Omega_I, \Omega_J$, there will always be a (time-independent) constant of motion, because the conditions of Thm.~\ref{thm:vandermonde} also directly imply the existence of a monomial eigenfunction. Indeed, choosing $\mu_j = S_j/2$ for $j\in I$, $\mu_j = T_j/2$ for $j \in J$, and $\mu_j = 0$ otherwise makes $z_1^{\mu_1}...z_N^{\mu_N}$ an eigenfunction with eigenvalue $+i(\Sigma_I\Omega_I + \Sigma_J\Omega_J)$. Hence, the product of eigenfunctions $z^{\bm\mu}V_{IJ}(\bm z)$ gives a constant of motion, whose real form is
\begin{align}\label{eq:conserved_vandermonde}
    \frac{\prod_{p<q\in I}\sin\left(\frac{\theta_p - \theta_q}{2}\right)^{\sigma_{pq}}}{\prod_{r<s\in J}\sin\left(\frac{\theta_r - \theta_s}{2}\right)^{\tau_{rs}}}\,.
\end{align}
For $m = n = 2$ and $\sigma_{12}=\tau_{34} = 1$, one recovers the result of Ref.~\cite{Burylko2022}, which can be expressed as the ratio $\sin((\theta_1 - \theta_2)/2)/\sin((\theta_3 - \theta_4)/2)$.
  
Bipartite graphs of phase oscillators enable various synchrony and stability patterns~\cite{Samani2008, Punetha2015, Montbrio2018, Nicolaou2019a, Peron2020} and interestingly, the conditions of Thm.~\ref{thm:vandermonde} are general enough to guarantee the existence of the conserved quantity for particular Janus oscillators~\cite{Nicolaou2019a, Peron2020} and excitation-inhibition Kuramoto oscillators~\cite{Montbrio2018} [Sec.~\ref{SIsubsec:janusEIKM}]. Figure~\ref{fig:fig2}B shows three examples of bipartite graphs admitting this conserved quantity. Since there are only two vertices in each layer of the motif~(b1), the edge weights from $I$ to $J$ (resp. $J$ to $I$) need to be identical (resp. with a different weight). By contrast, when one layer, say $I$, contains 3 or more vertices, vertices in $J$ can receive edges whose weights differ from one vertex of $I$ to another by condition~\ref{itm:2.3} (e.g., motifs~(b2) and (b3)).

The conserved quantity in Eq.~\eqref{eq:conserved_vandermonde} differs from those encountered in WS theory, but it is similarly built from products and ratios of sine factors involving half-angles. In the next subsection, we reframe WS theory within the Koopman framework and revisit the conservation of cross-ratios with an emphasis on their network-theoretic interpretation.

\subsection{Conserved cross-ratios}
\label{subsec:crossratios}
As a starting point, it is instructive to look more carefully at the Koopman generator and address the case of identical oscillators ($\omega_j = \omega\in \mathbb{R}$, $\alpha_{jk} = 0$ and $W_{jk} = 1$ for all $j,k\in\mathcal{V}$). Under these conditions, the generator becomes
\begin{align}\label{eq:kooku_identical}
    \mathcal{K}_{\mathrm{id}} = p(\bm{z})L_{-1} + i\omega L_0 - \overline{p(\bm{z})}L_1\,,
\end{align}
which is associated with an element of the Lie algebra for the projective special unitary group $\mathrm{PSU}(1, 1)$, where $L_n = \sum_{j=1}^N z_j^{n+1}\partial_j\,$ for $n\in\{-1,0,1\}$ are the generators of the algebra and $p(\bm{z}) = (1/2)\sum_{j=1}^Nz_j$. 

More generally, consider any dynamics with Koopman generator of the form 
\begin{align}\label{eq:generalLm101}
    \mathcal{K}_{\alpha\beta\gamma} = \alpha(t,\bm{z})L_{-1} + \beta(t,\bm{z})L_0 + \gamma(t, \bm{z})L_1\,,
\end{align}
including the Winfree model~\cite{Winfree1967}, the theta model~\cite{Ermentrout1986} or any Riccati-type dynamics. If one finds a joint invariant~\cite{Olver1993, Olver1995} for $L_{-1}, L_0, L_1$, i.e., such that $L_n[C(\bm{z})] = 0$ for all $n\in\{-1, 0, 1\}$, then $C$ is a constant of motion. The method of characteristics makes it possible \mbox{to deduce} such a joint invariant~: the cross-ratio
\begin{equation}\label{eq:cross_ratio}
c_{abcd}(\bm{z}) = \frac{(z_{c}-z_{a})(z_{d}-z_{b})}{(z_{c}-z_{b})(z_{d}-z_{a})}
\end{equation}
for distinct indices $a,b,c,d\in\mathcal{V}$ [Sec.~\ref{subsec:joint_invariants}]. This represents a simple systematic method for deriving the conservation of cross-ratios for phase oscillators, which is also already understood in the literature through projective geometry and complex analysis~\cite{Marvel2009, Stewart2011}.

Knowing that cross-ratios are constants of motion for any dynamics with a vector field of the form~\eqref{eq:generalLm101}, what are the conditions on $A$ so that $\mathcal{K}[c_{abcd}(\bm{z})] = 0$~? Answering this question brings the network interpretation of cross-ratio conservation into focus. The following theorem formalizes this interpretation by giving conditions that are both necessary and sufficient [proof outline in Materials and Methods]. 

\begin{theorem}[Cross-ratio conservation]\label{thm:crossratios}
The cross-ratio $c_{abcd}$ is a constant of motion of the Kuramoto model~\eqref{eq:kuramoto} if and only if the vertices $a$, $b$, $c$, and $d$ of the graph described by the complex matrix in Eq.~\eqref{eq:complex_matrix} have the same:
\begingroup
    \renewcommand{\theenumi}{3.\arabic{enumi}}
    \renewcommand{\labelenumi}{\theenumi.}
\begin{enumerate}
\item \label{itm:3.1} outgoing interactions within $\{a,b,c,d\}$, i.e.,
\begin{align*}
\begin{aligned}
    A_{ba} &= A_{ca} = A_{da} =: \mathcal{A}_a\,,\\ 
    A_{ab} &= A_{cb} = A_{db} =: \mathcal{A}_b\,, 
\end{aligned}
\qquad
\begin{aligned}
    A_{ac} &= A_{bc} = A_{dc} =: \mathcal{A}_c\,, \\
    A_{ad} &= A_{bd} = A_{cd} =: \mathcal{A}_d\,,
\end{aligned}
\end{align*}

\item  \label{itm:3.2} incoming interactions from $\mathcal{V}\setminus\{a, b, c, d\}$, i.e.,
\begin{equation*}
    A_{ak} = A_{bk} = A_{ck} = A_{dk}\,,\quad \forall k\in\mathcal{V}\setminus\{a, b, c, d\}
\end{equation*}

\item \label{itm:3.3} shifted natural frequencies
\begin{align*}
\omega_j - 2\,\mathrm{Im}(\mathcal{A}_j) = \omega_k - 2\,\mathrm{Im}(\mathcal{A}_k),\quad \forall j,k \in \{a,b,c,d\}\,.
\end{align*}

\end{enumerate}
\endgroup
\end{theorem}
As illustrated in Fig.~\ref{fig:fig2}C, condition~\ref{itm:3.1} constrains the possible directed network motifs that allow a cross-ratio to be a constant of motion. Condition~\ref{itm:3.2} restricts how these motifs can receive incoming edges from other vertices, e.g., a vertex $k_1$ can send equally-weighted edges to $\{a,b,c,d\}$ and another vertex $k_2$ can do the same with different equally-weighted edges. Finally, condition~\ref{itm:3.3} requires the oscillators to have the same effective natural frequency. Basic examples of motifs admitting conserved cross-ratios are given in Fig.~\ref{fig:fig2}C and in Sec.~\ref{subsec:basix_examples_thm3}. 

Thm.~\ref{thm:crossratios} also provides the necessary and sufficient conditions to have $N - 3$ conserved cross-ratios. Indeed, the model maximally has $N - 3$ functionally independent conserved cross-ratios if and only if
\begingroup
\renewcommand{\theenumi}{3\Alph{enumi}}
\renewcommand{\labelenumi}{\theenumi.}
\begin{enumerate}
    \item \label{itm:3A} $A_{j\ell} = A_{k\ell} =: \mathcal{A}_\ell$ for all $\ell\in\mathcal{V}$\,,
    \item \label{itm:3B} $\omega_j - 2\,\mathrm{Im}(\mathcal{A}_j) = \omega_k - 2\,\mathrm{Im}(\mathcal{A}_k)$\,,
\end{enumerate}
\endgroup
\noindent for all pairs $(j,k)$ with $j,k\in\mathcal{V}$ and $k,\ell\neq j$ [Cor.~\ref{cor:max_integrals_of_motion}]. These conditions were previously recognized to be sufficient~\cite{Lohe2017} and we add that they are also necessary. Thus, there are $N$ different directed graphs (non-isomorphic, weakly connected, binary) leading to $N - 3$ integrals, and complex weights satisfying \ref{itm:3A} can be included.

When $\omega_j = \omega$, $\alpha_{jk} = \delta - \pi/2$, $\sin\delta = 0$, and \mbox{$W_{jk} = 1/N$} for all $j,k$ in Eq.~\eqref{eq:kuramoto}, the model satisfies conditions~\ref{itm:3A}-\ref{itm:3B}, reduces to a Hamiltonian system~\cite{Watanabe1994}, and admits an elegant interpretation in hyperbolic geometry~\cite{Chen2017b, Chen2019, Lipton2021}. In this regime, there are $N-2$ WS integrals~: 
\begin{align}\label{eq:ws_integral}
    C^{\mathrm{ws}}_{j_1...j_N} = S_{j_1j_2}S_{j_2j_3}...S_{j_{N-1}j_N}S_{j_Nj_1}\,,
\end{align}
where \mbox{$S_{jk}(\bm\theta) = \sin((\theta_j-\theta_k)/2)$}~\cite{Watanabe1993, Watanabe1994, Goebel1995}. An example of a motif admitting six WS integrals is shown in Fig.~\ref{fig:fig2}~(c1), with the corresponding conditions stated in Lem.~\ref{lem:ws}. As shown in Ref.~\cite{Marvel2009}, cross-ratios are ratios of WS integrals, i.e., $c_{abcd} = -C^{\mathrm{ws}}_{acdb\,j_5...j_N}/C^{\mathrm{ws}}_{adcb\,j_5...j_N} = S_{ca}S_{db}/S_{cb}S_{da}$. Away from this special parameter regime, however, the WS quantities in Eq.~\eqref{eq:ws_integral} are generally not constants of motion. For the Koopman generator in Eq.~\eqref{eq:generalLm101}, they are instead Darboux functions, $\mathcal{K}_{\alpha\beta\gamma}[C^{\mathrm{ws}}_{j_1...j_N}] = \Lambda(\bm z)C^{\mathrm{ws}}_{j_1...j_N}$, and they all share the same cofactor $\Lambda(\bm z) = 2(\gamma(t, \bm{z})p(\bm z) - \alpha(t, \bm{z})\overline{p(\bm z)})$, which is another way to explain the conservation of their ratios~[Sec.~\ref{SIsubsec:WS_darboux}]. 

Thus far, Theorems \ref{thm:monomials}, \ref{thm:vandermonde}, and \ref{thm:crossratios} have provided general conditions that closely link the existence of constants of motion to specific motifs in the graph connecting the oscillators. On the one hand, these oscillator motifs, taken in isolation, can admit different Koopman eigenfunctions simultaneously, because the conditions of Theorems~\ref{thm:monomials}, \ref{thm:vandermonde}, and \ref{thm:crossratios} can overlap. This is illustrated in several examples in Fig.~\ref{fig:fig2}. On the other hand, since there is no restriction on the outgoing edges from these motifs to the rest of the network (white vertices in Fig.~\ref{fig:fig2}) and the motifs admitting conserved cross-ratios can have incoming edges in general, it is possible to connect the motifs in multiple ways. Each source subgraph then imposes a potentially complex forcing pattern on the rest of the network, and if such a motif admits a Koopman eigenfunction with imaginary eigenvalue $i\lambda$, the possible synchronization frequencies are constrained to harmonics, i.e., integer multiples of $\lambda$. Altogether, this general setup yields a broad family of weighted, directed, signed, and modular networks with nontrivial synchronization that admit Koopman eigenfunctions and conserved quantities.

From there, it is natural to ask whether the conserved quantities are associated with any symmetry in the oscillators' connections or dynamics. In fact, the conditions established in this section are not conditions on the existence of graph automorphisms (a.k.a., network symmetries~\cite{MacArthur2008}), known to have crucial implications for cluster synchronization~\cite{Nicosia2013, Sorrentino2016, Pecora2014, Nishikawa2016a, Schaub2016, Cho2017, Aguiar2018, Makse2025}. However, another type of symmetry will be useful, that is, Lie symmetries (a.k.a., point symmetries, continuous symmetries): transformations that map a solution of a system of differential equations to another solution  [Fig.~\ref{fig:fig1}B]~\cite{Olver1993}. Indeed, in the next section, we use continuous symmetries to generate additional constants of motion and extend the families of conserved quantities presented in Sec.~\ref{sec:constants}, thereby completing the list of quantities presented in Table~\ref{tab:eigenfunctions}.

\section{Symmetry-generated constants of motion}
\label{sec:symmetries}

For variational problems, Noether's theorem guarantees that continuous symmetries are associated with constants of motion, thus allowing a reduction of the system's order~\cite{Noether1918, Olver1993}. Although there is no guarantee that such symmetries underlie the presence of constants of motion in non-Lagrangian and non-Hamiltonian systems~\cite{Anco1997, Anco1998}, finding a Lie symmetry may still allow one to build new functionally independent constants of motion from known ones~\cite{Gonzalez-Gascon1978}. 


In this section, we introduce a lemma that formulates a criterion for Lie symmetries in terms of the Koopman generator. The lemma then leads us to find continuous Lie symmetries in the Kuramoto model that generate functionally independent constants of motion. 

\subsection{Koopmanian condition for continuous symmetries in dynamical systems}
The general method to derive Lie symmetries for differential equations involves computing the prolongation of the symmetry group action or its generators $\mathcal{S}$~\cite[Theorem 2.71]{Olver1993}, often requiring lengthy calculations. For Euler-Lagrange problems and Hamiltonian systems, the prolongation condition amounts to identifying commuting operators. As shown below, such a simplification is also possible for first-order ODEs. Indeed, for $\dot{\bm y} = \bm F (t, \bm y)$ with $\mathcal{U} = \partial_t + \mathcal{K}$ and $\mathcal{K} = \sum_{j = 1}^N F_j(t, \bm{y})\partial_j$, the next lemma gives the symmetry criterion in terms of operator commutation [Sec.~\ref{SIsubsec:inf_crit}].\\
\noindent\textbf{Lemma (Infinitesimal condition for Lie symmetries)}
    \textit{A connected local group of transformations $G$ acting on an open subset of $\mathbb{R} \times \mathbb{R}^N$ is a symmetry group of the first-order ODEs $\dot{\bm y} = \bm F (t, \bm y)$ if and only if
    \begin{align}\label{eq:symkoo}
    [\mathcal{U}, \mathcal{S}] - \mathcal{U}[\xi(t, \bm{y})]\,\mathcal{U} = 0
    \end{align}
    for every generator $\mathcal{S} = \xi(t, \bm{y})\partial_t + \sum_{j=1}^N\phi_j(t, \bm{y})\partial_{j}$ of $G$.}
\\
For symmetry generators where $\xi$ is constant and $\phi_1,...,\phi_N$ are time-independent, the infinitesimal criterion~\eqref{eq:symkoo} takes the more familiar and elegant form
\begin{align}\label{eq:symkook}
    [\mathcal{K}, \mathcal{S}] = 0\,.
\end{align}
If $\gamma_{\varepsilon} = \exp(\varepsilon \mathcal{S}) \in G$ is the symmetry generated by $\mathcal{S}$, the associated (non infinitesimal) criterion to Eq.~\eqref{eq:symkook} is $(D\gamma_\varepsilon)\bm{F} = \bm{F}(\gamma_\varepsilon)$, which is the usual $G$-equivariance condition~\cite{Golubitsky2002, Field2007, Burylko2022, Golubitsky2023}.

To find a symmetry generator, one uses a general symmetry generator~$\mathcal{S}$ (or a specific ansatz on the form of $\mathcal{S}$), which leads to partial differential equations called the determining equations. Solving these equations hopefully lead to specific forms of $\xi, \phi_1,...,\phi_N$. Although Eqs.~(\ref{eq:symkoo}-\ref{eq:symkook}) enable using commutation relations to search for symmetries, there is no general procedure to obtain particular solutions of the determining equations~\cite{Bluman2002}: this is the \emph{art} of Lie's method~\cite{Olver1993}.

When successful, however, this search is especially valuable~: Lie symmetries can generate constants of motion and Koopman eigenfunctions from known ones. This can be easily shown using the previous lemma. Indeed, for a constant of motion $C$ and a symmetry generator $\mathcal{S}$, Eq.~\eqref{eq:symkoo} and $\mathcal{U}[C] = 0$ directly lead to 
\begin{align}\label{eq:symmetry_generated_constants}
    \mathcal{U}[\mathcal{S}[C]] = \mathcal{S}[\mathcal{U}[C]] + \mathcal{U}[\xi]\,\mathcal{U}[C] = 0\,,
\end{align}
meaning that $\mathcal{S}[C]$ is conserved. For $\xi = 0$, the same calculation shows that if \mbox{$[\mathcal{U}, \mathcal{S}] = 0$} and $\psi$ is a Koopman eigenfunction, then $\mathcal{S}[\psi]$ is a Koopman eigenfunction with the same eigenvalue and $\mathcal{S}[\psi]/\psi$ is a constant of motion. Yet, these calculations do not guarantee that $\mathcal{S}[C]$ or $\mathcal{S}[\psi]$ are functionally independent of the known constants of motion and eigenfunctions, as the next subsection will illustrate.

\subsection{Continuous symmetries of the Kuramoto model}
\label{subsec:symkur}
For the Kuramoto model, the generator $\mathcal{U} = \partial_t + \mathcal{K}$ with $\mathcal{K}$ given in Eq.~\eqref{eq:kooku} acts on observables depending on $(t, \bm{z})$ and we want the symmetries to be automorphisms of the $N$-torus~[Fig.~\ref{fig:fig1}C] potentially acting on time. Using the criterion~\eqref{eq:symkoo}, it is straightforward to show that the global dilatation generator $i\,L_0$ (rotation of all the oscillators), the Koopman generator $\mathcal{K}$, and the trivial generator $f(t)\,\mathcal{U}$ for some smooth function $f$ are symmetry generators. Yet, the symmetry generators $iL_0$, $\mathcal{K}$, $f(t)\mathcal{U}$ do not generate independent constants of motion from the monomials, the Vandermonde-ratios or the cross-ratios.

To go further, a simple intuition comes in handy~: the symmetries must map periodic solutions to periodic solutions. This restricts $\phi_1,...,\phi_N$ in $\mathcal{S}$ to periodic functions, enabling their expansion in Fourier series $\phi_\ell(\bm z) = \sum_{\bm{p}\in\mathbb{Z}^N} \varphi_{\ell \bm{p}}z^{\bm{p}}$, under the assumptions $\xi = 0$ and $\partial_t\phi_{\ell}(t, \bm{z}) = 0$. After using commutation relations and simplifying, the general determining equations take the form of an infinite differential-algebraic system of equations [Sec.~\ref{SIsubsec:general_det_eq}]. To narrow our search, we limit $\bm{p}$ to a finite subset of $\mathbb{Z}^N$ with fixed total degree $\sum_{j=1}^Np_j$, leading to a finite overdetermined linear system [Sec.~\ref{SIsubsec:general_det_eq}-\ref{SIsubsec:determining_matrix}]
\begin{align}
    D(A)\bm{\varphi}=\bm 0\,,
\end{align}
where the ``determining matrix'' $D(A)$ depends only on the interaction matrix $A$.

Finding a symmetry then amounts to identifying $A$ for which the determining matrix has a zero singular value, with the corresponding right singular vector $\bm{\varphi}$ giving the coefficients of a symmetry generator.
This simplification highlights the link between continuous Lie symmetries and the underlying network structure in the Kuramoto model.
Reversing Yang's statement~\cite{Yang1980}, one can say that interaction dictates symmetry, a viewpoint in the spirit of the classic work of Golubitsky, Stewart and collaborators~\cite{Golubitsky1999, Golubitsky2002, Golubitsky2023, Makse2025}.

This approach, combined with symbolic calculations in basic examples~[Sec.~\ref{SIsubsec:proof_thm4}], leads us to infer a family of network motifs admitting continuous symmetries in the Kuramoto model. Whenever there is a source oscillator $s$ with natural frequency $\omega_{s}$ connecting disjoint subgraphs with vertex sets $\mathcal{W}_1,...,\mathcal{W}_r$ in the network, the Koopman generators of the subgraphs in the rotating frame of $s$,
\begin{align}\label{eq:peripheral_symmetry}
    \mathcal{S}_{\eta} = \mathcal{K}_{\eta} - i\omega_{s} L_0^{\eta}\,,\quad\eta\in\{1,...,r\}\,,\,r>1\,,
\end{align}
are Lie symmetry generators, where $L_0^{\eta} = \sum_{j\in\mathcal{W}_{\eta}}z_j\partial_j$ and $\mathcal{K}_{\eta} = \sum_{j\in\mathcal{W}_{\eta}}\sum_{k\in\mathcal{W}_{\eta}\cup\{s\}}(A_{jk}z_k - \bar{A}_{jk}\bar{z}_kz_j^2)\partial_j$ [Lem.~\ref{lem:sym_KL}]. Given a solution $(z_j(t))_{j\in\mathcal{V}}$, such symmetries make the oscillators in the subgraphs evolve in time in the frame of $s$ while leaving the dynamics of the other oscillators unchanged, which gives another solution $(\tilde{z}_j(t))_{j\in\mathcal{V}}$.

In some cases, it is possible to derive the explicit form of the symmetries. Let $\omega_s = 0$ and $\theta_s(0) = 0$ without loss of generality and 
\begin{align}\label{eq:peripheral_symmetry_identical}
    \mathcal{S}_\eta = \rho_\eta(\bm z) L_{-1}^\eta + i\Omega_\eta L_{0}^\eta - \overline{\rho_\eta(\bm z)} L_{1}^\eta
\end{align}
with $\rho_\eta(\bm z) = \mathcal{A}_s + \sum_{k\in\mathcal{W}_\eta} \mathcal{A}_{\eta k}z_k$, $\Omega_\eta = \omega_{\ell_\eta} - 2\Imag(\mathcal{A}_{\eta \ell_\eta})$, $\ell_\eta\in\mathcal{W}_\eta$, and $\mathcal{A}_s, \mathcal{A}_{\eta k} \in \mathbb{C}$ (e.g., when the vertices in $\mathcal{W}_\eta$ with $\# \mathcal{W}_\eta \geq 4$ satisfy the conditions from Thm.~\ref{thm:crossratios}). The generator $\mathcal{S}_\eta$ is thus associated with $\mathrm{PSU}(1,1)$ and given a solution $z_j(t)\in\mathbb{T}$ for all $t$ and $j\in\mathcal{V}$, the symmetry acts as a Möbius transformation to generate the transformed solution
\begin{align}\label{eq:mobius_symmetry}
    \tilde{z}_j(t) = \exp(\varepsilon \mathcal{S}_\eta)z_j(t) =     \frac{e^{i\phi_\eta(\varepsilon)}z_j(t) + Z_\eta(\varepsilon)}{1 + e^{i\phi_\eta(\varepsilon)}\overline{Z_\eta(\varepsilon)}z_j(t)}
\end{align}
where $j\in \mathcal{W}_\eta$, $(Z_\eta(0), \phi_\eta(0)) = (0, 0)$, and
\begin{align}
    \od[]{Z_\eta}{\varepsilon} &= F_\eta + i\Omega_\eta Z_\eta - \bar{F}_\eta Z_\eta^2~\label{eq:Z_variable}\\
    \od[]{\phi_\eta}{\varepsilon} &= \Omega_\eta + i(\bar{F}_\eta Z_\eta - F_\eta\bar{Z}_\eta)~\label{eq:phi_variable}\\
    F_\eta(Z_\eta, \phi_\eta, t) &= \mathcal{A}_s + \sum_{k\in\mathcal{W}_\eta} \mathcal{A}_{\eta k}\frac{e^{i\phi_\eta}z_k(t) + Z_\eta}{1 + e^{i\phi_\eta}\overline{Z_\eta}z_k(t)}~\label{eq:F_function}\,.
\end{align}
In this case, the corresponding symmetry subgroup is therefore a Möbius group. We describe how to apply the above methodology numerically in Sec.~\ref{subsec:explicit_sym} and we provide a concrete example in Sec.~\ref{subsec:viz}.

We now turn to the generation of new constants of motion from Lie symmetries in the Kuramoto model.

\subsection{Peripheral constants of motion}
\label{subsec:peripheralevolution}

The symmetry generators in Eq.~\eqref{eq:peripheral_symmetry} enable us to extract new forms of constants of motion in the Kuramoto model~[Fig.~\ref{fig:fig2}D, example in Eq.~\eqref{eq:symgen_example}], as stated in the next theorem [proof outline in Materials and Methods]. 
\begin{theorem}[Peripheral constants of motion]\label{thm:thm4}
   Consider that the Kuramoto model in Eq.~\eqref{eq:kuramoto} has a symmetry generator $\mathcal{S}_\eta$ as defined in Eq.~\eqref{eq:peripheral_symmetry} related to the subgraph induced by $\mathcal{W}_\eta$ and the source oscillator $s$.

   \vspace{0.1cm}
   
   \noindent {\normalfont 4A.} If four vertices $a,b,c,d \in\mathcal{V}\setminus\{s\}$ have 
   \begingroup
   \setlength{\leftmargini}{4em}
    \renewcommand{\theenumi}{4.\arabic{enumi}A}
    \renewcommand{\labelenumi}{\theenumi.}
   \begin{enumerate}
   \item \label{itm:4.1A} a unique incoming edge with weight $\mathcal{A}_s$ from $s$;
   \item \label{itm:4.2A} identical natural frequencies $\omega$;
   \item \label{itm:4.3A} and one, two or three of them belong to $\mathcal{W}_\eta$,
   \end{enumerate}
   \endgroup
   \noindent then the cross-ratio $c_{abcd}$ and $\mathcal{S}_\eta[c_{abcd}]$ are functionally independent constants of motion.
\vspace{0.1cm}

   \noindent{\normalfont 4B.} If three vertices $u, v, w\in\mathcal{V}\setminus\{s\}$ have
   \begingroup
   \setlength{\leftmargini}{4em}
    \renewcommand{\theenumi}{4.\arabic{enumi}B}
    \renewcommand{\labelenumi}{\theenumi.}
   \begin{enumerate}
   \item \label{itm:4.1B} a unique incoming edge with weight $\mathcal{A}_s$ from $s$;
   \item \label{itm:4.2B} identical natural frequencies $\omega = \omega_s - 2\,\mathrm{Im}(\mathcal{A}_s)$;
   \item \label{itm:4.3B} and one or two of them belong to $\mathcal{W}_\eta$,
   \end{enumerate}
   \endgroup
\noindent then the cross-ratio $c_{suvw}$ and $\mathcal{S}_\eta[c_{suvw}]$ are functionally independent constants of motion.
\end{theorem}
Because the symmetry acts on peripheral oscillators of a source, we refer to $\mathcal{S}_\eta[c_{abcd}]$ and $\mathcal{S}_\eta[c_{suvw}]$ as ``peripheral constants of motion''.

As a consequence of the preceding theorem, if there is a source star with $n\geq 4$ leaves having identical frequencies $\omega$ (to satisfy \ref{itm:4.2A}), where all edges from the source $s$ to the leaves have identical complex weight $\mathcal{A}_{s}$ (to satisfy \ref{itm:4.1A}), Thm.~\ref{thm:crossratios} (and Lem.~\ref{lem:independence_cross_ratios}) implies that there are $n - 3$ conserved cross-ratios $(c_{\rho})_{\rho\in \{1,...,n-3\}}$ associated with the leaves. There is also one additional conserved and functionally independent cross-ratio $c_s$ depending on the core if $\omega = \omega_s - 2\,\mathrm{Im}(\mathcal{A}_s)$ (condition~\ref{itm:4.2B}). 

Now, recall that there are no restrictions on the outgoing edges from the $n$ leaves to conserve the related cross-ratios and that $B_n$ is the Bell number. Hence, there are $r = B_n - 1$ ways of partitioning the leaves in at least two sets (to satisfy \ref{itm:4.3A} and \ref{itm:4.3B}) while including them as sources in their respective subgraphs' vertex sets $\mathcal{W}_1,...,\mathcal{W}_r$ [Fig.~\ref{fig:fig2}D] (to satisfy \ref{itm:4.1A} and \ref{itm:4.1B}). This setup makes $r$ symmetry generators act on at least one of the $n-3$ cross-ratios in such a way that $\mathcal{S}_{\eta}[c_{\rho}]$ for some $\rho$ and $\eta$ are conserved. If $\omega_s = \omega + 2\,\mathrm{Im}(\mathcal{A}_s)$ is satisfied, $\mathcal{S}_{\eta}[c_{s}(\bm{z})]$ is also conserved for all $\eta$. Yet, not all of these constants are functionally independent; their independence hinges on the specific network structure.

We have thus identified symmetries that generate constants of motion through the time evolution of peripheral parts in networks admitting conserved cross-ratios. We now consider translations of the phases and the phase lags. Except for monomials, these transformations provide parameterized versions of the conserved quantities introduced above.

\subsection{Parametrized families of conserved quantities}
\label{subsec:parametrized_families}
Although global phase translations generated by $iL_0$ are symmetries, non-uniform translations generated by $iL_0^{\bm{\alpha}} = i\sum_{j=1}^N \alpha_j z_j \partial_j = \sum_{j=1}^N \alpha_j\partial_{\theta_j}$ with $\alpha_j \in\mathbb{R}$ are not symmetries since \mbox{$[\mathcal{K}, L_0^{\bm{\alpha}}] \neq 0$}. However, if $g_j(\bm\theta) := \omega_j + \sum_{k=1}^NW_{jk}\sin(\theta_k - \theta_j - \alpha_{jk})$, then  
\begin{align}\label{eq:translation_dynamics}
    e^{iL_0^{\bm\alpha}}[g_j(\bm\theta)] = \omega_j + \sum_{k=1}^NW_{jk}\sin(\theta_k - \theta_j -\alpha_j - \alpha_{jk} + \alpha_k)
\end{align}
and replacing $\alpha_{jk}$ by \mbox{$-\alpha_j + \alpha_{jk} + \alpha_k$} gives the original vector field. If the phase lags are treated as state variables satisfying $\dot\alpha_{jk} = 0$, this equivalently means that
\begin{align}\label{eq:translation_generator}
\mathcal{S}_{\bm{\alpha}} &= \sum_{j=1}^N \alpha_j \partial_{\theta_j}-\sum_{p,q=1}^N (\alpha_p-\alpha_q)\,\partial_{\alpha_{pq}}\,,
\end{align}
generates the translation symmetry $e^{\varepsilon \mathcal{S}_{\bm{\alpha}}}$, where
\begin{align}\label{eq:translation}
    e^{\varepsilon \mathcal{S}_{\bm{\alpha}}}f(\bm\theta, \bm\alpha) = f(\bm \theta + \varepsilon\bm\alpha\,, \,\,\alpha + \varepsilon(\bm{1}\bm{\alpha}^\top-\bm{\alpha}\bm{1}^\top))\,,
\end{align}
for some function $f$ and $\bm{\alpha} = (\alpha_1,...,\alpha_N)$ [Sec.~\ref{SIsubsec:parametrized_families}]. 
The associated translation group $G_{\bm{\alpha}}$ is a symmetry group of the extended system~\eqref{eq:kuramoto} with $\dot\alpha_{jk} = 0$ for all $j,k$. 

These translations do not produce fundamentally different algebraic types of Koopman eigenfunctions or constants of motion, but they generate $G_{\bm{\alpha}}$-orbits of such quantities. In other words, if $\psi$ is a Vandermonde-ratio eigenfunction, a cross-ratio, a WS integral, or a peripheral constant of motion, then $\exp(\varepsilon\mathcal{S}_{\bm{\alpha}})[\psi]$ gives the corresponding members of the family, as presented in Table~\ref{tab:eigenfunctions} for $\varepsilon = -1$. The conditions of the previous theorems are also transformed accordingly under the same action to preserve these members.

For Vandermonde-ratio eigenfunctions, applying the symmetry to condition~\ref{itm:2.3} of Thm.~\ref{thm:vandermonde} yields 
\begin{enumerate}
    \item[2.3.'] $A_{pr} = e^{i\alpha_p}\mathcal A\,T_re^{-i\alpha_r}$ and $ A_{rp} = e^{i\alpha_r}\bar{\mathcal{A}}\,S_pe^{-i\alpha_p}$,
\end{enumerate}
for all $p\in I$, $r\in J$. 
In the case of cross-ratios, choosing $a$ as the label of reference, the transformation gives
    \begin{align}\label{eq:cross_ratio_family}
        e^{-\mathcal{S}_{\bm{\alpha}}}[c_{abcd}(\bm{z})] 
        = \frac{(e^{i\beta_c}z_c - z_a)(e^{i\beta_d}z_d - e^{i\beta_b}z_b)}{(e^{i\beta_c}z_c - e^{i\beta_b}z_b)(e^{i\beta_d}z_d - z_a)}\,,
    \end{align}
for all $\beta_\ell = \alpha_a - \alpha_\ell \in \mathbb{R}$ with $\ell \in \mathcal{W} := \{a,b,c,d\}$, forming a three-parameter family of cross-ratios. The conservation of these cross-ratios in the Kuramoto model~\eqref{eq:kuramoto} follows from applying the symmetry to the conditions of Thm.~\ref{thm:crossratios}, leading to
\begingroup
    \renewcommand{\theenumi}{3.\arabic{enumi}'}
    \renewcommand{\labelenumi}{\theenumi.}
\begin{enumerate}
\item \label{itm:3.1p} $e^{i\beta_\ell}A_{k\ell} =: \mathcal{A}_\ell \in \mathbb{C}$,  $\forall\ell \in \mathcal{W}$ and $\forall k\in \mathcal{W}\setminus \{\ell\}$;
\item \label{itm:3.2p} $A_{ak} = e^{i\beta_b}A_{bk} = e^{i\beta_c}A_{ck} = e^{i\beta_d}A_{dk}$, \mbox{$\forall k\in\mathcal{V}\setminus \mathcal{W}$};
\item \label{itm:3.3p} $\omega_a - 2\Imag(\mathcal{A}_a) = \omega_\ell - 2\Imag(e^{-i\beta_\ell}\mathcal{A}_\ell)$, $\forall \ell \in \{b, c, d\}$.
\end{enumerate}
\endgroup
\noindent Such a family of cross-ratios is not specific to the Kuramoto model: the same construction holds for general complex Riccati-type dynamics, such as those considered by Cestnik and Martens~\cite{Cestnik2024}. The families of peripheral constants of motion and WS integrals are also obtained in the same way [Sec.~\ref{SIsubsec:parametrized_families}].

While the preceding theorems establish the formal groundwork, they also enable practical insights into synchronization and conserved phase relationships in complex empirical networks. The next section provides examples that make these theorems and their implications more explicit.

\section{Examples and applications}
\label{sec:examples}

In this section, we first use a basic example to illustrate the geometric restrictions that constants of motion and Lie symmetries impose on the oscillators. Then, we shortly present an application to the synchronization of conformist-contrarian oscillators. We finally show that various network motifs admitting conserved quantities are present in empirical networks. 

\subsection{Visualizing constants of motion and symmetries}
\label{subsec:viz}
Let us now consider the example illustrated in Fig.~\ref{fig:fig3}. The simplest part of the graph in panel A is the directed star of five oscillators (vertices 1 to 5) with Koopman generator $\partial_t + i\omega_1z_1\partial_1 + \mathcal{K}_2 + \mathcal{K}_3 + \mathcal{K}_4 + \mathcal{K}_5$, where $\mathcal{K}_{\eta} = (i\omega z_\eta + \mathcal{A}_1z_1 - \bar{\mathcal{A}}_1\bar{z}_1z_\eta^2)\,\partial_\eta$ for $\eta \in\{2,3,4,5\}$ and $\mathcal{A}_1 \in\mathbb{C}\setminus\{0\}$. The presence of a source $s = 1$ implies that there is a monomial eigenfunction $z_1$ with eigenvalue $i\omega_1$ [Thm.~\ref{thm:monomials}] and $z_1e^{-i\omega_1 t}$ is conserved. Thm.~\ref{thm:crossratios} then guarantees the conservation of $c_{2345}$. There are also five symmetry generators $\mathcal{S}_{\eta} = \mathcal{K}_{\eta} - i\omega_1 L_0^\eta$ with $\eta\in\{2,...,6\}$ [Lem.~\ref{lem:sym_KL}] and Thm.~\ref{thm:thm4} ensures that $\mathcal{S}_\eta[c_{2345}]$ for all $\eta\in\{2,3,4,5\}$ are conserved. If $\omega_1 = \omega + 2\,\Imag(\mathcal{A}_1)$, then $c_{1234}$ is also conserved along with $\mathcal{S}_\eta[c_{1234}]$ for all $\eta\in\{2,3,4\}$.

Altogether, there are five independent constants of motion, say $z_1e^{-i\omega_1 t}$, $c_{1234}$, $c_{2345}$, $\mathcal{S}_2[c_{2345}]$, $\mathcal{S}_3[c_{2345}]$, for this star and this subsystem is completely integrable [Sec.~\ref{SIsubsec:basic_example_thm4}]. For $\mathcal{A}_1 = \sigma/4\in\mathbb{R}$, the real form of the peripheral constant of motion generated by $\mathcal{S}_2$ and $c_{2345}$ (assumed positive) is
\begin{align}\label{eq:symgen_example}
    \mathcal{S}_2[(2/\sigma)\ln c_{2345}] = \frac{C_{12}S_{12}S_{45}}{S_{42}S_{52}}
\end{align}
with $S_{jk} = \sin\left((\theta_j - \theta_k)/2\right)$ and $C_{jk} = \cos\left((\theta_j - \theta_k)/2\right)$. 

The vertices 6 and 7 admit a second monomial eigenfunction. One can thus form the time-independent constant of motion $z_1^{\nu_1}z_6^{\nu_6}z_7^{\nu_7}$ and illustrate the level set of its real form, i.e., a plane [Fig.~\ref{fig:fig3}B].
Also, the Lie symmetries $\exp(\varepsilon (\mathcal{S}_3 + \mathcal{S}_4 + \mathcal{S}_5))$ or $\exp(\varepsilon\mathcal{S}_6)$ act as Möbius transformations and $\exp(\varepsilon (\mathcal{S}_3 + \mathcal{S}_4 + \mathcal{S}_5))$ maps the solution from one invariant set to another [Fig.~\ref{fig:fig3} C and D]. The system is set in the reference frame of the source ($\omega_1 = 0$), making it possible to illustrate the level sets entirely.

\subsection{Perturbing the synchronization of conformist-contrarian oscillators}
\label{sec:conformist_contrarians}
Consider again the example in Fig.~\ref{fig:fig3}. Lohe~\cite{Lohe2017} showed that for the isolated part $\mathcal{W}_6$ ($\mathcal{A}_\mathrm{pert} = 0$) satisfying conditions~\ref{itm:3A}-\ref{itm:3B}, the coupling parameter $\gamma = \sum_{\ell \in \mathcal{W}_6}\mathrm{Re}(\mathcal{A}_\ell)$ measures whether the conformists or the contrarians dominate the long-term behavior of the system. In short, for $\gamma < 0$, the system becomes frequency synchronized; for $\gamma = 0$, periodic states are possible; and for $\gamma > 0$, it synchronizes perfectly. Connecting only a single leader to $\mathcal{W}_6$ can preserve all the $\#\mathcal{W}_6 - 3$ cross-ratios (condition~\ref{itm:3.2}), while inducing notable changes in the synchrony of~$\mathcal{W}_6$. 

When the contrarians prevail in $\mathcal{W}_6$ ($\gamma < 0$), the independent conformist leader~1 disrupts the frequency synchronization. In its frame, the leader is stationary and drives the oscillators in $\mathcal{W}_6$ to stop moving. But if its influence is not strong enough, the conformists and contrarians continue to push the system toward frequency synchronization. These opposing forces lead to the limit cycle in Fig.~\ref{fig:fig3}F. If the leader exerts a stronger influence, every oscillator in $\mathcal{W}_6$ converges to a fixed position [Fig.~\ref{fig:fig3}G]. This also explains the shift of the synchronization threshold toward negative values of $\gamma$ [Fig.~\ref{fig:fig3}E].

\subsection{Detecting network motifs admitting conserved quantities in empirical networks}

Generalized Kuramoto models arise as limiting cases of distinct models of complex systems. For instance, a phase reduction of the Wilson-Cowan model~\cite{Wilson1973}, which describes the activity of excitatory and inhibitory regions in the brain, leads to a generalized Kuramoto model as in Eq.~\eqref{eq:kuramoto}~\cite{Schuster1990, Pietras2019}. For power grids, approximating power balance equations leads to a second-order Kuramoto model providing insights into grid robustness~\cite{Filatrella2008, Rohden2012, Zhang2019b}. The model is also interpretable as a model of social opinion formation~\cite{Hong2011_pre}, which motivates the terminology used in the previous subsection. 

In all these examples, the phase description of the Kuramoto model should not be seen as a realistic mechanistic model of the underlying system, but rather as a mathematically tractable limit that preserves certain qualitative features of the collective dynamics. In fact, the Kuramoto model has even been used as a tool to reveal the topological scales~\cite{Arenas2006} or the detailed connectivity~\cite{Timme2007} of complex networks through synchronization. In a similar spirit, the previous theorems establish how conserved phase relationships are linked to specific network motifs. But do such motifs actually appear in empirical networks~?

Answering this question requires the detection of network motifs~\cite{Milo2002, Alon2007, Yu2020, Makse2025} that admit constants of motion; a task that quickly becomes computationally intensive.
For this application, we restrict our attention to a subset of these motifs, namely sources, 2-sources (e.g., (a2) in Fig.~\ref{fig:fig2}), and network motifs admitting conserved cross-ratios.
In Materials and Methods, we introduce a method to detect these motifs using the binary versions of empirical networks' weight matrices.

By analyzing 652 empirical networks---including social networks, power grids, and connectomes---we found 498 networks admitting at least one constant of motion [Fig.~\ref{fig:fig4}C] and we establish lower bounds on the number of conserved quantities~[Fig.~\ref{fig:fig4}D]. The average lower bound on the number of constants of motion, divided by $N$, reaches 16\% for undirected networks and 29\% for directed networks. These large averages are mostly explained by the presence of sources~[Fig.~\ref{fig:fig4}E], stars (e.g., in power grids), and maximal cliques with at least four vertices (e.g., in social networks). Yet, we also found less typical motifs such as 2-sources~(e.g., Fig.~\ref{fig:fig4}A) and those admitting conserved cross-ratios that are neither stars nor maximal cliques (e.g., Fig.~\ref{fig:fig4}B).

On the one hand, the presence of sources and 2-sources indicates that it is possible for some oscillators to have their phase constrained to (hyper)planes, as illustrated in Fig.~\ref{fig:fig3}B. On the other hand, having a motif admitting a conserved cross-ratio with a given non-zero value means that the phases of the involved oscillators cannot cross. While the concrete consequences of these observations are application-dependent, our results shed light on specific structural patterns that can sustain stable quantities over time in complex oscillator systems. 


\section*{DISCUSSION}

Spectral properties of linear operators, such as adjacency, Laplacian, and Jacobian matrices, have been central to understanding nonlinear dynamics on complex networks, from stability~\cite{May2001,Neri2020,Meena2023} and control~\cite{Liu2016a,Forrow2018} to critical transitions~\cite{Restrepo2005,Restrepo2008,Karrer2014,VanMieghem2012,Castellano2017a} and dimension reduction~\cite{Gfeller2007,Jiang2018,Laurence2019,Thibeault2024}. Yet, the spectral properties of the Koopman generator for such dynamics on networks remain largely unexplored. The Koopman generator encodes both the structure and the dynamics, making it a particularly relevant linear operator to characterize the structure-function relationship in complex systems. Uncovering its spectral properties is challenging, however.

In this paper, we identified various network motifs admitting Koopman eigenfunctions in the celebrated Kuramoto model. We combined these eigenfunctions and found Lie symmetries that generate conserved quantities determined by heterogeneous connection patterns between the oscillators. This provides different possibilities for obtaining $1 \leq n \leq N$ constants of motion, which, by the preimage theorem, constrains trajectories to a manifold of dimension $N-n$. The results thus offer a clear path toward partial integrability of general coupled oscillator dynamics~\cite{Lohe2018, Jacimovic2018, Lohe2020, Lipton2021, Cestnik2024} on more complex networks. 

Another direction concerns a deeper investigation of continuous symmetries, but also various other types of symmetries, to generate other Koopman eigenfunctions and constants of motion. For instance, we have not considered discrete symmetries~\cite{Burylko2022} and since constants of motion are not necessarily linked to continuous symmetries in general, adjoint symmetries may also provide complementary insights~\cite{Anco1997, Anco1998}. Network-based notions of symmetry, including graph automorphisms, groupoid formalisms, and graph fibrations~\cite{Nijholt2016, Morone2019a, Morone2020, Golubitsky2023, Makse2025}, also offer promising avenues for relating conserved quantities to network patterns.

Detecting the motifs in empirical networks from the four theorems in their full generality (for directed, weighted, and signed networks) remains an open problem. The conditions of the theorems are also expected to be satisfied only approximately for some motifs in empirical networks. The prevalence of these motifs, as well as their ability to support observables that are nearly conserved over finite time intervals, is yet to be characterized. In fact, beyond partial integrability, the framework provides a natural basis for approximate dimension reductions of oscillator networks~\cite{Ott2008, Gfeller2008, Gottwald2015} using nonlinear observables rather than linear ones~\cite{Gao2016,Laurence2019,Thibeault2020, Vegue2023}. In this regard, spectral methods from Koopman theory~\cite{Rowley2009, Budisic2012, Brunton2022, Colbrook2023} and perturbative approaches~\cite{Vlasov2016a} are promising avenues for future work.

More generally, the present framework suggests an inverse spectral perspective for dynamical systems on networks, where one seeks classes of dynamics admitting prescribed Koopman eigenfunctions or conserved quantities. In this respect, the framework is not restricted to Kuramoto-type dynamics and naturally extends to much broader classes of nonlinear systems, as will be further developed in a subsequent paper.



\subsection*{MATERIALS AND METHODS}






\subsubsection*{Proof outline for Thm.~1}
Applying $\mathcal{K}$ directly to $z^{\bm\mu}$ yields
\begin{equation}\label{eq:eigen_equation_monomial}
\sum_{\substack{j,k\in\mathcal{W}\\ k\neq j}} \left(A_{jk}\mu_j - \bar{A}_{kj}\mu_k\right)z^{\bm \mu - \bm e_j + \bm e_k} + \sum_{j\in\mathcal{W}}\sum_{k\in\mathcal{V}\setminus\mathcal{W}} \left(A_{jk}\mu_j z^{\bm \mu - \bm e_j + \bm e_k} - \bar{A}_{jk}\mu_jz^{\bm \mu + \bm e_j - \bm e_k }\right)+ iz^{\bm \mu}\sum_{j\in\mathcal{W}}\mu_j\omega_j\,,
\end{equation}
where $(\bm e_j)_\ell = \delta_{j\ell}$ and all monomials are linearly independent. Clearly, if $z^{\bm{\mu}}$ is an eigenfunction, its eigenvalue is $i\sum_{j\in\mathcal{W}}\mu_j\omega_j$. The necessary and sufficient conditions on $\bm \mu$ and $A$ for the eigenvalue equation to be satisfied become $A_{jk} = 0$ for all $j\in\mathcal{W}$ and $k\in\mathcal{V}\setminus\mathcal{W}$, and $A_{jk}\mu_j = \bar{A}_{kj}\mu_k$ for all $j,k\in\mathcal{W}$ with $j\neq k$. In terms of $W$ and $\alpha$, these conditions are equivalent to $W_{jk} = 0$ for all $j\in\mathcal{W}$ and $k\in\mathcal{V}\setminus\mathcal{W}$ (condition~\ref{itm:1.1}), and
\begin{equation}\label{eq:nec_suf_cond_Walpha_mm}
    \mu_jW_{jk}e^{i\alpha_{jk}} = \mu_kW_{kj}e^{-i\alpha_{kj}}
\end{equation}
for all $j,k\in\mathcal{W}$ with $j\neq k$. Since $|\alpha_{jk}| < \pi/2$ and $\mu_jW_{jk} \in \mathbb{R}$, it is straightforward to show that this is equivalent to $\alpha_{jk} = - \alpha_{kj}$ (condition~\ref{itm:1.4}) and $\mu_jW_{jk} = \mu_kW_{kj}$. The latter means that $(W_{jk})_{j,k\in\mathcal{W}}$ is symmetrizable and is satisfied if and only if the second and third conditions of the theorem are satisfied, by Lem.~\ref{lem:symmetrizable_equivalence}. The detailed proof is in Sec.~\ref{subsec:proof_thm1}, along with a discussion on the extension of the theorem for phase lags in the interval $(-\pi/2, \pi/2]$. 

When the exponents $\mu_j$ are not integers, the monomial eigenfunctions $z^{\bm\mu}=e^{i\bm\mu^\top\bm\theta}$ are naturally understood as smooth functions on $\mathbb R^N$, the universal cover of the torus $\mathbb T^N$, where the phase variables $\theta_j$ are not identified modulo $2\pi$. Ultimately, the corresponding real observable related to the eigenfunction is the linear function $\sum_{j=1}\mu_j\theta_j$ with $\theta_j\in\mathbb{R}$. However, the monomial observables themselves do not define globally continuous functions on $\mathbb T^N$, since $e^{i\mu_j(\theta_j+2\pi)} = e^{i2\pi\mu_j}e^{i\mu_j\theta_j}$, which differs from $e^{i\mu_j\theta_j}$ when $\mu_j\notin\mathbb Z$. Nevertheless, they remain smooth after restriction to any simply connected coordinate patch of $\mathbb T^N$, where continuous phase coordinates can be chosen locally.

\subsubsection*{Proof outline for Thm.~2}
Computing $\mathcal K[\ln V_{IJ}(\bm z)]$ and applying the first two conditions give $\mathcal K[\ln V_{IJ}(\bm z)] = F_{\mathrm{diagonal}} + F_{\mathrm{bipartite}}$ with
\begin{align}
    F_{\mathrm{diagonal}} &= \sum_{\substack{p,q\in I\\ p\neq q}}
\sigma_{pq}z_q\frac{A_{pp} -\bar A_{pp}}{z_p-z_q}-\sum_{\substack{r,s\in J\\ r\neq s}}\tau_{rs}z_s\frac{A_{ss}-\bar A_{ss}}{z_s-z_r}\\
    F_{\mathrm{bipartite}} &= \sum_{\substack{p,q\in I\\ p\neq q}}\sum_{k\in J}\sigma_{pq}\,z_q\,\frac{A_{pk}z_k\bar z_p-\bar A_{pk}\bar z_k z_p}{z_p-z_q}-\sum_{\substack{r,s\in J\\ r\neq s}}\sum_{k\in I}\tau_{rs}z_s\,\frac{A_{sk}z_k\bar z_s-\bar A_{sk}\bar z_k z_s}{z_s-z_r}\,.
\end{align}
Recalling that $A_{jj} = i\omega_j/2$, the fourth condition implies that $F_{\mathrm{diagonal}} = -i(\Sigma_I\Omega_I + \Sigma_J\Omega_J)$, i.e., the eigenvalue. It remains to show that $F_{\mathrm{bipartite}} = 0$. Applying the third condition yields
\begin{align}\label{eq:lastrelation_mm}
    F_{\mathrm{bipartite}}=\,\mathcal{A}\,\Big(\sum_{r\in J}T_rz_r\sum_{\substack{p,q\in I\\ p\neq q}}\sigma_{pq}\,\,\frac{ \bar z_p z_q}{z_p-z_q}
    + \sum_{p\in I}S_p\bar z_p\sum_{\substack{r,s\in J\\ r\neq s}}\tau_{rs}\,\frac{z_r^2}{z_r-z_s}\Big)\,.
\end{align}
Yet, note that 
\begin{align}\label{eq:vandermonde_relations_mm}
    \sum_{\substack{p,q\in I\\ p\neq q}}\frac{\sigma_{pq}\bar z_p z_q}{z_p-z_q} = -\sum_{p\in I}S_p\bar z_p\,,\qquad
    \sum_{\substack{r,s\in J\\ r\neq s}}\,\frac{\tau_{rs}z_r^2}{z_r-z_s} = \sum_{r\in J} T_rz_r\,.
\end{align}
Substituting these relations in Eq.~\eqref{eq:lastrelation_mm} makes $F_{\mathrm{bipartite}}$ vanish, which demonstrates the result. The detailed proof is given in Sec.~\ref{SIsubsec:proof_thm2}.

\subsubsection*{Proof outline for Thm.~3}
The Koopman generator can be expressed as $\mathcal{K} = \bm p^\top \bm L_{-1} - \bar{\bm p}^\top \bm L_1$ with $\bm p = A\bm z = (p_1,\; \ldots,\;p_N)^\top$, $\bm L_n =\left(\ell_1^{n},\; \ldots,\; \ell_N^{n}\right)^\top$ and $\ell_{j}^n=z_j^{n+1}\partial_j$ for all $n\in\mathbb{Z}$ and $j\in\{1,\ldots,N\}$. Simple calculations lead to
\begin{align*}
    \bm p^\top \bm L_{-1} [\ln c_{abcd}(\bm z)]&= \frac{p_{c} - p_{a}}{z_c - z_a} + \frac{p_{d} - p_{b}}{z_d - z_b} - \frac{p_{c} - p_{b}}{z_c - z_b} - \frac{p_{d} - p_{a}}{z_d - z_a}\\
    \bar{\bm p}^\top \bm L_1 [\ln c_{abcd}(\bm z)] &= \frac{\bar{p}_{c}z_c^2 - \bar{p}_{a}z_a^2}{z_c - z_a} + \frac{\bar{p}_{d}z_d^2 - \bar{p}_{b}z_b^2}{z_d - z_b} - \frac{\bar{p}_{c}z_c^2 - \bar{p}_{b}z_b^2}{z_c - z_b} - \frac{\bar{p}_{d}z_d^2 - \bar{p}_{a}z_a^2}{z_d - z_a}\,.
\end{align*}
Further algebraic manipulations then give expressions for $\mathcal{K}[\ln c_{abcd}]$ involving only factorized differences between pairs of elements of $A$. To prove sufficiency, this enables the direct use of the first and second conditions. Only a sum of independent monomials remains with coefficients involving the diagonal terms of $A$ (i.e., the natural frequencies) and $\mathcal{A}_a$, $\mathcal{A}_b$, $\mathcal{A}_c$, $\mathcal{A}_d$. Using the third condition cancels each of these coefficients. To prove necessity, it is convenient to separate the sum over $k\in\mathcal{V}$ as $k\in\{a,b,c,d\}$ and $k\notin\{a,b,c,d\}$. Then, one expands the expressions and collects all the terms with the same monomial. This leads to a linear system of equations that can be solved to deduce the conditions and their necessity. The detailed proof is provided in Sec.~\ref{SIsubsec:proof_thm3}.

\subsubsection*{Proof outline for Thm.~4}
Condition~\ref{itm:4.1A} implies that the four vertices are mutually disconnected ($A_{jk} = 0$ for all $j,k\in\{a,b,c,d\}$ with $j\neq k$) as they can only have an incoming edge from the source. Therefore, condition~\ref{itm:3.1} is satisfied. Condition~\ref{itm:4.1A} also states that the weights of these incoming edges are all equal to $\mathcal{A}_s$, meaning that condition~\ref{itm:3.2} holds. Then, condition~\ref{itm:3.3} is also fulfilled from condition~\ref{itm:4.2A} and the fact that $A_{jk} = 0$ for all $j,k\in\{a,b,c,d\}$ with $j\neq k$. Altogether, Thm.~\ref{thm:crossratios} guarantees the conservation of $c_{abcd}$. Lemma~\ref{lem:sym_KL} then shows that $\mathcal{S}_\eta = \mathcal{K}_\eta - i\omega_sL_0^\eta$ is a symmetry generator and condition~\ref{itm:4.3A} ensures that $\mathcal{S}_\eta[c_{abcd}]$ is not zero. Therefore, $\mathcal{K}[\mathcal{S}_{\eta}[c_{abcd}]] = \mathcal{S}_{\eta}[\mathcal{K}[c_{abcd}]] = 0$, i.e., $\mathcal{S}_\eta[c_{abcd}]$ is another constant of motion. Since $\mathcal{K}_\eta$ depends on the source's state $z_s$, $\mathcal{S}_\eta[c_{abcd}]$ also does. Therefore, $\mathcal{S}_\eta[c_{abcd}]$ is functionally independent of $c_{abcd}$, which only depends on $z_a, z_b, z_c, z_d$. The proof of part B is similar, with an additional step to show the functional independence. A detailed proof is given in Sec.~\ref{SIsubsec:proof_thm4}.

\subsubsection*{Settings for Fig.~3 E, F, G}
\label{sec:param_fig3}

The set $\mathcal{W}_6$ in Fig.~3A admits $\#\mathcal{W}_6 - 3$ cross-ratios because its oscillators satisfy conditions~\ref{itm:3A}-\ref{itm:3B}. Assuming $\omega_1 = 0$ and $\theta_1(0) = 0$ without loss of generality, the dynamics becomes $\dot{z}_j = \rho(\bm z) + i\Omega z_j - \overline{\rho(\bm z)}z_j^2$ with $j\in\mathcal{W}_6$, $\rho(\bm z) = \mathcal{A}_{\mathrm{pert}} + \sum_{\ell \in \mathcal{W}_6}\mathcal{A}_\ell z_\ell$, $\mathcal{A}_{\mathrm{pert}}\in\mathbb{C}$, $\mathcal{A}_\ell = r_\ell e^{-i\beta_\ell}\in\mathbb{C}$, and $\Omega = \omega_{x} - 2\Imag(\mathcal{A}_x)$ for any $x\in \mathcal{W}_6$. The parameters $\omega_x$, $r_\ell$, $\beta_\ell$ for all $\ell\in\mathcal{W}_6$, are drawn from arbitrarily chosen Gaussian distributions. Each $r_\ell$ is also set to zero with probability $1-p_\ell$. In this setup, if $r_\ell < 0$, the $\ell$-th oscillator is a contrarian and if $r_\ell > 0$, the $\ell$-th oscillator is a conformist. If $r_\ell = 0$, the oscillator is neutral. Using the WS transformation in Eq.~\eqref{eq:mobius_symmetry} for $j\in\mathcal{W}_6$
\begin{align*}
    z_j(t) = \frac{e^{i\phi(t)}w_j + Z(t)}{1 + e^{i\phi(t)}\overline{Z(t)}w_j}\,, \quad j\in\mathcal{W}_6
\end{align*}
leads to the differential equations
\begin{align}
    \od[]{Z}{t} &= F(Z, \phi) + i\Omega Z - \overline{F(Z, \phi)}Z^2\,\,,\label{eq:ws1_W6}\\
    \od[]{\phi}{t} &= \Omega + i(\overline{F(Z, \phi)}Z - F(Z, \phi)\bar{Z})\,,\label{eq:ws2_W6}\\
    F(Z, \phi) &= \mathcal{A}_{\mathrm{pert}} + \sum_{\ell\in\mathcal{W}_6} \mathcal{A}_{\ell}\frac{e^{i\phi}w_\ell + Z}{1 + e^{i\phi}\overline{Z}w_\ell}\,,
\end{align}
which are integrated using DOPRI45 to generate Fig.~3 E, F, G. Assuming there is no majority cluster~\cite[p.217]{Watanabe1994}, the initial conditions $Z(0), \phi(0)$, and $(w_j)_{j\in \mathcal{W}_6} =(e^{i\psi_j})_{j\in \mathcal{W}_6}$ are obtained from a real optimization problem designed to satisfy the conditions $\sum_{k\in\mathcal{W}_6} \cos\psi_k = 0$, $\sum_{k\in\mathcal{W}_6} \sin\psi_k = 0$, and the arbitrary condition $\psi_x = 0$ for some $x\in\mathcal{W}_6$.


\subsubsection*{Detection of motifs in Fig.~4}
\label{sec:detection}

Sources are vertices with in-degree zero and 2-sources are pairs of mutually connected vertices with in-degree one, so these monomials are easily found numerically. To detect the motifs admitting conserved cross-ratios, we consider the binary version of the weight matrices representing the empirical networks. From Thm~\ref{thm:crossratios}, detecting a motif admitting one or more conserved cross-ratios requires finding a set of at least four rows in the weight matrix $W$ with identical elements except for those associated with the diagonal (imagine a shuffled version of $W$ in Fig.~\ref{fig:complex_weight_matrix}). This exception related to the diagonal elements makes the detection more subtle, as one cannot directly use a standard similarity matrix to identify identical rows.

To detect the motifs admitting conserved cross-ratios, we set the weights in $W$ that are greater than a tolerance (10$^{-8}$) to one, we fix the diagonal elements to 0, and we keep only the weakly connected component, yielding a binary adjacency matrix $B$. Then, we construct a special similarity matrix $\Sigma$ from the binary matrix $B$, for which the off-diagonal zero entries are replaced by the imaginary unit $i$. Let us denote this new complex matrix $Q$ and the similarity matrix as the Hermitian matrix $\Sigma =(QQ^\dagger - I)/(N-2)$, where $\dagger$ is the Hermitian conjugation. In this way, an element $\Sigma_{jk} = 1$ indicates that the row $\bm{b}_j$ and the row $\bm{b}_k$ in $B$ are identical up to the diagonal elements, as desired. Indeed, if $\bm{q}_j$ and $\bm{q}_k$ are the associated (complex) rows in $Q$, the division of $QQ^\dagger$ by $N-2$ normalizes the scalar product $\bm{q}_j^\dagger \bm{q}_k$. The subtraction of $I/(N-2)$ also ensures that the diagonal elements of $\Sigma$ are ones. Once $\Sigma$ is obtained, we define a new graph $\mathcal{G}_\Gamma$ with adjacency matrix $\Gamma$ such that $\Gamma_{jk} = 1$ if $\Sigma_{jk} = 1$ and $\Gamma_{jk} = 0$ otherwise. Finding motifs admitting conserved cross-ratios then amounts to finding the maximal cliques of size greater than or equal to 4 in $\mathcal{G}_\Gamma$, which is achieved using the Bron-Kerbosh algorithm implemented in graph-tool.

\clearpage




\begin{figure}[t]
    \centering
    \includegraphics[width=\linewidth]{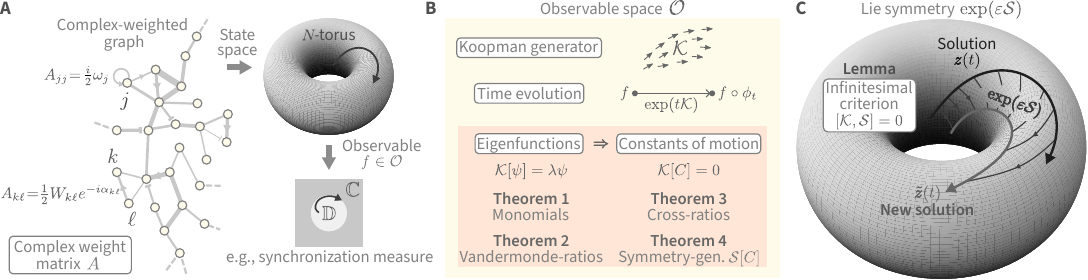}

    \vspace{-0.3cm}
    \caption{\small \textbf{The Kuramoto model, its constants of motion, and its Lie symmetries under Koopman's perspective.} (\textbf{A}) The complex weight matrix $A$ in Eq.~\eqref{eq:complex_matrix} encodes all the parameters of any network of Kuramoto oscillators, whose evolution is described by $\bm{z} = (z_1,...,z_N)$ on the $N$-torus. Observables $f$ for the model are complex-valued functions on the $N$-torus. For instance, the observable $(1/N)\sum_{j=1}^Nz_j$ lies in the closed unit disk $\mathbb{D}$ and its modulus, the Kuramoto order parameter~\cite{Schroder2017}, measures synchronization. (\textbf{B}) The observables belong to a function space $\mathcal{O}$. The Koopman generator $\mathcal{K}$ is the total derivative $\mathrm{d}/\mathrm{d}t$ that generates the time evolution of the observables through the Koopman operator $\exp(t\mathcal{K})$, which composes the observables with the flow $\phi_t$ of the dynamics. An eigenfunction $\psi$ of $\mathcal{K}$ (e.g., monomials in Thm.~\ref{thm:monomials}) gives key information about the dynamics (e.g., isostables~\cite{Mauroy2013} as level sets of $|\psi|$). Notably, an eigenfunction with eigenvalue 0 is a constant of motion $C$ (e.g., cross-ratios in Thm.~\ref{thm:crossratios}). 
    The existence of an eigenfunction $\psi$ with eigenvalue $\lambda$ directly provides a constant of motion $C = \psi e^{-\lambda t}$. (\textbf{C})~A Lie symmetry transforms a solution of the Kuramoto model to a new solution. A transformation is a symmetry provided that an infinitesimal criterion is satisfied: a symmetry generator $\mathcal{S}$ commutes with the Koopman generator $\mathcal{K}$, under specific conditions. Indeed, the general form of the criterion is slightly more subtle and is provided in Eq.~\eqref{eq:symkoo}.}
    \label{fig:fig1}
\end{figure}

\begin{figure}[t]
    \centering
    \includegraphics[width=\linewidth]{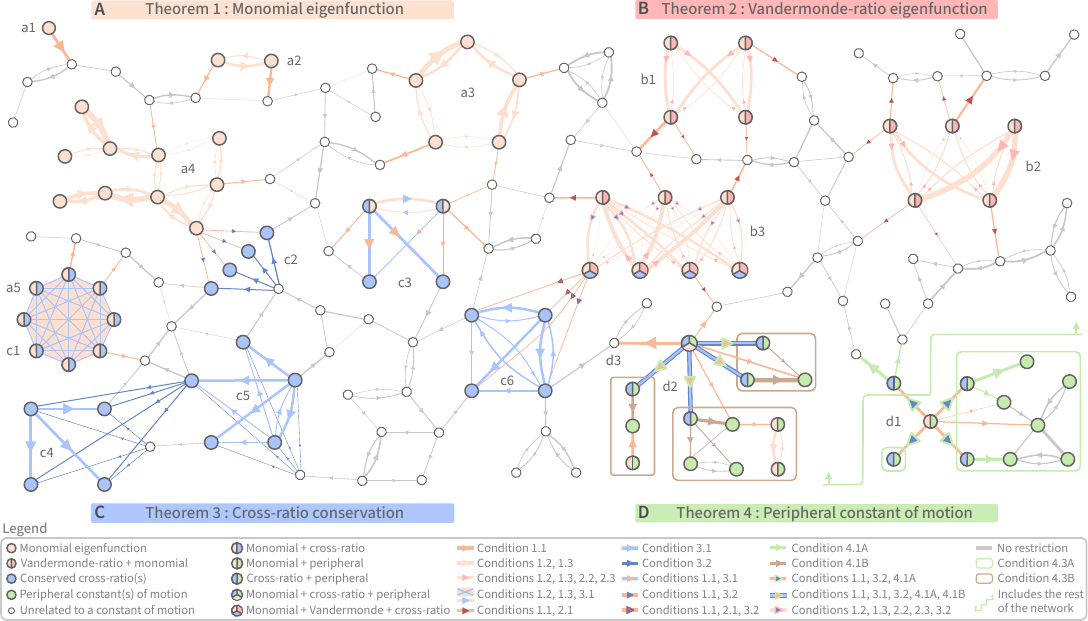} 
    \caption{\small \textbf{Network of Kuramoto oscillators with motifs supporting various constants of motion.} (\textbf{A}) Examples for Thm.~\ref{thm:monomials}. 
    (a1) A single source vertex satisfies Thm.~\ref{thm:monomials}. 
    (a2) Two oscillators satisfying \ref{itm:1.1} and \ref{itm:1.2} always satisfy \ref{itm:1.3}. 
    (a3) A 5-cycle illustrating \ref{itm:1.3}. 
    (a4) A more complex motif yielding a monomial eigenfunction. 
    (a5) A source complete graph of oscillators with antisymmetric phase-lag matrix induces a monomial eigenfunction. 
    (\textbf{B}) Examples for Thm.~\ref{thm:vandermonde}. 
    (b1) Smallest weighted bipartite graph admitting a Vandermonde-ratio eigenfunction. 
    (b2) Bipartite graph with layers of different sizes inducing a Vandermonde-ratio eigenfunction. 
    (b3)~Bipartite motif with a layer of 4 vertices that also admits a conserved cross-ratio. 
    (\textbf{C}) Examples for Thm.~\ref{thm:crossratios}. 
    (c1) Eight globally coupled identical oscillators lead to $5$ conserved cross-ratios.
    (c2) An empty subgraph (\ref{itm:3.1}) of identical-frequency oscillators (\ref{itm:3.3}) yields a conserved cross-ratio.
    (c3) Motif inducing a conserved cross-ratio with two vertices admitting a monomial eigenfunction. 
    (c4) The smallest directed star inducing a conserved cross-ratio. 
    (c5) Five-vertex motif admitting $2$ functionally independent cross-ratios.
    (c6)~Non-complete, non-empty and non-star graph yielding a conserved cross-ratio.
    (\textbf{D}) Examples for Thm.~\ref{thm:thm4}. 
    (d1) Four (blue and green) vertices with identical natural frequencies (\ref{itm:4.2A}) only receiving from a source with identical weights (\ref{itm:4.1A}) and distributed into 3 disjoint parts (\ref{itm:4.3A}) admit 3 distinct symmetry generators and peripheral constants of motion. 
    (d2) Four (blue and green) vertices with frequency $\omega_s - 2\,\mathrm{Im}(\mathcal{A}_s)$ (\ref{itm:4.2B}) receiving one edge from source $s$ with weight $\mathcal{A}_s$ (\ref{itm:4.1B}) are distributed in 3 disjoint parts (\ref{itm:4.3B}). This yields 2 conserved cross-ratios and 3 symmetries acting on them to form peripheral constants of motion. 
    (d3) The subgraph (i.e., the rest of the network) admits a symmetry but it yields no constant of motion. 
    }
    \label{fig:fig2}
\end{figure}


\begin{figure}[t]
    \centering
 \includegraphics[width=\linewidth]{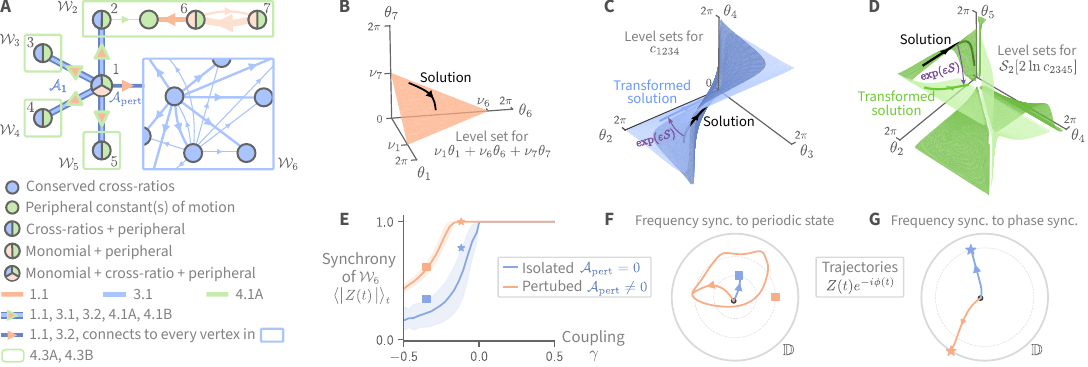}
  
    \caption{\small \textbf{Illustration of constants of motion, Lie symmetries, and the synchronization of 101 conformist-contrarian oscillators.} (\textbf{A}) Complex-weighted graph for the example. Edges between vertices in $\mathcal{W}_6$ exist with Bernoulli probability $p_\ell$ and have random weights $\mathcal{A}_\ell \in \mathbb{C}$ (Gaussian modulus and phase). The operators $\mathcal{S}_2$, ..., $\mathcal{S}_6$ related to $\mathcal{W}_2$, ..., $\mathcal{W}_6$ form an abelian Lie subalgebra of symmetry generators~\eqref{eq:peripheral_symmetry}. 
    (\textbf{B})~The trajectories of $\theta_1, \theta_6, \theta_7$ are constrained by a plane representing the level set of the conserved quantity $\nu_1\theta_1 + \nu_6\theta_6 + \nu_7\theta_7$ (i.e., the argument of the conserved monomial $z_1^{\nu_1}z_6^{\nu_6}z_7^{\nu_7}$). (\textbf{C}, \textbf{D}) The symmetry $\exp(\varepsilon \mathcal{S})$ with $\mathcal{S} = \mathcal{S}_3 + \mathcal{S}_4 + \mathcal{S}_5$ and $\varepsilon = 5$ sends a solution (black curves) constrained by the invariant sets of $c_{1234} \approx 5.0$ and $\mathcal{S}_2[2\ln c_{2345}] \approx 4.2$ [Eq.~\eqref{eq:symgen_example}] to a new solution (blue and green curves) evolving in other invariant sets ($c_{1234} \approx 1.3$, $\mathcal{S}_2[2\ln c_{2345}] \approx -0.6$). (\textbf{E}) Synchronization transitions measured by the time-averaged modulus $\langle|Z(t)|\rangle_t$ vs. the coupling $\gamma$ when the oscillators from $\mathcal{W}_6$ are isolated and when they are perturbed ($\mathcal{A}_\mathrm{pert} = 20e^{-i\alpha_s}$, $\alpha_s \in \mathcal{N}(1, 10^{-4})$). Squares and stars indicate corresponding trajectories in (\textbf{F}) and (\textbf{G}); $(Z(t),\phi(t))$ is obtained by integrating the WS equations of $\mathcal{W}_6$ [Materials and Methods]. 
    (\textbf{F}) Trajectory of 
    $Z(t)e^{-i\phi(t)}$ when $\mathcal{A}_\mathrm{pert} = 0$ reaches an equilibrium point $|Z(t)| < 1$ (blue curve, frequency synchronization). Under the same conditions, except $\mathcal{A}_\mathrm{pert} = 20e^{-0.892i}$, the order parameter instead reaches a limit cycle (orange curve). (\textbf{G}) Trajectory of $Z(t)e^{-i\phi(t)}$ for $\mathcal{A}_\mathrm{pert} = 0$ characterizing frequency synchronization. Under the same conditions, except $\mathcal{A}_\mathrm{pert} = 20e^{-0.823i}$, phase synchronization is achieved (orange curve reaching $|Z(t)| = 1$).
}
    \label{fig:fig3}
\end{figure}


\begin{figure}
    \centering
    \includegraphics[width=\linewidth]{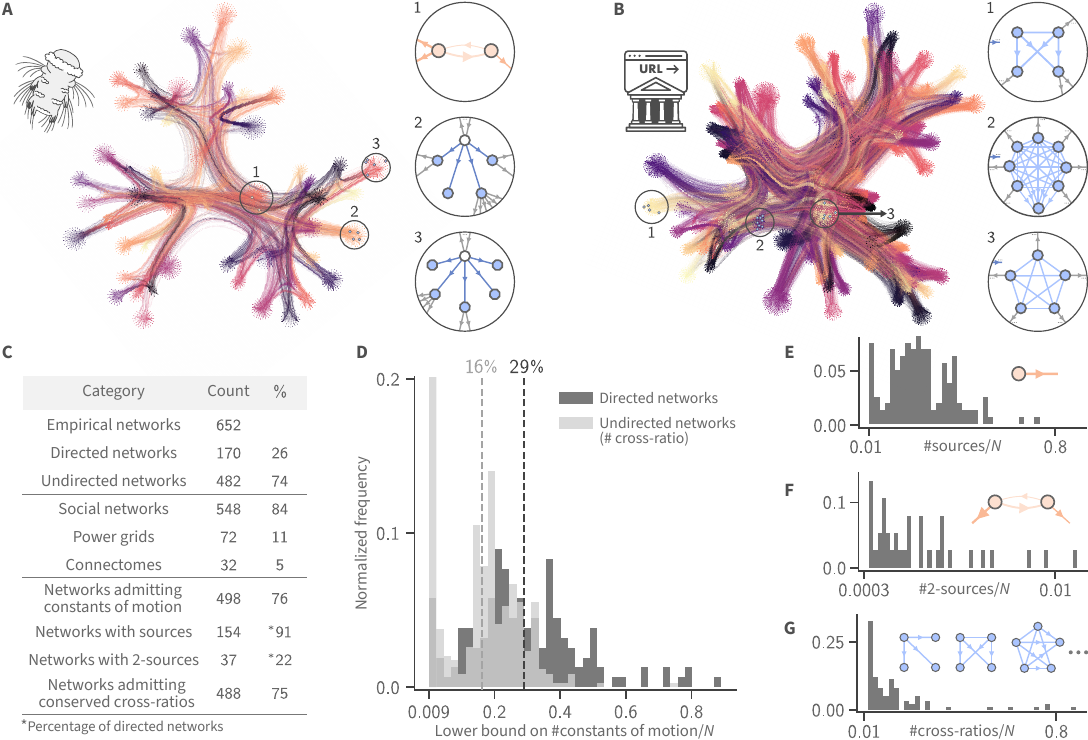}

    \vspace{-0.2cm}
    
    \caption{\small\textbf{Presence of network motifs admitting conserved quantities in empirical networks.} (\textbf{A}, \textbf{B}) Examples of motifs admitting constants of motion found in the connectome of the annelid \textit{P. dumerilii}~\cite{Veraszt2025} and the web-based network of US government agencies (California)~\cite{Kosack2018}. The blue arrows pointing toward the motifs in (\textbf{B}) mean that each vertex in the motifs has many incoming edges (satisfying \ref{itm:3.2}) from vertices outside the motif. The grey arrows with the ellipses indicate that the vertices have many outgoing edges. (\textbf{C}) Composition of the dataset, with networks drawn mainly from \href{https://networks.skewed.de/}{Netzschleuder} and \href{https://github.com/MATPOWER/matpower/tree/master/data}{Matpower}, and number of admitted constants of motion. (\textbf{D}) Histogram of the lower bound on the number of (time-dependent) constants of motion in the 498 networks admitting at least one conserved quantity. We count the number of sources, 2-sources (each associated with a conserved quantity $z^{\bm\mu} e^{-\lambda t}$), and for each motif of size $m$ admitting conserved cross-ratios, we count the $m - 3$ functionally independent ones. We exclude the conserved cross-ratios associated with the presence of at least four single sources (similarly for (\textbf{G})). For undirected networks, there are no monomial eigenfunctions by Thm.~\ref{thm:monomials} (the undirected empirical networks are of course not empty, not complete, considered connected, and they have no sources). Moreover, the only two motifs that admit conserved cross-ratios by Thm.~\ref{thm:crossratios} are mutually disconnected vertices in a motif of at least four vertices (forming the periphery of star graphs with one or multiple cores in connected networks) or completely connected motifs of at least four vertices (a maximal clique). (\textbf{E}, \textbf{F}, \textbf{G}) Histograms of the number of sources, 2-sources, and motifs admitting cross-ratios.}
    \label{fig:fig4}
\end{figure}


	

\begin{table} 
\centering
\caption{\small\textbf{Eigenfunctions and constants of motion of the Koopman generator in Eq.~\eqref{eq:kooku} for the Kuramoto model~\eqref{eq:kuramoto}.} The table lists the phase-shifted families of eigenfunctions with the shifts given by $\bm\alpha = (\alpha_1,...,\alpha_N) \in \mathbb{R}^N$ [Sec.~\ref{subsec:parametrized_families}]. Products of these eigenfunctions with themselves and with other Koopman eigenfunctions yield different forms of eigenfunctions. In particular, monomial and Vandermonde-ratio eigenfunctions allows to generate different constants of motion. The vector $\bm\mu$ is in $\mathbb{R}^N$ and contains the exponents of the monomial eigenfunction. For the Vandermonde-ratio eigenfunction, the vertex subsets $I,J\subset \mathcal{V}$ are disjoint, $\sigma_{pq},\tau_{rs}$ for all $p,q\in I$, $r,s \in J$ are real constants, $\Omega_{I},\Omega_{J}\in\mathbb{R}$ are natural frequencies for the oscillators in $I$ and $J$ respectively, $\Sigma_I, \Sigma_J\in\mathbb{R}$ are rescaled in-strengths in $J$ and $I$ respectively, and $\beta_{jk} = \alpha_j - \alpha_k \in \mathbb{R}$. The constants $\beta:=\alpha_a-\alpha_b, \gamma:=\alpha_a-\alpha_c, \delta:=\alpha_a-\alpha_d$ are also real numbers that parametrize the cross-ratio family. For the family of Watanabe-Strogatz (WS) integrals, the constants are $\beta_u:=\alpha_{j_{u+1}}-\alpha_{j_u}$ for all $u$ with $j_{\ell+1}=j_1$ and they satisfy $\beta_1+...+\beta_\ell=0$. The operator $\mathcal{S}_\eta$ is the continuous symmetry generator in Eq.~\eqref{eq:peripheral_symmetry} and $c_{abcd}$ is the parametrized cross-ratio presented in the table. Note that we have defined $S_{jk}^\phi = \sin\left((\theta_j - \theta_k + \phi)/{2}\right)$ to simplify the expressions.}
\label{tab:eigenfunctions} 
	\scriptsize
    \renewcommand{\arraystretch}{2}
\begin{tabular}{lcccc} 
	\\
	\hline
Eigenfunction family & Form in $z_1,...,z_N$ & Form in $\theta_1,...,\theta_N$ & Eigenvalue & Section \\
		\hline
Monomial $z^{\bm\mu}$ 
&
\(\displaystyle\prod_{j=1}^N z_j^{\mu_j}\)
&
\(\displaystyle\exp\left(i{\textstyle\sum_{j=1}^N} \mu_j\theta_j\right)\)
&
\(\displaystyle i\bm{\mu}^\top\bm{\omega}\)
&
\ref{subsec:monomials}
\\[0.25cm]

\mbox{Vandermonde-ratio $V_{IJ}$} 
&
\(\displaystyle\frac{\prod_{p<q\in I}(\bar{z}_p-e^{-i\beta_{pq}}\bar{z}_q)^{\sigma_{pq}}}{\prod_{r<s\in J}(z_r-e^{i\beta_{rs}}z_s)^{\tau_{rs}}}\)
&
\(\displaystyle\frac{\prod_{p<q\in I}\Big(e^{-i(\theta_p + \theta_q)/2}\,S_{pq}^{\beta_{pq}}\Big)^{\sigma_{pq}}}{\prod_{r<s\in J}\Big(e^{i(\theta_r + \theta_s)/2}\,S_{rs}^{\beta_{rs}}\Big)^{\tau_{rs}}}\)
&
\(\displaystyle -i\left(\Sigma_I\Omega_I + \Sigma_J\Omega_J\right)\)
&
\shortstack[c]{
\ref{subsec:vandermonde}, 
\ref{subsec:parametrized_families}
}
\\[0.28cm]

\mbox{Cross-ratio $c_{abcd}$}
&
\(\displaystyle\frac{(e^{i\gamma}z_c-z_a)(e^{i\delta}z_d-e^{i\beta}z_b)}{(e^{i\gamma}z_c-e^{i\beta}z_b)(e^{i\delta}z_d-z_a)}\)
&
\(\displaystyle\frac{S_{ca}^\gamma S_{db}^{\delta-\beta}}{S_{cb}^{\gamma - \beta}S_{da}^{\delta}}\)
&
\(\displaystyle 0\)
&
\shortstack[c]{
\ref{subsec:crossratios},
\ref{subsec:parametrized_families}
}
\\[0.28cm]

\shortstack[l]{$^*$Watanabe-Strogatz $C^{\mathrm{ws}}_{\bm j}$}
&
\(\displaystyle\prod_{u=1}^{\ell}\left(1-e^{i\beta_u}\frac{z_{j_u}}{z_{j_{u+1}}}\right)\)
&
\(\displaystyle\prod_{u=1}^{\ell}S_{j_u j_{u+1}}^{\beta_u}\)
&
\(\displaystyle \phantom{^*}0^*\)
&
\shortstack[c]{
\ref{subsec:crossratios}, 
\ref{subsec:parametrized_families}
}

\\[0cm]

\shortstack[l]{Peripheral}
&
\(\displaystyle\mathcal S_{\eta}\!\left[c_{abcd}(\bm{z})\right]\)
&
\(\displaystyle\mathcal S_{\eta}\!\left[c_{abcd}(\bm{\theta})\right]\)
&
\(\displaystyle 0\)
&
\ref{subsec:peripheralevolution}, \ref{subsec:parametrized_families}
\\
\hline
\end{tabular}

\vspace{0.1cm}

{\scriptsize $^*$ Under the conditions ensuring cross-ratio conservation, the $C^{\mathrm{ws}}_{\bm j}$ are instead Darboux functions sharing a common cofactor~[Sec.~\ref{subsec:crossratios}, \ref{SIsubsec:WS_darboux}].} 
\end{table}


\clearpage 

%
\bibliography{paper_kuramoto_koopman} 
\bibliographystyle{sciencemag}

%
%
%
%
%
%


\paragraph*{Acknowledgments}
We thank I.~Mezić, Z.~Nicolaou, G.~St-Onge, I.~Joseph, and R. Lambiotte for their constructive and encouraging feedback on this work. 

\paragraph*{Funding:}
This work was supported by the Fonds de recherche du Qu\'ebec -- Nature et technologies (V.~T., B.~C.), the Natural Sciences and Engineering Research Council of Canada (B.~C., A.~A., P.~D.), and the Sentinelle Nord program of Universit\'e Laval, funded by the Canada First Research Excellence Fund (A.~A.). 

\paragraph*{Author contributions:}
\phantom{.}\\
Conceptualization: VT, PD\\ Writing—original draft: VT\\ 
Supplementary material: VT, with support from BC and PD\\
Software: VT, with assistance
from BC and AA\\
Visualization and numerical experiments: VT.\\
Formal analysis: VT, BC, PD.\\ Interpretation of the results: VT, BC, AA, PD\\
Writing—review and editing: VT, BC, AA, PD \\Supervision: PD, AA




\paragraph*{Competing interests:}
All authors declare that they have no competing interests.

\paragraph*{Data and materials availability:}
%
%
%
%
The symbolic calculations (\textit{Mathematica}, \textit{Matlab}), the \textit{Python} code to generate Fig.~\ref{fig:fig3} and \ref{fig:fig4}, and the network dataset are available on Zenodo (\url{https://doi.org/10.5281/zenodo.20599776}). No physical materials were generated in this work.


\subsection*{Supplementary materials}
Supplementary Text : Detailed proofs, context, and examples for the four theorems\\
Section S1. Introduction to the Kuramoto model under Koopman's perspective\\
Section S2. Monomial eigenfunctions and their conservation\\
Section S3. Vandermonde-ratio eigenfunctions\\
Section S4. Conservation of cross-ratios\\
Section S5. Continuous symmetries and the generation of new constants of motion\\
Fig. S1. Cross-ratio as a projective invariant\\
Fig. S2. Simplified history of weight matrices in WS theory\\
Fig. S3. 4-vertices motifs admitting a conserved cross-ratio\\
Fig. S4. 5-vertices motifs admitting two conserved cross-ratios\\
Fig. S5. Illustration of a graph with symmetry generators $\mathcal{S}_1,...,\mathcal{S}_r$\\
Fig. S6. Basic example with coexisting conserved quantities of different types\\
References \textit{(148-\arabic{enumiv})}


\newpage

\let\addcontentsline\oldaddcontentsline

\newpage


\renewcommand{\thefigure}{S\arabic{figure}}
\renewcommand{\thetable}{S\arabic{table}}
\renewcommand{\theequation}{S\arabic{equation}}
\renewcommand{\thepage}{S\arabic{page}}
\setcounter{figure}{0}
\setcounter{table}{0}
\setcounter{equation}{0}
\setcounter{page}{1} 
\setcounter{equation}{0}
\setcounter{figure}{0}
\setcounter{table}{0}
\setcounter{section}{0}
\setcounter{page}{1}
\setcounter{theorem}{0}

\renewcommand{\theequation}{S\arabic{equation}}
\renewcommand{\thefigure}{S\arabic{figure}}
\renewcommand{\thetable}{S\Roman{table}}
\renewcommand{\thetheorem}{S\arabic{theorem}}
\renewcommand{\thelemma}{S\arabic{lemma}}
\renewcommand{\thesection}{S\arabic{section}}

\renewcommand{\theHsection}{S\arabic{section}}
\renewcommand{\theHfigure}{Supplement.\thefigure}  
\thispagestyle{empty}


\begin{center}
\section*{Supplementary Materials for\\ \scititle}

\let\addcontentsline\oldaddcontentsline


Vincent~Thibeault$^{\ast}$,
Benjamin~Claveau,
Antoine~Allard,
Patrick~Desrosiers\\ 
\small$^\ast$Corresponding author. Email: vincent.thibeault.1@ulaval.ca\\
\end{center}

\subsubsection*{This PDF file includes:}
Section S1. Introduction to the Kuramoto model under Koopman's perspective\\
Section S2. Monomial eigenfunctions and their conservation\\
Section S3. Vandermonde-ratio eigenfunctions\\
Section S4. Conservation of cross-ratios\\
Section S5. Continuous symmetries and the generation of new constants of motion\\
Fig. S1. Cross-ratio as a projective invariant\\
Fig. S2. Simplified history of weight matrices in WS theory\\
Fig. S3. 4-vertices motifs admitting a conserved cross-ratio\\
Fig. S4. 5-vertices motifs admitting two conserved cross-ratios\\
Fig. S5. Illustration of a graph with symmetry generators $\mathcal{S}_1,...,\mathcal{S}_r$\\
Fig. S6. Basic example with coexisting conserved quantities of different types\\
References \textit{(148-\arabic{enumiv})}

\newpage

\tableofcontents






\section{Introduction to the Kuramoto model under Koopman's perspective}
\label{sec:intro_kooku}

Koopman theory was initially developed to formulate classical dynamics using linear operators on spaces of observables—mirroring the structure of quantum mechanics~\cite{Koopman1931,Koopman1932, VonNeumann1932, Carleman1932}.
Despite this elegant framework, the theory long lacked concrete theoretical examples for complex systems, where intricate interactions among many constituents give rise to emergent phenomena such as synchronization. Since, moreover, these dynamical systems are transformed into infinite-dimensional systems as a price to pay for having a linear operator, this has cast doubt among many researchers on the usefulness of Koopman theory. For these reasons, one of our goals is to provide concrete analytical results on a widely-used dynamics on networks. For its 50th anniversary and its substantial impact on the study of complex systems, the Kuramoto model is a natural choice to achieve this goal. In this section, we introduce different descriptions of the Kuramoto model and derive the Koopman generator $\mathcal{K}$ of the Kuramoto model presented in the paper.


\subsection{Different descriptions for the Kuramoto model}
\label{subsec:descriptions}

First of all, we provide more precise descriptions of the Kuramoto model. Note that we define a more general form of the model than the original one~\cite{Kuramoto1975} and the Kuramoto-Sakaguchi model (or Sakaguchi-Kuramoto model if preferred)~\cite{Sakaguchi1986}, where $\alpha_{jk} = \alpha$ for all $j,k$. To highlight the original contribution of Kuramoto, we will stick to the name ``Kuramoto model''.
\begin{definition}\label{def:kuramoto}
    The Kuramoto model is an initial value problem of dimension $N$ such that, for $\theta_j:\mathscr{T} \to E$ with $\mathscr{T}, E \subset \mathbb{R}$,
\begin{align}
    \dot{\theta}_j(t) &= \omega_j + \sum_{k=1}^N W_{jk} \sin(\theta_k(t) - \theta_j(t) - \alpha_{jk})\,,\qquad j\in\{1,...,N\} \label{eq:thetakur}\\
    \theta_j(0) &= \vartheta_j \in E\label{eq:theta0kur} 
\end{align}
where $t\in \mathscr{T}$, $\omega_j, W_{jk} \in \mathbb{R}$ and $-\pi/2 < \alpha_{jk} \leq \pi/2$ for all $j,k$ with $W_{jj} = 0$, $\alpha_{jj} = 0$ for all $j$.
\end{definition}
\begin{remark}\label{rem:frequency_shift}
    In the definition, we set $W_{jj} = 0$, $\alpha_{jj} = 0$ for all $j$ without loss of generality. Indeed, consider that $W_{jj} \neq 0$, $\alpha_{jj} \neq 0$ for all $j$. Then, Eq.~\eqref{eq:thetakur} can be expressed as 
    \begin{align*}
        \dot{\theta}_j(t) &= \omega_j - W_{jj}\sin(\alpha_{jj}) + \sum_{k\neq j} W_{jk} \sin(\theta_k(t) - \theta_j(t) - \alpha_{jk})
    \end{align*}
    which means that the self-interaction term $W_{jj}\sin(\alpha_{jj})$ only acts as a shift to the natural frequency and one can always redefine $\omega_j$ as $\omega_j - W_{jj}\sin(\alpha_{jj})$ without loss of generality. Note also that we could absorb the coupling constant in $W$, but it is useful to control the global weight of the interactions, which can be proportional to $1/N$.
\end{remark}
\begin{lemma}\label{lem:unicite_kuramoto}
    There exists a constant $a > 0$ such that the Kuramoto model possesses a unique solution $\theta_1(t),...,\theta_N(t)$ on $t\in[-a, a]$.
\end{lemma}
\begin{proof}
    By the fundamental existence-uniqueness Thm.~\cite[p.74]{Perko2001}, it is sufficient to show that the partial derivatives of
    \begin{align*}
        f_j(\vartheta_1, ...,\vartheta_N) = \omega_j + \sum_{k=1}^N W_{jk}\sin(\vartheta_k - \vartheta_j - \alpha_{jk})\,,\qquad j\in\{1,...,N\}\,,
    \end{align*}
     forming the vector field $\bm f = (f_1,...,f_N)$ of the model exist and are continuous. For all $j,\ell$, recalling that $W_{jj}=0$ and $\alpha_{jj} = 0$, the partial derivatives are
     \begin{align*}
         \pd[]{f_j(\vartheta_1, ...,\vartheta_N)}{\vartheta_\ell} = \sum_{k\neq j}W_{jk}\pd[]{\sin(\vartheta_k - \vartheta_j - \alpha_{jk})}{\vartheta_\ell} = \begin{cases} 
    - \sum_{k\neq j}W_{jk}\cos(\vartheta_k - \vartheta_j - \alpha_{jk})\,, & \text{if }  j= \ell\\
     W_{j\ell}\cos(\vartheta_\ell - \vartheta_j - \alpha_{j\ell}), & \text{if } j\neq \ell
\end{cases} \, ,
     \end{align*}
     which are evidently continuous functions on $\mathbb{R}^N$.
\end{proof}
It is often convenient to rather work with the model described in the complex plane.
\begin{lemma}\label{lem:correspondance_complexe}
    Let $z_j(t) = e^{i\theta_j(t)}$ with $\theta_j$ and $t$ as in Def.~\ref{def:kuramoto}. The initial value problem in $z_1,...,z_N$ related to Def.~\ref{def:kuramoto} is 
\begin{align}
    \dot{z}_j(t) &= p_j(\bm z(t)) + i\omega_j z_j(t) - \overline{p_j(\bm z(t))}z_j(t)^2\,,\qquad p_j(\bm z(t))  = \frac{1}{2}\textstyle{\sum_{k=1}^N} W_{jk}e^{-i\alpha_{jk}}z_k(t) \label{eq:zkur}\\
    z_j(0) &= e^{i\vartheta_j}  \in \mathbb{T}\label{eq:z0kur}
\end{align}
for all $j\in\{1,...,N\}$.
\end{lemma}
\begin{proof}
    The derivative of $z_j$ is $\dot{z}_j = iz_j\dot{\theta}_j$. Substituting Eq.~\eqref{eq:thetakur} and expressing the sine function with complex exponentials readily yields the result. 
\end{proof}
\begin{remark}
    Note that $z_j(t) = e^{i\theta_j(t)}$ is not bijective when $\theta_j(t)\in\mathbb{R}$ (e.g., $\theta_j(t) = 0$ or $2\pi$ both yield $z_j(t) = 1$). Yet, restricting the initial condition $\vartheta_1, ...,\vartheta_N$ such that $\vartheta_j \in [0, 2\pi)$ and assuming $\theta_j(t)$ is continuous in time for all $j$ is sufficient to guarantee the correspondence of the trajectories for the dynamics in $\bm{\theta}$ and $\bm{z}$.
\end{remark}
Another useful formulation of the model, where all the parameters are regrouped in only one matrix, is the following.
\begin{lemma}\label{lem:correspondance_complexe_A}
    The initial value problem~(\ref{eq:zkur}-\ref{eq:z0kur}) is equivalent to 
\begin{align}
    \dot{z}_j(t) &= \sum_{k=1}^NA_{jk}z_k(t) - \left(\sum_{k=1}^N\bar{A}_{jk}\bar{z}_k(t) \right)z_j(t)^2 \label{eq:zkurA}\\
    z_j(0) &= e^{i\vartheta_j}  \in \mathbb{T}\label{eq:z0kurA}\,,
\end{align}
where $A$ is a complex matrix of interactions satisfying
\begin{align}
    A = \frac{1}{2}\left( W \circ e^{-i\alpha} + i\,\mathrm{diag}(\bm{\omega})\right)\,,
\end{align}
where $e^{-i\alpha} = (e^{-i\alpha_{jk}})_{j,k}$, $\bm{\omega} = (\omega_1,...,\omega_N)$, $\circ$ is the element-wise product and $\mathrm{diag}(W) = \mathrm{diag}(\alpha) = \bm{0}$. There exists a constant $a > 0$ such that the problem~(\ref{eq:zkurA}-\ref{eq:z0kurA}) possesses a unique solution $z_1(t),..., z_N(t)$ on $t\in[-a, a]$.
\end{lemma}
\begin{proof}
    From Lem.~\ref{lem:correspondance_complexe}, the model is equivalently described by
    \begin{align}\label{eq:kurz}
        \dot{z}_j = \frac{1}{2}\sum_{k=1}^N W_{jk}e^{-i\alpha_{jk}}z_k + i\omega_j z_j - \overline{\left(\frac{1}{2}\sum_{k=1}^N W_{jk}e^{-i\alpha_{jk}}z_k\right)}z_j^2\,.
    \end{align}
    The term related to the natural frequencies can be separated such that
    \begin{align*}
        i\omega_j z_j = \frac{i}{2}\omega_jz_j + \frac{i}{2}\omega_jz_j = \frac{i}{2}\omega_jz_j - \overline{\left(\frac{i}{2}\omega_j\bar{z}_j\right)}\,.
    \end{align*}
    Since $\bar{z}_j = \bar{z}_j^2z_j$, then
    \begin{align*}
        i\omega_j z_j = \frac{i}{2}\omega_jz_j - \overline{\left(\frac{i}{2}\omega_jz_j\right)}z_j^2 = \frac{i}{2}\sum_{k=1}^N \omega_k z_k \delta_{jk} - \overline{\left(\frac{i}{2}\sum_{k=1}^N \omega_k z_k \delta_{jk}\right)}z_j^2\,.
    \end{align*}
    The substitution of the latter equation into Eq.~\eqref{eq:kurz} gives 
    \begin{align*}
         \dot{z}_j = \sum_{k=1}^N \left(\frac{1}{2}W_{jk}e^{-i\alpha_{jk}} + \frac{i}{2} \omega_k \delta_{jk} \right)z_k - \overline{\left(\sum_{k=1}^N \left(\frac{1}{2}W_{jk}e^{-i\alpha_{jk}} + \frac{i}{2} \omega_k \delta_{jk} \right)z_k\right)}z_j^2\,,
    \end{align*}
    which is the desired result by defining $A_{jk} = \frac{1}{2}W_{jk}e^{-i\alpha_{jk}} + \frac{i}{2} \omega_k \delta_{jk}$ for all $j,k$. The proof of uniqueness of the solutions is similar to Lem.~\ref{lem:unicite_kuramoto}: the elements of the Jacobian matrix of $\bm{F} = (F_1,...,F_N)$ with $F_j(\bm{w}) = \sum_{k=1}^NA_{jk}w_k - \left(\sum_{k=1}^N\bar{A}_{jk}\bar{w}_k \right)w_j^2$ and $\bar{w}_k = 1/w_k$ are
    \begin{align*}
        \pd[]{F_j(\bm w)}{w_\ell} = A_{j\ell} + \bar{A}_{j\ell}\bar{w}_\ell^2w_j^2 - 2\sum_{k=1}^N \bar{A}_{jk}\bar{w}_kw_j\delta_{j\ell}
    \end{align*}
    and thus, the partial derivatives exist and are continuous on $\mathbb{T}^N$.
\end{proof}
\begin{remark}
    The first term $\frac{1}{2}W\circ e^{-i\alpha}$ encodes the interaction between the oscillators ($\mathrm{diag}(W) = \mathrm{diag}(\alpha) = \bm{0}$), while the natural frequencies are the self-interaction terms (self-loops with imaginary weights).
\end{remark}

\subsection{Koopman generator for the Kuramoto model}
\label{SIsubsec:koopman_generator}

We refer to standard articles such as Refs.~\cite{Budisic2012, Brunton2022} for a more complete introduction to Koopman theory. One can also see subsection~\ref{SIsubsec:inf_crit} for the definition of the Koopman operator and its generator for general non-autonomous systems.

For the Kuramoto model in terms of the phases in Def.~\ref{def:kuramoto}, the Koopman generator is 
\begin{align}\label{eq:kooku_phase}
    \mathcal{K} = \sum_{j=1}^N \left(\omega_j + \sum_{k=1}^NW_{jk}\sin(\theta_k - \theta_j - \alpha_{jk})\right)\pd[]{}{\theta_j}
\end{align}
and acts on functions of the phases $(\theta_1,...,\theta_N)\in\mathbb{R}^N$ (recall the abuse of notation $\theta_j(t) = \theta_j \in \mathbb{R}$), giving another real function. The Koopman generator is a Lie derivative $\mathcal{L}_{\bm F}$ along the dynamics' vector field $\bm F$ acting on scalar functions. The related Koopman operator encodes all the information on the flow when the state space is compact~\cite[Theorem 4.13]{Eisner2015}. 

Under the change of variables $z_j = e^{i\theta_j}$ for all $j$, the partial derivatives for the phases become $\partial/\partial \theta_j = iz_j\partial/\partial z_j$ and one readily gets the Koopman generator of the Kuramoto model under the form given in Lem.~\ref{lem:correspondance_complexe_A}:
\begin{align}\label{eq:kookuSI}
    \mathcal{K} = \sum_{j,k=1}^N\left(A_{jk}z_k - \bar{A}_{jk}\bar{z}_kz_j^2\right)\pd[]{}{z_j}\,,
\end{align}
In matrix form, one can write
\begin{equation}\label{eq:generator_nonidentical}
    \mathcal{K} = \bm z^\top A^\top \bm L_{-1} - \bar{\bm z}^\top \bar{A}^\top \bm L_1\,,
\end{equation}
where 
\begin{equation}
\bm L_n =\left( z_1^{n+1}\frac{\partial}{\partial z_1},\; \ldots,\; z_N^{n+1}\frac{\partial}{\partial z_N}\right)^\top
\,,\qquad n\in\mathbb{Z}.
\end{equation}
are the vectorial Euler differential operators. The elements of the vectorial operators $\bm L_n$ form a larger algebra. Indeed, defining
\begin{equation}\label{eq:elements_euler}
    \ell_{j}^n=z_j^{n+1}\frac{\partial}{\partial z_j}
,\qquad n\in\mathbb{Z},\quad j\in\{1,\ldots,N\},
\end{equation}
the Koopman generator becomes
\begin{align}\label{eq:euler_calK}
    \mathcal{K} = \sum_{j,k = 1}^N A_{jk}z_k\ell_j^{-1} - \sum_{j,k = 1}^N \bar{A}_{jk}z_k^{-1}\ell_j^{1}\,.
\end{align}
We adopt $\bm{L}_n$ and $\ell_{j}^n$ as notation by analogy with the elements of the Witt (more generally, Virasoro) algebra, to which the special linear algebra naturally forms a subalgebra.

All the forms of the Koopman generator can be useful depending on the context. We will typically use Eq.~\eqref{eq:kookuSI} as in the paper, because of its simple form.

In the next section, we show how one can find monomial eigenfunctions for the generator $\mathcal{K}$ and how they lead to constants of motion.

\clearpage
\section{Monomial eigenfunctions and their conservation}

Among the most fundamental objects in Koopman theory are the eigenfunctions. The goal of this section is to provide the proof of Thm.~1 on monomial eigenfunctions from the main text and the details regarding their conservation, along with some basic examples. 

\subsection{Proof of Theorem 1: Monomial eigenfunctions}
\label{subsec:proof_thm1}
 Before providing the proof of the first theorem of the paper, we need to introduce an important lemma that defines symmetrizable matrices in a more general way than Ref.~\cite{McKee2021}. 
\begin{lemma}\label{lem:symmetrizable_equivalence}
    Consider a real matrix $B$ of size $b\times b$.
    The following statements are equivalent:
    \begin{enumerate}
        \item $B$ is symmetrizable;
        \item $DB$ is symmetric, where $D$ is a real diagonal matrix with nonzero diagonal elements;
        \item $\mu_jB_{jk} = \mu_kB_{kj}$ for some $\mu_j,\mu_k\in\mathbb{R}\setminus\{0\}$ for all $1 \leq j < k \leq b$;
        \item \begin{enumerate}
            \item $B_{jk} \neq 0$ whenever $B_{kj}\neq 0$ for all $1 \leq j < k \leq b$;
            \item $B_{i_1 i_2}B_{i_2 i_3} ... B_{i_{\eta-1}i_{\eta}}B_{{i_{\eta} i_1}} = B_{{i_1 i_{\eta}}}B_{i_{\eta}i_{\eta-1}} ... B_{i_3 i_2}B_{i_2 i_1}$ for all sequences $i_1, i_2, ..., i_\eta$ of elements of $\{1, ..., b\}$.
        \end{enumerate}
    \end{enumerate}
\end{lemma}
\begin{proof}
$(1 \Leftrightarrow 2)$ By definition.   
\par
$(2 \Leftrightarrow 3)$ Since all the elements of $\bm \mu=(\mu_j)_{j=1}^b$ and those of the diagonal of $D$ are nonzero, suppose that $D = \mathrm{diag}(\bm \mu)$.
Element-wise, $DB = (DB)^\top$ is then equivalent to $\mu_jB_{jk} = \mu_kB_{kj}$ for all $1 \leq j < k \leq b$.
\par
$(3 \Rightarrow 4)$ First, all the elements of $\bm \mu$ are nonzero and thus, $\mu_jB_{jk} = \mu_k B_{kj}$ implies that $B_{jk} = \mu_k B_{kj}/\mu_j$ for all $j,k\in \{1, ..., b\}$. Consequently, $B_{kj} \neq 0$ implies that $B_{jk} \neq 0$.
Second, multiplying together $\mu_k B_{kj} = \mu_j B_{jk}$ for any sequence $i_1, i_2, ..., i_\eta$ of elements of $\{1, ..., b\}$ gives
\begin{equation}
    \mu_{i_1}\mu_{i_2}...\mu_{i_{\eta-1}}\mu_{i_{\eta}} B_{i_1 i_2}B_{i_2 i_3} ... B_{i_{\eta-1}i_{\eta}}B_{{i_{\eta} i_1}} = \mu_{i_1}\mu_{i_{\eta}}...\mu_{i_{3}}\mu_{i_{2}} B_{{i_1 i_{\eta}}}B_{i_{\eta}i_{\eta-1}} ... B_{i_3 i_2}B_{i_2 i_1}\,.
\end{equation}
But $\mu_j \neq 0$ for all $j$  and therefore, 
\begin{equation}
    B_{i_1 i_2}B_{i_2 i_3} ... B_{i_{\eta-1}i_{\eta}}B_{{i_{\eta} i_1}} = B_{{i_1 i_{\eta}}}B_{i_{\eta}i_{\eta-1}} ... B_{i_3 i_2}B_{i_2 i_1}\,.
\end{equation}
\par
$(3\Leftarrow 4)$  
    Let us interpret $B$ as the weight matrix of a graph.
    If the graph is not connected, simply repeat the following procedure for each connected component.
    Otherwise, $B$ can be interpreted as the weight matrix of a strongly connected graph (because of condition~4(a)).
    If $b=1$, then the implication is trivially satisfied. Otherwise, for $b>1$, choose $\mu_\ell\in \mathbb{R}\setminus\{0\}$ for some $\ell\in\{1, ..., b\}$. Since the graph is strongly connected, there exists a $j\in\{1,...,b\}$ with $j\neq \ell$ such that $B_{\ell j}/B_{j\ell}$ is a well-defined, nonzero ratio from condition 4(a). Their product allows defining a new nonzero real number $\mu_j$ such that
    \begin{equation}\label{eq:mu_construction}
        \mu_j := \frac{B_{\ell j}}{B_{j\ell}} \mu_\ell\,.
    \end{equation}
    Again, because the graph is strongly connected, one can repeat the procedure iteratively to set $\mu_k$ for all $k\in\{1,...,b\}\setminus\{j,\ell\}$.
    This implies that all the elements of $\mu_1,...,\mu_b$ satisfy $\mu_pB_{pq} = \mu_qB_{qp}$ for at least one given pair $(p, q)$, because $\mu_p$ was built from $\mu_q$ using Eq.~\eqref{eq:mu_construction}.
    \par
    Let us now deal with null matrix elements and ensure that $\mu_jB_{jk} = \mu_kB_{kj}$ \textit{for all} $1 \leq j < k \leq b$ with $j\neq k$.
    Condition 4(a) ensures that $B_{kj}\neq 0 \Rightarrow B_{jk}\neq 0$. The contrapositive of this statement is that $B_{jk} = 0 \Rightarrow B_{kj}=0$. But condition 4(a) applies for all $1 \leq j < k \leq b$, so $B_{kj} = 0$ also implies that $B_{jk} = 0$. When $B_{jk} = B_{kj} = 0$, condition 3 is trivially satisfied for any finite values of $\mu_j$ and $\mu_k$.
    \par
    For nonzero $B_{jk}$ and $B_{kj}$, consider the cycle condition 4(b) with a sequence $i_1, i_2, ..., i_{\eta-1}, i_\eta$ of elements in $\{1, ..., b\}$ where $i_1=j$, $i_{\eta} = k$ and $B_{i_m i_{m+1}} \neq 0$ for all $m\in\{1, ..., \eta-1\}$.
    Since the graph is connected, this type of sequence exists and can be chosen in such a way that, from the building procedure of $\mu_1,...,\mu_b$, 
    \begin{equation}
        \mu_{i_m}B_{i_mi_{m+1}} = \mu_{i_{m+1}}B_{i_{m+1}i_m}\,,
    \end{equation}
    or equivalently, since all elements involved are nonzero,
    \begin{equation}\label{eq:mu_and_B_ratios}
        \mu_{i_{m+1}}/\mu_{i_m} = B_{i_mi_{m+1}} / B_{i_{m+1}i_m}\,.
    \end{equation}
    According to the cycle condition 4(b),
    \begin{equation}
        B_{j\, i_2}B_{i_2 i_3} ... B_{i_{\eta-1}k}B_{kj} = B_{jk}B_{k\,i_{\eta-1}} ... B_{i_3 i_2}B_{i_2 j}\,.
    \end{equation}
    Since $B_{kj}\neq 0$ and $B_{i_m i_{m+1}} \neq 0$ for all $m$, it can be rewritten as
    \begin{equation}
        \frac{B_{jk}}{B_{kj}} = \frac{B_{j\, i_2}B_{i_2 i_3} ... B_{i_{\eta-1}k}}{B_{k\,i_{\eta-1}}B_{i_{\eta-1}i_\eta} ... B_{i_2 j}}\,,
    \end{equation}
    and Eq.~\eqref{eq:mu_and_B_ratios} then implies
    \begin{equation}
        \frac{B_{jk}}{B_{kj}} = \frac{\mu_{i_{2}} ... \mu_{i_{\eta-1}}\mu_{k}}{\mu_{j}\mu_{i_{2}} ... \mu_{i_{\eta-1}}} = \frac{\mu_k}{\mu_j}\,.
    \end{equation}
    Therefore, $\mu_j B_{jk} = \mu_k B_{kj}$ for all $1 \leq j < k \leq b$ as desired.
\end{proof}
The relevance of the latter lemma lies in the fact that it enables stating Thm.~1 solely in terms of the weight matrix and the phase lags.
\begin{theorem}\label{thm:existence_fpmonom}[Thm.~\ref{thm:monomials} of the paper]
    Let $\mathcal{W} \subset \mathcal{V}$ be a non-empty subset of vertices such that $|\alpha_{jk}| < \pi/2$ for all $j,k\in\mathcal{W}$. Let $\bm \mu = (\mu_1\,\,\,\cdots\,\,\,\mu_N)^\top \in \mathbb{R}^N$ satisfy $\mu_j \neq 0$ if and only if $j \in \mathcal{W}$. There exists a $\bm{\mu}$ such that the monomial $z^{\bm \mu} := z_1^{\mu_1}...z_N^{\mu_N}$ is an eigenfunction of the Koopman generator $\mathcal{K}$ in Eq.~\eqref{eq:kookuSI} if and only if 
    \begin{enumerate}
        \item $W_{jk} = 0$ for all $j\in \mathcal{W}$ and $k\in \mathcal{V}\setminus\mathcal{W}$;
        \item $W_{jk}\neq 0$ whenever $W_{kj}\neq 0$ for all $j,k\in \mathcal{W}$ ;
        \item $W_{i_1 i_2} ... W_{i_{\eta-1}i_{\eta}}W_{{i_{\eta} i_1}} = W_{{i_1 i_{\eta}}}W_{i_{\eta}i_{\eta-1}} ... W_{i_2 i_1}$ for all sequences $i_1, i_2, ..., i_\eta$ of elements of~$\mathcal{W}$;
        \item $\alpha_{jk} = -\alpha_{kj}$ whenever  $j,k\in \mathcal{W}$,  $j\neq k$, $W_{jk}\neq 0$.
    \end{enumerate}
    If $z^{\bm \mu}$ is an eigenfunction, then its eigenvalue is $i\bm{\mu}^\top\bm{\omega}$.
\end{theorem}
\begin{proof}
The action of the Koopman generator on a monomial $z^{\bm \mu}$ is
\begin{equation}
    \mathcal{K}\,[z^{\bm \mu}] = \sum_{j\in\mathcal{W}}\sum_{k\in\mathcal{V}} \left(A_{jk}\mu_j z^{\bm \mu - \bm e_j + \bm e_k} - \bar{A}_{jk}\mu_jz^{\bm \mu - \bm e_k + \bm e_j}\right)\,,
\end{equation}
where $(\bm e_j)_\ell = \delta_{j\ell}$. Splitting the sum over $k\in\mathcal{V}$ to $\mathcal{W}$ and $\mathcal{V}\setminus\mathcal{W}$ yields
\begin{equation}
    \mathcal{K}\,[z^{\bm \mu}] = \sum_{j\in\mathcal{W}}\sum_{k\in\mathcal{W}} \left(A_{jk}\mu_j - \bar{A}_{kj}\mu_k\right)z^{\bm \mu - \bm e_j + \bm e_k} + \sum_{j\in\mathcal{W}}\sum_{k\in\mathcal{V}\setminus\mathcal{W}} \left(A_{jk}\mu_j z^{\bm \mu - \bm e_j + \bm e_k} - \bar{A}_{jk}\mu_jz^{\bm \mu - \bm e_k + \bm e_j}\right)\,.
\end{equation}
The diagonal terms can also be separated from the off-diagonal ones to obtain
\begin{equation}\label{eq:KOG_action_monomial}
    \mathcal{K}\,[z^{\bm \mu}] = \sum_{j\in\mathcal{W}}\sum_{\substack{k\in\mathcal{W}\\ k\neq j}} \left(A_{jk}\mu_j - \bar{A}_{kj}\mu_k\right)z^{\bm \mu - \bm e_j + \bm e_k} + \sum_{j\in\mathcal{W}}\sum_{k\in\mathcal{V}\setminus\mathcal{W}} \left(A_{jk}\mu_j z^{\bm \mu - \bm e_j + \bm e_k} - \bar{A}_{jk}\mu_jz^{\bm \mu - \bm e_k + \bm e_j}\right) + iz^{\bm \mu}\sum_{j\in\mathcal{W}}\omega_j \mu_j\,.
\end{equation}
The monomial $z^{\bm \mu}$ is an eigenfunction of the Koopman generator if and only if it satisfies the eigenvalue equation
\begin{equation}
    \mathcal{K}[z^{\bm \mu}] = \lambda z^{\bm \mu}\,,
\end{equation}
which is equivalent, by Eq.~\eqref{eq:KOG_action_monomial}, to
\begin{equation}
\sum_{j\in\mathcal{W}}\sum_{\substack{k\in\mathcal{W}\\ k\neq j}} \left(A_{jk}\mu_j - \bar{A}_{kj}\mu_k\right)z^{\bm \mu - \bm e_j + \bm e_k} + \sum_{j\in\mathcal{W}}\sum_{k\in\mathcal{V}\setminus\mathcal{W}} \left(A_{jk}\mu_j z^{\bm \mu - \bm e_j + \bm e_k} - \bar{A}_{jk}\mu_jz^{\bm \mu + \bm e_j - \bm e_k }\right) + iz^{\bm \mu}\sum_{j\in\mathcal{W}}\mu_j\omega_j  = \lambda z^{\bm \mu}\,.
\end{equation}
All monomials on the left-hand side are linearly independent. Clearly, if $z^{\bm{\mu}}$ is an eigenfunction, its eigenvalue is $i\sum_{j\in\mathcal{W}}\mu_j\omega_j $. Also, the necessary and sufficient conditions on $\bm \mu$ and $A$ for the eigenvalue equation to be satisfied with eigenvalue $i\sum_{j\in\mathcal{W}}\mu_j\omega_j $ are $A_{jk} = 0$ for all $j\in\mathcal{W}$ and $k\in\mathcal{V}\setminus\mathcal{W}$, and
\begin{equation}\label{eq:nec_suf_cond_eigenfunc}
    A_{jk}\mu_j = \bar{A}_{kj}\mu_k
\end{equation}
for all $j,k\in\mathcal{W}$ with $j\neq k$.
In terms of the weight matrix $W$ and the phase-lag matrix $\alpha$, these conditions are equivalent to $W_{jk} = 0$ for all $j\in\mathcal{W}$ and $k\in\mathcal{V}\setminus\mathcal{W}$ (condition 1), and
\begin{equation}\label{eq:nec_suf_cond_Walpha}
    \mu_jW_{jk}e^{i\alpha_{jk}} = \mu_kW_{kj}e^{-i\alpha_{kj}}
\end{equation}
for all $j,k\in\mathcal{W}$ with $j\neq k$. 
The two complex numbers in Eq.~\eqref{eq:nec_suf_cond_Walpha} are equal if and only if their modulus coincide and, when their modulus is nonzero, their principal argument also coincide. In other words, Eq.~\eqref{eq:nec_suf_cond_Walpha} is satisfied if and only if either one of the following conditions is satisfied:
\begin{enumerate}
\item 
$|\mu_jW_{jk}| = |\mu_k W_{kj}| = 0$;
\item $|\mu_jW_{jk}| = |\mu_k W_{kj}| \neq 0$ and $\alpha_{jk} + \mathrm{Arg}(\mu_jW_{jk}) = -\alpha_{kj} + \mathrm{Arg}(\mu_kW_{kj})$.
\end{enumerate}
In the last condition, $\mathrm{Arg}(z)\in(-\pi,\pi]$ denotes the principal argument of $z$. The first condition is equivalent to  $W_{jk} = W_{kj} = 0$ because $\mu_j$ is non-zero for each $j\in\mathcal{W}$.
For the second condition, recall that $|\alpha_{jk}| < \pi/2$, while $\mu_jW_{jk} \in \mathbb{R}$ is equivalent to the fact that the arguments $\mathrm{Arg}(\mu_jW_{jk})$ and $\mathrm{Arg}(\mu_kW_{kj})$ are 0 or $\pi$.
Therefore, $\alpha_{jk} + \mathrm{Arg}(\mu_jW_{jk}) = -\alpha_{kj} + \mathrm{Arg}(\mu_kW_{kj})$ if and only if $\alpha_{jk} = -\alpha_{kj}$ and $\mathrm{Arg}(\mu_jW_{jk}) = \mathrm{Arg}(\mu_kW_{kj})$.
The second condition is thus equivalent to $\alpha_{jk} = -\alpha_{kj}$,  $\mathrm{Arg}(\mu_jW_{jk}) = \mathrm{Arg}(\mu_kW_{kj})$, and $|\mu_jW_{jk}| = |\mu_k W_{kj}|\neq0$.  
This first equation provides the fourth condition 4 of the theorem: $\alpha_{jk} = -\alpha_{kj}$ for all $j,k\in\mathcal{W}$ such that $W_{jk} \neq 0$. Together, $\mathrm{Arg}(\mu_jW_{jk}) = \mathrm{Arg}(\mu_kW_{kj})$ and $|\mu_jW_{jk}| = |\mu_k W_{kj}|\neq0$ are equivalent to $\mu_jW_{jk} = \mu_k W_{kj}$. Then, by Lem.~\ref{lem:symmetrizable_equivalence}, there exists at least one $\bm\mu$ such that $\mu_jW_{jk} = \mu_k W_{kj}$ for all $j,k\in\mathcal{W}$ with $j\neq k$ if and only if $W_{jk}\neq 0$ whenever $W_{kj}\neq 0$ for all $j,k\in \mathcal{W}$ (condition 2) and $W_{i_1 i_2}W_{i_2 i_3} ... W_{i_{\eta-1}i_{\eta}}W_{{i_{\eta} i_1}} = W_{{i_1 i_{\eta}}}W_{i_{\eta}i_{\eta-1}} ... W_{i_3 i_2}W_{i_2 i_1}$ for all sequences $i_1, i_2, ..., i_\eta$ of elements of $\mathcal{W}$ (condition 3).
Altogether, conditions 1-4 are necessary and sufficient for the monomial $z^{\bm \mu}$ to be an eigenfunction  of $\mathcal{K}$.
\end{proof}


\begin{remark}\label{rem:phaselags_monomials}
    The theorem can be extended to phase lags in the interval $-\frac{\pi}{2} < \alpha_{jk} \leq \frac{\pi}{2}$ for all $j,k\in\mathcal{W}$, but $\pi/2$ phase lags imply different specific conditions that we wanted to avoid for the sake of simplicity. Indeed, if $\alpha_{jk} = \pi/2$ for some $j,k\in\mathcal{W}$ with $W_{jk} \neq 0$, Eq.~\eqref{eq:nec_suf_cond_Walpha} becomes $i\mu_jW_{jk} = \mu_k W_{kj}e^{-i\alpha_{kj}}$. Since $-\pi/2$ phase lags are excluded from the interval and $\mu_jW_{jk} \in \mathbb{R}$ for all $j,k\in\mathcal{W}$, the matching of the modulus and complex phase results in $\alpha_{jk} = \alpha_{kj} = \pi/2$ and $\mu_jW_{jk} = -\mu_k W_{kj}$. If $\mu_jW_{jk} = -\mu_k W_{kj}$ for all $j,k \in \mathcal{W}$, $(W_{jk})_{j,k\in \mathcal{W}}$ is skew-symmetrizable. If there is a mix of phase lags equal to $\pi/2$ and in the interval $(-\pi/2, \pi/2)$, the symmetrizability conditions (2 and 3) apply not directly to $W$, but to another matrix equal to $W$ up to the sign inversion of some elements.
\end{remark}

\subsection{Procedure for constructing the monomial eigenfunction exponent \texorpdfstring{$\bm{\mu}$}{Lg}}
\label{subsec:mu_procedure}
The procedure for constructing $\bm{\mu}$ from $W$ is identical to the one from the proof of Lemma~\ref{lem:symmetrizable_equivalence}. For the sake of clarity, we reproduce the procedure explicitly in this subsection with very few adjustments.

First, define $B$ as the submatrix related to the subgraph with vertex set $\mathcal{W}$ of size $b$ from Thm.~\ref{thm:existence_fpmonom}.
If the subgraph is not connected, one can repeat the following procedure for each connected component. In the connected case, condition~2 of Thm.~\ref{thm:existence_fpmonom} implies that $B$ is the weight matrix of a strongly connected graph. 
If $b = 1$, then $\bm\mu$ contains a single arbitrary nonzero element. Otherwise, for $b>1$, choose $\mu_\ell\in \mathbb{R}\setminus\{0\}$ for some $\ell\in\mathcal{W}$. Since the graph is strongly connected, there exists a $j\in\mathcal{W}$ with $j\neq \ell$ such that $B_{\ell j}/B_{j\ell}$ is a well-defined, nonzero ratio from condition 2 of Thm.~\ref{thm:existence_fpmonom}. Their product allows defining a new nonzero real number $\mu_j$ such that
\begin{equation}\label{eq:mu_construction_2}
    \mu_j := \frac{B_{\ell j}}{B_{j\ell}} \mu_\ell\,.
\end{equation}
Again, because the graph is strongly connected, one can repeat the procedure iteratively to set $\mu_k$ for all $k\in\mathcal{W}\setminus\{j,\ell\}$, which completes the construction of $\bm{\mu}$.


\subsection{Monomials as constants of motion}
Having a monomial eigenfunction $z^{\bm \mu}$ implies the existence of a constant of motion with the form $z^{\bm{\mu}}e^{-i\bm{\mu}^\top \bm{\omega}t}$. If $\bm{\mu}^\top\bm\omega = 0$, the time dependence disappears, which is convenient when making a change of variables.
This is, however, a rather specific case, because the powers of monomial eigenfunctions are determined by the weight matrix, thus restricting the natural frequencies satisfying the orthogonality condition.
In the more general case, it is possible to combine eigenfunctions with nonzero eigenvalues to obtain constants of motion with no time dependence, as presented in the following lemma. 

\begin{lemma}\label{lem:monom_constants}
    Let the Kuramoto model in Lem.~\ref{lem:correspondance_complexe_A} have natural frequencies $\bm{\omega} = (\omega_1\,\cdots\,\omega_N)$ and Koopman generator $\mathcal{K}$. Suppose that $\mathcal{K}$ admits $1 \leq q\leq N$ functionally independent monomial eigenfunctions $z^{\bm\mu_1},\ldots,z^{\bm\mu_q}$, where $\bm{\mu}_{\rho}\in\mathbb{R}^N$ for each $1\leq \rho\leq q$, whose corresponding eigenvalues are given by the vector $\bm \lambda = (i\bm\mu_1^\top \bm{\omega}\,\,\, \cdots\,\,\, i\bm\mu_q^\top \bm{\omega})^\top$. If all eigenvalues are nonzero, then there are $q-1$ functionally independent monomial constants of motion $z^{\bm{\nu}_1},\ldots, z^{\bm{\nu}_{q-1}}$, defined by the matrix equation
    $
        (\bm{\nu}_{1} \; \cdots \; \bm{\nu}_{q-1}) = (\bm{\mu}_{1} \; \cdots \; \bm{\mu}_{q}) O
    $
    where $O$ is a real ${q \times (q-1)}$ matrix having linearly independent columns orthogonal to $\bm{\lambda}$.
\end{lemma}
\begin{proof}
    Denote $\psi_{\rho}(\bm{z}) = z^{\bm\mu_\rho}$ for all $\rho \in \{1,...,q\}$. Since $\bm{z}\in\mathbb{T}^N$, these $q$ monomials are non-vanishing eigenfunctions of $\mathcal{K}$. Moreover, since the eigenfunctions are also assumed to be functionally independent, the vectors $\bm \mu_1,\ldots,\bm\mu_q$ are linearly independent and thus span a $q$-dimensional subspace in $\mathbb{R}^N$.
    
    Now, it is known that the product of non-vanishing eigenfunctions yields an eigenfunction with sums of eigenvalues~\cite[Proposition~5]{Budisic2012}. Explicitly, for any real numbers $a_1,...,a_q$,
    \begin{align*}
        \mathcal{K}\left[\prod_{\rho=1}^q \psi_{\rho}^{a_\rho}\right] = \left(\sum_{\eta=1}^q a_\eta \lambda_\eta\right) \,\prod_{\rho=1}^q \psi_{\rho}^{a_\rho}\,.
    \end{align*}
    Then, $\prod_{\rho=1}^q \psi_{\rho}^{a_\rho}$ is a constant of motion if and only if the new orthogonality condition $\bm{a}^\top \bm \lambda = 0$ is met. This condition is nontrivial since, by assumption, no component of $\bm \lambda$ is zero. Clearly, in this case, the imaginary part of $\bm \lambda$ lies in $\mathbb{R}^q$ and has a $(q-1)$-dimensional orthogonal complement.  
    This implies that we can find $q - 1$ linearly independent vectors $\bm{a}_{\tau} = (a_{\tau 1} \,\cdots\, a_{\tau q})^{\top} \in \mathbb{R}^q$, for $\tau \in \{1, \ldots, q - 1\}$, that are orthogonal to $\bm \lambda$.
    Since monomials of linearly independent powers are functionally independent, there are $q - 1$ constants of motion having the form
    \begin{align*}
        \Psi_{\tau}(\bm{z}) = \prod_{\rho=1}^q \psi_{\rho}^{a_{\tau\rho}}(\bm{z}) = \prod_{\rho=1}^q \left(z_1^{a_{\tau\rho}(\bm\mu_\rho)_1}\cdots\,z_N^{a_{\tau\rho}(\bm\mu_\rho)_N}\right) = \prod_{j=1}^N z_j^{\sum_{\rho=1}^q(\bm \mu_\rho)_ja_{\tau \rho}}=z^{\bm\nu_\tau},\qquad \tau\in\{1,..., q-1\},
    \end{align*}
    where $\bm\nu_\tau = U \bm{a}_{\tau}$ with $U = (\bm{\mu}_{1} \; \cdots \; \bm{\mu}_{q})$. Altogether, the vectors $\bm\nu_1,\ldots, \bm\nu_{q-1}$, which define the monomial constants of motion, satisfy the matrix equation  $(\bm{\nu}_{1} \;\cdots \; \bm{\nu}_{q-1}) = UO$, where $O = (\bm{a}_1\;\cdots\; \bm{a}_{q-1})$. The latter is a real $q \times (q-1)$ matrix with linearly independent columns satisfying $O^\top \bm \lambda = \bm 0$, and the lemma follows. 
\end{proof}

\begin{remark}
The above lemma remains valid if the condition “all eigenvalues are nonzero” is replaced by “at least one eigenvalue is nonzero.” However, in this case, some eigenfunctions $z^{\bm{\mu}_\rho}$ are themselves constants of motion, and thus some of the resulting constants of motion $z^{\bm{\nu}_\tau}$ may factor through them — that is, they include terms like $z^{\bm{\mu}_\rho}$ as multiplicative factors. As a result, the two sets of constants of motion become functionally dependent. The condition that all eigenvalues are nonzero therefore ensures the cleanest setting, where the only functionally independent constants of motion are precisely the monomials $z^{\bm{\nu}_\tau}$.
\end{remark}

\subsection{Basic examples for Theorem 1 and the conservation of monomials}
\label{sec:basic_example_monomial}

\begin{example}
    Consider a (sink) directed star of five Kuramoto oscillators such that 
    \begin{align*}
        \dot{z}_1 &= i\omega_1 z_1 + (A_{12}z_2 - \bar{A}_{12}\bar{z}_2z_1^2) + (A_{13}z_3 - \bar{A}_{13}\bar{z}_3z_1^2) + (A_{14}z_4 - \bar{A}_{14}\bar{z}_4z_1^2) + (A_{15}z_5 - \bar{A}_{15}\bar{z}_5z_1^2)\\
        \dot{z}_k &= i\omega_k z_k\,, \quad k\in\{2,3,4,5\}.
    \end{align*}
    Clearly, the last four equations readily inform that $z_2, z_3, z_4, z_5$ are $q = 4$ monomial eigenfunctions with respective eigenvalues $\bm \lambda = (i\omega_2\,\,\, i\omega_3\,\,\, i\omega_4\,\,\, i\omega_5)$ (let's assume they are not zero). 
    From Lem.~\ref{lem:monom_constants}, those eigenfunctions can be combined to obtain 3 functionally independent constants of motion.
    For example, $z_2^{\omega_3}z_3^{-\omega_2}$, $z_3^{\omega_4}z_4^{-\omega_3}$ and $z_4^{\omega_5}z_5^{-\omega_4}$ is a set of independent constants of motion, along with the time-dependent integral $z_5e^{-i \omega_5 t}$.
\end{example}

\begin{example}
    Consider the system of 10 Kuramoto oscillators associated with the complex matrix
    {\footnotesize\begin{equation}
        A = 
        \begin{bmatrix}
            i\omega_1/2 & \mathcal{B}_{12}e^{i\alpha_{12}} & \mathcal{B}_{13}e^{i\alpha_{13}} & 0 &  0 & 0 & 0 & 0 & 0 & 0\\
            \mathcal{B}_{21}e^{-i\alpha_{12}} & i\omega_2/2 & \mathcal{B}_{23}e^{i\alpha_{23}} & 0  & 0 & 0 & 0 & 0 & 0 & 0\\
            \mathcal{B}_{31}e^{-i\alpha_{13}} & \mathcal{B}_{32}e^{-i\alpha_{23}} & i\omega_3/2 & 0 & 0 & 0 & 0 & 0 & 0 &  0\\
            0 & 0 & 0 & i\omega_4/2 & 0  & 0 & 0 & 0 & 0 & 0\\
            0 & 0 & 0 & 0 & i\omega_5/2 & \mathcal{C}_{12}e^{i\alpha_{56}}  & 0 & 0 & 0 & 0\\
            0 & 0 & 0 & 0 & \mathcal{C}_{21}e^{-i\alpha_{56}} & i\omega_6/2 & 0 & 0 & 0 & 0\\
            0 & 0 & 0 & 0 & 0 & 0 & i\omega_7/2 & \mathcal{D}_{12}e^{i\alpha_{78}} &  0 & 0\\
            0 & 0 & 0 & 0 & 0 & 0 & \mathcal{D}_{21}e^{-i\alpha_{78}} & i\omega_8/2  & 0 & 0\\
            A_{9,1} & A_{9,2} & A_{9,3} & A_{9,4} & A_{9,5} & A_{9,6} & A_{9,7} & A_{9,8}  & i\omega_{9}/2 & A_{9, 10}\\
            A_{10,1} & A_{10, 2} & A_{10, 3} & A_{10, 4} & A_{10, 5} & A_{10, 6} & A_{10, 7} & A_{10, 8} & A_{10, 9}  & i\omega_{10}/2
        \end{bmatrix}\,,
    \end{equation}}
    where $\mathcal{B}$, $\mathcal{C}$, $\mathcal{D}$ are respectively $3\times 3$, $2\times 2$ and $2\times 2$ real matrices with null diagonal and nonzero off-diagonal elements.
    Let $\bm{\mu}_{\mathcal{B}} = (\mu_{\mathcal{B}}^1, \mu_{\mathcal{B}}^2, \mu_{\mathcal{B}}^3, 0, 0, 0, 0, 0, 0, 0)^\top$, $\bm{\mu}_{\mathcal{C}} = (0, 0, 0, 0, \mu_{\mathcal{C}}^5, \mu_{\mathcal{C}}^6, 0, 0, 0, 0)^\top$, and $\bm{\mu}_{\mathcal{D}} = (0, 0, 0, 0, 0, 0, \mu_{\mathcal{D}}^7, \mu_{\mathcal{D}}^8, 0, 0)^\top$.
    Since $2\times 2$ matrices are always symmetrizable, let $\text{diag}(\mu_{\mathcal{C}}^5, \mu_{\mathcal{C}}^6)\mathcal{C}$ and $\text{diag}(\mu_{\mathcal{D}}^7, \mu_{\mathcal{D}}^8)\mathcal{D}$ be symmetric and also consider that $\mathcal{B}$ is symmetrizable, i.e., $\text{diag}(\mu_{\mathcal{B}}^1, \mu_{\mathcal{B}}^2, \mu_{\mathcal{B}}^3)\mathcal{B}$ is symmetric. The second and third conditions of Thm.~\ref{thm:existence_fpmonom} are thus satisfied. Moreover, the first and fourth conditions of Thm.~\ref{thm:existence_fpmonom} are satisfied by construction. Altogether Thm.~\ref{thm:existence_fpmonom} guarantees that the Koopman generator possesses four monomial eigenfunctions
    \begin{equation}
        \psi_\mathcal{B}(\bm z) = z^{\bm \mu_\mathcal{B}}\,, \quad \psi_4(\bm z) = z_4\,, \quad \psi_\mathcal{C}(\bm z) = z^{\bm \mu_\mathcal{C}}\,, \quad \psi_\mathcal{D}(\bm z) = z^{\bm \mu_\mathcal{D}}\,,
    \end{equation}
    where
    \begin{align}
        \lambda_{\mathcal{B}} = i\,\tilde{\omega}_\mathcal{B}\,,\quad \lambda_4 = i\,\omega_4\,,\quad\lambda_{\mathcal{C}} = i\,\tilde{\omega}_\mathcal{C}\,,\quad\lambda_{\mathcal{D}} = i\,\tilde{\omega}_\mathcal{D}\, 
    \end{align}
    and
    \begin{align}
        \tilde{\omega}_\mathcal{B} :=  \mu_{\mathcal{B}}^1 \omega_1 + \mu_{\mathcal{B}}^2 \omega_2 + \mu_{\mathcal{B}}^3 \omega_3\,,\quad
        \tilde{\omega}_{\mathcal{C}} :=  \mu_{\mathcal{C}}^5 \omega_5 +  \mu_{\mathcal{C}}^6 \omega_6\,,\quad
        \tilde{\omega}_{\mathcal{D}} :=  \mu_{\mathcal{D}}^7 \omega_7 +  \mu_{\mathcal{D}}^8 \omega_8 \,.
    \end{align}
    To further specify the example, suppose that $\omega_4,\tilde{\omega}_{\mathcal{B}},\tilde{\omega}_{\mathcal{D}}\in\mathbb{R}\setminus\{0\}$ and $\tilde{\omega}_{\mathcal{C}} = 0$. Since $\psi_\mathcal{C}$ is a monomial eigenfunction of null eigenvalue, it is a constant of motion and there are $q = 3$ monomial eigenfunctions with nonzero eigenvalues. We can construct $q - 1 = 2$ time-independent conserved monomials according to Lem.~\ref{lem:monom_constants}.
    The constants of motion are $z^{\bm \nu_1}$ and $z^{\bm \nu_2}$, where $V = UO$ for the $10 \times 3$ matrix $U = (\bm{\mu}_{\mathcal{B}}\,\,\,\bm{e}_4\,\,\,\bm{\mu}_\mathcal{D})$ for $\bm{e}_4$ the unit vector with the fourth entry being unity, $V = (\bm \nu_1\,\,\, \bm \nu_2)$, and $O$ is a $3\times 2$ matrix where the columns are linearly independent vectors which are orthogonal to $\tilde{\bm \omega} = (\tilde{\omega}_\mathcal{B}\, \, \, \omega_4\,\,\, \tilde{\omega}_\mathcal{D})^\top$.
    For example, the matrix $O$ can be of the form
    \begin{equation}
        O = 
        \begin{bmatrix}
            \omega_4 & 0\\
            -\tilde{\omega}_\mathcal{B} & \tilde{\omega}_\mathcal{D}\\
            0 & -\omega_4
        \end{bmatrix}\,.
    \end{equation}
    Then, the monomial constants of motion are
    \begin{align}
        z^{\bm \nu_1} &= \psi_\mathcal{B}^{\omega_4}\psi_4^{-\tilde{\omega}_\mathcal{B}} = z^{\omega_4\bm \mu_\mathcal{B}}z_4^{-\tilde{\omega}_\mathcal{B}} = z_1^{\mu_\mathcal{B}^1\omega_4}z_2^{\mu_\mathcal{B}^2\omega_4}z_3^{\mu_\mathcal{B}^3\omega_4}z_4^{-(\mu_{\mathcal{B}}^1 \omega_1 + \mu_{\mathcal{B}}^2 \omega_2 + \mu_{\mathcal{B}}^3 \omega_3)}\,,\\
        z^{\bm \nu_2} &= \psi_4^{\tilde{\omega}_\mathcal{D}}\psi_\mathcal{D}^{-\omega_4} = z_4^{\tilde{\omega}_\mathcal{D}}z^{-\omega_4 \bm \mu_\mathcal{D}} =  z_4^{\mu_{\mathcal{D}}^7 \omega_7 +  \mu_{\mathcal{D}}^8 \omega_8}z_7^{-\mu_\mathcal{D}^7\omega_4}z_8^{-\mu_\mathcal{D}^8\omega_4}\,.
    \end{align}
\end{example}

\clearpage 

\section{Vandermonde-ratio eigenfunctions}
\label{sec:vandermonde}

In the last section, we have established the class of networks that admit monomial eigenfunctions and we have shown how to combine these Koopman eigenfunctions to define monomial constants of motion. As explained in the main text, monomial eigenfunctions can also be combined with Vandermonde-ratio eigenfunctions to generate another family of constants of motion. In the following subsection, we provide the proof for the second theorem of the paper.

\subsection{Proof of Theorem 2: Vandermonde-ratio eigenfunction}
\label{SIsubsec:proof_thm2}
\begin{theorem}[Thm.~\ref{thm:vandermonde} of the paper]\label{thm:vandermonde_SI}
Let $I$ and $J$ be disjoint subsets of $\mathcal{V}$ of respective sizes $m\geq 2$ and $n\geq 2$ with $m + n \leq N$. Define
\begin{align*}
V_{IJ}(z)
&:=
\frac{\displaystyle \prod_{\substack{p,q\in I\\ p<q}}(\bar z_p-\bar z_q)^{\sigma_{pq}}}
     {\displaystyle \prod_{\substack{r,s\in J\\ r<s}}(z_r-z_s)^{\tau_{rs}}}\,, 
\qquad
S_p:=\sum_{q\in I\setminus\{p\}}\sigma_{pq}\,,\qquad T_r:=\sum_{s\in J\setminus\{r\}}\tau_{rs}\,,
\end{align*}
where $\sigma_{pq} = \sigma_{qp}\in\mathbb{R}\setminus\{0\}$ for all $p,q\in I$ and $\tau_{rs} = \tau_{sr}\in\mathbb{R}\setminus\{0\}$ for all $r,s\in J$. If
\begin{enumerate}
\item $A_{jk} =0$ for all $j\in I\cup J$ and $k\in \mathcal V\setminus (I\cup J)$\,,
\item $A_{pq}=0$ for all distinct $p,q\in I$ and $A_{rs}=0$ for all distinct $r,s\in J$\,,
\item $A_{pr}=\mathcal A\,T_r$ and $A_{rp}=\bar{\mathcal{A}}\,S_p$ for all $p\in I$, $r\in J$ and $\mathcal{A}\in\mathbb{C}$,
\item $\omega_p = \Omega_I\in \mathbb{R}$ for all $p \in I$ and $\omega_r = \Omega_J\in \mathbb{R}$ for all $r \in J$\,,
\end{enumerate}
then $V_{IJ}(z)$ is an eigenfunction of $\mathcal{K}$ with eigenvalue
\begin{align*}
    \lambda = -i(\Sigma_I\Omega_I + \Sigma_J\Omega_J) \qquad \text{where} \qquad \Sigma_I = \frac{1}{2}\sum_{p\in I}S_p\,,\,\quad \Sigma_J = \frac{1}{2}\sum_{r\in J}T_{r}\,.
\end{align*}
\end{theorem}
\begin{proof}
The logarithm of the eigenfunction is
\begin{align*}
\ln V_{IJ}(z)
&=
\sum_{\substack{p,q\in I\\ p<q}}\sigma_{pq}\ln(\bar z_p-\bar z_q)
-
\sum_{\substack{r,s\in J\\ r<s}}\tau_{rs}\ln(z_r-z_s)
\end{align*}
and its derivative is
\begin{align*}
\partial_j \ln V_{IJ}(z)
=
\sum_{\substack{p,q\in I\\ p<q}}\sigma_{pq}\,
\frac{\bar z_p z_q\,\delta_{jp}-z_p\bar z_q\,\delta_{jq}}{z_p-z_q}\,
\bm{1}_{j \in I}
+
\sum_{\substack{r,s\in J\\ r<s}}\tau_{rs}\,
\frac{\delta_{js}-\delta_{jr}}{z_r-z_s}\,
\bm{1}_{j \in J}\,.
\end{align*}
With $\sigma_{pq} =\sigma_{qp}$ for all $p,q\in I$, $\tau_{rs}=\tau_{sr}$ for all $r,s \in J$, and the cancellation of the terms through the Kronecker deltas, one obtains
\begin{align*}
2\,\partial_j \ln V_{IJ}(z)
&=\Big(\sum_{\substack{q\in I\\ q\neq j}}\sigma_{jq}\frac{\bar z_j z_q}{z_j-z_q}-\sum_{\substack{p\in I\\ p\neq j}}\sigma_{pj}\frac{z_p\bar z_j}{z_p-z_j}\Big)\bm{1}_{j \in I}+\Big(\sum_{\substack{r\in J\\ r\neq j}}\tau_{rj}\frac{1}{z_r-z_j}-\sum_{\substack{s\in J\\ s\neq j}}\tau_{js}\frac{1}{z_j-z_s}\Big)\bm{1}_{j \in J}\,.
\end{align*}
The derivative of the logarithm is thus simplified to
\begin{align*}
\partial_j \ln V_{IJ}(z)
&=
\sum_{\substack{q\in I\\ q\neq j}}
\sigma_{jq}\frac{z_q\bar z_j}{z_j-z_q}\,\bm{1}_{j \in I}
-
\sum_{\substack{r\in J\\ r\neq j}}
\tau_{rj}\frac{1}{z_j-z_r}\,\bm{1}_{j \in J}\,.
\end{align*}
The Koopman generator applied to $\ln V_{IJ}(z)$ is then
\begin{align*}
\mathcal K[\ln V_{IJ}(z)]
&=
\sum_{\substack{j,q,k\in I\\ j\neq q}}\sigma_{jq}z_q\,\frac{A_{jk}z_k\bar z_j-\bar A_{jk}\bar z_k z_j}{z_j-z_q}-\sum_{\substack{j,r, k\in J\\ j\neq r}}\tau_{rj}z_j\,\frac{A_{jk}z_k-\bar A_{jk}\bar z_k z_j^2}{z_j-z_r}
\\
&\quad + \sum_{\substack{j,q\in I\\ j\neq q}}\sum_{k\in J} \sigma_{jq}z_q\, \frac{A_{jk}z_k\bar z_j-\bar A_{jk}\bar z_k z_j}{z_j-z_q} - \sum_{\substack{j,r\in J\\ j\neq r}}\sum_{k\in I} \tau_{rj}z_j\, \frac{A_{jk}z_k-\bar A_{jk}\bar z_k z_j^2}{z_j-z_r}\\
&\quad + \sum_{\substack{j,q\in I\\ j\neq q}}\sum_{k\in \mathcal V\setminus(I\cup J)} \sigma_{jq}z_q\, \frac{A_{jk}z_k\bar z_j-\bar A_{jk}\bar z_k z_j}{z_j-z_q} - \sum_{\substack{j,r\in J\\ j\neq r}}\sum_{k\in \mathcal V\setminus(I\cup J)} \tau_{rj}z_j\, \frac{A_{jk}z_k-\bar A_{jk}\bar z_k z_j^2}{z_j-z_r}\,,
\end{align*}
where the first two terms are the intra-block terms (within $I, J$), the second and the third terms are the bipartite terms (between $I$ and $J$), and the last two terms are the inputs in $I, J$ from the rest of the graph. If the first and second conditions are applied, only the diagonal part of the intra-block terms remain along with the bipartite terms, so the latter expression becomes $\mathcal K[\ln V_{IJ}(z)] = F_{\mathrm{diagonal}} + F_{\mathrm{bipartite}}$ with
\begin{align*}
    F_{\mathrm{diagonal}} &= \sum_{\substack{p,q\in I\\ p\neq q}}
\sigma_{pq}z_q\frac{A_{pp} -\bar A_{pp}}{z_p-z_q}-\sum_{\substack{r,s\in J\\ r\neq s}}\tau_{rs}z_s\frac{A_{ss}-\bar A_{ss}}{z_s-z_r}\,,\\
    F_{\mathrm{bipartite}} &= \sum_{\substack{p,q\in I\\ p\neq q}}\sum_{k\in J}\sigma_{pq}\,z_q\,\frac{A_{pk}z_k\bar z_p-\bar A_{pk}\bar z_k z_p}{z_p-z_q}-\sum_{\substack{r,s\in J\\ r\neq s}}\sum_{k\in I}\tau_{rs}z_s\,\frac{A_{sk}z_k\bar z_s-\bar A_{sk}\bar z_k z_s}{z_s-z_r}\,.
\end{align*}
Using $A_{pp} = i\omega_p/2$, the diagonal part is
\begin{align*}
    F_{\mathrm{diagonal}}=\sum_{\substack{p,q\in I\\ p\neq q}}
\frac{i\omega_p\sigma_{pq}z_q}{z_p-z_q}-\sum_{\substack{r,s\in J\\ r\neq s}}\frac{i\omega_s\tau_{rs}z_s}{z_s-z_r}=i\sum_{\substack{p,q\in I\\ p<q}}
\sigma_{pq}\frac{\omega_pz_q - \omega_qz_p}{z_p-z_q}-i\sum_{\substack{r,s\in J\\ r<s}}\tau_{rs}\frac{\omega_rz_r - \omega_sz_s}{z_r-z_s}\,.
\end{align*}
If $\omega_p = \Omega_I$ for all $p\in I$ and $\omega_r = \Omega_J$ for all $r\in J$ (fourth condition), then 
\begin{align*}
    F_{\mathrm{diagonal}} = -i\,\Omega_I\sum_{\substack{p<q\in I}}\sigma_{pq} - i\,\Omega_J \sum_{\substack{r<s\in J}} \tau_{rs} = -i(\Sigma_I\Omega_I + \Sigma_J\Omega_J)\,.
\end{align*}
It thus remains to show that $F_{\mathrm{bipartite}} = 0$ by the third condition to have $\mathcal K[\ln V_{IJ}(z)] = \lambda$ with $\lambda = -i(\Sigma_I\Omega_I + \Sigma_J\Omega_J)$ as the eigenvalue. Applying the third condition ($A_{pr}=T_r\,\mathcal A$ and $A_{rp}=S_p\,\bar{\mathcal{A}}$ for all $p\in I$, $r\in J$ and $\mathcal{A}\in\mathbb{C}$) gives
\begin{align*}  F_{\mathrm{bipartite}}&=\sum_{\substack{p,q\in I\\ p\neq q}}\sum_{r\in J}T_r\sigma_{pq}\,z_q\,\frac{\mathcal Az_r\bar z_p-\bar{\mathcal{A}}\bar z_r z_p}{z_p-z_q}-\sum_{\substack{r,s\in J\\ r\neq s}}\sum_{p\in I}S_p\tau_{rs}z_r\,\frac{\,\bar{\mathcal{A}}z_p\bar z_r-\mathcal{A}\bar z_p z_r}{z_r-z_s}\,,
\end{align*}
which is equivalent to
\begin{align*}
    F_{\mathrm{bipartite}}=\,\,&\mathcal A\sum_{r\in J}T_rz_r\sum_{\substack{p,q\in I\\ p\neq q}}\sigma_{pq}\,\,\frac{ \bar z_p z_q}{z_p-z_q}
    -\bar{\mathcal{A}}\sum_{r\in J}T_r\bar z_r\sum_{\substack{p,q\in I\\ p\neq q}}\sigma_{pq}\,\,\frac{ z_pz_q}{z_p-z_q}
    \\&-\bar{\mathcal{A}}\sum_{p\in I}S_pz_p\sum_{\substack{r,s\in J\\ r\neq s}}\tau_{rs}\,\frac{1}{z_r-z_s}
    +\mathcal{A}\sum_{p\in I}S_p\bar z_p\sum_{\substack{r,s\in J\\ r\neq s}}\tau_{rs}\,\frac{z_r^2}{z_r-z_s}\,.
\end{align*}
The second and third terms vanish by antisymmetry, which leaves  
\begin{align}\label{eq:lastrelation}
    F_{\mathrm{bipartite}}=\,\mathcal{A}\,\Big(\sum_{r\in J}T_rz_r\sum_{\substack{p,q\in I\\ p\neq q}}\sigma_{pq}\,\,\frac{ \bar z_p z_q}{z_p-z_q}
    + \sum_{p\in I}S_p\bar z_p\sum_{\substack{r,s\in J\\ r\neq s}}\tau_{rs}\,\frac{z_r^2}{z_r-z_s}\Big)\,.
\end{align}
But note that 
\begin{align*}
    \sum_{\substack{p,q\in I\\ p\neq q}}\sigma_{pq}\,\,\frac{ \bar z_p z_q}{z_p-z_q} &= \sum_{\substack{p,q\in I\\ p<q}}\sigma_{pq}\,\,\frac{ \bar z_p z_q - z_p \bar z_q}{z_p-z_q} = 
    -\sum_{\substack{p,q\in I\\ p<q}}\sigma_{pq}(\bar z_p + \bar z_q )= -\sum_{p\in I}\left(\sum_{q>p}\sigma_{pq} +\sum_{q<p}\sigma_{pq}\right)\bar z_p \\&= -\sum_{p\in I}S_p\bar z_p\,, \\
    \sum_{\substack{r,s\in J\\ r\neq s}}\tau_{rs}\,\frac{z_r^2}{z_r-z_s} &= \sum_{\substack{r,s\in J\\ r<s}}\tau_{rs}\,\frac{z_r^2 - z_s^2}{z_r - z_s} = \sum_{\substack{r,s\in J\\ r<s}}\tau_{rs}\,(z_r + z_s) = \sum_{r\in J}\left(\sum_{s>r}\tau_{rs} +\sum_{s<r}\tau_{rs}\right)z_r \\&= \sum_{r\in J} T_rz_r\,.
\end{align*}
Inserting these relations in Eq.~\eqref{eq:lastrelation} completes the proof~:
\begin{align*}
    F_{\mathrm{bipartite}} = \,\mathcal{A}\,\Big(\Big(\sum_{r\in J}T_rz_r\Big)\Big(-\sum_{p\in I}S_p\bar z_p\Big)
    + \Big(\sum_{p\in I}S_p\bar z_p\Big)\Big(\sum_{r\in J} T_rz_r\Big)\Big) = 0\,.
\end{align*}
\end{proof}
\begin{remark}
    When there are four oscillators, $V_{IJ}(\bm{z}) = \frac{(\bar{z}_a - \bar{z}_b)^\sigma}{(z_c - z_d)^\tau}$ with $\sigma, \tau \in \mathbb{R}\setminus\{0\}$, we also proved that the conditions are necessary (not shown here).
\end{remark}

\begin{remark}\label{remark:regularity-vandermonde}
    When the exponents $\sigma_{pq}$ or $\tau_{rs}$ are not integers, the Vandermonde-ratio eigenfunctions are naturally understood on the universal cover $\mathbb R^N$ of the torus $\mathbb T^N$, where the phase variables $\theta_j$ are not identified modulo $2\pi$. Indeed,
\[
e^{i\theta_p}-e^{i\theta_q}
=
2i\,e^{i(\theta_p+\theta_q)/2}
\sin\!\left(\frac{\theta_p-\theta_q}{2}\right)
\]
so that the eigenfunctions can be locally expressed in terms of smooth phase variables on $\mathbb R^N$. However, for non-integer exponents, the complex powers $(z_p-z_q)^{\sigma_{pq}}$ and $(z_r-z_s)^{\tau_{rs}}$ require a choice of branch and therefore do not define globally continuous observables on $\mathbb T^N$. In addition, extra regularity issues arise near collision configurations where $z_p=z_q$ (equivalently $\theta_p-\theta_q\in2\pi\mathbb Z$), since the sine factors vanish there. Nevertheless, away from these collision sets, the eigenfunctions remain smooth after restriction to simply connected coordinate patches where continuous phase coordinates and consistent branches can be chosen locally.
\end{remark}

\subsection{Conserved quantity for Janus and excitation-inhibition oscillators}
\label{SIsubsec:janusEIKM}
In this subsection, we consider two models from the literature, namely Janus oscillators and excitation-inhibition Kuramoto oscillators, that can be adapted to admit Vandermonde-ratio eigenfunctions and their corresponding conserved quantities.
\begin{example}\label{ex:janus}
    In Ref.~\cite{Nicolaou2019a}, Janus oscillators are introduced to show how various synchronization phenomena, such as chimera states, explosive synchronization, and asymmetry-induced synchronization, can co-occur within the same system of oscillators. Following Sec.~IIIC of Ref.~\cite{Nicolaou2019a} and Ref.~\cite{Peron2020}, these oscillators can be described by a particular Kuramoto model where
\begin{align*}
    \dot{\theta}_p &= \Omega + \sum_{r\in J} W_{pr}\sin(\theta_r - \theta_p)\,,\quad p\in I\,,\\
    \dot{\theta}_r &= -\Omega + \sum_{p\in I} W_{rp}\sin(\theta_p - \theta_r)\,,\quad r\in J\,,
\end{align*}
with $\Omega_I = -\Omega_J =: \Omega$ and
\begin{align*}
    W = \begin{pmatrix}
        0_{m\times m} & S\\ T & 0_{n\times n}
    \end{pmatrix}\,, \quad S\in\mathbb{R}^{m\times n}, \quad T\in\mathbb{R}^{n\times m}\,.
\end{align*}
This form readily satisfies the second and fourth conditions of Thm.~\ref{thm:vandermonde_SI}. Condition 1 is trivially satisfied since no other oscillators than those in the bipartite graph are considered. Note also that more generally, one has the freedom to set $\Omega_I$ and $\Omega_J$ as desired (e.g., $\Omega_I = -c \Omega_J$ with $c>0$). In Refs.~\cite{Nicolaou2019a, Peron2020}, $m = n$ and the matrices $S,T$ have particular forms, so that the Janus oscillators form a ring or a symmetric network ($T = S^\top = \beta I + \sigma B$, for constants $\beta, \sigma$ and a matrix $B$). Yet, it is interesting to note that it suffices to choose $S$ and $T$ to satisfy the third condition of Thm.~\ref{thm:vandermonde_SI} (yielding a specific network of Janus oscillators, in which one could add phase lags) to get a Vandermonde-ratio eigenfunction.
\end{example}

\begin{example}\label{ex:EIKM}
In 2018, E.~Montbrió and D.~Pazó~\cite{Montbrio2018} introduced a Kuramoto model with excitatory (E) and inhibitory (I) coupling to reproduce the generation of EI-based neuronal rhythms observed in the brain. In the paper, they also consider the effect of having bipartite coupling in their Fig.~2,\,3,\,4 and it yields non trivial phase and bifurcation diagrams. Below, we start from their EI Kuramoto model and make adjustments so that it satisfies the conditions of Thm.~\ref{thm:vandermonde_SI}. First, we rename the set $I$ as $E$ and the set $J$ as $I$ to avoid confusion, but abandon the subscripts $\sigma \in \{E, I\}$. Then, by setting the noise $\xi_i$ to zero and imposing conditions 2,\,3,\,4 of Thm.~\ref{thm:vandermonde_SI} (condition 1 is automatically satisfied) yields 
\begin{align*}
    \dot{\theta}_p &= \Omega_{E} + K\sum_{r \in I} T_{r}\sin(\theta_r - \theta_p + \pi/2) = \Omega_{E} + K\sum_{r \in I} T_{r}\cos(\theta_r - \theta_p)\,,\quad p\in E\,,\\
    \dot{\theta}_r &= \Omega_{I} + K\sum_{p\in E} S_{p}\sin(\theta_p - \theta_r - \pi/2) = \Omega_{I} - K\sum_{p\in E} S_{p}\cos(\theta_p - \theta_r)\,,\quad r \in I\,,
\end{align*}
where $K\in\mathbb{R}$ and the coupling constants that would shift the natural frequencies in the original model are absorbed in $\Omega_E$, $\Omega_I$. In other terms, with $\mathcal{A} = Ke^{i\pi/2}/2 = iK/2$, the corresponding complex weight matrix $A$ to this EI Kuramoto model is
\begin{align*}
    A = \frac{i}{2}\begin{pmatrix}
        \Omega_E & \cdots & 0  & K T_1 & \cdots & K T_n\\
        \vdots & \cdots & \vdots  & \vdots & \cdots & \vdots\\
        0 & \cdots & \Omega_E  & K T_1 & \cdots & K T_n\\
        -K S_1 & \cdots & -K S_m & \Omega_I & \cdots & 0  & \\
        \vdots & \cdots & \vdots  & \vdots & \cdots & \vdots\\
        -K S_1 & \cdots & -K S_m & 0 & \cdots & \Omega_I  & \\
    \end{pmatrix}\,.
\end{align*}
\end{example}
In both these examples, monomial eigenfunctions, Vandermonde-ratio eigenfunctions, and conserved cross-ratios are possible according to the first three theorems of the paper.

\clearpage

\section{Conservation of cross-ratios}

In 1994, Watanabe and Strogatz found constants of motion for identical phase oscillators, shaping subsequent years of theoretical studies on such oscillators. Fifteen years later~\cite{Marvel2009}, these constants of motion were linked to the cross-ratios. Finding various constants of motion for phase oscillators on general heterogeneous networks, however, remained challenging. In the last two sections, we found general conditions to have monomial and Vandermonde-ratio eigenfunctions and they strongly depend on the network structure; in this section, we provide the necessary and sufficient conditions for the conservation of cross-ratios in the Kuramoto model, with a focus on the network-theoretic interpretation.
The first three subsections introduce the cross-ratios and their properties (functional independence and their joint invariance under the special linear algebra). Section~\ref{SIsubsec:proof_thm3} contains the proof of Theorem 3 from the main text while its corollaries are in Section~\ref{SIsubsec:corollaries_thm3}. Finally, we explain in Sec.~\ref{SIsubsec:WS_darboux} how WS integrals are rather Darboux functions with identical cofactors when the Koopman generator has the form that allows cross-ratios to be conserved.

\subsection{Introduction to cross-ratios}
In this subsection, we present some facts about the cross-ratios (also called the anharmonic ratio), which are central quantities in the paper. The cross-ratio of four different points $z_{a}, z_{b}, z_{c}, z_{d}$ in $\mathbb{C}\cup \{\infty\}$ is
\begin{equation}
c_{abcd}(\bm z) = (z_{a},z_{b}~; z_{c},z_{d}) =\frac{(z_{c}-z_{a})(z_{d}-z_{b})}{(z_{c}-z_{b})(z_{d}-z_{a})}\,\label{eq:cross_ratio_SI} 
\end{equation}
and we will use the notation $c_{abcd}(\bm z)$, $(z_{a},z_{b}~; z_{c},z_{d})$ or even $\gamma_{abcd}$ at our convenience. 

The cross-ratios are the only projective invariant of a quadruple of collinear points (see Fig.~\ref{fig:cross_ratio_projection}), i.e.,
\begin{equation*}
    (z_{a},z_{b}~; z_{c},z_{d}) = (z_{a}',z_{b}'~; z_{c}',z_{d}')\,,
\end{equation*}
which gives them a special place in projective geometry.
\begin{figure}[b]
    \centering
    \includegraphics[width=0.3\linewidth]{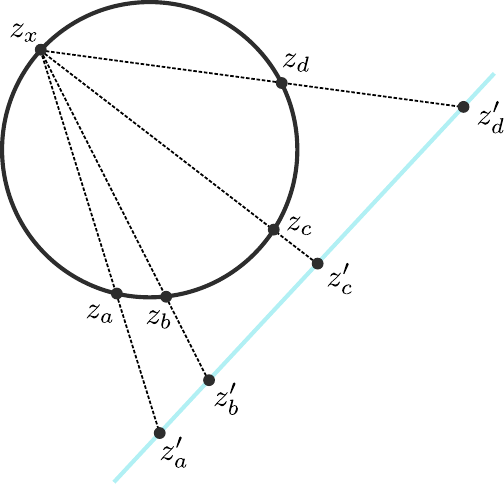}
    \caption[Cross-ratio as a projective invariant]{\textbf{Cross-ratio as a projective invariant}. Projection of the different points $z_a, z_b, z_c, z_d \in \mathbb{T}^1\setminus\{z_x\}$ to a line in $\mathbb{C}$ from a point $z_x \in \mathbb{T}^1$.}
    \label{fig:cross_ratio_projection}
\end{figure}
They are also invariant under M\"{o}bius transformations
\begin{align}
    M_{\alpha ,\beta, \gamma, \delta}(z)=\frac{\alpha z+\beta}{\gamma z+\delta}\,,
\end{align}
where $z,\alpha ,\beta, \gamma, \delta$ are complex numbers and $\alpha\delta-\beta\gamma\neq 0$, i.e.,
\begin{equation}
     (M_{\alpha ,\beta, \gamma, \delta}(z_{a}),M_{\alpha ,\beta, \gamma, \delta}(z_{b})~; M_{\alpha ,\beta, \gamma, \delta}(z_{c}), M_{\alpha ,\beta, \gamma, \delta}(z_{d})) = (z_{a},z_{b}~; z_{c},z_{d})\,.
\end{equation}
A cross-ratio is real if and only if the four points are distributed on a circle (concyclic points) or on a line (collinear points). In the case of interest in the paper, the cross-ratios depend on the state vector $(z_1,...,z_N) = (e^{i\theta_1},..., e^{i\theta_N})$, describing the positions of the oscillators rotating on the unit circle, so the values of the cross-ratios belong to $\mathbb{R}\cup \{\infty\}$ and the cross-ratios can be expressed in terms of the phases $\theta_{a}$, $\theta_{b}$, $\theta_{c}$, $\theta_{d}$:
\begin{align}
c_{abcd}(\bm z) &= \frac{(e^{i\theta_c}-e^{i\theta_a})(e^{i\theta_d}-e^{i\theta_b})}{(e^{i\theta_c}-e^{i\theta_b})(e^{i\theta_d}-e^{i\theta_a})}\nonumber\\ 
&= \frac{e^{i\left(\frac{\theta_c+\theta_a}{2}\right)}\left[e^{i\left(\frac{\theta_c-\theta_a}{2}\right)}-e^{-i\left(\frac{\theta_c-\theta_a}{2}\right)}\right]\,\,e^{i\left(\frac{\theta_d+\theta_b}{2}\right)}\left[e^{i\left(\frac{\theta_d-\theta_b}{2}\right)}-e^{-i\left(\frac{\theta_d-\theta_b}{2}\right)}\right]}{e^{i\left(\frac{\theta_c+\theta_b}{2}\right)}\left[e^{i\left(\frac{\theta_c-\theta_b}{2}\right)}-e^{-i\left(\frac{\theta_c-\theta_b}{2}\right)}\right]\,\,e^{i\left(\frac{\theta_d+\theta_a}{2}\right)}\left[e^{i\left(\frac{\theta_d-\theta_a}{2}\right)}-e^{-i\left(\frac{\theta_d-\theta_a}{2}\right)}\right]}\nonumber\\ &= \frac{\sin\left(\frac{\theta_c - \theta_a}{2}\right)\sin\left(\frac{\theta_d - \theta_b}{2}\right)}{\sin\left(\frac{\theta_c - \theta_b}{2}\right)\sin\left(\frac{\theta_d - \theta_a}{2}\right)}\,.
\end{align}
Different perspectives are given in group theory, hyperbolic geometry~\cite[Chap. 11]{Pressley2010} and other fields for the cross-ratios that we will not put forward here. Yet, we will address two other properties in more detail, that is, the fact that they are the joint invariants of a Lie algebra and their functional dependencies.

\subsection{Cross-ratios as joint invariants of the special linear algebra}
\label{subsec:joint_invariants}
In this subsection, we show that the cross-ratios are joint invariants of 
\begin{align*}
    L_{-1} := \sum_{j=1}^n \pd[]{}{z_j}\,,\qquad L_{0} := \sum_{j=1}^n z_j\pd[]{}{z_j}\,, \qquad L_{1} := \sum_{j=1}^n z_j^2\pd[]{}{z_j}\,,
\end{align*}
where $4 \leq n \leq N$  and $L_{-1}, L_{0}, L_{1}$ are associated with the basis elements of $\mathfrak{sl}_2$. This is an old, known result~\cite{Olver1995} and method~\cite{Olver1993} that we present here for the sake of completeness and because we did not see it explicitly elsewhere.

The idea is to successively (1) apply the method of characteristics to the partial differential equation $L_k[\eta] = 0$ for some $k\in\{-1,0,1\}$, (2) find the characteristic curves (invariants of $L_k$) and (3) use them as new coordinates for the next $L_\ell[\eta] = 0$ for some $\ell\in\{-1,0,1\}\setminus\{k\}$ and (4) repeat steps (1) to (3) until the three partial differential equations are treated. The easiest way to proceed is to address the partial differential equations in this order: $L_{-1}[\eta] = 0$, $L_0[\eta] = 0$ and then, $L_1[\eta] = 0$.

First, we find the general form of the invariants of $L_{-1}$. They obey the first-order partial differential equation
\begin{align}\label{eq:Lm1_PDE}
    L_{-1}[\eta] = \sum_{j=1}^n \pd[]{\eta}{z_j} = 0\,,
\end{align}
where $\eta$ is a complex-valued function of $\mathbb{T}^N$. The characteristic equations are $\dif z_a = \dif z_j$ for some $a \in \{1,...,n\}$ and for all $j\in\{1,...,n\}\setminus\{a\}$. Integrating yields $n-1$ functionally independent characteristic curves having the form $z_j - z_a = C_{ja} \in \mathbb{C}$.

Let $\Delta_j := z_j - z_a$. The change of coordinates from $z_1,...,z_n$ to $z_a$, $(\Delta_j)_{j\neq a}$ gives $\partial/\partial z_j = \partial/\partial \Delta_j$ for $j \neq a$ and 
\begin{align}\label{eq:L0_PDE}
    L_{0}[\eta] = \sum_{j=1}^n z_j\pd[]{\eta}{z_j} = z_a\pd[]{\eta}{z_a} + \sum_{j\neq a}(\Delta_j + z_a)\pd[]{\eta}{\Delta_j} =  z_aL_{-1}[\eta] + \sum_{j\neq a}\Delta_j\pd[]{\eta}{\Delta_j}= \sum_{j\neq a}\Delta_j\pd[]{\eta}{\Delta_j} = 0\,,
\end{align}
because $L_{-1}[\eta] = 0$. Hence, the characteristic equations and curves are respectively $\dif \Delta_b / \Delta_b = \dif \Delta_j / \Delta_j$ and $\Delta_j/\Delta_b = C'_{jb}\in\mathbb{C}$ for some $b \in\{1,...,n\}\setminus\{a\}$ and for all $j \in \{1,...,n\}\setminus\{a,b\}$.

Applying the same change of coordinates to the partial differential equation for $L_1$ gives
\begin{align}\label{eq:L1eq}
    L_1[\eta] = \sum_{j=1}^n z_j^2\pd[]{\eta}{z_j} = z_a^2L_{-1}[\eta] +2z_a\sum_{j\neq a} \Delta_j\pd[]{\eta}{\Delta_j} + \sum_{j\neq a}\Delta_j^2\pd[]{\eta}{\Delta_j} = 0\,.
\end{align}
But again, $L_{-1}[\eta] = 0$ and $\sum_{j\neq a} \Delta_j\pd[]{\eta}{\Delta_j} = L_{0}[\eta] - L_{-1}[\eta] = 0$. Therefore, Eq.~\eqref{eq:L1eq} is simplified to
\begin{align}\label{eq:L1eq2}
    \sum_{j\neq a}\Delta_j^2\pd[]{\eta}{\Delta_j} = 0\,.
\end{align}
With the characteristic curves from Eq.~\eqref{eq:L0_PDE}, define $\rho_j = \Delta_j/\Delta_b$ and the new coordinates $z_a, \Delta_b, (\rho_j)_{j\neq a,b}$. The partial derivatives $\partial/\partial \Delta_j$ for $j\neq a,b$ become $\Delta_b^{-1}\partial/\partial \rho_j$ and making the change of coordinates for Eq.~\eqref{eq:L1eq2}, one obtains
\begin{align*}
    \Delta_b\pd[]{\eta}{\Delta_b} + \sum_{j\neq a,b}\rho_j^2\pd[]{\eta}{\rho_j} = 0\,.
\end{align*}
Adding $ 0 = \sum_{j\neq a,b}\Delta_j\partial\eta/\partial\Delta_j - \sum_{j\neq a,b}\Delta_j\partial\eta/\partial\Delta_j$ yields
\begin{align*}
    \sum_{j\neq a}\Delta_j\pd[]{\eta}{\Delta_j} - \sum_{j\neq a,b}(\Delta_b\rho_j)\left(\frac{1}{\Delta_b}\pd[]{\eta}{\rho_j}\right)  + \sum_{j\neq a,b}\rho_j^2\pd[]{\eta}{\rho_j} = 0\,,
\end{align*}
but Eq.~\eqref{eq:L0_PDE} implies that the first term vanishes, which gives
\begin{align*}
   \sum_{j\neq a,b} \rho_j(\rho_j-1)\pd[]{\eta}{\rho_j}= 0\,.
\end{align*}
For all $c,d \neq a, b$, the method of characteristics leads to
\begin{align*}
    \frac{\dif \rho_c}{\rho_c(\rho_c - 1)} = \frac{\dif \rho_d}{\rho_d(\rho_d - 1)}\,,
\end{align*}
and the characteristic curves
\begin{align}
    \frac{\rho_c(1- \rho_d)}{\rho_d (1-\rho_c)} = C''_{cd}\,.
\end{align}
Altogether, by returning to the original variables, the joint invariants of $L_{-1}$, $L_0$ and $L_1$ are such that
\begin{align}\label{eq:derived_cross_ratios}
    \eta(\bm{z}) = \frac{\rho_c(1- \rho_d)}{\rho_d (1-\rho_c)} &= \frac{\frac{z_c - z_a}{z_b - z_a}\left(1- \frac{z_d - z_a}{z_b - z_a}\right)}{\frac{z_d - z_a}{z_b - z_a} \left(1-\frac{z_c - z_a}{z_b - z_a}\right)} = \frac{(z_c - z_a) (z_d - z_b)}{(z_c - z_b)(z_d - z_a)} = c_{abcd}(\bm{z})\,,
\end{align}
that is, the joint invariants are cross-ratios.

\subsection{Functional independence of cross-ratios}
\label{subsec:functional_independence}
In this subsection, to make the presentation more self-contained, we demonstrate a known result (e.g., Ref.~\cite[Example 2.35]{Olver1995} or Appendix of Ref.~\cite{Marvel2009}) about the functional independence of cross-ratios using the classical criterion on the rank of a Jacobian matrix~\cite[p.85]{Olver1993}. We will use the result not only for globally coupled oscillators ($n=N$), but for any subgraph of $4\leq n < N$ vertices with a certain number of functionally independent cross-ratios. 

Each cross-ratio $c_{abcd}$ depends on four indices $\{a,b,c,d\}$ with $a,b,c,d\in\{1,...,n\}$, where the order of the indices $a$, $b$, $c$, and $d$ may change the cross-ratio. Since there are $n!/(n-4)!$ (falling factorial $(n)_4$) ways to select an ordered list of 4 distinct elements from a set of $n$ items (number of $4$-permutations in a set of size $n$), there are $n!/(n-4)!$ cross-ratios. Yet, only some of them are independent as stated in the next lemma.
\begin{lemma}\label{lem:independence_cross_ratios}
    Among the $n!/(n-4)!$ cross-ratios $c_{abcd}$ such that $a,b,c,d \in\{1, ... , n\}$ with $a\neq b \neq c \neq d$, only $n-3$ cross-ratios form a functionally independent set, such as
    \begin{equation}\label{eq:cross_ratio_fct_ind}
        \{c_{1234}, c_{2345}, ..., c_{(n-3)(n-2)(n-1)n}\}\,.
    \end{equation}
\end{lemma}
\begin{proof}
    A large portion of the proof is taken from Ref.~\cite{Marvel2009}. 
    First, it is straightforward to show that all cross-ratios with permutations of the same 4 indices are functionally dependent.
    For a given cross-ratio $c_{abcd}$,
\begin{align}
    c_{badc}(\bm z) &= c_{abcd}(\bm z)\,;\label{ben:permut_1}\\
    c_{cdab}(\bm z) &= c_{abcd}(\bm z)\,;\label{ben:permut_2}\\
    c_{dcba}(\bm z) &= c_{abcd}(\bm z)\,;\label{ben:permut_3}\\
    c_{abdc}(\bm z) &= \frac{1}{c_{abcd}(\bm z)}\label{ben:1/c}\,;\\
    c_{acbd}(\bm z) &= 1 - c_{abcd}(\bm z)\,;\label{ben:1-c}\\
    c_{acdb}(\bm z) &= \frac{1}{c_{acbd}(\bm z)} = \frac{1}{1 - c_{abcd}(\bm z)}\quad \text{by Eqs.~(\ref{ben:1/c}-\ref{ben:1-c})}\,;\label{ben:1/1-c}\\
    c_{adcb}(\bm z) &= 1 - c_{acdb}(\bm z) = 1 - \frac{1}{1 - c_{abcd}(\bm z)} = \frac{c_{abcd}(\bm z)}{c_{abcd}(\bm z) - 1}\quad \text{by Eqs.~(\ref{ben:1-c}-\ref{ben:1/1-c})}\,;\label{ben:c/c-1}\\
    c_{adbc}(\bm z) &= \frac{1}{c_{adcb}(\bm z)} = \frac{c_{abcd}(\bm z) - 1}{c_{abcd}(\bm z)}\quad \text{by Eqs.~(\ref{ben:1/c}-\ref{ben:c/c-1})}\,.\label{ben:c-1/c}
\end{align}
The 15 other permutations can be obtained by permuting the indices of the cross-ratios in Eqs.~(\ref{ben:1/c}-\ref{ben:c-1/c}) according to Eqs.~(\ref{ben:permut_1}-\ref{ben:permut_3}). 
Thus, cross-ratios with all 24 permutations of the same 4 indices are functionally dependent upon a single cross-ratio. 
For the rest of the proof, permutations can therefore be omitted without loss of generality.
\par
Next, we demonstrate the functional independence of the $n-3$ cross-ratios
\begin{equation*}
    \{c_{1234}, c_{2345}, ..., c_{(n-3)(n-2)(n-1)n}\}\,.
\end{equation*}
Consider the function $\bm \zeta: \mathbb{T}^n \mapsto \mathbb{R}^{n-3}$ defined as
\begin{equation*}
    \bm \zeta(\bm z) = (c_{1234}(\bm z), c_{2345}(\bm z), ..., c_{(n-3)(n-2)(n-1)n}(\bm z))\,.
\end{equation*}
Its $(n-3)\times n$ Jacobian matrix is
\begin{equation*}
    D\bm \zeta(\bm z) = \begin{bmatrix}
        \partial_{1}c_{1234}(\bm z) & \partial_{2}c_{1234}(\bm{z}) & \cdots & \partial_{n}c_{1234}(\bm z)\\
        \partial_{1}c_{2345}(\bm z) & \partial_{2}c_{2345}(\bm{z}) & \cdots & \partial_{n}c_{2345}(\bm z)\\
        \vdots & \vdots & \ddots & \vdots\\
        \partial_{1}c_{(n-3)(n-2)(n-1)n}(\bm z) & \partial_{2}c_{(n-3)(n-2)(n-1)n}(\bm z) & \cdots & \partial_{n}c_{(n-3)(n-2)(n-1)n}(\bm z)\\
    \end{bmatrix}\,,
\end{equation*}
where $\partial_i = \pd{}{z_i}$. On one hand, the cross-ratios only depend on the variables given by their four indices, leading to $\partial_i c_{abcd}(\bm z) = 0$ for all $i \notin \{a, b, c, d\}$ and all $\bm z \in \mathbb{T}^n$. On the other hand, the derivatives of the cross-ratios with respect to each index are
\begin{equation}\label{eq:derivees_cr}
    \begin{aligned}
        \partial_a c_{abcd}(\bm z) &= \frac{(z_d - z_b)(z_c - z_d)}{(z_c - z_b)(z_d - z_a)^2}\,,\\
    \partial_b c_{abcd}(\bm z) &= \frac{(z_c - z_a)(z_d - z_c)}{(z_c - z_b)^2(z_d - z_a)}\,,
    \end{aligned}
    \quad\quad
    \begin{aligned}
        \partial_c c_{abcd}(\bm z) &= \frac{(z_a - z_b)(z_d - z_b)}{(z_c - z_b)^2(z_d - z_a)}\,,\\
    \partial_d c_{abcd}(\bm z) &= \frac{(z_c - z_a)(z_b - z_a)}{(z_c - z_b)(z_d - z_a)^2}\,.
    \end{aligned}
\end{equation}
Hence, the Jacobian matrix becomes 
\begin{equation*}
    D\bm \zeta(\bm z) = \begin{bmatrix}
        \partial_{1}c_{1234}(\bm z) & \partial_{2}c_{1234}(\bm z) & \partial_{3}c_{1234}(\bm z) & \partial_{4}c_{1234}(\bm z) & 0 & 0 & 0 & \cdots\\
        0 & \partial_{2}c_{2345}(\bm z) & \partial_{3}c_{2345}(\bm z) & \partial_{4}c_{2345}(\bm z) & \partial_5c_{2345}(\bm z) & 0 & 0 & \cdots\\
        0 & 0 & \partial_{3}c_{3456}(\bm z) & \partial_{4}c_{3456}(\bm z) & \partial_{5}c_{3456}(\bm z) & \partial_{6}c_{3456}(\bm z) & 0 & \cdots\\
        0 & 0 & 0 & \partial_{4}c_{4567}(\bm z) & \partial_{5}c_{4567}(\bm z) & \partial_{6}c_{4567}(\bm z) & \partial_{7}c_{4567}(\bm z) & \cdots\\
        \vdots & \vdots & \vdots & \vdots & \vdots & \vdots & \vdots & \ddots
    \end{bmatrix}.
\end{equation*}
From there, the cross-ratios are functionally independent if $D\bm \zeta(\bm z)$ has full rank $n - 3$ for all $\bm z \in \mathbb{T}^n$ such that $z_j \neq z_k$ for all $j,k\in \{1,...,n\}$~\cite[p.85]{Olver1993}.
Taking only the $n-3$ first columns yields an upper triangular submatrix $T(\bm z)$. 
The determinant of the matrix is then simply the product of the diagonal elements, i.e.
\begin{equation}\label{eq:det_jacobian}
    \det(T(\bm z)) = \partial_1c_{1234}(\bm z)\,\,\partial_2c_{2345}(\bm z)\,\,\partial_3c_{3456}(\bm z)\,\,...\,\,\partial_{n-3}c_{(n-3)(n-2)(n-1)n}(\bm z)\,.
\end{equation}
For any distinct $\bm z$ (no superimposed oscillators), the derivatives of the cross-ratios in Eq.~\eqref{eq:det_jacobian} cannot vanish by the form of Eqs.~\eqref{eq:derivees_cr}. 
Thus, the determinant of the submatrix $T(\bm z)$ cannot be null, implying that $T(\bm z)$ is invertible and has full rank, that is, $\mathrm{rank}(T(\bm z)) = n - 3$. 
Moreover, $T(\bm z)$ is the largest invertible square submatrix of $D\bm \zeta(\bm z)$. Recalling that the rank of a matrix is the size of the largest invertible square submatrix, one obtains that $\mathrm{rank}(D\bm \zeta(\bm z)) = n-3$ for any distinct $\bm z \in \mathbb{T}^n$. 
We thus conclude by Ref.~\cite[Theorem 2.16]{Olver1993} that the $n-3$ cross-ratios are functionally independent.\\
\par
Finally, the $n-3$ independent cross-ratios of Eq.~\eqref{eq:cross_ratio_fct_ind} can be combined to obtain every possible cross-ratio $c_{pqrs}$ where $p<q<r<s\leq n$. 
To alleviate the notation, define $\gamma_{pqrs} = c_{pqrs}(\bm z)$.
First, two cross-ratios can be multiplied to obtain a third one as
\begin{align}
    \gamma_{abcd}\gamma_{becd} &= \gamma_{aecd}\label{ben:multiplication_cr}
\end{align}
Using this property combined with the permutation relations of Eqs.~(\ref{ben:permut_1}-\ref{ben:c-1/c}), we can define three functions $F$, $G$ and $H$ which take two cross-ratios with the indices ordered from lowest to highest and generate a new cross-ratio which also has indices ordered from lowest to highest. 
Explicitly, these functions are
\begin{align}
    F(\gamma_{abcd}, \gamma_{bcde}) &= \frac{1}{1-\gamma_{abcd}(\gamma_{bcde} - 1)/\gamma_{bcde}} = \gamma_{acde}\,,\\
    G(\gamma_{abcd}, \gamma_{bcde}) &= \frac{1}{1 - (1-\gamma_{abcd})(1 - \gamma_{bcde})} = \gamma_{abde}\,,\\
    H(\gamma_{abcd}, \gamma_{bcde}) &= \frac{1}{1 - (\gamma_{abcd} - 1)\gamma_{bcde}/\gamma_{abcd}} = \gamma_{abce}\,.
\end{align}
Every cross-ratio with growing indices can be generated by applying those three functions on the $n-3$ independent cross-ratios with 4 consecutive indices.
First, cross-ratios where there is a gap between the first and the second indices can be expressed by applying $F$ iteratively on cross-ratios with consecutive indices from the lower bound of the gap to the higher bound.
The same can be said for gaps between the second and the third indices by applying $G$ and for gaps between the third and the fourth indices by applying $H$. 
Any cross-ratio with four indices in growing order $\gamma_{pqrs}$ can thus be expressed by applying the $F$ function $q-p-1$ times, then the $G$ function $r-q-1$ times, then the $H$ function $s-r-1$ times on the $n-3$ functionally independent ones.
Since all cross-ratios are functionally dependent on the $n-3$ independent ones with growing, consecutive indices, we conclude that $n-3$ is the maximum number of functionally independent cross-ratios.
\end{proof}

\subsection{Proof of Theorem 3: Cross-ratios as constants of motion}
\label{SIsubsec:proof_thm3}

We now recall the third theorem and proceed with its proof. Some elementary---but lengthy---steps of the proof relied on symbolic calculations that were performed in Matlab (\textit{symbolic\_calculations\_theorem\_generalized.m}). Remember that we use the convention that $A_{jk}$ is the (complex) weight of the interaction from $k$ to $j$. Note also that conditions 1, 2, 3 respectively correspond to conditions 3.1, 3.2, 3.3 of the main text.
\begin{theorem}[Thm.~\ref{thm:crossratios} of the paper]\label{thm:cte_mvt_kuramoto_graphe} 
Consider the $N$-dimensional Kuramoto model on a graph described by a $N\times N$ real matrix $W$, a $N\times N$ phase-lag matrix $\alpha$, and a natural frequency vector $\bm{\omega} = (\omega_j)_{j=1}^N$ [Definition~\ref{def:kuramoto}]. 
The cross-ratio $c_{abcd}$~\eqref{eq:cross_ratio_SI} is a constant of motion in the model if and only if the vertices $a$, $b$, $c$, and $d$ of the graph described by the complex matrix
\begin{align}
     A = \frac{1}{2}\left( W \circ e^{-i\alpha} + i\,\mathrm{diag}(\bm{\omega})\right)\quad \text{with} \quad e^{-i\alpha} =  (e^{-i\alpha_{jk}})_{j,k\in\{1,...,N\}}
\end{align}
have the same:
\begin{enumerate}
\item outgoing interactions within $\{a,b,c,d\}$, i.e., 
\begin{align}\label{eq:identical_insiders_outcoming_edge_A}
\begin{aligned}
    A_{ba} &= A_{ca} = A_{da} =: \mathcal{A}_a\,,\\ 
    A_{ab} &= A_{cb} = A_{db} =: \mathcal{A}_b\,, 
\end{aligned}
\qquad
\begin{aligned}
    A_{ac} &= A_{bc} = A_{dc} =: \mathcal{A}_c\,, \\
    A_{ad} &= A_{bd} = A_{cd} =: \mathcal{A}_d\,;     
\end{aligned}
\end{align}
\item incoming interactions from the vertices outside $\{a,b,c,d\}$, i.e.,
\begin{equation}\label{eq:identical_outsiders_incoming_edge_A}
    A_{ak} = A_{bk} = A_{ck} = A_{dk},\quad \forall\, k\in\{1,...,N\}\setminus\{a, b, c, d\}\,;
\end{equation}
\item shifted natural frequencies
\begin{align}\label{eq:imag_cond}
\omega_a - 2\Imag(\mathcal{A}_a) = \omega_b - 2\Imag(\mathcal{A}_b) = \omega_c - 2\Imag(\mathcal{A}_c) = \omega_d - 2\Imag(\mathcal{A}_d)\,.
\end{align}
\end{enumerate}
\end{theorem}
\begin{proof}
By Lem.~\ref{lem:correspondance_complexe_A}, the Kuramoto model can be described by
\begin{align}
    \dot{z}_j = \sum_{k}A_{jk}z_k - \left(\sum_k\bar{A}_{jk}\bar{z}_k \right)z_j^2
\end{align}
with $z_j = e^{i\theta_j}$ and the complex matrix of interactions
\begin{align}
    A = \frac{1}{2}\left( W \circ e^{-i\alpha} + i\,\mathrm{diag}(\bm{\omega})\right)\,,
\end{align}
where $e^{-i\alpha} = (e^{-i\alpha_{jk}})_{j,k}$, $\bm{\omega} = (\omega_1,...,\omega_N)$, $\circ$ is the element-wise product, and we recall that without loss of generality one can assume that the diagonal elements of $W$ and $\alpha$ are zero. The Koopman generator is thus
\begin{align*}
    \mathcal{K} = \bm p^\top \bm L_{-1} - \bar{\bm p}^\top \bm L_1\,,
\end{align*}
where we have introduced $\bm p = A\bm z$ to simplify the expressions. Saying that the cross-ratio $c_{abcd}$ is a constant of motion in the model is equivalent to the condition
\begin{equation}
    \mathcal{K}[c_{abcd}](\bm{z}) = 0\,,
\end{equation}
i.e., the cross-ratio is an eigenfunction with eigenvalue 0 of the Koopman generator. The property 
\begin{equation*}
    \frac{\partial c_{abcd}(\bm z)}{\partial z_j} = c_{abcd}(\bm z)\,\frac{\partial \ln(c_{abcd}(\bm z))}{\partial z_j}
\end{equation*}
together with the properties of the logarithm imply that 
\begin{equation*}
    \mathcal{K}[c_{abcd}(\bm z)] = c_{abcd}(\bm z)\, \mathcal{K}[\ln(c_{abcd}(\bm z))] = c_{abcd}(\bm z)\, \mathcal{K}[\ln(z_c - z_a) + \ln(z_d - z_b) - \ln(z_c - z_b) - \ln(z_d - z_a)].
\end{equation*}
Using the relations
\begin{equation*}
    \ell_j^n[\ln(z_x - z_y)] = z_j^{n+1} \, \frac{\delta_{jx} - \delta_{jy}}{z_x - z_y}\,,\qquad \forall\, n\in\mathbb{Z}\,,
\end{equation*}
where $\ell_j^n$ is the $j$-th element of the $n$-th Euler's operator defined in Eq.~\eqref{eq:elements_euler} for all $j \in \{1,...,N\}$, leads to 
\begin{equation*}
    \sum_{j=1}^N \beta_j \ell_j^n[\ln(z_x - z_y)] =  \frac{\beta_{x}z_x^{n+1} - \beta_{y}z_y^{n+1}}{z_x - z_y}
\end{equation*}
for all $n\in\mathbb{Z}$ and some arbitrary constants $\beta_1,...,\beta_N$. The above identity applied to each term of the generator yields
\begin{align*}
    \bm p^\top \bm L_{-1} [c_{abcd}(\bm z)]&= c_{abcd}(\bm z)\,\left(\frac{p_{c} - p_{a}}{z_c - z_a} + \frac{p_{d} - p_{b}}{z_d - z_b} - \frac{p_{c} - p_{b}}{z_c - z_b} - \frac{p_{d} - p_{a}}{z_d - z_a}\right)\\
    \bar{\bm p}^\top \bm L_1 [c_{abcd}(\bm z)] &= c_{abcd}(\bm z)\,\left(\frac{\bar{p}_{c}z_c^2 - \bar{p}_{a}z_a^2}{z_c - z_a} + \frac{\bar{p}_{d}z_d^2 - \bar{p}_{b}z_b^2}{z_d - z_b} - \frac{\bar{p}_{c}z_c^2 - \bar{p}_{b}z_b^2}{z_c - z_b} - \frac{\bar{p}_{d}z_d^2 - \bar{p}_{a}z_a^2}{z_d - z_a}\right)\,.
\end{align*}
The factorization of $1/[(z_c - z_a)(z_d - z_b)(z_c - z_b)(z_d - z_a)]$ in the last three equations, the simplification
\begin{equation*}
    \gamma_{abcd}(\bm{z}) := \frac{c_{abcd}(\bm z)}{(z_c - z_a)(z_d - z_b)(z_c - z_b)(z_d - z_a)} = \frac{1}{(z_c - z_b)^2(z_d - z_a)^2},
\end{equation*}
and elementary algebraic manipulations give
{\footnotesize
\begin{align}
    \bm p^\top \bm L_{-1} [c_{abcd}(\bm z)]& = \gamma_{abcd}(\bm{z}) \sum_{k=1}^N \Big[(A_{ck} - A_{ak})(z_d - z_b)(z_c - z_b)(z_d - z_a)+ (A_{dk} - A_{bk})(z_c - z_a)(z_c - z_b)(z_d - z_a) \nonumber\\& \phantom{\quad\gamma_{abcd}(\bm{z}) \sum}- (A_{ck} - A_{bk})(z_c - z_a)(z_d - z_b)(z_d - z_a) - (A_{dk} - A_{ak})(z_c - z_a)(z_d - z_b)(z_c - z_b)\Big]\,z_k\label{eq:Lm1c}\\
    \bar{\bm p}^\top \bm L_1 [c_{abcd}(\bm z)] &= \gamma_{abcd}(\bm{z}) \sum_{k=1}^N \Big[(\bar{A}_{ak} - \bar{A}_{bk})z_a^2z_b^2(z_c - z_d) + (\bar{A}_{ck} - \bar{A}_{ak})z_a^2z_c^2(z_b- z_d) + (\bar{A}_{ak} - \bar{A}_{dk})z_a^2z_d^2(z_b-z_c)  \nonumber\\& \phantom{\quad\gamma_{abcd}(\bm{z}) \sum}+ (\bar{A}_{bk} - \bar{A}_{ck})z_b^2z_c^2(z_a  -  z_d) + (\bar{A}_{dk} - \bar{A}_{bk})z_b^2z_d^2(z_a-z_c) + (\bar{A}_{ck} - \bar{A}_{dk})z_c^2z_d^2(z_a - z_b)  \Big]\bar{z}_k\,.\label{eq:L1c}
\end{align}}
Regrouping Eqs.~(\ref{eq:Lm1c}-\ref{eq:L1c}) gives the expression for $\mathcal{K}[c_{abcd}](\bm{z})$, i.e,
{\footnotesize\begin{align}\label{eq:generator_crossratio_A}
    \mathcal{K}[c_{abcd}](\bm{z}) &= \gamma_{abcd}(\bm{z})\Big[\textstyle{\sum_{k\in\{a,b,c,d\}}} \Big((A_{ck} - A_{ak})(z_d - z_b)(z_c - z_b)(z_d - z_a)z_k + (A_{dk} - A_{bk})(z_c - z_a)(z_c - z_b)(z_d - z_a)z_k \nonumber\\& + (A_{bk} - A_{ck})(z_c - z_a)(z_d - z_b)(z_d - z_a)z_k + (A_{ak} - A_{dk})(z_c - z_a)(z_d - z_b)(z_c - z_b)z_k \nonumber\\& - (\bar{A}_{ak} - \bar{A}_{bk})z_a^2z_b^2(z_c - z_d)\bar{z}_k - (\bar{A}_{ck} - \bar{A}_{ak})z_a^2z_c^2(z_b- z_d)\bar{z}_k - (\bar{A}_{ak} - \bar{A}_{dk})z_a^2z_d^2(z_b-z_c)\bar{z}_k  \nonumber\\& - (\bar{A}_{bk} - \bar{A}_{ck})z_b^2z_c^2(z_a  -  z_d)\bar{z}_k - (\bar{A}_{dk} - \bar{A}_{bk})z_b^2z_d^2(z_a-z_c)\bar{z}_k - (\bar{A}_{ck} - \bar{A}_{dk})z_c^2z_d^2(z_a - z_b)\bar{z}_k\Big)\nonumber\\ + &\textstyle{\sum_{k\notin \{a,b,c,d\}}} \Big((A_{ck} - A_{ak})(z_d - z_b)(z_c - z_b)(z_d - z_a)z_k + (A_{dk} - A_{bk})(z_c - z_a)(z_c - z_b)(z_d - z_a)z_k \nonumber\\& + (A_{bk} - A_{ck})(z_c - z_a)(z_d - z_b)(z_d - z_a)z_k + (A_{ak} - A_{dk})(z_c - z_a)(z_d - z_b)(z_c - z_b)z_k \nonumber\\& - (\bar{A}_{ak} - \bar{A}_{bk})z_a^2z_b^2(z_c - z_d)\bar{z}_k - (\bar{A}_{ck} - \bar{A}_{ak})z_a^2z_c^2(z_b- z_d)\bar{z}_k - (\bar{A}_{ak} - \bar{A}_{dk})z_a^2z_d^2(z_b-z_c)\bar{z}_k  \nonumber\\& - (\bar{A}_{bk} - \bar{A}_{ck})z_b^2z_c^2(z_a  -  z_d)\bar{z}_k - (\bar{A}_{dk} - \bar{A}_{bk})z_b^2z_d^2(z_a-z_c)\bar{z}_k - (\bar{A}_{ck} - \bar{A}_{dk})z_c^2z_d^2(z_a - z_b)\bar{z}_k\Big)\Big]\,,
\end{align}}
where we have separated the sum over $k\in\{1,...,N\}$ into $k\in\{a,b,c,d\}$ and $k\notin \{a,b,c,d\}$. 

($\Leftarrow$) In Eq.~\eqref{eq:generator_crossratio_A}, only differences between the complex matrix elements $A_{ak}$, $A_{bk}$, $A_{ck}$, $A_{dk}$ and their conjugate appear. One can readily apply equation~\eqref{eq:identical_outsiders_incoming_edge_A} (condition 2) to cancel each term in the summation on $k\notin\{a,b,c,d\}$~:
{\footnotesize\begin{align*}
    \mathcal{K}[c_{abcd}](\bm{z}) &= \gamma_{abcd}(\bm{z})\Big[\textstyle{\sum_{k\in\{a,b,c,d\}}} \Big((A_{ck} - A_{ak})(z_d - z_b)(z_c - z_b)(z_d - z_a)z_k + (A_{dk} - A_{bk})(z_c - z_a)(z_c - z_b)(z_d - z_a)z_k \nonumber\\& + (A_{bk} - A_{ck})(z_c - z_a)(z_d - z_b)(z_d - z_a)z_k + (A_{ak} - A_{dk})(z_c - z_a)(z_d - z_b)(z_c - z_b)z_k \nonumber\\& - (\bar{A}_{ak} - \bar{A}_{bk})z_a^2z_b^2(z_c - z_d)\bar{z}_k - (\bar{A}_{ck} - \bar{A}_{ak})z_a^2z_c^2(z_b- z_d)\bar{z}_k - (\bar{A}_{ak} - \bar{A}_{dk})z_a^2z_d^2(z_b-z_c)\bar{z}_k  \nonumber\\& - (\bar{A}_{bk} - \bar{A}_{ck})z_b^2z_c^2(z_a  -  z_d)\bar{z}_k - (\bar{A}_{dk} - \bar{A}_{bk})z_b^2z_d^2(z_a-z_c)\bar{z}_k - (\bar{A}_{ck} - \bar{A}_{dk})z_c^2z_d^2(z_a - z_b)\bar{z}_k\Big)\Big]\,.
\end{align*}}
Performing the summation explicitly and using Eq.~\eqref{eq:identical_insiders_outcoming_edge_A} (condition 1) give
{\footnotesize\begin{align*}
    \mathcal{K}[c_{abcd}](\bm{z}) &= \gamma_{abcd}(\bm{z})\Big[(\mathcal{A}_a - A_{aa})(z_d - z_b)(z_c - z_b)(z_d - z_a)z_a + (A_{aa} - \mathcal{A}_a)(z_c - z_a)(z_d - z_b)(z_c - z_b)z_a \\&- (\bar{A}_{aa} - \bar{\mathcal{A}_a})z_a^2z_b^2(z_c - z_d)\bar{z}_a - (\bar{\mathcal{A}_a} - \bar{A}_{aa})z_a^2z_c^2(z_b- z_d)\bar{z}_a - (\bar{A}_{aa} - \bar{\mathcal{A}_a})z_a^2z_d^2(z_b-z_c)\bar{z}_a \\& + (\mathcal{A}_b - A_{bb})(z_c - z_a)(z_c - z_b)(z_d - z_a)z_b + (A_{bb} - \mathcal{A}_b)(z_c - z_a)(z_d - z_b)(z_d - z_a)z_b \\& - (\bar{\mathcal{A}}_b - \bar{A}_{bb})z_a^2z_b^2(z_c - z_d)\bar{z}_b - (\bar{A}_{bb} - \bar{\mathcal{A}}_b)z_b^2z_c^2(z_a  -  z_d)\bar{z}_b - (\bar{\mathcal{A}}_b - \bar{A}_{bb})z_b^2z_d^2(z_a-z_c)\bar{z}_b \\&
    + (A_{cc} - \mathcal{A}_c)(z_d - z_b)(z_c - z_b)(z_d - z_a)z_c + (\mathcal{A}_c - A_{cc})(z_c - z_a)(z_d - z_b)(z_d - z_a)z_c\\& - (\bar{A}_{cc} - \bar{\mathcal{A}}_c)z_a^2z_c^2(z_b- z_d)\bar{z}_c - (\bar{\mathcal{A}}_c - \bar{A}_{cc})z_b^2z_c^2(z_a  -  z_d)\bar{z}_c - (\bar{A}_{cc} - \bar{\mathcal{A}}_c)z_c^2z_d^2(z_a - z_b)\bar{z}_c \\& + (A_{dd} - \mathcal{A}_d)(z_c - z_a)(z_c - z_b)(z_d - z_a)z_d + (\mathcal{A}_d - A_{dd})(z_c - z_a)(z_d - z_b)(z_c - z_b)z_d \\&- (\bar{\mathcal{A}}_d - \bar{A}_{dd})z_a^2z_d^2(z_b-z_c)\bar{z}_d - (\bar{A}_{dd} - \bar{\mathcal{A}}_d)z_b^2z_d^2(z_a-z_c)\bar{z}_d - (\bar{\mathcal{A}}_d - \bar{A}_{dd})z_c^2z_d^2(z_a - z_b)\bar{z}_d\Big]\,.
\end{align*}}
The expansion and the simplification of the latter equation enable regrouping the monomials and writing
{\footnotesize\begin{align*}
     &\mathcal{K}[c_{abcd}](\bm{z}) = \\
      &\gamma_{abcd}(\bm{z})
  \Big[(A_{cc} - \bar{A}_{cc} - A_{bb} + \bar{A}_{bb} - \mathcal{A}_c + \bar{\mathcal{A}}_c + \mathcal{A}_b - \bar{\mathcal{A}}_b)z_a^2z_bz_c
     + (A_{bb} - \bar{A}_{bb} - A_{dd} + \bar{A}_{dd} - \mathcal{A}_b + \bar{\mathcal{A}}_b + \mathcal{A}_d - \bar{\mathcal{A}}_d)z_a^2z_bz_d\\& 
     + (A_{dd} - \bar{A}_{dd} - A_{cc} + \bar{A}_{cc} - \mathcal{A}_d + \bar{\mathcal{A}}_d + \mathcal{A}_c - \bar{\mathcal{A}}_c)z_a^2z_cz_d
     + (A_{aa} - \bar{A}_{aa} - A_{cc} + \bar{A}_{cc} - \mathcal{A}_a + \bar{\mathcal{A}}_a + \mathcal{A}_c - \bar{\mathcal{A}}_c)z_az_b^2z_c\\& 
     + (A_{dd} - \bar{A}_{dd} - A_{aa} + \bar{A}_{aa} - \mathcal{A}_d + \bar{\mathcal{A}}_d + \mathcal{A}_a - \bar{\mathcal{A}}_a)z_az_b^2z_d
     + (A_{bb} - \bar{A}_{bb} - A_{aa} + \bar{A}_{aa} - \mathcal{A}_b + \bar{\mathcal{A}}_b + \mathcal{A}_a - \bar{\mathcal{A}}_a)z_az_bz_c^2\\& 
     + (A_{aa} - \bar{A}_{aa} - A_{bb} + \bar{A}_{bb} - \mathcal{A}_a + \bar{\mathcal{A}}_a + \mathcal{A}_b - \bar{\mathcal{A}}_b)z_az_bz_d^2
     + (A_{aa} - \bar{A}_{aa} - A_{dd} + \bar{A}_{dd} - \mathcal{A}_a + \bar{\mathcal{A}}_a + \mathcal{A}_d - \bar{\mathcal{A}}_d)z_az_c^2z_d\\& 
     + (A_{cc} - \bar{A}_{cc} - A_{aa} + \bar{A}_{aa} - \mathcal{A}_c + \bar{\mathcal{A}}_c + \mathcal{A}_a - \bar{\mathcal{A}}_a)z_az_cz_d^2 
     + (A_{cc} - \bar{A}_{cc} - A_{dd} + \bar{A}_{dd} - \mathcal{A}_c + \bar{\mathcal{A}}_c + \mathcal{A}_d - \bar{\mathcal{A}}_d)z_b^2z_cz_d\\& 
     + (A_{dd} - \bar{A}_{dd} - A_{bb} + \bar{A}_{bb} - \mathcal{A}_d + \bar{\mathcal{A}}_d + \mathcal{A}_b - \bar{\mathcal{A}}_b)z_bz_c^2z_d 
     + (A_{bb} - \bar{A}_{bb} - A_{cc} + \bar{A}_{cc} - \mathcal{A}_b + \bar{\mathcal{A}}_b + \mathcal{A}_c - \bar{\mathcal{A}}_c)z_bz_cz_d^2\Big]\,,
\end{align*}}
which is equivalent to
{\footnotesize\begin{align*}
     \mathcal{K}[c_{abcd}](\bm{z}) = 2i\gamma_{abcd}(\bm{z}) \Big[&(\Imag(A_{cc} - A_{bb}) - \Imag(\mathcal{A}_c - \mathcal{A}_b))z_a^2z_bz_c
     + (\Imag(A_{bb} - A_{dd}) - \Imag(\mathcal{A}_b - \mathcal{A}_d))z_a^2z_bz_d\\& 
     + (\Imag(A_{dd} - A_{cc}) - \Imag(\mathcal{A}_d - \mathcal{A}_c))z_a^2z_cz_d
     + (\Imag(A_{aa} - A_{cc}) - \Imag(\mathcal{A}_a - \mathcal{A}_c))z_az_b^2z_c\\& 
     + (\Imag(A_{dd} - A_{aa}) - \Imag(\mathcal{A}_d - \mathcal{A}_a))z_az_b^2z_d
     + (\Imag(A_{bb} - A_{aa}) - \Imag(\mathcal{A}_b - \mathcal{A}_a))z_az_bz_c^2\\& 
     + (\Imag(A_{aa} - A_{bb}) - \Imag(\mathcal{A}_a - \mathcal{A}_b))z_az_bz_d^2
     + (\Imag(A_{aa} - A_{dd}) - \Imag(\mathcal{A}_a - \mathcal{A}_d))z_az_c^2z_d\\& 
     + (\Imag(A_{cc} - A_{aa}) - \Imag(\mathcal{A}_c - \mathcal{A}_a))z_az_cz_d^2 
     + (\Imag(A_{cc} - A_{dd}) - \Imag(\mathcal{A}_c - \mathcal{A}_d))z_b^2z_cz_d\\& 
     + (\Imag(A_{dd} - A_{bb}) - \Imag(\mathcal{A}_d - \mathcal{A}_b))z_bz_c^2z_d 
     + (\Imag(A_{bb} - A_{cc}) - \Imag(\mathcal{A}_b - \mathcal{A}_c))z_bz_cz_d^2\Big]\,.
\end{align*}}
Since $\Imag(A_{jj}) = \omega_j/2$, one gets
{\footnotesize\begin{align*}
     \mathcal{K}[c_{abcd}](\bm{z}) = i\gamma_{abcd}(\bm{z}) \Big[&((\omega_{c} - \omega_{b}) - 2\Imag(\mathcal{A}_c - \mathcal{A}_b))z_a^2z_bz_c
     + ((\omega_{b} - \omega_{d}) - 2\Imag(\mathcal{A}_b - \mathcal{A}_d))z_a^2z_bz_d\\& 
     + ((\omega_{d} - \omega_{c}) - 2\Imag(\mathcal{A}_d - \mathcal{A}_c))z_a^2z_cz_d
     + ((\omega_{a} - \omega_{c}) - 2\Imag(\mathcal{A}_a - \mathcal{A}_c))z_az_b^2z_c\\& 
     + ((\omega_{d} - \omega_{a}) - 2\Imag(\mathcal{A}_d - \mathcal{A}_a))z_az_b^2z_d
     + ((\omega_{b} - \omega_{a}) - 2\Imag(\mathcal{A}_b - \mathcal{A}_a))z_az_bz_c^2\\& 
     + ((\omega_{a} - \omega_{b}) - 2\Imag(\mathcal{A}_a - \mathcal{A}_b))z_az_bz_d^2
     + ((\omega_{a} - \omega_{d}) - 2\Imag(\mathcal{A}_a - \mathcal{A}_d))z_az_c^2z_d\\& 
     + ((\omega_{c} - \omega_{a}) - 2\Imag(\mathcal{A}_c - \mathcal{A}_a))z_az_cz_d^2 
     + ((\omega_{c} - \omega_{d}) - 2\Imag(\mathcal{A}_c - \mathcal{A}_d))z_b^2z_cz_d\\& 
     + ((\omega_{d} - \omega_{b}) - 2\Imag(\mathcal{A}_d - \mathcal{A}_b))z_bz_c^2z_d 
     + ((\omega_{b} - \omega_{c}) - 2\Imag(\mathcal{A}_b - \mathcal{A}_c))z_bz_cz_d^2\Big]\,.
\end{align*}}
Using Eq.~\eqref{eq:imag_cond} (condition 3) makes each term fall, yielding $\mathcal{K}[c_{abcd}](\bm{z}) = 0$ and the sufficiency of the three conditions of the theorem.

($\Rightarrow$) For the necessary conditions, we have to solve $\mathcal{K}[c_{abcd}](\bm{z}) = 0$ in terms of $A_{ak}, A_{bk}, A_{ck}, A_{dk}$ for all $k \in \{1,..., N\}$. The monomials resulting from the last summation on $k\notin \{a,b,c,d\}$ in Eq.~\eqref{eq:generator_crossratio_A} are all independent from the other terms of the expression, which only depend on $z_a, z_b, z_c, z_d$. They can thus be treated separately, i.e., such that
\begin{align}
    0 &= \textstyle{\sum_{k\in\{a,b,c,d\}}} \Big[(A_{ck} - A_{ak})(z_d - z_b)(z_c - z_b)(z_d - z_a)z_k + (A_{dk} - A_{bk})(z_c - z_a)(z_c - z_b)(z_d - z_a)z_k \nonumber\\& + (A_{bk} - A_{ck})(z_c - z_a)(z_d - z_b)(z_d - z_a)z_k + (A_{ak} - A_{dk})(z_c - z_a)(z_d - z_b)(z_c - z_b)z_k \nonumber\\& - (\bar{A}_{ak} - \bar{A}_{bk})z_a^2z_b^2(z_c - z_d)\bar{z}_k - (\bar{A}_{ck} - \bar{A}_{ak})z_a^2z_c^2(z_b- z_d)\bar{z}_k - (\bar{A}_{ak} - \bar{A}_{dk})z_a^2z_d^2(z_b-z_c)\bar{z}_k  \nonumber\\& - (\bar{A}_{bk} - \bar{A}_{ck})z_b^2z_c^2(z_a  -  z_d)\bar{z}_k - (\bar{A}_{dk} - \bar{A}_{bk})z_b^2z_d^2(z_a-z_c)\bar{z}_k - (\bar{A}_{ck} - \bar{A}_{dk})z_c^2z_d^2(z_a - z_b)\bar{z}_k\Big]\label{eq:generator_crossratio_abcd_A}\\ 
    0 &= \textstyle{\sum_{k\notin \{a,b,c,d\}}} \Big[(A_{ck} - A_{ak})(z_d - z_b)(z_c - z_b)(z_d - z_a)z_k + (A_{dk} - A_{bk})(z_c - z_a)(z_c - z_b)(z_d - z_a)z_k \nonumber\\& + (A_{bk} - A_{ck})(z_c - z_a)(z_d - z_b)(z_d - z_a)z_k + (A_{ak} - A_{dk})(z_c - z_a)(z_d - z_b)(z_c - z_b)z_k \nonumber\\& - (\bar{A}_{ak} - \bar{A}_{bk})z_a^2z_b^2(z_c - z_d)\bar{z}_k - (\bar{A}_{ck} - \bar{A}_{ak})z_a^2z_c^2(z_b- z_d)\bar{z}_k - (\bar{A}_{ak} - \bar{A}_{dk})z_a^2z_d^2(z_b-z_c)\bar{z}_k  \nonumber\\& - (\bar{A}_{bk} - \bar{A}_{ck})z_b^2z_c^2(z_a  -  z_d)\bar{z}_k - (\bar{A}_{dk} - \bar{A}_{bk})z_b^2z_d^2(z_a-z_c)\bar{z}_k - (\bar{A}_{ck} - \bar{A}_{dk})z_c^2z_d^2(z_a - z_b)\bar{z}_k\Big]\label{eq:generator_crossratio_notin_abcd_A}\,,
\end{align}
where we have multiplied each equation by $\gamma_{abcd}^{-1}(\bm{z}) = (z_c - z_b)^2(z_d - z_a)^2$.

On the one hand, the expanded form of the summand in Eq.~\eqref{eq:generator_crossratio_notin_abcd_A} for a given $k$ contains 24 independent monomials:
\begin{align*}
&\phantom{+}((A_{dk} - A_{ak})z_bz_c^2z_k^2 + (A_{ak} - A_{ck})z_bz_d^2z_k^2 + (A_{ak} - A_{bk})z_c^2z_dz_k^2 + (A_{bk} - A_{ak})z_cz_d^2z_k^2\\&
+ (A_{ck} - A_{dk})z_a^2z_bz_k^2  + (A_{dk} - A_{bk})z_a^2z_cz_k^2 + (A_{bk} - A_{ck})z_a^2z_dz_k^2 + (A_{dk} - A_{ck})z_az_b^2z_k^2\\&
+ (A_{bk} - A_{dk})z_az_c^2z_k^2 + (A_{ak} - A_{dk})z_b^2z_cz_k^2 + (A_{ck} - A_{ak})z_b^2z_dz_k^2  + (A_{ck} - A_{bk})z_az_d^2z_k^2 \\&
+ (\bar{A}_{dk} - \bar{A}_{bk})z_b^2z_cz_d^2  + (\bar{A}_{dk} - \bar{A}_{ck})z_az_c^2z_d^2  + (\bar{A}_{ck} - \bar{A}_{dk})z_bz_c^2z_d^2 + (\bar{A}_{bk} - \bar{A}_{ck})z_b^2z_c^2z_d \\&
+ (\bar{A}_{ck} - \bar{A}_{ak})z_a^2z_c^2z_d + (\bar{A}_{ak} - \bar{A}_{dk})z_a^2z_cz_d^2+ (\bar{A}_{ck} - \bar{A}_{bk})z_az_b^2z_c^2 + (\bar{A}_{bk} - \bar{A}_{dk})z_az_b^2z_d^2  \\&
+ (\bar{A}_{bk} - \bar{A}_{ak})z_a^2z_b^2z_c + (\bar{A}_{ak} - \bar{A}_{bk})z_a^2z_b^2z_d + (\bar{A}_{ak} - \bar{A}_{ck})z_a^2z_bz_c^2 + (\bar{A}_{dk} - \bar{A}_{ak})z_a^2z_bz_d^2)/z_k\,.
\end{align*}
Since $k\notin\{a,b,c,d\}$, there are thus $N-4$ groups of 24 monomials. The groups are all independent from one to another, because every group has a unique monomial depending on the group index $k\notin\{a,b,c,d\}$. The $24(N-4)$ coefficients in front of the monomials must thus be zero to satisfy Eq.~\eqref{eq:generator_crossratio_notin_abcd_A} since $z_x^u z_y^v z_z^w$ with $u,v,w\in\{1, 2\}$ and $1/z_k$ are not zero. There are 12 coefficients having the form $A_{xk} - A_{yk}$ and they must vanish (i.e., $A_{xk} = A_{yk}$) for all pairs of indices $(x,y)$ with $x,y \in \{a,b,c,d\}$ and $k\notin \{a,b,c,d\}$. This readily implies $\bar{A}_{xk} = \bar{A}_{yk}$, meaning that the 12 other conditions with the form $\bar{A}_{xk} - \bar{A}_{yk} = 0$ are redundant. Therefore, the condition
\begin{equation*}
    A_{ak} = A_{bk} = A_{ck} = A_{dk},\quad \forall k\in\{1,...,N\}\setminus\{a,b,c,d\}\,,
\end{equation*}
i.e., Eq.~\eqref{eq:identical_outsiders_incoming_edge_A} of the second condition is necessary.

\noindent\mbox{On the other hand, expanding Eq.~\eqref{eq:generator_crossratio_abcd_A} yields an equation with 42 independent monomials:}
\begin{align*}
    0 &=  ((A_{ca} - A_{da})z_a^4z_b^2z_cz_d + (A_{da} - A_{ba})z_a^4z_bz_c^2z_d 
+ (A_{ba} - A_{ca})z_a^4z_bz_cz_d^2 + (\bar{A}_{bd} - \bar{A}_{ad})z_a^3z_b^3z_c^2 
\\&+ (A_{cb} - A_{ca} + A_{da} - A_{db} - \bar{A}_{ac} + \bar{A}_{ad} + \bar{A}_{bc} - \bar{A}_{bd})z_a^3z_b^3z_cz_d 
+ (\bar{A}_{ac} - \bar{A}_{bc})z_a^3z_b^3z_d^2 \\&+ (\bar{A}_{ad} - \bar{A}_{cd})z_a^3z_b^2z_c^3 
+ (A_{cc} - A_{bb} + A_{db} - A_{dc} - \bar{A}_{ab} + \bar{A}_{ac} + \bar{A}_{bb} - \bar{A}_{cc})z_a^3z_b^2z_c^2z_d \\&
+ (A_{bb} - A_{cb} + A_{cd} - A_{dd} + \bar{A}_{ab} - \bar{A}_{bb} - \bar{A}_{ad} + \bar{A}_{dd})z_a^3z_b^2z_cz_d^2
+ (\bar{A}_{dc} - \bar{A}_{ac})z_a^3z_b^2z_d^3 \\&+ (A_{ba} - A_{bc} - A_{da} + A_{dc} + \bar{A}_{ab} - \bar{A}_{ad} - \bar{A}_{cb} + \bar{A}_{cd})z_a^3z_bz_c^3z_d  + (A_{db} - A_{cb})z_a^2z_b^4z_cz_d
\\&+ (A_{bc} - A_{bd} - A_{cc} + A_{dd} - \bar{A}_{ac} + \bar{A}_{ad} + \bar{A}_{cc} - \bar{A}_{dd})z_a^3z_bz_c^2z_d^2 + (\bar{A}_{cb} - \bar{A}_{ab})z_a^3z_c^3z_d^2
\\&+ (A_{bd} - A_{ba} + A_{ca} - A_{cd} - \bar{A}_{ab} + \bar{A}_{ac} + \bar{A}_{db} - \bar{A}_{dc})z_a^3z_bz_cz_d^3 + (\bar{A}_{ab} - \bar{A}_{db})z_a^3z_c^2z_d^3
\\& + (\bar{A}_{cd} - \bar{A}_{bd})z_a^2z_b^3z_c^3
+ (A_{aa} - A_{cc} - A_{da} + A_{dc} - \bar{A}_{aa} + \bar{A}_{ba} - \bar{A}_{bc} + \bar{A}_{cc})z_a^2z_b^3z_c^2z_d
\\&+ (A_{ca} - A_{aa} - A_{cd} + A_{dd} + \bar{A}_{aa} - \bar{A}_{ba} + \bar{A}_{bd} - \bar{A}_{dd})z_a^2z_b^3z_cz_d^2
+ (\bar{A}_{bc} - \bar{A}_{dc})z_a^2z_b^3z_d^3 \\&+ (A_{bb} - A_{aa} + A_{da} - A_{db} + \bar{A}_{aa} - \bar{A}_{bb} - \bar{A}_{ca} + \bar{A}_{cb})z_a^2z_b^2z_c^3z_d + (A_{ab} - A_{db})z_az_b^4z_c^2z_d\\&
+ (A_{aa} - A_{bb} - A_{ca} + A_{cb} - \bar{A}_{aa} + \bar{A}_{bb} + \bar{A}_{da} - \bar{A}_{db})z_a^2z_b^2z_cz_d^3    + (A_{bc} - A_{dc})z_a^2z_bz_c^4z_d\\& + (A_{aa} - A_{ba} + A_{bd} - A_{dd} - \bar{A}_{aa} + \bar{A}_{ca} - \bar{A}_{cd} + \bar{A}_{dd} + i\omega_a - i\omega_d)z_a^2z_bz_c^3z_d^2+ (\bar{A}_{db} - \bar{A}_{cb})z_a^2z_c^3z_d^3\\&
+ (A_{ba} - A_{aa} - A_{bc} + A_{cc} + \bar{A}_{aa} - \bar{A}_{da} - \bar{A}_{cc} + \bar{A}_{dc})z_a^2z_bz_c^2z_d^3 + (A_{cd} - A_{bd})z_a^2z_bz_cz_d^4\\&
+ (A_{cb} - A_{ab})z_az_b^4z_cz_d^2 + (A_{ac} - A_{ab} + A_{db} - A_{dc} - \bar{A}_{ba} + \bar{A}_{ca} + \bar{A}_{bd} - \bar{A}_{cd})z_az_b^3z_c^3z_d
\\&+ (A_{ad} - A_{ac} + A_{cc} - A_{dd} + \bar{A}_{bc} - \bar{A}_{bd} - \bar{A}_{cc} + \bar{A}_{dd})z_az_b^3z_c^2z_d^2 + (A_{ac} - A_{bc})z_az_bz_c^4z_d^2\\&
+ (A_{ab} - A_{ad} - A_{cb} + A_{cd} + \bar{A}_{ba} - \bar{A}_{bc} - \bar{A}_{da} + \bar{A}_{dc})z_az_b^3z_cz_d^3  
+ (A_{dc} - A_{ac})z_az_b^2z_c^4z_d\\& + (A_{ab} - A_{ad} - A_{bb} + A_{dd} + \bar{A}_{bb} - \bar{A}_{cb} + \bar{A}_{cd} - \bar{A}_{dd})z_az_b^2z_c^3z_d^2 + (A_{ad} - A_{cd})z_az_b^2z_cz_d^4\\&
+ (A_{ac} - A_{ab} + A_{bb} - A_{cc} - \bar{A}_{bb} + \bar{A}_{cc} + \bar{A}_{db} - \bar{A}_{dc})z_az_b^2z_c^2z_d^3 + (A_{bd} - A_{ad})z_az_bz_c^2z_d^4 \\&
 + (A_{ad} - A_{ac} + A_{bc} - A_{bd} - \bar{A}_{ca} + \bar{A}_{cb} + \bar{A}_{da} - \bar{A}_{db})z_az_bz_c^3z_d^3
+ (\bar{A}_{ba} - \bar{A}_{ca})z_b^3z_c^3z_d^2 \\&+ (\bar{A}_{da} - \bar{A}_{ba})z_b^3z_c^2z_d^3
+ (\bar{A}_{ca} - \bar{A}_{da})z_b^2z_c^3z_d^3)/(z_az_bz_cz_d)\,,
\end{align*}
where $1/z_az_bz_cz_d$ is not zero. The monomials $z_a^{u_1} z_b^{u_2} z_c^{u_3} z_d^{u_4}$ (with $u_1, u_2, u_3, u_4 \in\{0,1,2,3,4\}$ such that $u_1+u_2+u_3+u_4 = 8$) are independent and therefore, all the coefficients must be zero, yielding a linear system of 42 complex equations:
\begin{align*}
\begin{aligned}
A_{ca} - A_{da} &= 0\phantom{\bar{A}}\\ 
A_{da} - A_{ba} &= 0\phantom{\bar{A}}\\ 
A_{ba} - A_{ca} &= 0\phantom{\bar{A}}\\ 
A_{db} - A_{cb} &= 0\phantom{\bar{A}}\\
A_{ab} - A_{db} &= 0\phantom{\bar{A}}\\
A_{cb} - A_{ab} &= 0\phantom{\bar{A}}\\
A_{bc} - A_{dc} &= 0\phantom{\bar{A}}\\
A_{dc} - A_{ac} &= 0\phantom{\bar{A}}\\
A_{ac} - A_{bc} &= 0\phantom{\bar{A}}\\
A_{cd} - A_{bd} &= 0\phantom{\bar{A}}\\
A_{ad} - A_{cd} &= 0\phantom{\bar{A}}\\
A_{bd} - A_{ad} &= 0\phantom{\bar{A}}
\end{aligned}
\quad\quad
\begin{aligned}
\bar{A}_{ba} - \bar{A}_{ca} &= 0\\
\bar{A}_{da} - \bar{A}_{ba} &= 0\\
\bar{A}_{ca} - \bar{A}_{da} &= 0\\
\bar{A}_{ab} - \bar{A}_{db} &= 0\\
\bar{A}_{cb} - \bar{A}_{ab} &= 0\\
\bar{A}_{db} - \bar{A}_{cb} &= 0\\
\bar{A}_{dc} - \bar{A}_{ac} &= 0\\
\bar{A}_{bc} - \bar{A}_{dc} &= 0\\
\bar{A}_{ac} - \bar{A}_{bc} &= 0\\
\bar{A}_{ad} - \bar{A}_{cd} &= 0\\ 
\bar{A}_{bd} - \bar{A}_{ad} &= 0\\ 
\bar{A}_{cd} - \bar{A}_{bd} &= 0
\end{aligned}
\quad\quad
\begin{aligned}
A_{cb} - A_{ca} + A_{da} - A_{db} - \bar{A}_{ac} + \bar{A}_{ad} + \bar{A}_{bc} - \bar{A}_{bd} &= 0\\ 
A_{cc} - A_{bb} + A_{db} - A_{dc} - \bar{A}_{ab} + \bar{A}_{ac} + \bar{A}_{bb} - \bar{A}_{cc} &= 0\\
A_{bb} - A_{cb} + A_{cd} - A_{dd} + \bar{A}_{ab} - \bar{A}_{bb} - \bar{A}_{ad} + \bar{A}_{dd} &= 0\\
A_{ba} - A_{bc} - A_{da} + A_{dc} + \bar{A}_{ab} - \bar{A}_{ad} - \bar{A}_{cb} + \bar{A}_{cd} &= 0\\
A_{bc} - A_{bd} - A_{cc} + A_{dd} - \bar{A}_{ac} + \bar{A}_{ad} + \bar{A}_{cc} - \bar{A}_{dd} &= 0\\
A_{bd} - A_{ba} + A_{ca} - A_{cd} - \bar{A}_{ab} + \bar{A}_{ac} + \bar{A}_{db} - \bar{A}_{dc} &= 0\\
A_{aa} - A_{cc} - A_{da} + A_{dc} - \bar{A}_{aa} + \bar{A}_{ba} - \bar{A}_{bc} + \bar{A}_{cc} &= 0\\
A_{ca} - A_{aa} - A_{cd} + A_{dd} + \bar{A}_{aa} - \bar{A}_{ba} + \bar{A}_{bd} - \bar{A}_{dd} &= 0\\
A_{bb} - A_{aa} + A_{da} - A_{db} + \bar{A}_{aa} - \bar{A}_{bb} - \bar{A}_{ca} + \bar{A}_{cb} &= 0\\
A_{aa} - A_{bb} - A_{ca} + A_{cb} - \bar{A}_{aa} + \bar{A}_{bb} + \bar{A}_{da} - \bar{A}_{db} &= 0\\
A_{aa} - A_{ba} + A_{bd} - A_{dd} - \bar{A}_{aa} + \bar{A}_{ca} - \bar{A}_{cd} + \bar{A}_{dd} &= 0\\
A_{ba} - A_{aa} - A_{bc} + A_{cc} + \bar{A}_{aa} - \bar{A}_{da} - \bar{A}_{cc} + \bar{A}_{dc} &= 0\\
A_{ac} - A_{ab} + A_{db} - A_{dc} - \bar{A}_{ba} + \bar{A}_{ca} + \bar{A}_{bd} - \bar{A}_{cd} &= 0\\
A_{ad} - A_{ac} + A_{cc} - A_{dd} + \bar{A}_{bc} - \bar{A}_{bd} - \bar{A}_{cc} + \bar{A}_{dd} &= 0\\
A_{ab} - A_{ad} - A_{cb} + A_{cd} + \bar{A}_{ba} - \bar{A}_{bc} - \bar{A}_{da} + \bar{A}_{dc} &= 0\\
A_{ab} - A_{ad} - A_{bb} + A_{dd} + \bar{A}_{bb} - \bar{A}_{cb} + \bar{A}_{cd} - \bar{A}_{dd} &= 0\\
A_{ac} - A_{ab} + A_{bb} - A_{cc} - \bar{A}_{bb} + \bar{A}_{cc} + \bar{A}_{db} - \bar{A}_{dc} &= 0\\
A_{ad} - A_{ac} + A_{bc} - A_{bd} - \bar{A}_{ca} + \bar{A}_{cb} + \bar{A}_{da} - \bar{A}_{db} &= 0
\end{aligned}
\end{align*}
Half of the equations with two terms are complex conjugate of the other half and are thus redundant. This readily leads to the necessity of the first condition:
\begin{equation}\label{eq:outgoing_abcd}
\begin{aligned}
    A_{ba} = A_{ca} &= A_{da} =: \mathcal{A}_a\\
    A_{ab} = A_{cb} &= A_{db} =: \mathcal{A}_b\\
    A_{ac} = A_{bc} &= A_{dc} =: \mathcal{A}_c\\
    A_{ad} = A_{bd} &= A_{cd} =: \mathcal{A}_d.
\end{aligned}
\end{equation}
 The equations with more than two terms can be rearranged as
{\footnotesize\begin{align*}
\begin{aligned}
A_{cb} - A_{db} + A_{da} - A_{ca} + \bar{A}_{bc} - \bar{A}_{ac} + \bar{A}_{ad} - \bar{A}_{bd} &= 0\\ 
A_{cc} - A_{bb} + A_{db} - A_{dc} + \bar{A}_{bb} - \bar{A}_{ab} + \bar{A}_{ac} - \bar{A}_{cc} &= 0\\
A_{bb} - A_{cb} + A_{cd} - A_{dd} + \bar{A}_{ab} - \bar{A}_{bb} + \bar{A}_{dd} - \bar{A}_{ad} &= 0\\
A_{ba} - A_{da} + A_{dc} - A_{bc} + \bar{A}_{ab} - \bar{A}_{cb} + \bar{A}_{cd} - \bar{A}_{ad} &= 0\\
A_{bc} - A_{cc} + A_{dd} - A_{bd} + \bar{A}_{cc} - \bar{A}_{ac} + \bar{A}_{ad} - \bar{A}_{dd} &= 0\\
A_{bd} - A_{cd} + A_{ca} - A_{ba} + \bar{A}_{ac} - \bar{A}_{dc} + \bar{A}_{db} - \bar{A}_{ab} &= 0\\
A_{aa} - A_{da} + A_{dc} - A_{cc} + \bar{A}_{ba} - \bar{A}_{aa} + \bar{A}_{cc} - \bar{A}_{bc} &= 0\\
A_{ca} - A_{aa} + A_{dd} - A_{cd} + \bar{A}_{aa} - \bar{A}_{ba} + \bar{A}_{bd} - \bar{A}_{dd} &= 0\\
A_{bb} - A_{db} + A_{da} - A_{aa} + \bar{A}_{aa} - \bar{A}_{ca} + \bar{A}_{cb} - \bar{A}_{bb} &= 0\\
\end{aligned}
\qquad
\begin{aligned}
A_{aa} - A_{ca} + A_{cb} - A_{bb} + \bar{A}_{bb} - \bar{A}_{db} + \bar{A}_{da} - \bar{A}_{aa} &= 0\\
A_{aa} - A_{ba} + A_{bd} - A_{dd} + \bar{A}_{ca} - \bar{A}_{aa} + \bar{A}_{dd} - \bar{A}_{cd} &= 0\\
A_{ba} - A_{aa} + A_{cc} - A_{bc} + \bar{A}_{aa} - \bar{A}_{da} + \bar{A}_{dc} - \bar{A}_{cc} &= 0\\
A_{ac} - A_{dc} + A_{db} - A_{ab} + \bar{A}_{ca} - \bar{A}_{ba} + \bar{A}_{bd} - \bar{A}_{cd} &= 0\\
A_{ad} - A_{dd} + A_{cc} - A_{ac} + \bar{A}_{bc} - \bar{A}_{cc} + \bar{A}_{dd} - \bar{A}_{bd} &= 0\\
A_{ab} - A_{cb} + A_{cd} - A_{ad} + \bar{A}_{ba} - \bar{A}_{da} + \bar{A}_{dc} - \bar{A}_{bc} &= 0\\
A_{ab} - A_{bb} + A_{dd} - A_{ad} + \bar{A}_{bb} - \bar{A}_{cb} + \bar{A}_{cd} - \bar{A}_{dd} &= 0\\
A_{ac} - A_{cc} + A_{bb} - A_{ab} + \bar{A}_{cc} - \bar{A}_{dc} + \bar{A}_{db} - \bar{A}_{bb} &= 0\\
A_{ad} - A_{bd} + A_{bc} - A_{ac} + \bar{A}_{da} - \bar{A}_{ca} + \bar{A}_{cb} - \bar{A}_{db} &= 0\,.
\end{aligned}
\end{align*}}
Using Eqs.~\eqref{eq:outgoing_abcd} cancels every equation from the latter system that does not involve a self-interaction term $A_{jj}$ for some $j\in\{a,b,c,d\}$. This leads to
{\footnotesize\begin{align*}
\begin{aligned}
A_{cc} - A_{dc} + A_{db} - A_{bb} + \bar{A}_{bb} - \bar{A}_{ab} + \bar{A}_{ac} - \bar{A}_{cc} &= 0\\
A_{bb} - A_{cb} + A_{cd} - A_{dd} + \bar{A}_{ab} - \bar{A}_{bb} + \bar{A}_{dd} - \bar{A}_{ad} &= 0\\
A_{bc} - A_{cc} + A_{dd} - A_{bd} + \bar{A}_{cc} - \bar{A}_{ac} + \bar{A}_{ad} - \bar{A}_{dd} &= 0\\
A_{aa} - A_{da} + A_{dc} - A_{cc} + \bar{A}_{ba} - \bar{A}_{aa} + \bar{A}_{cc} - \bar{A}_{bc} &= 0\\
A_{ca} - A_{aa} + A_{dd} - A_{cd} + \bar{A}_{aa} - \bar{A}_{ba} + \bar{A}_{bd} - \bar{A}_{dd} &= 0\\
A_{bb} - A_{db} + A_{da} - A_{aa} + \bar{A}_{aa} - \bar{A}_{ca} + \bar{A}_{cb} - \bar{A}_{bb} &= 0
\end{aligned}
\qquad
\begin{aligned}
A_{aa} - A_{ca} + A_{cb} - A_{bb} + \bar{A}_{bb} - \bar{A}_{db} + \bar{A}_{da} - \bar{A}_{aa} &= 0\\
A_{aa} - A_{ba} + A_{bd} - A_{dd} + \bar{A}_{ca} - \bar{A}_{aa} + \bar{A}_{dd} - \bar{A}_{cd} &= 0\\
A_{ba} - A_{aa} + A_{cc} - A_{bc} + \bar{A}_{aa} - \bar{A}_{da} + \bar{A}_{dc} - \bar{A}_{cc} &= 0\\
A_{ad} - A_{dd} + A_{cc} - A_{ac} + \bar{A}_{bc} - \bar{A}_{cc} + \bar{A}_{dd} - \bar{A}_{bd} &= 0\\
A_{ab} - A_{bb} + A_{dd} - A_{ad} + \bar{A}_{bb} - \bar{A}_{cb} + \bar{A}_{cd} - \bar{A}_{dd} &= 0\\
A_{ac} - A_{cc} + A_{bb} - A_{ab} + \bar{A}_{cc} - \bar{A}_{dc} + \bar{A}_{db} - \bar{A}_{bb} &= 0\,.
\end{aligned}
\end{align*}}
Half of these equations are the complex conjugate of the other half and are thus redundant:
\begin{align*}
A_{cc} - A_{dc} + A_{db} - A_{bb} + \bar{A}_{bb} - \bar{A}_{ab} + \bar{A}_{ac} - \bar{A}_{cc} &= 0\\
A_{bb} - A_{cb} + A_{cd} - A_{dd} + \bar{A}_{ab} - \bar{A}_{bb} + \bar{A}_{dd} - \bar{A}_{ad} &= 0\\
A_{bc} - A_{cc} + A_{dd} - A_{bd} + \bar{A}_{cc} - \bar{A}_{ac} + \bar{A}_{ad} - \bar{A}_{dd} &= 0\\
A_{aa} - A_{da} + A_{dc} - A_{cc} + \bar{A}_{ba} - \bar{A}_{aa} + \bar{A}_{cc} - \bar{A}_{bc} &= 0\\
A_{ca} - A_{aa} + A_{dd} - A_{cd} + \bar{A}_{aa} - \bar{A}_{ba} + \bar{A}_{bd} - \bar{A}_{dd} &= 0\\
A_{bb} - A_{db} + A_{da} - A_{aa} + \bar{A}_{aa} - \bar{A}_{ca} + \bar{A}_{cb} - \bar{A}_{bb} &= 0\,.
\end{align*}
Using Eqs.~\eqref{eq:outgoing_abcd} once again to express the latter equations in terms of $\mathcal{A}_a,\mathcal{A}_b,\mathcal{A}_c,\mathcal{A}_d$ makes it evident that they are constraints for the self-interaction parameters:
\begin{align*}
(A_{cc} - \bar{A}_{cc}) - (A_{bb} - \bar{A}_{bb}) &= (\mathcal{A}_c - \bar{\mathcal{A}}_c) - (\mathcal{A}_b - \bar{\mathcal{A}}_b)\\
(A_{bb} - \bar{A}_{bb}) - (A_{dd} - \bar{A}_{dd}) &= (\mathcal{A}_b - \bar{\mathcal{A}}_b) - (\mathcal{A}_d - \bar{\mathcal{A}}_d)\\
(A_{dd} - \bar{A}_{dd}) - (A_{cc} - \bar{A}_{cc}) &= (\mathcal{A}_d - \bar{\mathcal{A}}_d) - (\mathcal{A}_c - \bar{\mathcal{A}}_c)\\
(A_{aa} - \bar{A}_{aa}) - (A_{cc} - \bar{A}_{cc}) &= (\mathcal{A}_a - \bar{\mathcal{A}}_a) - (\mathcal{A}_c - \bar{\mathcal{A}}_c)\\
(A_{dd} - \bar{A}_{dd}) - (A_{aa} - \bar{A}_{aa}) &= (\mathcal{A}_d - \bar{\mathcal{A}}_d) - (\mathcal{A}_a - \bar{\mathcal{A}}_a)\\
(A_{bb} - \bar{A}_{bb}) - (A_{aa} - \bar{A}_{aa}) &= (\mathcal{A}_b - \bar{\mathcal{A}}_b) - (\mathcal{A}_a - \bar{\mathcal{A}}_a)\,.
\end{align*}
The equations can be written in real form as
\begin{align*}
\Imag(A_{cc} - A_{bb}) &= \Imag(\mathcal{A}_c - \mathcal{A}_b)\\
\Imag(A_{bb} - A_{dd}) &= \Imag(\mathcal{A}_b - \mathcal{A}_d)\\
\Imag(A_{dd} - A_{cc}) &= \Imag(\mathcal{A}_d - \mathcal{A}_c)\\
\Imag(A_{aa} - A_{cc}) &= \Imag(\mathcal{A}_a - \mathcal{A}_c)\\
\Imag(A_{dd} - A_{aa}) &= \Imag(\mathcal{A}_d - \mathcal{A}_a)\\
\Imag(A_{bb} - A_{aa}) &= \Imag(\mathcal{A}_b - \mathcal{A}_a)\,.
\end{align*}
One observes that only three of them are independent, which leads to
\begin{align*}
\Imag(A_{bb} - A_{aa}) = \Imag(\mathcal{A}_b - \mathcal{A}_a)\,,\quad
\Imag(A_{cc} - A_{aa}) = \Imag(\mathcal{A}_c - \mathcal{A}_a)\,,\quad
\Imag(A_{dd} - A_{aa}) = \Imag(\mathcal{A}_d - \mathcal{A}_a)\,,
\end{align*}
and
\begin{align*}
\omega_b - \omega_a = 2\Imag(\mathcal{A}_b - \mathcal{A}_a)\,,\qquad
\omega_c - \omega_a = 2\Imag(\mathcal{A}_c - \mathcal{A}_a)\,,\qquad
\omega_d - \omega_a = 2\Imag(\mathcal{A}_d - \mathcal{A}_a)\,.
\end{align*}
Combining these finally provides the third condition
\begin{align*}
    \omega_a - 2\Imag(\mathcal{A}_a) = \omega_b - 2\Imag(\mathcal{A}_b) = \omega_c - 2\Imag(\mathcal{A}_c) = \omega_d - 2\Imag(\mathcal{A}_d)\,.
\end{align*}
Altogether, the three conditions are necessary and sufficient to have $\mathcal{K}[c_{abcd}](\bm{z}) = 0$ and the proof of the theorem is complete.
\end{proof}
\begin{remark}
    It may be surprising to observe that the natural frequencies do not have to be identical when nontrivial phase lags are present. Yet, as shown in subsection~\ref{subsec:basix_examples_thm3}, the third condition in fact guarantees that the oscillators whose positions on the unit circle participate in a conserved cross-ratio have the same effective frequency. Also, in terms of the original parameters, the third condition is equivalent to 
\begin{align*}
    \omega_a - W_{k_a a}\sin\alpha_{k_a a} = \omega_b -  W_{k_b b}\sin\alpha_{k_b b} = \omega_c -  W_{k_c c}\sin\alpha_{k_c c} = \omega_d -  W_{k_d d}\sin\alpha_{k_d d}\,,
\end{align*}
where $k_j$ takes any value within $\{a,b,c,d\}\setminus\{j\}$ for $j\in\{a,b,c,d\}$.
\end{remark}
\begin{remark}
    Note that there is no restriction on the outgoing edges from the vertices involved in conserved cross-ratios to the vertices not involved in a conserved cross-ratio. Indeed, although the conditions of the theorem make the equations for the oscillators admitting a conserved cross-ratio identical, their initial conditions typically differ and their contributions within the whole network are generally heterogeneous.
\end{remark}
\begin{remark}
    If the phases are not equal at time $t = 0$ and are involved in a conserved cross-ratio, the cross-ratio is not zero. Since it is conserved, the cross-ratio remains non-zero throughout time and the phases cannot cross (become equal) for all finite values of $t$~\cite{Lohe2017}.
\end{remark}

\subsection{Basic examples for Theorem 3}
\label{subsec:basix_examples_thm3}
One of the simplest, but instructive, examples with a conserved cross-ratio is the following.
\begin{example}
    Consider a graph of $N= 5$ vertices with complex weight matrix 
    \begin{align*}
    A = \begin{pmatrix}
        i\omega_1/2 & \mathcal{A}_2 & \mathcal{A}_3 & \mathcal{A}_4 & \mathcal{A}_5\\
        \mathcal{A}_1 & i\omega_2 /2 & \mathcal{A}_3 & \mathcal{A}_4 & \mathcal{A}_5\\
        \mathcal{A}_1 & \mathcal{A}_2 & i\omega_3/2 & \mathcal{A}_4 & \mathcal{A}_5\\
        \mathcal{A}_1 & \mathcal{A}_2 & \mathcal{A}_3 & i\omega_4/2 & \mathcal{A}_5\\
        A_{51} & A_{52} & A_{53} & A_{54} & i\omega/2
    \end{pmatrix}\,,
\end{align*}
where $\omega$ and $\omega_1$ are fixed to arbitrary real values while 
\begin{align*}
    \omega_j = \omega_1 + 2\Imag(\mathcal{A}_j - \mathcal{A}_1)\,,\quad j\in \{2,3,4\}\,.
\end{align*}
By construction, the effective natural frequency of oscillators 2, 3, 4 is $\Omega = \omega_1 - 2\Imag(\mathcal{A}_1)$. In fact, one can readily verify that this yields identical equations for oscillators 1,2,3,4 (although they have different contributions to oscillator~5): $\dot{z}_j = \rho(\bm{z})  + i \Omega z_j - \overline{\rho(\bm{z})} z_j^2$ where $j\in\{1,2,3,4\}$ and $\rho(\bm{z}) = \sum_{k=1}^N \mathcal{A}_{ k}z_k$. In such a case, from Thm.~\ref{thm:cte_mvt_kuramoto_graphe} and Lem.~\ref{lem:independence_cross_ratios}, there is only one functionally independent cross-ratio, say
\begin{align*}
    c_{1234}(z) = \frac{(z_{3}-z_{1})(z_{4}-z_{2})}{(z_{3}-z_{2})(z_{4}-z_{1})}\,,
\end{align*}
that is conserved.
\end{example}
Let's now present a more general example in terms of network structure.
\begin{example}
    Consider a graph with $N = 13$ vertices with complex weight matrix
    {\footnotesize\begin{align*}
        A = \begin{pmatrix}
        i\omega_1/2 & \mathcal{A}_{1,2} & \mathcal{A}_{1,3} & \mathcal{A}_{1,4} & \mathcal{A}_{1,5} & \mathcal{A}_{1,6}& \mathcal{A}_{1,7} & \mathcal{A}_{1,8}& \mathcal{A}_{1,9} & \mathcal{A}_{1,10} & \mathcal{A}_{1,11} & \mathcal{A}_{1,12} & \mathcal{A}_{1,13}\\
        \mathcal{A}_{1,1} & i\omega_2/2 & \mathcal{A}_{1,3} & \mathcal{A}_{1,4} & \mathcal{A}_{1,5} & \mathcal{A}_{1,6}& \mathcal{A}_{1,7} & \mathcal{A}_{1,8}& \mathcal{A}_{1,9} & \mathcal{A}_{1,10} & \mathcal{A}_{1,11} & \mathcal{A}_{1,12} & \mathcal{A}_{1,13}\\
        \mathcal{A}_{1,1} & \mathcal{A}_{1,2} & i\omega_3/2 & \mathcal{A}_{1,4} & \mathcal{A}_{1,5} & \mathcal{A}_{1,6}& \mathcal{A}_{1,7} & \mathcal{A}_{1,8}& \mathcal{A}_{1,9} & \mathcal{A}_{1,10} & \mathcal{A}_{1,11} & \mathcal{A}_{1,12} & \mathcal{A}_{1,13}\\
        \mathcal{A}_{1,1} & \mathcal{A}_{1,2} & \mathcal{A}_{1,3} & i\omega_4/2 & \mathcal{A}_{1,5} & \mathcal{A}_{1,6}& \mathcal{A}_{1,7} & \mathcal{A}_{1,8}& \mathcal{A}_{1,9} & \mathcal{A}_{1,10} & \mathcal{A}_{1,11} & \mathcal{A}_{1,12} & \mathcal{A}_{1,13}\\
        \hline 
        \mathcal{A}_{2,1} & \mathcal{A}_{2,2} & \mathcal{A}_{2,3} & \mathcal{A}_{2,4} & i\omega_5/2 & \mathcal{A}_{2,6}& \mathcal{A}_{2,7} & \mathcal{A}_{2,8}& \mathcal{A}_{2,9} & \mathcal{A}_{2,10} & \mathcal{A}_{2,11} & \mathcal{A}_{2,12} & \mathcal{A}_{2,13}\\
        \mathcal{A}_{2,1} & \mathcal{A}_{2,2} & \mathcal{A}_{2,3} & \mathcal{A}_{2,4} & \mathcal{A}_{2,5} & i\omega_6/2 & \mathcal{A}_{2,7} & \mathcal{A}_{2,8}& \mathcal{A}_{2,9} & \mathcal{A}_{2,10} & \mathcal{A}_{2,11} & \mathcal{A}_{2,12} & \mathcal{A}_{2,13}\\
        \mathcal{A}_{2,1} & \mathcal{A}_{2,2} & \mathcal{A}_{2,3} & \mathcal{A}_{2,4} & \mathcal{A}_{2,5} & \mathcal{A}_{2,6}& i\omega_7/2 & \mathcal{A}_{2,8}& \mathcal{A}_{2,9} & \mathcal{A}_{2,10} & \mathcal{A}_{2,11} & \mathcal{A}_{2,12} & \mathcal{A}_{2,13}\\
        \mathcal{A}_{2,1} & \mathcal{A}_{2,2} & \mathcal{A}_{2,3} & \mathcal{A}_{2,4} & \mathcal{A}_{2,5} & \mathcal{A}_{2,6}& \mathcal{A}_{2,7} & i\omega_8/2& \mathcal{A}_{2,9} & \mathcal{A}_{2,10} & \mathcal{A}_{2,11} & \mathcal{A}_{2,12} & \mathcal{A}_{2,13}\\
        \mathcal{A}_{2,1} & \mathcal{A}_{2,2} & \mathcal{A}_{2,3} & \mathcal{A}_{2,4} & \mathcal{A}_{2,5} & \mathcal{A}_{2,6}& \mathcal{A}_{2,7} & \mathcal{A}_{2,8}& i\omega_9/2 & \mathcal{A}_{2,10} & \mathcal{A}_{2,11} & \mathcal{A}_{2,12} & \mathcal{A}_{2,13}\\
        \mathcal{A}_{2,1} & \mathcal{A}_{2,2} & \mathcal{A}_{2,3} & \mathcal{A}_{2,4} & \mathcal{A}_{2,5} & \mathcal{A}_{2,6}& \mathcal{A}_{2,7} & \mathcal{A}_{2,8}& \mathcal{A}_{2,9} & i\omega_{10}/2 & \mathcal{A}_{2,11} & \mathcal{A}_{2,12} & \mathcal{A}_{2,13}\\
        \hline
        A_{11,1} & A_{11,2} & A_{11,3} & A_{11,4} & A_{11,5}& A_{11,6}& A_{11,7}& A_{11,8}& A_{11,9}& A_{11,10}& i\omega_{11}/2& A_{11,12}& A_{11,13}\\
        A_{12,1} & A_{12,2} & A_{12,3} & A_{12,4} & A_{12,5}& A_{12,6}& A_{12,7}& A_{12,8}& A_{12,9}& A_{12,10}& A_{12,11}& i\omega_{12}/2& A_{12,13}\\
        A_{13,1} & A_{13,2} & A_{13,3} & A_{13,4} & A_{13,5}& A_{13,6}& A_{13,7}& A_{13,8}& A_{13,9}& A_{13,10}& A_{13,11}& A_{13,12}& i\omega_{13}/2
    \end{pmatrix}\,,
    \end{align*}}
    where
    \begin{align*}
    \omega_j = \begin{cases} 
    \text{arbitrary real number} & \text{if } j\in\{1, 5, 11, 12, 13\},\\
    \omega_1 + 2\Imag(\mathcal{A}_{1,j} - \mathcal{A}_{1,1})& \text{if } j\in \{2,3,4\}, \\
    \omega_5 + 2\Imag(\mathcal{A}_{2,j} - \mathcal{A}_{2,5}) & \text{if } j\in \{6,7,8,9,10\}.
\end{cases}
\end{align*}
The effective natural frequencies within each partially integrable part are $\Omega_1 = \omega_1 - 2\Imag(\mathcal{A}_{1,1})$ and $\Omega_2 = \omega_5 - 2\Imag(\mathcal{A}_{2,5})$. Following Thm.~\ref{thm:cte_mvt_kuramoto_graphe} and Lem.~\ref{lem:independence_cross_ratios}, the following cross-ratios are functionally independent constants of motion:
    \begin{align*}
        c_{1,2,3,4}(z) &= \frac{(z_{3}-z_{1})(z_{4}-z_{2})}{(z_{3}-z_{2})(z_{4}-z_{1})}
    \end{align*}
related to the first four oscillators and
    \begin{align*}
        c_{5,6,7,8}(z) = \frac{(z_{7}-z_{5})(z_{8}-z_{6})}{(z_{7}-z_{6})(z_{8}-z_{5})}\,,\quad
        c_{6,7,8,9}(z) = \frac{(z_{8}-z_{6})(z_{9}-z_{7})}{(z_{8}-z_{7})(z_{9}-z_{6})}\,,\quad
        c_{7,8,9,10}(z) = \frac{(z_{9}-z_{7})(z_{10}-z_{8})}{(z_{9}-z_{8})(z_{10}-z_{7})}
    \end{align*}
for oscillators 5 to 10.
\end{example}

\begin{figure}[t]
    \centering
    \includegraphics[width=0.9\linewidth]{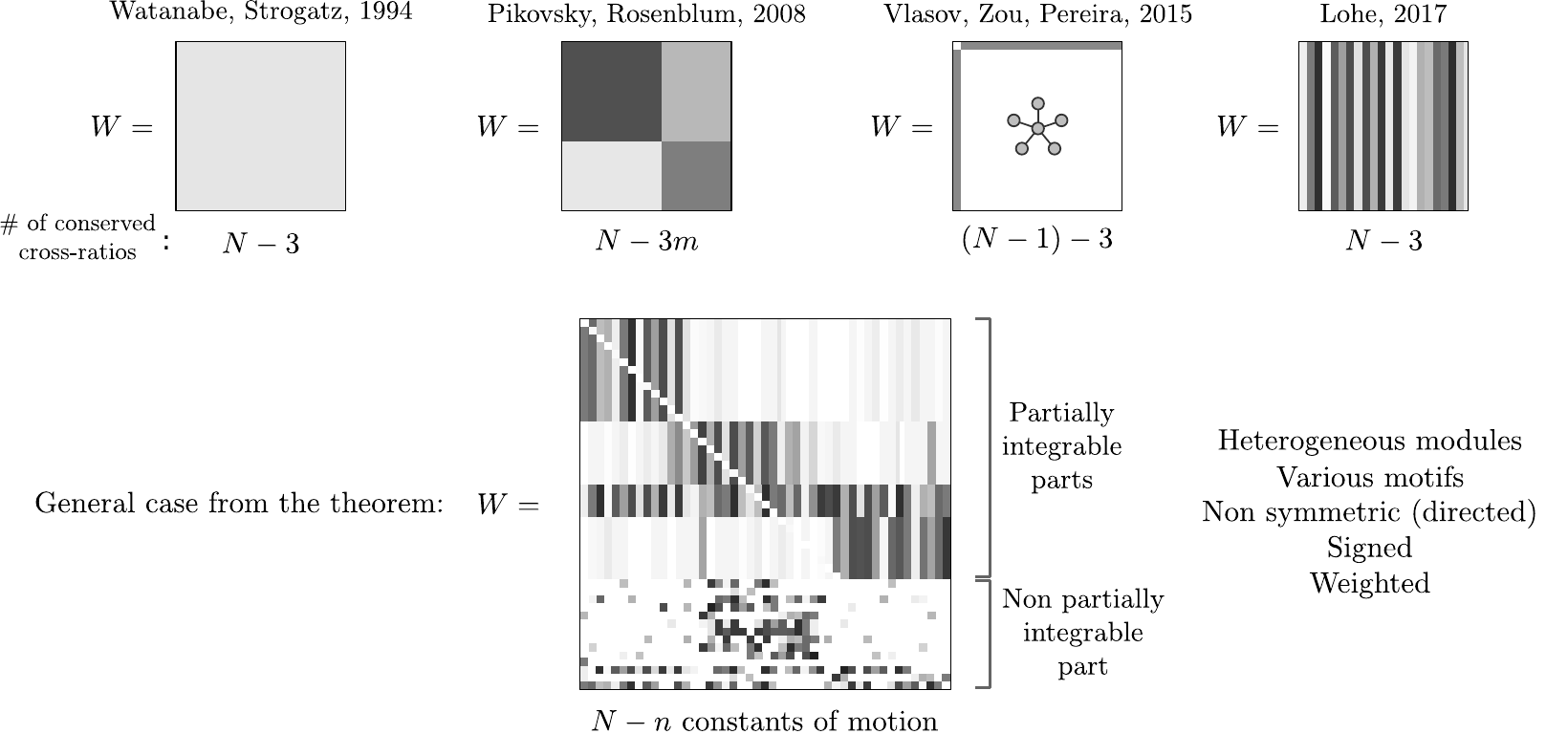}
    \caption[Simplified history of weight matrices in WS theory]{\textbf{Simplified history of weight matrices in WS theory}. Previous results on the possible forms of weight matrices allowing conserved cross-ratios and a generic example of weight matrix satisfying the conditions of Thm.~\ref{thm:cte_mvt_kuramoto_graphe}, without considering natural frequencies and phase lags for simplicity.}
    \label{fig:complex_weight_matrix}
\end{figure}

In Fig.~\ref{fig:complex_weight_matrix}, we illustrate the weight matrix of a more general network of Kuramoto oscillators with conserved cross-ratios.

\subsection{Corollaries of Theorem 3}
\label{SIsubsec:corollaries_thm3}
In this subsection, we provide some consequences of Thm.~\ref{thm:cte_mvt_kuramoto_graphe}. First, the theorem readily gives the necessary and sufficient conditions to have $N - 3$ constants of motion having the form of cross-ratios. The sufficiency is known from the excellent work of Lohe~\cite{Lohe2017, Lohe2019}. The following corollary formalizes this result while adding the necessity of the conditions.
\begin{corollary}\label{cor:max_integrals_of_motion}
    The $N$-dimensional Kuramoto model on a graph with complex weight matrix $A$ admits the maximum number of functionally independent cross-ratios as constants of motion, namely $N-3$, if and only if the following two conditions are satisfied:
    \begin{enumerate}
        \item $A_{j\ell} = A_{k\ell} =: \mathcal{A}_\ell$ for all $\ell\in\{1,...,N\}$ and for all pairs $(j,k)$ with $j,k\in\{1,...,N\}$ and $k,\ell\neq j\,;$
        \item $\omega_j - 2\Imag(\mathcal{A}_j) = \omega_k - 2\Imag(\mathcal{A}_k)$ for all pairs $(j,k)$ with $j,k\in\{1,...,N\}$ and $k\neq j\,.$
    \end{enumerate}
\end{corollary}
\begin{proof}
    ($\Leftarrow$) Assume that the first condition of the corollary holds. Then, conditions 1 and 2 of Thm.~\ref{thm:cte_mvt_kuramoto_graphe} are automatically satisfied. Moreover, if the second condition of the corollary also holds, then condition 3 of Thm.~\ref{thm:cte_mvt_kuramoto_graphe} is satisfied as well. 
    Therefore, a Kuramoto system that satisfies the two conditions of the corollary admits each cross-ratio $c_{abcd}$ as a constant of motion.
    Now, according to Lem.~\ref{lem:independence_cross_ratios}, only $N-3$ of them can be functionally independent, meaning that $N-3$ is the maximal number of functionally independent cross-ratios being constants of motion. 

    ($\Rightarrow$) We will prove the contrapositive. Let $(j, k)$, where $j,k\in\{1, \cdots, N\}$ and $k\neq j$, be a pair for which condition 1 or condition 2 of the corollary is not satisfied.
    Without loss of generality, relabel the oscillators in such a way that $j=N-1$ and $k=N$.
    Consider the cross-ratio $c_{(N-3)(N-2)(N-1)N}$.
    This cross-ratio cannot be conserved, because condition 1 of the corollary not being satisfied implies that condition 1 of Thm.~\ref{thm:cte_mvt_kuramoto_graphe} is not satisfied, and 2 of the corollary not being satisfied implies that condition 3 of Thm.~\ref{thm:cte_mvt_kuramoto_graphe} is not satisfied.
    Moreover, by the same reasoning, any cross-ratio involving $N-1$ and $N$ cannot be conserved.
    Consider the $N-4$ cross-ratios in 
    \begin{equation}
        \{c_{1234}, c_{2345}, ..., c_{(N-4)(N-3)(N-2)(N-1)}\}.
    \end{equation}
    According to Lem.~\ref{lem:independence_cross_ratios}, any $c_{abcd}$ with $a,b,c,d\in\{1, \cdots, N-1\}$ is functionally dependent on those $N-4$ cross-ratios, so any additional independent cross-ratio must involve oscillator $N$.
    Since all permutations of the indices of a cross-ratio are functionally dependent, consider without loss of generality that this new independent cross-ratio is $c_{aNbc}$, where $a,b,c\in\{1, \cdots, N-1\}$.
    However, the $N-3$ cross-ratios in
    \begin{equation}
        \{c_{1234}, c_{2345}, ..., c_{(N-4)(N-3)(N-2)(N-1)}, c_{aNbc}\}
    \end{equation}
    cannot all be conserved. Indeed, consider the cross-ratio $c_{(N-1)abc}$, which is dependent on the $N-4$ first cross-ratios. Then, by Eq.~\eqref{ben:multiplication_cr}, 
    \begin{equation}
        c_{(N-1)abc}\,c_{aNbc} = c_{(N-1)Nbc}\,,
    \end{equation}
    but recall that any cross-ratio involving oscillators $N-1$ and $N$ cannot be conserved. Therefore, if either condition 1 or condition 2 is not satisfied, then the model cannot admit $N-3$ conserved cross-ratios.
\end{proof}
From a graph-theoretical perspective, the latter corollary implies that graphs other than the complete graph or the star graph also admit $N-3$ conserved cross-ratios. Consider the following simple example.
\begin{example}
    Consider binary matrices $A$ satisfying the first condition of Corollary~\ref{cor:max_integrals_of_motion}, disregarding the diagonal. There are $2^N$ such matrices, corresponding to all possible binary choices for each of the $N$ columns. Each of these matrices defines a graph. Between them, there are many graph isomorphisms. Starting from the complete graph, for which all matrix elements are equal to 1: if one changes no column, there is 1 possible graph; if one changes a column of ones into a column of zeros, there are $N$ isomorphic graphs; if one changes two such columns, there are $\binom{N}{2}$ isomorphic graphs, etc. Generally, if one changes $k$ columns of ones into columns of zeros, there are $\binom{N}{k}$ isomorphic graphs.  Summing over all the isomorphic graphs yields the total number of possibilities from the binomial theorem $\sum_{k=0}^N \binom{N}{k} = 2^N$. There are thus $N$ non-isomorphic, weakly connected, binary graphs leading to $N - 3$ conserved cross-ratios in the Kuramoto dynamics. The 16 four-vertex graphs and the 32 five-vertex graphs including isomorphisms, all supporting the maximal number of conserved cross-ratios, are presented in Fig.~\ref{fig:motifs_cte_mvt} and Fig.~\ref{fig:max_reduc_5nodes}.
\end{example}
\begin{figure}[h]
    \centering
    \includegraphics[width=0.9\linewidth]{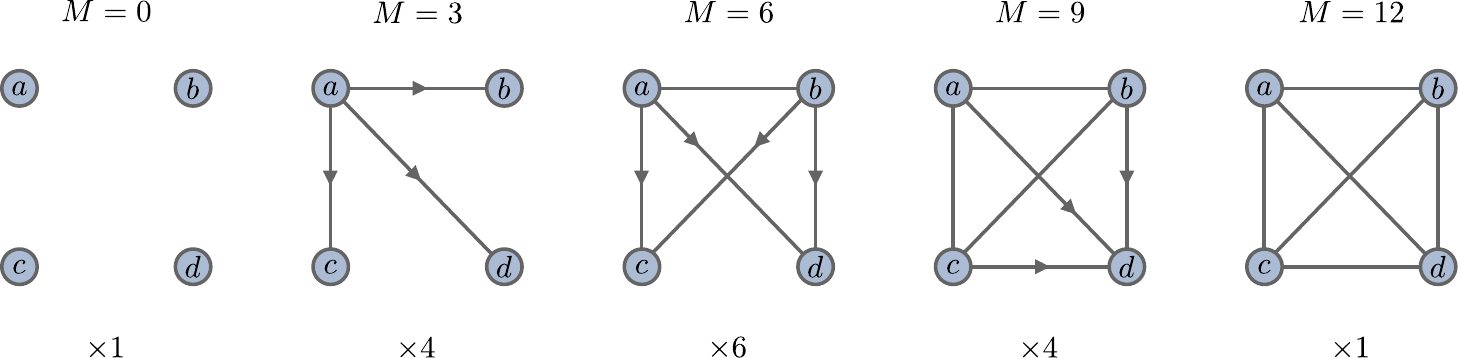}
    \caption[4-vertices motifs admitting a conserved cross-ratio]{\textbf{4-vertices motifs admitting a conserved cross-ratio} (considering that the conditions on the frequencies and the phase lags are satisfied). The number of arcs (oriented edges) $M$ is specified above the graph while the number of isomorphic graphs is specified below the graphs. Note that there are 0, 1, 2, 3, 0 vertices in a ``source subgraph'' respectively for $M = 0, 3, 6, 9, 12$, allowing the possibility of having monomial eigenfunctions if all the conditions of Thm.~\ref{thm:existence_fpmonom} are satisfied.}
    \label{fig:motifs_cte_mvt}
\end{figure}
\begin{figure}[h]
    \centering    
    \includegraphics[width=\linewidth]{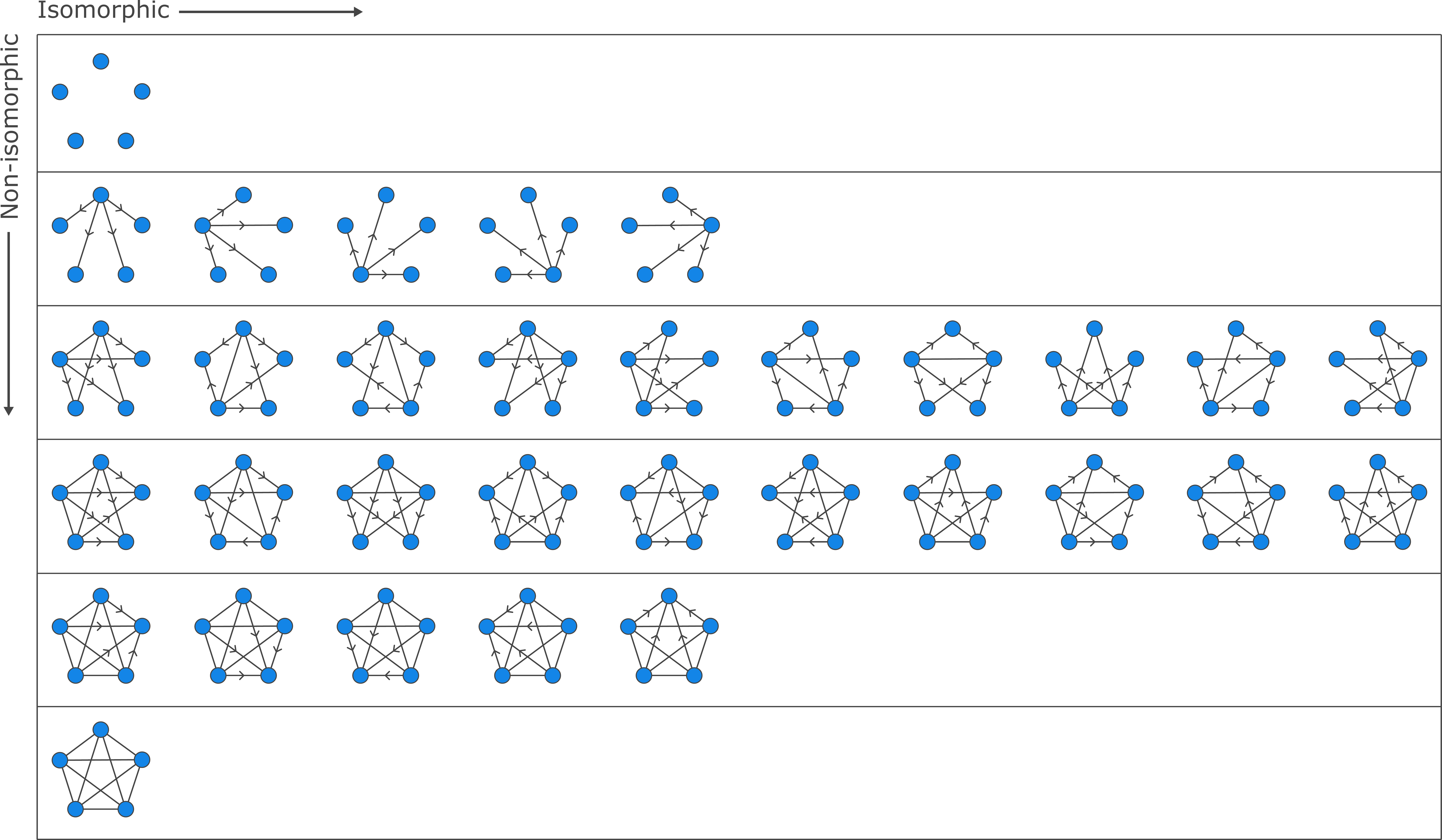}
    \caption[5-vertices motifs admitting two conserved cross-ratios]{\textbf{5-vertices motifs admitting two conserved cross-ratios}. All 32 possible binary graphs with 5 vertices admitting the maximal number of 2 functionally independent conserved cross-ratios in the corresponding Kuramoto model (considering that the conditions on the frequencies and the phase lags are satisfied). The 6 non-isomorphic graphs are displayed vertically and the corresponding isomorphisms are displayed horizontally.}
    \label{fig:max_reduc_5nodes}
\end{figure}
The next corollary is the equivalent of the theorem when there is no phase lag between the oscillators. 
\begin{corollary}\label{cor:cte_mvt_kuramoto_graphe}
Consider the $N$-dimensional Kuramoto model~\eqref{eq:thetakur} on a graph described by some $N\times N$ real matrix $W$, with natural frequencies $(\omega_j)_{j=1}^N$ and zero phase lags. The cross-ratio $c_{abcd}$~\eqref{eq:cross_ratio_SI} is a constant of motion in the model if and only if the vertices $a$, $b$, $c$, and $d$ have the same:
\begin{enumerate}
\item outgoing edges within $\{a,b,c,d\}$, i.e.,
\begin{align}\label{eq:identical_insiders_outcoming_edge}
    W_{ba} = W_{ca} = W_{da}\,,\quad
    W_{ab} = W_{cb} = W_{db}\,,\quad
    W_{ac} = W_{bc} = W_{dc}\,,\quad
    W_{ad} = W_{bd} = W_{cd}\,;
\end{align}    

\item incoming edges from the vertices outside $\{a,b,c,d\}$ in the graph, i.e.,
\begin{equation}\label{eq:identical_outsiders_incoming_edge}
    W_{ak} = W_{bk} = W_{ck} = W_{dk},\quad \forall\, k\in\{1,...,N\}\setminus\{a, b, c, d\}\,;
\end{equation}

\item natural frequencies, i.e.,
\begin{equation}\label{eq:identical_frequencies}
    \omega_{a} = \omega_{b} = \omega_{c} = \omega_{d}\,;
\end{equation}
\end{enumerate}
\end{corollary}
\begin{proof}
When the phase lags are zero, $A = \frac{1}{2}(W  + i\,\mathrm{diag}(\bm{\omega}))$. In this case, the first two conditions of Thm.~\ref{thm:cte_mvt_kuramoto_graphe}, which only involve non-diagonal terms of $A$, coincide with the first two conditions stated in the corollary. Using the explicit form of the third condition in Thm.~\ref{thm:cte_mvt_kuramoto_graphe}, that is,
\begin{align*}
    \omega_a -  W_{k_a a}\sin\alpha_{k_a a} = \omega_b -  W_{k_b b}\sin\alpha_{k_b b} = \omega_c -  W_{k_c c}\sin\alpha_{k_c c} = \omega_d -  W_{k_d d}\sin\alpha_{k_d d}\,,
\end{align*}
where $k_j$ takes any value within $\{a,b,c,d\}\setminus\{j\}$ for $j\in\{a,b,c,d\}$, the sine terms vanish due to the zero phase lags. This directly yields $\omega_a = \omega_b = \omega_c = \omega_d$, as stated in the third condition of the corollary.
\end{proof}
\begin{remark}
    The latter corollary was also verified with symbolic calculations in Matlab (\textit{symbolic\_calculations\_theorem.m}) and in Mathematica (\textit{KMK\_constants\_of\_motion.nb}) in Ref.~\cite{Thibeault2026_koopman_kuramoto}.
\end{remark}
In the following corollary, we only focus on the weight matrix conditions to admit a conserved cross-ratio when there is a star $S_N$ with $N$ vertices in an undirected graph. 
\begin{corollary}\label{cor:star}
    Consider the Kuramoto model~\eqref{eq:kuramoto} on an undirected graph where the phase lags are zero and the graph contains a star $S_n$ of $n$ vertices with peripheral oscillators having identical frequencies and only a connection to the core. Then, $n$ must be at least 5 for the star to admit a conserved cross-ratio. 
\end{corollary}
\begin{proof}
    The cross-ratio involves 4 vertices and the stars $S_1$ (trivial graph), $S_2$ (path), or $S_3$ (path) are readily excluded. For $n=4$, denote the core by $a$ and the periphery by $\{b, c, d\}$  without loss of generality. The core is connected to all vertices in the periphery, so in particular, $W_{ab} = 1$. However, $W_{cb} = 0 \neq W_{ab}$ and thus the first condition in Corollary~\ref{cor:cte_mvt_kuramoto_graphe} is not satisfied. 
    
    For $n = 5$, let the core be labeled $e$ and the periphery $\{a,b,c,d\}$. The first condition is readily satisfied since there is no edge between the peripheral vertices in Corollary~\ref{cor:cte_mvt_kuramoto_graphe}. The second condition in Corollary~\ref{cor:cte_mvt_kuramoto_graphe} is also satisfied since $W_{ae} = W_{be} = W_{ce} = W_{de} = 1$. Setting the natural frequencies of the vertices $a,b,c,d$ to be identical, the Kuramoto model on the star $S_5$ admits the cross-ratio $c_{abcd}$ as a constant of motion by Corollary~\ref{cor:cte_mvt_kuramoto_graphe}.
\end{proof}
\begin{remark}
    In the directed case, the smallest star that admits a conserved cross-ratio is composed of 4 vertices, as shown in Fig.~\ref{fig:motifs_cte_mvt}.
\end{remark}

\subsection{WS integrals to Darboux functions}
\label{SIsubsec:WS_darboux}

For $\omega_j = 0$, $\alpha_{jk} = \delta - \pi/2$ with $\sin\delta = 0$, and $W_{jk} = 1/N$ for all $j,k\in\mathcal{V}$, the model admits $N - 2$  functionally independent constants of motion of the form
\begin{align}\label{eq:ws_integral_SI}
    C^{\mathrm{ws}}_{j_1...j_N} = S_{j_1j_2}S_{j_2j_3}...S_{j_{N-1}j_N}S_{j_Nj_1}\,,
\end{align}
where $S_{jk} = \sin\left((\theta_j - \theta_k)/2\right)$. J.W.~Swift found these integrals for $N =3,\,4$ and conjectured the general case, which was then proved by Watanabe and Strogatz~\cite{Watanabe1993, Watanabe1994, Goebel1995}. They are now often called ``WS integrals". As noted in Ref.~\cite{Watanabe1994}, these constants of motion are Hénon-type integrals: they are indexed by cyclic labels and are built as products of sine factors, reminiscent of the product-type integrals in Toda lattices~\cite{Henon1974, Flaschka1974, Toda1981}. 

Cross-ratios are directly linked to WS integrals~\cite{Marvel2009}~:
\begin{align}\label{eq:crossWS}
    c_{abcd} = -\frac{C^{\mathrm{ws}}_{acdb\,j_5...j_N}}{C^{\mathrm{ws}}_{adcb\,j_5...j_N}} = \frac{S_{ca}S_{db}}{S_{cb}S_{da}}\,.
\end{align}
Although Eq.~\ref{eq:crossWS} was established in Ref.~\cite{Marvel2009} in the context where $\omega_j = 0$, $\alpha_{jk} = \delta - \pi/2$, $\sin\delta = 0$, and $W_{jk} = 1/N$ for all $j,k$, the relation is always valid. Therefore, one might ask what happens when the conditions of Thm.~\ref{thm:crossratios} are satisfied or, more generally, for any dynamics whose Koopman generator has the form 
\begin{align}\label{eq:generalLm101_SI}
    \mathcal{K}_{\alpha\beta\gamma} = \alpha(t, \bm z)L_{-1} + \beta(t, \bm z)L_{0} + \gamma(t, \bm z)L_{1}\,.
\end{align}
Are the WS quantities in Eq.~\eqref{eq:ws_integral_SI} first integrals under these conditions ?

Theorem~\ref{thm:crossratios} and  Eq.~\eqref{eq:generalLm101_SI} do not imply that the WS integrals are conserved even if the relation $c_{abcd} = -C^{\mathrm{ws}}_{acdb\,j_5...j_\ell}/C^{\mathrm{ws}}_{adcb\,j_5...j_\ell}$ for $4\leq\ell\leq N$ still holds. They do suggest that $C^{\mathrm{ws}}_{acdb\,j_5...j_\ell}$ and $C^{\mathrm{ws}}_{adcb\,j_5...j_\ell}$ may be eigenfunctions of the Koopman generator $\mathcal{K}$ such that the ratio of these eigenfunctions yields another eigenfunction of null eigenvalue---i.e., a constant of motion---similarly to the cancellation of the eigenvalues using monomial and Vandermonde-ratio eigenfunctions. 
In fact, WS quantities in Eq.~\eqref{eq:ws_integral_SI} are, however, more general Darboux functions~\cite{Goriely2001, Zhang2017} of the form
\begin{align*}
    \mathcal{K}_{\alpha\beta\gamma}[C^{\mathrm{ws}}_{\bm j}(\bm z)] = \Lambda_{\bm j}(\bm z)\,C^{\mathrm{ws}}_{\bm j}(\bm z)\,,
\end{align*}
where $\Lambda_{\bm j}(\bm z)$ is called a cofactor, $C^{\mathrm{ws}}_{\bm j}(\bm z) = \prod_{u=1}^\ell(1 - z_{j_u}\bar{z}_{j_{u+1}})$ is the complex form of the WS function, and $\bm j = (j_1,...,j_\ell)$ with $j_{\ell +1} = j_1$. Unlike Koopman eigenfunctions, having two Darboux functions $\psi_1, \psi_2$ with (non constant) cofactors $\Lambda_1, \Lambda_2$ does not imply a conserved quantity, since $\mathcal{K}[\psi_1^{a_1}(\bm z)\psi_2^{a_2}(\bm z)] = (a_1\Lambda_1(\bm z) + a_2\Lambda_2(\bm z))\psi_1^{a_1}(\bm z)\psi_2^{a_2}(\bm z)$ and the real constants $a_1, a_2$ cannot necessarily be chosen to cancel the contributions of $\Lambda_1(\bm z)$ and $\Lambda_2(\bm z)$. 

The particularity of WS functions is that they share the same cofactor whenever the Koopman generator has the form in Eq.~\eqref{eq:generalLm101_SI}. Indeed, 
\begin{align*}
    L_{-1}[\ln C^{\mathrm{ws}}_{\bm j}(\bm z)] = -\textstyle{\sum_{u=1}^\ell} \bar{z}_{j_u}\,,\qquad\quad
    L_0[\ln C^{\mathrm{ws}}_{\bm j}(\bm z)] = 0\,,\qquad\quad
    L_1[\ln C^{\mathrm{ws}}_{\bm j}(\bm z)] = \textstyle{\sum_{u=1}^\ell} z_{j_u}\,.
\end{align*}
Thus, for all $t\in \mathbb{R}$, $\bm z \in \mathbb{T}^N$, and for a fixed set of labels ${j_1,\ldots,j_\ell}$, the cofactor is
\begin{align}\label{eq:darbouxWSgen}
    \Lambda(\bm z) := \Lambda_{\bm j}(\bm z) = \gamma(t, \bm{z})\sum_{u=1}^\ell z_{j_u} - \alpha(t, \bm{z})\sum_{u=1}^\ell \bar z_{j_{u}}
\end{align}
for any given ordering of the labels ${j_1,\ldots,j_\ell}$, which is a consequence of the fact that $\Lambda_{\bm j}(\bm z)$ depends on symmetric polynomials and functions $\alpha, \gamma$ do not depend on the indices $j_1,...,j_\ell$. Therefore, it suffices to take the ratio of two WS (Darboux) functions to get the conserved cross-ratio. 

In the specific case of the Kuramoto model, the conditions of Thm.~\ref{thm:crossratios} are satisfied and the cofactor becomes
\begin{align}\label{eq:darbouxWS}
    \Lambda(\bm z) 
    = -2\,\mathrm{Re}\,\left[\left(\sum_{u=1}^\ell\bar{z}_{j_u}\right)\left(\sum_{k=1}^N\mathcal{A}_kz_k\right)\right]\,, 
\end{align}
where $\mathcal{A}_1,...,\mathcal{A}_N \in \mathbb{C}$. If $\mathcal{A}_k = \frac{w_k}{2}e^{-i\alpha_k}$ with $w_k,\alpha_k \in \mathbb{R}$, the real form of the cofactor is 
\begin{align*}
    \Tilde{\Lambda}(\bm \theta) = -\sum_{u=1}^\ell\sum_{k=1}^N w_k \cos(\theta_k - \theta_{j_u} - \alpha_k)\,.
\end{align*}
From the cofactor in Eq.~\eqref{eq:darbouxWSgen}, the general condition to have a WS eigenfunction or integral ($\lambda = 0$ [Lem.~\ref{lem:ws}]) is therefore
\begin{align}\label{eq:genconditionWS}
     \gamma(t, \bm{z})\sum_{u=1}^\ell z_{j_u} - \alpha(t, \bm{z})\sum_{u=1}^\ell \bar z_{j_{u}} = \lambda \in \mathbb{C}\,.
\end{align}
Since $L_0[\ln C^{\mathrm{ws}}_{\bm j}(\bm z)] = 0$ (i.e., the WS function is an invariant of $L_0$), there is no restriction on $\beta(t,\bm z)$ to conserve the WS function. Hence, any dynamics $\dot{z}_j = \alpha(t,\bm{z}) + \beta(t,\bm{z})z_j + \gamma(t, \bm{z})z_j^2$ satisfying Eq.~\eqref{eq:genconditionWS} with $\lambda = 0$ will admit WS integrals. More explicitly, with $\alpha'(t, \bm z) := \alpha(t, \bm z)/(\sum_{u=1}^\ell z_{j_u})$ and $L' := (\sum_{u=1}^\ell z_{j_u})L_{-1} + (\sum_{u=1}^\ell \bar z_{j_u})L_1$, one gets $\mathcal{K}_{\alpha',\beta} = \alpha'(t, \bm z)L' + \beta(t,\bm z)L_0$, the WS quantity is a joint invariant of $L'$ and $L_0$, and it is thus conserved.

\clearpage

\section{Continuous symmetries and the generation of new constants of motion}
The concept of symmetry for differential equations has a long history that has flourished from the work of Sophus Lie to the work of Emmy Noether. Below, we only briefly present the theory for ordinary differential equations in order to present the symmetry criterion under Koopman's perspective and dive quickly into its application to the Kuramoto model. For more details, the reader is invited to visit Refs.~\cite{Sattinger1986, Stephani1989, Arnold1992, Olver1993, Gaeta1994, Ibragimov1999, Hydon2000, Bluman2002} and in particular, the great book by Peter Olver~\cite{Olver1993} that includes pertinent historical remarks, reproducible examples, and crucial theorems for general differential equations. The theorem of interest for us is based on the concept of prolongation of a vector field and gives us the necessary and sufficient conditions to have a symmetry group. Without giving the details, it is stated as follows.
\begin{theorem}[Theorem 2.71~\cite{Olver1993}]\label{thm:2.71}
    Let $\Delta(x, u^{(n)}) = 0$ be a nondegenerate system of $\ell$ differential equations. A connected local group of transformations $G$ acting on an open subset $M\subset X \times U$ is a symmetry group of the system if and only if
    \begin{align}\label{eq:inf_crit_n}
        \mathrm{pr}^{(n)}v[\Delta_\nu(x, u^{(n)})] = 0\,,\quad \nu \in\{1,...,\ell\}\,,\quad \text{whenever} \quad \Delta(x, u^{(n)}) = 0\,,
    \end{align}
    for every infinitesimal generator $v$ of $G$.
\end{theorem}
\begin{remark}\label{rem:smoothness}
    The function $\Delta$ from the $n$-jet space $X\times U^{(n)}$ to $\mathbb{R}^\ell$ is considered to be smooth in its arguments~\cite[p.96]{Olver1993}. Moreover, it is also assumed that the infinitesimal generator $v$ and its prolongations act on the space of smooth functions, a fact that we will use later. Finally, we refer to p.20-22 of Ref.~\cite{Olver1993} for the definition of a connected local group of transformations.
\end{remark}
In the next subsection, we use this general result for first-order ordinary differential equations and show that the infinitesimal criterion~\eqref{eq:inf_crit_n} is elegantly written in terms of the Koopman generator.

\subsection{Proof of the Lemma: Infinitesimal criterion of symmetry under Koopman's perspective}
\label{SIsubsec:inf_crit}

To use Thm.~\ref{thm:2.71} in our context, we first adapt Defs. 2.30 and 2.70 of Ref.~\cite{Olver1993} for systems of first-order ODEs.
\begin{definition}
    Let $\Delta_i(t, \bm{u}, \dot{\bm{u}}) = 0$ for $i \in \{1,...,N\}$ be a system of first-order ordinary differential equations. The system is of maximal rank if the $N \times (2N+1)$ Jacobian matrix of $\bm\Delta = (\Delta_1,...,\Delta_N)$,
    \begin{align*}
        J_{\bm \Delta}(t, \bm{u}, \dot{\bm{u}}) = 
        \begin{pmatrix}
            \pd[]{\Delta_1(t, \bm{u}, \dot{\bm{u}})}{t} & \pd[]{\Delta_1(t, \bm{u}, \dot{\bm{u}})}{u_1} & \cdots & \pd[]{\Delta_1(t, \bm{u}, \dot{\bm{u}})}{ u_N} & \pd[]{\Delta_1(t, \bm{u}, \dot{\bm{u}})}{ \dot{u}_1} & \cdots & \pd[]{\Delta_1(t, \bm{u}, \dot{\bm{u}})}{\dot{u}_N}\\
            \vdots & \vdots & \vdots & \vdots & \vdots & \vdots & \vdots\\
            \pd[]{\Delta_N(t, \bm{u}, \dot{\bm{u}})}{t} & \pd[]{\Delta_N(t, \bm{u}, \dot{\bm{u}})}{u_1} & \cdots & \pd[]{\Delta_N(t, \bm{u}, \dot{\bm{u}})}{ u_N} & \pd[]{\Delta_N(t, \bm{u}, \dot{\bm{u}})}{ \dot{u}_1} & \cdots & \pd[]{\Delta_N(t, \bm{u}, \dot{\bm{u}})}{\dot{u}_N}
        \end{pmatrix},
    \end{align*}
    is of rank $N$ for all $(t, \bm{u}, \dot{\bm{u}})$ such that $\bm\Delta(t, \bm{u}, \dot{\bm{u}}) = \bm 0$.
\end{definition}
\begin{definition}
    A system of $N$ first-order differential equations, $\bm \Delta(t, \bm{u}, \dot{\bm{u}}) = 0$ is \textit{locally solvable at the point} $(t_0, \bm{u}_0, \dot{\bm{u}}_0) \in \mathscr{G}_{\bm\Delta} = \{(t, \bm{u}, \dot{\bm{u}})\,|\,\bm\Delta(t, \bm{u}, \dot{\bm{u}}) = 0\}$ if there exists a smooth solution $\bm{u} = \bm{y}(t)$ of the system, defined for $t$ in a neighborhood of $t_0$, which has the prescribed ``initial condition" $\mathrm{pr}^{(1)}\bm{y}(t_0) = (\bm u_0, \dot{\bm{u}}_0)$. The system is \textit{locally solvable} if it is locally solvable at every point of $\mathscr{G}_{\bm\Delta}$. A system is \textit{nondegenerate} if at every point $(t_0, \bm{u}_0, \dot{\bm{u}}_0) \in \mathscr{G}_{\bm\Delta}$ it is both locally solvable and of maximal rank.
\end{definition}
Consider the system of first-order ordinary differential equations (henceforth, called the ``dynamics")
\begin{equation}\label{eq:sysODESI}
    \od[]{y_i}{t} = F_i(t, y_1,...,y_N)\,,\quad i \in\{1,...,N\}\,,
\end{equation}
with initial condition $y_i(t_0) = (\bm{u}_0)_i$ for all $i$, $t_0<t$, and $F_1,...,F_N$ are smooth ($\mathscr{C}^\infty$) in their arguments. Note that we could relax the differentiability requirements in principle, but we use smooth functions for simplicity and to be coherent with the approach and the results in Ref.~\cite[see p.4 and p.96]{Olver1993}. Let us define the smooth functions $\Delta_i$ for all $i\in\{1,...,N\}$ on the 1-jet space related to the $i$-th equation in the dynamics such that
\begin{equation}\label{eq:funODE}
    \Delta_i(t, \bm{u}, \dot{\bm{u}}) = \dot{u}_i - F_i(t, \bm{u})\,,
\end{equation}
where $\bm{u} = (u_1,...,u_N)\in\mathbb{R}^N$ and $\dot{\bm{u}} =  (\dot{u}_1,...,\dot{u}_N)\in\mathbb{R}^N$ are coordinates with $t\in\mathscr{T}$ for the jet space. In such case, the next lemma shows that there is no problem with the system regarding the condition of maximal rank and local solvability.
\begin{lemma}\label{lem:nondegenerate}
    The dynamics in Eq.~\eqref{eq:sysODESI} is nondegenerate.
\end{lemma}
\begin{proof}
    The $N \times (2N+1)$ Jacobian matrix of $\bm \Delta = (\Delta_1,...,\Delta_N)$ where $\Delta_i$ is given in Eq.~\eqref{eq:funODE} is 
    \begin{align*}
        J_{\bm\Delta}(t, \bm{u}, \dot{\bm{u}}) = 
        \begin{pmatrix}
            -\pd[]{F_1(t, \bm{u})}{t} & -\pd[]{F_1(t, \bm{u})}{u_1} & ... & -\pd[]{F_1(t, \bm{u})}{u_N} & 1 & ... & 0\\
            \vdots & \vdots & ... & \vdots & \vdots & \ddots & \vdots\\
            -\pd[]{F_N(t, \bm{u})}{t} & -\pd[]{F_N(t, \bm{u})}{u_1} & ... & -\pd[]{F_N(t, \bm{u})}{u_N} & 0 & ... & 1
        \end{pmatrix}
    \end{align*}
     and is of rank $N$ with respect to all $(t, \bm{u}, \dot{\bm{u}})$, because of the $N\times N$ identity submatrix in the last columns of $J_{\bm\Delta}(t, \bm{u}, \dot{\bm{u}})$. The dynamics is thus of maximal rank.

     Moreover, by assumption, the vector field in Eq.~\eqref{eq:sysODESI} is smooth. Thus, there exists a unique smooth solution $\bm{y}(t) = \bm{u}$ starting at $\bm{y}(t_0) = \bm{u}_0$ by Lem.~2.3 of Ref.~\cite{Teschl2012}. The existence of the solution and the form of the differential equations imply that $\dot{\bm{y}}(t_0) = \bm{F}(t_0, \bm{y}(t_0)) = \bm{F}(t_0, \bm{u}_0) =: \dot{\bm{u}}_0$. Such solution exists for any initial condition $(\bm{u}_0, \dot{\bm{u}}_0) = \mathrm{pr}^{(1)}\bm{y}(t_0)\in\mathscr{G}_\Delta$ and the system is therefore locally solvable and altogether, nondegenerate.
\end{proof}
Consider the set $\mathscr{O}$ of time-dependent smooth observables $f:S\subset \mathscr{T}\times U\to \mathbb{R}$. The Koopman operator for a non-autonomous dynamical system described by Eq.~\eqref{eq:sysODESI} is $U_{\varphi}: \mathscr{T}\times \mathscr{T} \times \mathscr{O} \to \mathscr{O}$ with $U_{\varphi_{t_0, t}}[f] := U_{\varphi}(t_0, t, f)$ and it acts on an observable $f$ such that
\begin{align*}
    U_{\varphi_{t_0, t}}[f] = f\circ \varphi_{t_0, t}
\end{align*}
with the properties $U_{\varphi_{t_0, t_0}} = \mathrm{id}$ and $U_{\varphi_{a,t + a}}\circ U_{\varphi_{t_0, a}} = U_{\varphi_{t_0, t+a}}$~\cite{Macesic2018}. There is a family $(\mathcal{U}_{t_0})_{t_0 \in \mathscr{T}}$ of Koopman generators
\begin{align*}
    \mathcal{U}_{t_0}[f](\bm{u}_0) = \left.\od[]{f(t, \bm{y}(t))}{t}\right|_{t = t_0}\,, 
\end{align*}
recalling that $\bm{y}(t_0) = \bm{u}_0$ and that $\mathcal{U}_t$ is also locally defined at some point $\bm{u}\in U$. Performing the total derivative gives the explicit form
\begin{align*}
    \mathcal{U}_{t_0}[f](\bm{u}_0) = \left.\left[\pd[]{f}{t}(t, \bm y(t)) + \bm{F}(t, \bm{y}(t))\cdot\nabla f(t, \bm{y}(t))\right]\right|_{t = t_0} = \pd[]{f}{t}(t_0, \bm u_0) + \bm{F}(t_0, \bm{u}_0)\cdot\nabla f(t_0, \bm{u}_0)\,.
\end{align*}
Therefore, the Koopman generator is
\begin{equation}\label{eq:calK}
    \mathcal{U} = \partial_t + \sum_{j = 1}^N F_j(t, \bm{u})\partial_j\,,
\end{equation}
where $(t, \bm{u}) \in S$, $\partial_t := \partial/\partial t$, $\partial_j := \partial/\partial u_j$ and where we have removed the time index of the generator for simplicity. Under these considerations, the Koopman generator $\mathcal{U}$ and the infinitesimal generator $v$ both act on smooth functions [Remark~\ref{rem:smoothness}] and can be manipulated together. The next lemma is the equivalent of Thm.~\ref{thm:2.71} for systems of first-order ODEs and provides the infinitesimal criterion of symmetry in terms of the Koopman generator.
\begin{lemma}[Lemma of the paper]\label{lem:symkoo}
    A connected local group of transformations $G$ acting on an open subset $S\subset \mathscr{T} \times U$ is a symmetry group of the dynamics in Eq.~\eqref{eq:sysODESI} if and only if 
    \begin{align}\label{eq:symkooSI}
    [\mathcal{U}, v] - \mathcal{U}[\xi(t, \bm{u})]\mathcal{U} = 0
    \end{align}
    for every infinitesimal generator $v = \xi(t, \bm{u})\partial_t + \sum_{j=1}^N\phi_j(t, \bm{u})\partial_{j}$ of $G$.
\end{lemma}
\begin{proof}
To begin with, Lem.~\ref{lem:nondegenerate} ensures that the dynamics~\eqref{eq:sysODESI} is nondegenerate, which guarantees that Thm.~\ref{thm:2.71} can be applied. Now, since the dynamics is a first-order system of ODEs, only the first prolongation of the infinitesimal generator $v$ is needed. By the general prolongation formula in Ref.~\cite[Theorem 2.36]{Olver1993}, the first prolongation of $v$ is
\begin{equation*}
    \mathrm{pr}^{(1)}v =  v + \sum_{j=1}^N\phi_j^t(t, \bm{u}, \dot{\bm{u}}) \pd{}{\dot{u}_j} \,,
\end{equation*}
with 
\begin{equation}\label{eq:phitj}
    \phi^t_j = D_t(\phi_j - \xi \dot{u}_j)+ \xi \ddot{u}_j = \dot{\phi}_j + \sum_{k=1}^N(\partial_k \phi_j - \dot{\xi}\delta_{jk})F_k - \sum_{k=1}^N\partial_k\xi \,F_kF_j\,,
\end{equation}
where the superscript $t$ in $\phi^t_j$ is a label used to denote the component of the prolongation associated with $\partial/\partial\dot{u}_j$ and $D_t=d/dt$ is the total derivative. The dependencies on $t, \bm{u}, \dot{\bm{u}}$ are omitted in Eq.~\eqref{eq:phitj} to simplify the notation, as will be done from now on. The infinitesimal condition for $G$ to be a symmetry group is then
\begin{align*}
    \mathrm{pr}^{(1)}v[\Delta_i] &= \xi \dot{\Delta}_i + \sum_{j=1}^N\phi_j \partial_j\Delta_i + \sum_{j=1}^N\phi^t_j\pd{\Delta_i}{\dot{u}_j} = 0\,,
\end{align*}
for all $i \in \{1,...,N\}$. Inserting~Eq.~\eqref{eq:phitj}, performing the derivatives and rearranging leads to the infinitesimal criterion
\begin{equation}\label{eq:cond_inf_general}
    \dot{\phi}_i + \sum_{j=1}^N(\partial_j \phi_i - \dot{\xi}\delta_{ij})F_j - \sum_{j=1}^N\partial_j\xi F_jF_i = \xi \dot{F}_i + \sum_{j=1}^N \phi_j \partial_jF_i\,,\quad \forall i\in\{1,...,N\}\,.
\end{equation}
Applying Thm.~\ref{thm:2.71} to our particular case, a connected local group of transformations $G$ is a symmetry group if and only if Eqs.~\eqref{eq:cond_inf_general} are satisfied. It now remains to show that Eqs.~\eqref{eq:cond_inf_general} are equivalent to condition~\eqref{eq:symkooSI}.  The commutator of $\mathcal{U}$ and $v$ is
\begin{align*}
    [\mathcal{U}, v] &= [\partial_t + \sum_j F_j\partial_j, \xi\partial_t + \sum_i \phi_i \partial_i]
    = [\partial_t, \xi \partial_t] + \sum_i[\partial_t, \phi_i\partial_i] + \sum_j[F_j\partial_j,\xi\partial_t] + \sum_{i,j}[F_j\partial_j, \phi_i\partial_i]\,,
\end{align*}
where we have used the bilinearity of the commutator. More explicitly, the last equation is
\begin{align*}
    [\mathcal{U}, v] &= (\partial_t\xi) \partial_t + \sum_i (\partial_t\phi_i) \partial_i + \sum_j(F_j(\partial_j\xi)\partial_t - \xi(\partial_tF_j)\partial_j) + \sum_{i,j}(F_j(\partial_j\phi_i)\partial_i - \phi_i(\partial_iF_j)\partial_j)\,,
\end{align*}
henceforth considering that all sums run from 1 to $N$. Given some observable $f\in\mathscr{O}$, the expression can be rearranged to yield
\begin{align*}
    [\mathcal{U}, v]f &= \mathcal{U}[\xi]\partial_tf + \sum_i (\dot{\phi}_i + \sum_j\partial_j\phi_i \,F_j - \xi\dot{F}_i - \sum_j\phi_j \partial_jF_i)\partial_if\,.
\end{align*}
From the definition of the Koopman generator~\eqref{eq:calK}, the relation $\partial_t f = \mathcal{U}[f] - \sum_i F_i\partial_if$ holds and implies
\begin{align*}
    [\mathcal{U}, v]f = \mathcal{U}[\xi]\mathcal{U}[f] + \sum_{i}(\dot{\phi}_i - \partial_t\xi\,F_i + \sum_j\partial_j\phi_i F_j - \sum_j F_iF_j\partial_j\xi - \xi\dot{F}_i -\sum_j\phi_j\partial_jF_i)\partial_if\,.
\end{align*}
Writing $\partial_t\xi\,F_i$ as $\sum_j (\dot{\xi}\delta_{ij})F_j$ explicitly gives a sum over the $N$ equations of the infinitesimal symmetry condition:
\begin{align*}
    [\mathcal{U}, v] = \mathcal{U}[\xi]\mathcal{U} + \sum_{i}(\dot{\phi}_i + \sum_{j}(\partial_j \phi_i - \dot{\xi}\delta_{ij})F_j - \sum_{j}\partial_j\xi F_jF_i  - \xi\dot{F}_i -\sum_j\phi_j\partial_jF_i)\partial_i\,.
\end{align*}
On the one hand, if the $N$ infinitesimal conditions of symmetry in Eqs.~\eqref{eq:cond_inf_general} are satisfied, then
\begin{align*}
    [\mathcal{U}, v] - \mathcal{U}[\xi]\mathcal{U} = 0\,.
\end{align*}
On the other hand, if $[\mathcal{U}, v] - \mathcal{U}[\xi]\mathcal{U} = 0$, then
\begin{align*}
    \sum_{i}(\dot{\phi}_i + \sum_{j}(\partial_j \phi_i - \dot{\xi}\delta_{ij})F_j - \sum_{j}\partial_j\xi F_jF_i  - \xi\dot{F}_i -\sum_j\phi_j\partial_jF_i)\partial_i = 0\,.
\end{align*}
But each term of the sum over $i$ is independent, meaning that Eqs.~\eqref{eq:cond_inf_general} are satisfied and thus completing the proof.
\end{proof}

\subsection{Basic symmetries of the Kuramoto model}

The Koopman generator of the Kuramoto dynamics on functions of time and $\bm z\in\mathbb{T}^N$ is
\begin{align*}
    \mathcal{U} &= \partial_t + \sum_{j,k=1}^N(A_{jk}z_k - \bar{A}_{jk}\bar{z}_kz_j^{2})\partial_j = \partial_t + \mathcal{K}\,.
\end{align*}

The most general infinitesimal generator of (potential) symmetries is
\begin{align}\label{eq:generateur_general2}
    \mathcal{S} = \xi(t,z)\partial_t + \sum_{\ell=1}^N\phi_\ell(t, z)\partial_\ell\,.
\end{align}
Adapting Lem.~\ref{lem:symkoo} for coordinates of time and the $N$-torus, it is necessary and sufficient that the generator $\mathcal{S}$ satisfies
\begin{align}
    [\mathcal{U}, \mathcal{S}] - \mathcal{U}[\xi(t,z)]\,\mathcal{U} = 0\,.
\end{align}
It is easy to verify that
\begin{align*}
    \mathcal{S}_1 = iL_0\,,\qquad
    \mathcal{S}_2 = \mathcal{K}\,,\qquad
    \mathcal{S}_3 = f(t)\mathcal{U}
\end{align*}
are Lie symmetries of the Kuramoto dynamics. Note that the time translation generator $\partial_t$ is obtained with $\mathcal{U} - \mathcal{K}$ and is a symmetry generator, as expected of any autonomous dynamical system. In fact, denoting $\mathcal{S} = \mathcal{S}_\xi + \mathcal{S}_\phi$ with $\mathcal{S}_\xi = \xi(t, \bm{z})\partial_t$ and $\mathcal{S}_\phi = \sum_{\ell=1}^N\phi_\ell(t, \bm{z})\partial_{\ell}$, the infinitesimal criterion for the Kuramoto model becomes $[\mathcal{U}, \mathcal{S}_\phi] - \mathcal{U}[\xi(t, \bm{z})]\,\mathcal{K} = 0$, that is, the form of the condition for autonomous dynamical systems. 

On the one hand, if $\phi_1(t, \bm{z}),...,\phi_N(t, \bm{z})$ are zero, the condition becomes
$\mathcal{U}[\xi(t,\bm{z})] = 0$, meaning that $\xi(t,\bm{z})$ must be a constant of motion if $\xi(t,z)\partial_t$ is to be a generator of symmetry (also highlighted in Ref.~\cite{Schwarz1984}). Yet, such symmetries simply act as time translations. Indeed, considering that the conditions of Thm.~\ref{thm:cte_mvt_kuramoto_graphe} are satisfied for some quadruples of vertices, one can set $\xi(t,z)$ to be any of the conserved functionally independent cross-ratios, leading to the infinitesimal symmetry generators $\mathcal{S}_{abcd} = c_{abcd}(\bm z)\partial_t$ for all $a,b,c,d$ such that $\mathcal{U}[c_{abcd}(\bm z)] = \mathcal{K}[c_{abcd}(\bm z)] = 0$. The action of the symmetry group on the coordinates is thus $e^{\epsilon \mathcal{S}_{abcd}}t = t + \epsilon\,c_{abcd}(\bm z)$ and $e^{\epsilon \mathcal{S}_{abcd}}z_j = z_j$. Considering that $a,b,c,d$ belong to some partially integrable part $P$, using $z_j(t) = M_t(w_j)$ (notation of Ref.~\cite{Marvel2009}) and the fact that cross-ratios are invariant under M\"{o}bius transformations $M_t$ leads to
\begin{align*}
    c_{abcd}(\bm{z}(t)) = (M_{t}(w_a), M_{t}(w_b);M_{t}(w_c),M_{t}(w_d)) = (w_a, w_b\,;\,w_c, w_d) = c_{abcd}(\bm{w})
\end{align*}
and the action of the related symmetry group on a solution $\bm{z}(t)$ of the Kuramoto model is such that
\begin{align*}
    \tilde{\bm{z}}(t) := e^{\epsilon \mathcal{S}_{abcd}}\bm{z}(t) = \bm{z}(t + \epsilon\,c_{abcd}(\bm{w}))\,,
\end{align*}
where $\tilde{\bm{z}}(t)$ is an analogous time-translated solution of the Kuramoto model.

On the other hand, if $\xi(t, \bm{z}) = 0$, then the infinitesimal criterion~\eqref{eq:symkooSI} is $[\mathcal{U}, \mathcal{S}] = 0$. For $\mathcal{S} = \psi(t, \bm{z})\tilde{\mathcal{S}}$ with some smooth function $\psi$ and $\tilde{\mathcal{S}} = \sum_{j=1}^N \tilde{\phi}_j(t, \bm{z})\partial_j$, $[\mathcal{U}, \mathcal{S}] = \mathcal{U}[\psi(t,\bm{z})]\tilde{\mathcal{S}} + \psi(t, \bm{z})[\mathcal{U}, \tilde{\mathcal{S}}]$. Therefore, if $\psi(\bm{z})$ is a constant of motion and $\tilde{\mathcal{S}}$ is a symmetry generator, then $\psi(\bm{z})\tilde{\mathcal{S}}$ is also a symmetry generator, but its action remains the one of $\tilde{\mathcal{S}}$ or is not an automorphism of the $N$-torus, making it of little use.

We will need more than naive inspection to uncover additional symmetries. Therefore, the next two subsections are dedicated to deriving the determining equations and developing a method for obtaining particular solutions.

\subsection{General determining equations for the Kuramoto model}
\label{SIsubsec:general_det_eq}
To obtain the determining equations, it is useful to introduce basic commutation relations. This is the purpose of the next lemma.
\begin{lemma}\label{lem:commutations_kuramoto}
Consider the elements of the vectorial Euler differential operators defined in Eq.~\eqref{eq:elements_euler} and let $\mathcal{K}$ be defined by Eq.~\eqref{eq:euler_calK}. Then, the following commutation relations hold for all $j,k\in \{1,...,N\}$ and $m, n \in \mathbb{Z}$:
\begin{equation}\label{eq:commutation_ell_z}
    [\,\ell_{j}^m, \,\ell_{k}^n\,]=\delta_{jk}(n-m)\,\ell_j^{m+n}\,,\qquad     [\,z_{j}^m, \,z_{k}^n\,]=0\,,\qquad [\,\ell_{j}^m, \,z_{k}^n\,]=\delta_{jk} \,n \,z_{k}^{m+n}
\end{equation}
and
\begin{align}
    [\,\mathcal{K},\,z_j^n\,] &=  nz_j^n\sum_{k=1}^N \left(A_{jk} z_kz_j^{-1} - \bar{A}_{jk} z_k^{-1}z_j^{1}\right) = 2inz_j^n\Imag\left(\sum_{k=1}^NA_{jk}z_kz_j^{-1}\right)\,,\label{eq:commutation_koopman_kuramoto_1}
    \\ [\,\mathcal{K},\,\ell_j^n\,] &= (n+1)\Big(\sum_{k=1}^N A_{jk}z_k\Big)\ell_j^{n-1} - (n-1)\Big(\sum_{k=1}^N \bar{A}_{jk}z_k^{-1}\Big)\ell_j^{n+1} - z_j^{n+1}\Big(\sum_{k=1}^N A_{kj}\ell_k^{-1}\Big) - z_j^{n-1}\Big(\sum_{k=1}^N \bar{A}_{kj}\ell_k^{1} \Big)\,.
    \label{eq:commutation_koopman_kuramoto_2}
\end{align}
\end{lemma}
\begin{proof}
    The commutation relations in Eq.~\eqref{eq:commutation_ell_z} are obtained easily from the definition in Eq.~\eqref{eq:elements_euler}. Then,
    \begin{align*}
    [\,\mathcal{K},\,z_j^n\,] &= \sum_{q,k}A_{qk}[z_k\ell_q^{-1}, z_j^n] - \sum_{q,k}\bar{A}_{qk}[z_k^{-1}\ell^1_q, z_j^n]\\
    [\,\mathcal{K},\,\ell_j^n\,] &= \sum_{q,k}A_{qk}[z_k\ell_q^{-1}, \ell_j^n] - \sum_{q,k}\bar{A}_{qk}[z_k^{-1}\ell^1_q, \ell_j^n]\,.
    \end{align*}
    Using the linearity of the commutator, the general formula $[AB,C]=A[B,C]+[A,C]B$ and Eq.~\eqref{eq:commutation_ell_z} readily provides 
    \begin{align*}
    \begin{aligned}
        [z_k\ell_q^{-1}, z_j^n] &= \delta_{qj}nz_kz_j^{n-1}\\
        [z_k^{-1}\ell_q^{1}, z_j^n] &= \delta_{qj}nz_k^{-1}z_j^{n+1}
    \end{aligned}
        \qquad
    \begin{aligned}
        [z_k \ell_q^{-1}, \ell_j^n] &= \delta_{qj}(n+1)z_k\ell_q^{n-1} - \delta_{jk} z_k^{n+1}\ell_q^{-1}\\
        [z_k^{-1}\ell_q^{1}, \ell_j^n] &= \delta_{qj}(n-1)z_k^{-1}\ell_q^{n+1} + \delta_{jk} z_k^{n-1}\ell_q^{1}
    \end{aligned}
    \end{align*}
    and their substitution yields the desired results
    {\footnotesize\begin{align*}
    [\,\mathcal{K},\,z_j^n\,] &= \sum_{q,k}A_{qk}\delta_{qj}nz_kz_j^{n-1} - \sum_{q,k}\bar{A}_{qk}\delta_{qj}nz_k^{-1}z_j^{n+1} = n\sum_{k}A_{jk}z_kz_j^{n-1} - n\sum_{k}\bar{A}_{jk}z_k^{-1}z_j^{n+1}\\
    [\,\mathcal{K},\,\ell_j^n\,] &= \sum_{q,k}A_{qk}(\delta_{qj}(n+1)z_k\ell_q^{n-1} - \delta_{jk} z_k^{n+1}\ell_q^{-1}) - \sum_{q,k}\bar{A}_{qk}(\delta_{qj}(n-1)z_k^{-1}\ell_q^{n+1} + \delta_{jk} z_k^{n-1}\ell_q^{1})\\
    &= (n+1)\Big(\sum_{k}A_{jk}z_k\Big)\ell_j^{n-1} - z_j^{n+1}\Big(\sum_{q}A_{qj}\ell_q^{-1}\Big) - (n-1)\Big(\sum_{k}\bar{A}_{jk}z_k^{-1}\Big)\ell_j^{n+1} - z_j^{n-1}\Big(\sum_{q}\bar{A}_{qj}\ell_q^{1} \Big)\\
    &= (n+1)\Big(\sum_{k}A_{jk}z_k\Big)\ell_j^{n-1} - (n-1)\Big(\sum_{k}\bar{A}_{jk}z_k^{-1}\Big)\ell_j^{n+1} - z_j^{n+1}\Big(\sum_{k}A_{kj}\ell_k^{-1}\Big) - z_j^{n-1}\Big(\sum_{k}\bar{A}_{kj}\ell_k^{1} \Big)\,.
    \end{align*} }
\end{proof}
As mentioned in the main text, it is also useful to simplify calculations to restrict the general symmetry generator $\mathcal{S}$ to one where $\xi$ and $\phi_1,...,\phi_N$ are periodic functions, allowing us to expand them in Fourier series:
\begin{align*}
    \xi(t, z) = \sum_{\bm{p}\in\mathbb{Z}^N} \varepsilon_{\bm{p}}(t)z^{\bm{p}}\,,\qquad \phi_\ell(t, z) = \sum_{\bm{p}\in\mathbb{Z}^N} \varphi_{\ell \bm{p}}(t)z^{\bm{p}}\,, \qquad \text{with} \quad z^{\bm{p}} = \prod_{j=1}^N z_j^{p_j}\,.
\end{align*}
This assumption and some notation simplifications lead to
\begin{align*}
    [\mathcal{U}, \mathcal{S}] - \mathcal{U}[\xi(t,z)]\,\mathcal{U} = \sum_{\bm{p}} \Big[\mathcal{U},\varepsilon_{\bm{p}}(t)z^{\bm{p}}\partial_t\Big] + \sum_{\ell,\bm{p}}  \Big[\mathcal{U}, \varphi_{\ell \bm{p}}(t)z^{\bm{p}}\partial_\ell\Big] - \Big(\sum_{\bm{p}}  \mathcal{U}\Big[\varepsilon_{\bm{p}}(t)z^{\bm{p}}\Big]\Big)\,\mathcal{U}\,.
\end{align*}
After some manipulations using Lem.~\ref{lem:commutations_kuramoto} and simplifications, one finds
\begin{align*}
[\mathcal{U}, \mathcal{S}] - \mathcal{U}[\xi(t,z)]\,\mathcal{U} = &\sum_{\ell, \bm{p}}\dot\varphi_{\ell \bm{p}}(t)z^{\bm{p}}\partial_\ell + \sum_{\ell,\bm{p},j,k}\varphi_{\ell \bm{p}}(t)(A_{jk}[z_k\partial_j, z^{\bm{p}}\partial_\ell] - \bar{A}_{jk}[z_k^{-1}z_j^2\partial_j, z^{\bm{p}}\partial_\ell]) \\&- \sum_{\bm{p}}\left(\dot{\varepsilon}_{\bm{p}}(t) + \varepsilon_{\bm{p}}(t)\sum_{r,s}p_r(A_{rs}z_r^{-1}z_s - \bar{A}_{rs}z_rz_s^{-1})\right)z^{\bm{p}}\sum_{j,k}(A_{jk}z_k - \bar{A}_{jk}\bar{z}_kz_j^2)\partial_j\,.
\end{align*}
The commutation relations are explicitly given by
\begin{align*}
    [z_k\partial_j, z^{\bm{p}}\partial_\ell] &= p_j z^{\bm{p} - \bm e_j + \bm e_k}\partial_\ell - \delta_{k\ell}z^{\bm{p}}\partial_j\\
    [z_k^{-1}z_j^2\partial_j, z^{\bm{p}}\partial_\ell] &= p_j z^{\bm{p} +\bm e_j - \bm e_k}\partial_\ell + \delta_{k\ell}z^{\bm{p} +2\bm{e}_j - 2\bm{e}_k}\partial_j - 2z^{\bm{p} +\bm{e}_j - \bm{e}_k}\delta_{j\ell}\partial_j\,.
\end{align*}
Substituting these commutation relations into the infinitesimal condition yields, after simplifications,
{\footnotesize\begin{align*}
    [\mathcal{U}, \mathcal{S}] - \mathcal{U}[\xi(t,z)]\,\mathcal{U} = &\sum_{\ell, \bm{p}}\dot\varphi_{\ell \bm{p}}(t)z^{\bm{p}}\partial_\ell + \sum_{\ell,\bm{p},j,k}\varphi_{\ell \bm{p}}(t)
    A_{jk}p_jz^{\bm{p} - \bm{e}_j + \bm{e}_k}\partial_\ell - \sum_{\ell,\bm{p},j}\varphi_{\ell \bm{p}}(t)A_{j\ell}z^{\bm{p}}\partial_j\\ &- \sum_{\ell,\bm{p}, j,k}\varphi_{\ell \bm{p}}(t)\bar{A}_{jk}p_j z^{\bm{p} + \bm{e}_j - \bm{e}_k}\partial_\ell - \sum_{\ell,\bm{p}, j}\varphi_{\ell \bm{p}}(t)\bar{A}_{j\ell}z^{\bm{p} + 2\bm{e}_j - 2\bm{e}_\ell}\partial_j  + 2\sum_{\ell,\bm{p}, k}\varphi_{\ell \bm{p}}(t)\bar{A}_{\ell k}z^{\bm{p} + \bm{e}_\ell - \bm{e}_k}\partial_\ell 
    \\&- \sum_{\bm{p}}\left(\dot{\varepsilon}_{\bm{p}}(t) + \varepsilon_{\bm{p}}(t)\sum_{r,s}p_r (A_{rs}z_r^{-1}z_s - \bar{A}_{rs}z_rz_s^{-1})\right)z^{\bm{p}}\sum_{j,k}(A_{jk}z_k - \bar{A}_{jk}\bar{z}_kz_j^2)\partial_j\,.
\end{align*}}
From there, let's simplify again the equations to extract the determining equations. First,
{\footnotesize \begin{align*}
    [\mathcal{U}, \mathcal{S}] - \mathcal{U}[\xi(t,z)]\,\mathcal{U} = &\sum_{\ell, \bm{p}}\dot\varphi_{\ell \bm{p}}(t)z^{\bm{p}}\partial_\ell + \sum_{\ell,\bm{p},j,k}\varphi_{\ell \bm{p}}(t)
    A_{jk}p_jz^{\bm{p} - \bm{e}_j + \bm{e}_k}\partial_\ell - \sum_{\ell,\bm{p},j}\varphi_{\ell \bm{p}}(t)A_{j\ell}z^{\bm{p}}\partial_j\\ &- \sum_{\ell,\bm{p}, j,k}\varphi_{\ell \bm{p}}(t)\bar{A}_{jk}p_j z^{\bm{p} + \bm{e}_j - \bm{e}_k}\partial_\ell - \sum_{\ell,\bm{p}, j}\varphi_{\ell \bm{p}}(t)\bar{A}_{j\ell}z^{\bm{p} + 2\bm{e}_j - 2\bm{e}_\ell}\partial_j  + 2\sum_{\ell,\bm{p}, k}\varphi_{\ell \bm{p}}(t)\bar{A}_{\ell k}z^{\bm{p} + \bm{e}_\ell - \bm{e}_k}\partial_\ell 
    \\&-\sum_{\bm{p},j,k} \dot{\varepsilon}_{\bm{p}}(t)A_{jk}z^{\bm{p} +\bm{e}_k}\partial_j + \sum_{\bm{p},j,k} \dot{\varepsilon}_{\bm{p}}(t)\bar{A}_{jk}z^{\bm{p} + 2\bm{e}_j - \bm{e}_k}\partial_j
    \\& -\sum_{\bm{p},j,k,r,s} p_r \varepsilon_{\bm{p}}(t)A_{rs}A_{jk}z^{\bm{p} - \bm{e}_r + \bm{e}_s + \bm{e}_k}\partial_j + \sum_{\bm{p},j,k,r,s} p_r \varepsilon_{\bm{p}}(t)A_{rs}\bar{A}_{jk}z^{\bm{p} - \bm{e}_r + \bm{e}_s - \bm{e}_k + 2\bm{e}_j}\partial_j
    \\& +\sum_{\bm{p},j,k,r,s} p_r \varepsilon_{\bm{p}}(t)\bar{A}_{rs}A_{jk}z^{\bm{p} + \bm{e}_r - \bm{e}_s + \bm{e}_k}\partial_j + \sum_{\bm{p},j,k,r,s} p_r \varepsilon_{\bm{p}}(t)\bar{A}_{rs}\bar{A}_{jk}z^{\bm{p} + \bm{e}_r - \bm{e}_s - \bm{e}_k + 2\bm{e}_j}\partial_j\,.
\end{align*}}
Making the change of indices to yield $z^{\bm{p}}$ in every term leads to
\begin{align*}
    &[\mathcal{U}, \mathcal{S}] - \mathcal{U}[\xi(t,z)]\,\mathcal{U} = \sum_{\ell, \bm{p}}\dot\varphi_{\ell \bm{p}}(t)z^{\bm{p}}\partial_\ell + \sum_{\ell,\bm{p},j,k}(p_j + 1 - \delta_{jk})\varphi_{\ell \bm{p} + \bm{e}_j - \bm{e}_k}(t)
    A_{jk}z^{\bm{p}}\partial_\ell - \sum_{\ell,\bm{p},j}\varphi_{\ell \bm{p}}(t)A_{j\ell}z^{\bm{p}}\partial_j\\ &- \sum_{\ell,\bm{p}, j, k}(p_j - 1 + \delta_{jk})\varphi_{\ell \bm{p}-\bm{e}_j+\bm{e}_k}(t)\bar{A}_{jk}z^{\bm{p}}\partial_\ell - \sum_{\ell,\bm{p}, j}\varphi_{\ell \bm{p} - 2\bm{e}_j + 2\bm{e}_\ell}(t)\bar{A}_{j\ell}z^{\bm{p}}\partial_j  + 2\sum_{\ell,\bm{p}, k}\varphi_{\ell \bm{p} - \bm{e}_\ell + \bm{e}_k}(t)\bar{A}_{\ell k}z^{\bm{p}}\partial_\ell 
    \\&-\sum_{\bm{p},j,k} \dot{\varepsilon}_{\bm{p} -\bm{e}_k}(t)A_{jk}z^{\bm{p}}\partial_j + \sum_{\bm{p},j,k} \dot{\varepsilon}_{\bm{p} - 2\bm{e}_j + \bm{e}_k}(t)\bar{A}_{jk}z^{\bm{p}}\partial_j
     -\sum_{\bm{p},j,k,r,s} (p_r + 1 - \delta_{rs} - \delta_{rk}) \varepsilon_{\bm{p} + \bm{e}_r - \bm{e}_s - \bm{e}_k}(t)A_{rs}A_{jk}z^{\bm{p}}\partial_j 
    \\&+ \sum_{\bm{p},j,k,r,s} (p_r  + 1 - \delta_{rs} + \delta_{rk} - 2\delta_{rj}) \varepsilon_{\bm{p}  + \bm{e}_r - \bm{e}_s + \bm{e}_k - 2\bm{e}_j}(t)A_{rs}\bar{A}_{jk}z^{\bm{p}}\partial_j
    \\& +\sum_{\bm{p},j,k,r,s} (p_r  - 1 + \delta_{rs} - \delta_{rk})\varepsilon_{\bm{p} - \bm{e}_r + \bm{e}_s - \bm{e}_k}(t)\bar{A}_{rs}A_{jk}z^{\bm{p}}\partial_j 
    \\&+ \sum_{\bm{p},j,k,r,s} (p_r - 1 + \delta_{rs} + \delta_{rk} - 2\delta_{rj})\varepsilon_{\bm{p} - \bm{e}_r + \bm{e}_s + \bm{e}_k - 2\bm{e}_j}(t)\bar{A}_{rs}\bar{A}_{jk}z^{\bm{p}}\partial_j\,.
\end{align*}
The change of indices from $j$ to $\ell$ to get a factor $\partial_\ell$ in every sum implies 
\begin{align*}
    &[\mathcal{U}, \mathcal{S}] - \mathcal{U}[\xi(t,z)]\,\mathcal{U} = \sum_{\ell, \bm{p}}\dot\varphi_{\ell \bm{p}}(t)z^{\bm{p}}\partial_\ell + \sum_{\ell,\bm{p},j,k}(p_j + 1 - \delta_{jk})\varphi_{\ell \bm{p} + \bm{e}_j - \bm{e}_k}(t)
    A_{jk}z^{\bm{p}}\partial_\ell - \sum_{j,\bm{p},\ell}\varphi_{j \bm{p}}(t)A_{\ell j}z^{\bm{p}}\partial_\ell\\ &- \sum_{\ell,\bm{p}, j, k}(p_j - 1 + \delta_{jk})\varphi_{\ell \bm{p}-\bm{e}_j+\bm{e}_k}(t)\bar{A}_{jk}z^{\bm{p}}\partial_\ell - \sum_{j,\bm{p}, \ell}\varphi_{j \bm{p} - 2\bm{e}_\ell + 2\bm{e}_j}(t)\bar{A}_{\ell j}z^{\bm{p}}\partial_\ell  + 2\sum_{\ell,\bm{p}, k}\varphi_{\ell \bm{p} - \bm{e}_\ell + \bm{e}_k}(t)\bar{A}_{\ell k}z^{\bm{p}}\partial_\ell 
    \\&-\sum_{\bm{p},\ell,k} \dot{\varepsilon}_{\bm{p} -\bm{e}_k}(t)A_{\ell k}z^{\bm{p}}\partial_\ell + \sum_{\bm{p},\ell,k} \dot{\varepsilon}_{\bm{p} - 2\bm{e}_\ell + \bm{e}_k}(t)\bar{A}_{\ell k}z^{\bm{p}}\partial_\ell
     -\sum_{\bm{p},\ell,k,r,s} (p_r + 1 - \delta_{rs} - \delta_{rk}) \varepsilon_{\bm{p} + \bm{e}_r - \bm{e}_s - \bm{e}_k}(t)A_{rs}A_{\ell k}z^{\bm{p}}\partial_\ell 
    \\&+ \sum_{\bm{p},\ell,k,r,s} (p_r + 1 - \delta_{rs} + \delta_{rk} - 2\delta_{r\ell}) \varepsilon_{\bm{p}  + \bm{e}_r - \bm{e}_s + \bm{e}_k - 2\bm{e}_\ell}(t)A_{rs}\bar{A}_{\ell k}z^{\bm{p}}\partial_\ell
    \\& +\sum_{\bm{p},\ell,k,r,s} (p_r - 1 + \delta_{rs} - \delta_{rk})\varepsilon_{\bm{p} - \bm{e}_r + \bm{e}_s - \bm{e}_k}(t)\bar{A}_{rs}A_{\ell k}z^{\bm{p}}\partial_\ell 
    \\&+ \sum_{\bm{p},\ell,k,r,s} (p_r - 1 + \delta_{rs} + \delta_{rk} - 2\delta_{r\ell})\varepsilon_{\bm{p} - \bm{e}_r + \bm{e}_s + \bm{e}_k - 2\bm{e}_\ell}(t)\bar{A}_{rs}\bar{A}_{\ell k}z^{\bm{p}}\partial_\ell\,.
\end{align*}
The infinitesimal condition of symmetry then yields
\begin{align*}
    0 &= \dot\varphi_{\ell \bm{p}}(t) + \sum_{j,k}(p_j + 1 - \delta_{jk})A_{jk}\varphi_{\ell \bm{p} + \bm{e}_j - \bm{e}_k}(t)
     - \sum_{j}A_{\ell j}\varphi_{j \bm{p}}(t)- \sum_{j, k}(p_j - 1 + \delta_{jk})\bar{A}_{jk}\varphi_{\ell \bm{p}-\bm{e}_j+\bm{e}_k}(t) \\&- \sum_{j}\bar{A}_{\ell j}\varphi_{j \bm{p} - 2\bm{e}_\ell + 2\bm{e}_j}(t)  + 2\sum_{k}\bar{A}_{\ell k} \varphi_{\ell \bm{p} - \bm{e}_\ell + \bm{e}_k}(t)
    -\sum_{k} A_{\ell k}\dot{\varepsilon}_{\bm{p} -\bm{e}_k}(t) + \sum_{k} \bar{A}_{\ell k}\dot{\varepsilon}_{\bm{p} - 2\bm{e}_\ell + \bm{e}_k}(t)
    \\& -\sum_{k,r,s} (p_r + 1 - \delta_{rs} - \delta_{rk})A_{rs}A_{\ell k} \varepsilon_{\bm{p} + \bm{e}_r - \bm{e}_s - \bm{e}_k}(t) 
    + \sum_{k,r,s} (p_r  + 1 - \delta_{rs} + \delta_{rk} - 2\delta_{r\ell})A_{rs}\bar{A}_{\ell k} \varepsilon_{\bm{p}  + \bm{e}_r - \bm{e}_s + \bm{e}_k - 2\bm{e}_\ell}(t)
    \\& +\sum_{k,r,s} (p_r  - 1 + \delta_{rs} - \delta_{rk})\bar{A}_{rs}A_{\ell k}\varepsilon_{\bm{p} - \bm{e}_r + \bm{e}_s - \bm{e}_k}(t) 
    + \sum_{k,r,s} (p_r - 1 + \delta_{rs} + \delta_{rk} - 2\delta_{r\ell})\bar{A}_{rs}\bar{A}_{\ell k}\varepsilon_{\bm{p} - \bm{e}_r + \bm{e}_s + \bm{e}_k - 2\bm{e}_\ell}(t)
\end{align*}
for all $\ell\in\{1,...,N\}$, $\bm{p}\in\mathbb{Z}^N$. By rearranging, one finds the general determining equations 
\begin{align*}
    &\dot\varphi_{\ell \bm{p}}(t) -\sum_{k} A_{\ell k}\dot{\varepsilon}_{\bm{p} -\bm{e}_k}(t) + \sum_{k} \bar{A}_{\ell k}\dot{\varepsilon}_{\bm{p} - 2\bm{e}_\ell + \bm{e}_k}(t) = \sum_{k}[A_{\ell k}\varphi_{k \bm{p}}(t) + \bar{A}_{\ell k}\varphi_{k \bm{p} - 2\bm{e}_\ell + 2\bm{e}_k}(t) - 2\bar{A}_{\ell k} \varphi_{\ell \bm{p} - \bm{e}_\ell + \bm{e}_k}(t)]
\\&+ \sum_{j, k}(p_j - 1 + \delta_{jk})\bar{A}_{jk}\varphi_{\ell \bm{p}-\bm{e}_j+\bm{e}_k}(t) - \sum_{j,k}(p_j + 1 - \delta_{jk})A_{jk}\varphi_{\ell \bm{p} + \bm{e}_j - \bm{e}_k}(t)  \\&+ \sum_{k,r,s} (p_r + 1 - \delta_{rs} - \delta_{rk})A_{rs}A_{\ell k} \varepsilon_{\bm{p} + \bm{e}_r - \bm{e}_s - \bm{e}_k}(t) - \sum_{k,r,s} (p_r  + 1 - \delta_{rs} + \delta_{rk} - 2\delta_{r\ell})A_{rs}\bar{A}_{\ell k} \varepsilon_{\bm{p}  + \bm{e}_r - \bm{e}_s + \bm{e}_k - 2\bm{e}_\ell}(t)
\\&- \sum_{k,r,s} (p_r  - 1 + \delta_{rs} - \delta_{rk})\bar{A}_{rs}A_{\ell k}\varepsilon_{\bm{p} - \bm{e}_r + \bm{e}_s - \bm{e}_k}(t) - \sum_{k,r,s} (p_r - 1 + \delta_{rs} + \delta_{rk} - 2\delta_{r\ell})\bar{A}_{rs}\bar{A}_{\ell k}\varepsilon_{\bm{p} - \bm{e}_r + \bm{e}_s + \bm{e}_k - 2\bm{e}_\ell}(t)\,,
\end{align*}
for all $\ell\in\{1,...,N\}$, $\bm{p}\in\mathbb{Z}^N$, which represent an infinite-dimensional differential-algebraic system of equations. 

\subsection{Determining matrix and its singular vectors as symmetry generator coefficients}
\label{SIsubsec:determining_matrix}

Let's search for state-space symmetries, i.e., those where $\xi(t, z) = 0$. In this context, the determining equations become an infinite-dimensional system of ordinary differential equations
\begin{align*}
    \dot\varphi_{\ell \bm{p}}(t) = &\sum_{k}[A_{\ell k}\varphi_{k \bm{p}}(t) + \bar{A}_{\ell k}\varphi_{k \bm{p} - 2\bm{e}_\ell + 2\bm{e}_k}(t) - 2\bar{A}_{\ell k} \varphi_{\ell \bm{p} - \bm{e}_\ell + \bm{e}_k}(t)]
\\&+ \sum_{j, k}(p_j - 1 + \delta_{jk})\bar{A}_{jk}\varphi_{\ell \bm{p}-\bm{e}_j+\bm{e}_k}(t) - \sum_{j,k}(p_j + 1 - \delta_{jk})A_{jk}\varphi_{\ell \bm{p} + \bm{e}_j - \bm{e}_k}(t) \,,
\end{align*} 
where $\ell\in\{1,...,N\}$ and $\bm{p} \in\mathbb{Z}^N$. One notices that it
forms an infinite-dimensional \textit{linear} system of equations and that only the coefficients related to the monomials of the same total degree are dependent over one another. We can thus treat the different total degrees separately. However, we can seek a specific symmetry generator that has a finite number of nonzero coefficients ($\varphi_{\ell \bm{p}}(t))_{\ell,\bm{p}}$.
Consider that the nonzero coefficients are such that $\bm{p}$ is in a finite subset $\mathbb{P} \subset \mathbb{Z}^N$ including $d := N\cdot \# \mathbb{P}$ coefficients.
For some ordering of $(\ell,\bm{p})$, the determining equations become an overdetermined linear system of differential-algebraic equations described by 
\begin{align}
    \dot{\bm{\varphi}} &= \mathcal{M}\,\bm{\varphi}\,, \label{eq:det_diff}\\
    \bm{0} &= \mathcal{N}\,\bm{\varphi}\,,\label{eq:det_alg}
\end{align}
where we name $\mathcal{M}$ the differential determining matrix and $\mathcal{N}$ the algebraic determining matrix, as they entirely determine the possibility of having a generator of symmetry or not. As one should expect, $\mathcal{M}$ and $\mathcal{N}$ solely depend upon the elements of the complex weight matrix $A$. The differential determining matrix $\mathcal{M}$ is a $d \times d$ matrix, while $\mathcal{N}$ is a $r\times d$ matrix where $r$ is the number of equations such that $\dot\varphi_{\ell \bm{p}}(t) = 0$ for $\bm{p}\notin\mathbb{P}$. The algebraic equations appear since there are shifts of coefficients in these determining equations that yield nonzero coefficients. 

The solvability of Eqs~(\ref{eq:det_diff}-\ref{eq:det_alg}) is questionable since it is generally strongly overdetermined. Yet, we know that for a set $\mathbb{P}$ containing  sufficient $\bm{p}$'s of total degree one, there are at least two solutions ($L_0$ and $\mathcal{K}$). In fact, if $\dot{\bm{\varphi}} = \bm{0}$, then the following recurrence relations hold:
\begin{align}
    0 = &\sum_{k}\Big[A_{\ell k}\varphi_{k \bm{p}} + \bar{A}_{\ell k}\varphi_{k \bm{p} - 2\bm{e}_\ell + 2\bm{e}_k} - 2\bar{A}_{\ell k} \varphi_{\ell \bm{p} - \bm{e}_\ell + \bm{e}_k}\nonumber
    \\&+ \sum_{j}(p_j - 1 + \delta_{jk})\bar{A}_{jk}\varphi_{\ell \bm{p}-\bm{e}_j+\bm{e}_k} - \sum_{j}(p_j + 1 - \delta_{jk})A_{jk}\varphi_{\ell \bm{p} + \bm{e}_j - \bm{e}_k}\Big] \,.
\end{align}
In matrix form, 
\begin{align}
    D(A)\,\bm{\varphi} = \bm{0}
\end{align}
where $D(A)$ is a $m \times d$ ($m>d$) complex rectangular matrix depending on the elements of a complex square matrix $A$ with imaginary diagonal and $\bm{\varphi}$ is a complex vector of dimension $d\times 1$. Under what conditions on $A$ does the overdetermined system $D(A)\,\bm{\varphi} = \bm{0}$ admit nontrivial solutions?
The last equation means that the nullspace of $D(A)$ must have a dimension greater than or equal to one (the nullspace always contains the null vector, but if it only contains $\bm{0}$, its dimension is 0). By the rank-nullity theorem \cite{Horn2013},
\begin{align*}
    \dim(\mathrm{nullspace}(D(A))) = N\cdot\# \mathbb{P} - \mathrm{rank}(D(A))\,,
\end{align*}
and thus, in order to have a symmetry, the following inequality must be satisfied:
\begin{align*}
    \mathrm{rank}(D(A)) < N\cdot\# \mathbb{P}\,.
\end{align*}
In other words, the multiplicity of the zero singular value for $D(A)$ must be greater or equal to 1. The right singular vectors with singular value 0 are the coefficients of the time-independent symmetry generator.

\subsection{Proof of Theorem 4: Peripheral constants of motion}
\label{SIsubsec:proof_thm4}
Using diverse matrices $A$ for $N = 4$ and $N = 5$, we performed symbolic and numerical calculations (see Ref.~\cite{Thibeault2026_koopman_kuramoto}, koopman-kuramoto/symbolic/symmetries) to obtain the associated determining matrices $D(A)$, their singular value decomposition and from the singular vectors with zero singular value, symmetry generators. In this way, we inferred a class of symmetry generators that enable the creation of new constants of motion in the Kuramoto model on graph. This subsection is devoted to the proof of Thm.~\ref{thm:thm4} (Thm.~\ref{thm:thm4SI}) on these symmetry-generated constants of motion. To that end, we first introduce a lemma that specifies the conditions (see Fig.~\ref{fig:symmetry_graph}) under which the Kuramoto model admits such symmetry generators, along with their explicit form.
\begin{lemma}[{\small Time evolution of peripheral oscillators in the frame of their source is a symmetry}]\label{lem:sym_KL}\phantom{.}\\
   If there is a source oscillator with natural frequency $\omega_s$ and it has outgoing edges toward $r > 1$ disjoint subgraphs whose vertex sets are denoted by $\mathcal{W}_1, ..., \mathcal{W}_r$, then the Koopman generators of the subgraphs in the rotating frame of the source,
\begin{align}\label{eq:liesymSI}
    \mathcal{S}_{\eta} = \mathcal{K}_{\eta} - i\omega_{s} L_0^{\eta}\,,\quad\eta\in\{1,...,r\}\,,
\end{align}
are generators of Lie symmetries, where $\mathcal{K}_{\eta} = \sum_{j\in\mathcal{W}_{\eta}}\sum_{k\in\mathcal{W}_{\eta}\cup\{s\}}(A_{jk}z_k - \bar{A}_{jk}\bar{z}_kz_j^2)\partial_j$ 
 and $L_0^{\eta} = \sum_{j\in\mathcal{W}_{\eta}}z_j\partial_j$.
\end{lemma}
\begin{proof}
    The existence of a source connected to disjoint subsets implies that the Koopman generator splits as
    \begin{align}
        \mathcal{K} = \mathcal{K}_s + \sum_{\tau=1}^r\mathcal{K}_\tau = i\omega_sz_s\partial_s + \sum_{\tau=1}^r\sum_{j\in\mathcal{W}_\tau}\sum_{k\in\mathcal{W}_\tau\cup\{s\}}(A_{jk}z_k - \bar{A}_{jk}\bar{z}_kz_j^2)\partial_j\,.
    \end{align}
    Lemma~\ref{lem:symkoo}, stated in $\bm z$-coordinates on the $N$-torus, implies that satisfying the commutation relations $[\mathcal{K}, \mathcal{S}_\eta] = 0$ for all $\eta$ is a sufficient condition for the present lemma to hold. Now, as illustrated in Fig.~\ref{fig:symmetry_graph}, each subgraph $\mathcal{W}_1, \ldots, \mathcal{W}_r$ has a certain fraction of vertices contained in $\mathcal{R}_1, \ldots, \mathcal{R}_r$ that receive from the source. The generator of the source $\mathcal{K}_s$ only acts on the phase of these oscillators and it is convenient to split the generators related to the subgraphs as
    \begin{align*}
        \mathcal{K}_\tau = \mathcal{K}_{\mathcal{R}_\tau} + \mathcal{K}_{\mathcal{W}_{\tau}},\quad \tau\in\{1,...,r\}\,,
    \end{align*}
    where $\mathcal{K}_{\mathcal{R}_\tau} = \sum_{j\in\mathcal{R}_\tau}(A_{js}z_s - \bar{A}_{js}\bar{z}_sz_j^2)\partial_j$ and $\mathcal{K}_{\mathcal{W}_\tau} = \sum_{j,k\in\mathcal{W}_\tau}(A_{jk}z_k - \bar{A}_{jk}\bar{z}_kz_j^2)\partial_j$. Hence,
    \begin{align*}
        [\mathcal{K}, \mathcal{S}_\eta] = [\,i\omega_sz_s\partial_s + \textstyle{\sum_{\tau=1}^r}(\mathcal{K}_{\mathcal{R}_\tau} + \mathcal{K}_{\mathcal{W}_{\tau}})\,\,, \,\,\mathcal{K}_{\mathcal{R}_\eta} + \mathcal{K}_{\mathcal{W}_\eta} - i\omega_sL_0^{\eta}\,]\,.
    \end{align*}
    Using bilinearity and keeping only the nontrivial commutators yields
    \begin{align*}
        [\mathcal{K}, \mathcal{S}_\eta] = i\omega_s[z_s\partial_s, \mathcal{K}_{\mathcal{R}_\eta}] - i\omega_s[\mathcal{K}_{\mathcal{R}_\eta} + \mathcal{K}_{\mathcal{W}_\eta}, L_0^\eta]\,.
    \end{align*}
    But clearly, one also finds $[\mathcal{K}_{\mathcal{W}_\eta}, L_0^\eta] = 0$ with the commutation relations of Lem.~\ref{lem:commutations_kuramoto} (or the intuition that $L_0^{\eta}$ is the dilatation symmetry generator for $\mathcal{W}_\eta$ and thus commutes with $\mathcal{K}_{\mathcal{W}_\eta}$) and $[\mathcal{K}_{\mathcal{R}_\eta}, L_0^\eta] =  \sum_{k\in\mathcal{R}_\eta}[\mathcal{K}_{\mathcal{R}_\eta},z_k\partial_k]$. Hence,
    \begin{align*}
        [\mathcal{K}, \mathcal{S}_\eta] &= i\omega_s(\,\textstyle{\sum_{j\in\mathcal{R}_\eta}}[z_s\partial_s\,, \,(A_{js}z_s - \bar{A}_{js}\bar{z}_sz_j^2)\partial_j] - \textstyle{\sum_{j,k\in\mathcal{R}_\eta}}[(A_{js}z_s - \bar{A}_{js}\bar{z}_sz_j^2)\partial_j , z_k\partial_k]\,)\,.
    \end{align*}
    On the one hand, the first term is 
    \begin{align*}
        \sum_{j\in\mathcal{R}_\eta}[z_s\partial_s\,, \,(A_{js}z_s - \bar{A}_{js}\bar{z}_sz_j^2)\partial_j] = \sum_{j\in\mathcal{R}_\eta}(A_{js}z_s + \bar{A}_{js}\bar{z}_sz_j^2)\partial_j\,.
    \end{align*}
    On the other hand, the commutation relation $[\,\ell_{j}^m, \,\ell_{k}^n\,]=\delta_{jk}(n-m)\,\ell_j^{m+n}$ of Lem.~\ref{lem:commutations_kuramoto} implies that the second term is
    \begin{align*}
        \sum_{j,k\in\mathcal{R}_\eta}[(A_{js}z_s - \bar{A}_{js}\bar{z}_sz_j^2)\partial_j , z_k\partial_k] = \sum_{j,k\in\mathcal{R}_\eta}(A_{js}z_s[\ell_j^{-1} , \ell_k^0] - \bar{A}_{js}\bar{z}_s[\ell_j^{1} , \ell_k^0]) = \sum_{j\in\mathcal{R}_\eta}(A_{js}z_s + \bar{A}_{js}\bar{z}_sz_j^2)\partial_j \,.
    \end{align*}
    Consequently, $[\mathcal{K}, \mathcal{S}_\eta] = 0$ for all $\eta \in\{1,...,r\}$, so each $\mathcal{S}_\eta$ is indeed a Lie symmetry generator of the Kuramoto dynamics.  
\end{proof}
\begin{remark}
    For $r = 1$, the symmetry generator is $\mathcal{S} = \mathcal{K} - i\omega_sL_0$ and hence, $\mathcal{S}$ is linearly dependent on $\mathcal{K}$ and $L_0$. This means that the symmetry generator does not enable the creation of a new constant of motion in the way we do in Thm.~\ref{thm:thm4} of the main text.
\end{remark}
\begin{figure}
    \centering
    \includegraphics[width=\linewidth]{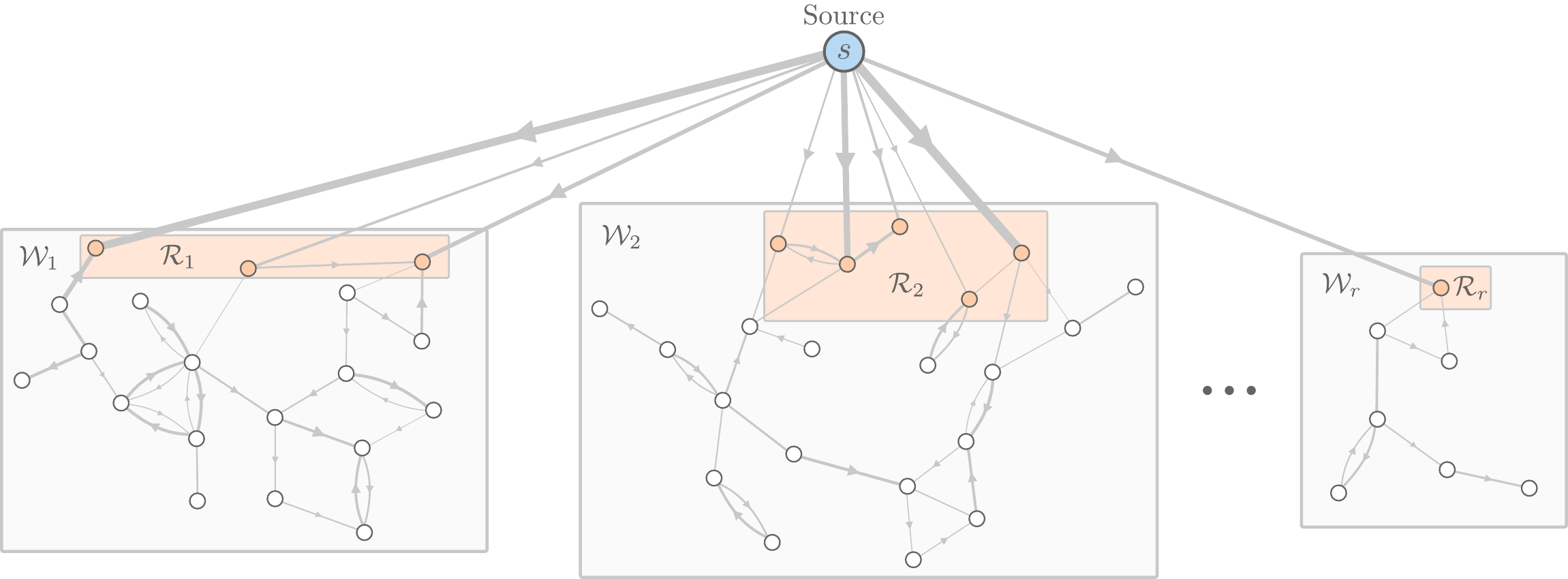}
    \caption[Illustration of a graph with symmetry generators $\mathcal{S}_1,...,\mathcal{S}_r$]{\textbf{Illustration of a graph with symmetry generators $\mathcal{S}_1,...,\mathcal{S}_r$} [Lem.~\ref{lem:sym_KL}]. A source $s$ is connected to $r$ subsets $\mathcal{R}_1, ...,\mathcal{R}_r$ of peripheral parts $\mathcal{W}_1, ..., \mathcal{W}_r$. Each of these parts is related to a symmetry generator $\mathcal{S}_1,...,\mathcal{S}_r$ that makes the oscillators of the peripheral parts evolve in time in the frame of the source.}
    \label{fig:symmetry_graph}
\end{figure}
In the next theorem, we use the possible coexistence of the symmetry generators in the last lemma and the conserved cross-ratios (Thm.~\ref{thm:cte_mvt_kuramoto_graphe}) to generate new functionally independent constants of motion.
\begin{theorem}[Thm.~\ref{thm:thm4} of the paper]\label{thm:thm4SI}
   Consider that the Kuramoto model in Def.~\eqref{def:kuramoto} has a symmetry generator $\mathcal{S}_\eta$ as defined in Eq.~\eqref{eq:liesymSI} related to the subgraph $\mathcal{W}_\eta$ and the source oscillator $s$.
   
\vspace{0.1cm}
   
   \noindent A. If four vertices $a,b,c,d \in\mathcal{V}\setminus\{s\}$ have \\
   \indent(A1) a unique incoming edge with weight $\mathcal{A}_s$ from $s$;\\
   \indent(A2) identical natural frequencies $\omega$;\\
   \indent(A3) and one, two or three of them belong to $\mathcal{W}_\eta$,\\
   then both the cross-ratio $c_{abcd}$ and $\mathcal{S}_\eta[c_{abcd}]$ are conserved and functionally independent.
   
\vspace{0.1cm}

   \noindent B. If three vertices $u, v, w\in\mathcal{V}\setminus\{s\}$ have \\
   \indent(B1) a unique incoming edge with weight $\mathcal{A}_s$ from $s$;\\
   \indent(B2) identical natural frequencies $\omega = \omega_s - 2\,\mathrm{Im}(\mathcal{A}_s)$;\\
   \indent(B3) and one or two of them belong to $\mathcal{W}_\eta$\\
then both the cross-ratio $c_{suvw}$ and $\mathcal{S}_\eta[c_{suvw}]$ are conserved and functionally independent.
\end{theorem}
\begin{proof}
    A. Condition~(A1) implies that the four vertices are mutually disconnected ($A_{jk} = 0$ for all $j,k\in\{a,b,c,d\}$ with $j\neq k$) as they can only have an incoming edge from the source. Therefore, condition~(1) from Thm.~\ref{thm:cte_mvt_kuramoto_graphe} is satisfied. Condition~(A1) also highlights that the weights of these incoming edges are all equal to $\mathcal{A}_s$, meaning that condition~2 of Thm.~\ref{thm:cte_mvt_kuramoto_graphe} holds. Then, condition~3 of Thm.~\ref{thm:cte_mvt_kuramoto_graphe} is also fulfilled from condition~(A2) and the fact that $A_{jk} = 0$ for all $j,k\in\{a,b,c,d\}$ with $j\neq k$. Altogether, Thm.~\ref{thm:cte_mvt_kuramoto_graphe} guarantees that $\mathcal{K}[c_{abcd}] = 0$, that is, $c_{abcd}$ is a constant of motion. Lemma~\ref{lem:sym_KL} shows that $\mathcal{S}_\eta = \mathcal{K}_\eta - i\omega_sL_0^\eta$ is a symmetry generator and condition~(A3) ensures that $\mathcal{S}_\eta[c_{abcd}]$ is not zero. Therefore, $\mathcal{K}\mathcal{S}_{\eta}[c_{abcd}] = \mathcal{S}_{\eta}\mathcal{K}[c_{abcd}] = 0$, i.e., $\mathcal{S}_\eta[c_{abcd}]$ is another constant of motion. Since $\mathcal{K}_\eta$ depends on the source's state $z_s$, $\mathcal{S}_{\eta}$ and $\mathcal{S}_\eta[c_{abcd}]$ also do. Therefore, $\mathcal{S}_\eta[c_{abcd}]$ is functionally independent of $c_{abcd}$, which only depends on $z_a, z_b, z_c, z_d$.

    B. Conditions~(B1) and (B2) imply that all the conditions of Thm.~\ref{thm:cte_mvt_kuramoto_graphe} are fulfilled and hence, $c_{suvw}$ is conserved. Then, condition~(B3) guarantees that $\mathcal{S}_\eta[c_{suvw}]$ is not zero and as in part A of the proof, $\mathcal{S}_\eta[c_{suvw}]$ is conserved. If there are vertices other than $u$, $v$ or $w$ belonging to $\mathcal{W}_\eta$, then $\mathcal{S}_\eta[c_{suvw}]$ is functionally independent of $c_{suvw}$. If only one or two vertices among $u$, $v$ or $w$ are in $\mathcal{W}_\eta$, the symmetry generator $\mathcal{S}_\eta$ can take the 6 different forms
    {\footnotesize\begin{align*}
        \begin{cases}
  \mathcal{S}_\eta^{x} = (i(\omega_x - \omega_s)z_x + A_{xs}z_s - \bar{A}_{xs}\bar{z}_sz_x^2)\partial_x & x\in\mathcal{W}_\eta \land x \in\{u,v,w\} \\
  \mathcal{S}_\eta^{xy} = (i(\omega_x - \omega_s)z_x + A_{xs}z_s - \bar{A}_{xs}\bar{z}_sz_x^2)\partial_x + (i(\omega_y - \omega_s)z_y + A_{ys}z_s - \bar{A}_{ys}\bar{z}_sz_y^2)\partial_y & x,y\in\mathcal{W}_\eta \land x,y \in\{u,v,w\}\,,
\end{cases}
    \end{align*}}
    where $\mathcal{S}_\eta^{xy} = \mathcal{S}_\eta^{yx}$ and $x\neq y$. The constants of motion $c_{suvw}$ and $\mathcal{S}_\eta[c_{suvw}]$ are functionally independent if the rank of the Jacobian matrix
    \begin{align*}
        J = \begin{pmatrix}
            \pd[]{c_{suvw}}{z_s}&\pd[]{c_{suvw}}{z_u}& \pd[]{c_{suvw}}{z_v} & \pd[]{c_{suvw}}{z_w}\\
            \pd[]{\mathcal{S}_\eta[c_{suvw}]}{z_s}&\pd[]{\mathcal{S}_\eta[c_{suvw}]}{z_u}& \pd[]{\mathcal{S}_\eta[c_{suvw}]}{z_v} & \pd[]{\mathcal{S}_\eta[c_{suvw}]}{z_w}
        \end{pmatrix}
    \end{align*}
     is 2, where $\mathcal{S}_{\eta}$ is either $\mathcal{S}_{\eta}^{u}$, $\mathcal{S}_{\eta}^{v}$, $\mathcal{S}_{\eta}^{w}$, $\mathcal{S}_{\eta}^{uv}$, $\mathcal{S}_{\eta}^{uw}$, or $\mathcal{S}_{\eta}^{vw}$. For the six Jacobian matrices, lengthy but straightforward calculations enable showing that their rank is 2 (see \textit{proof\_thm4\_partB.wls} in Ref.~\cite{Thibeault2026_koopman_kuramoto} for symbolic calculations).   
\end{proof}

\begin{remark}
    \textit{A priori}, one could hope to generate new constants of motion from the class of symmetry generators in Lem.~\ref{lem:sym_KL} and monomial eigenfunctions. However, if there is a subgraph with vertex set $\mathcal{M}$ supporting a monomial eigenfunction, it must be a source and there is no way to make $\mathcal{S}_\eta$ act on only a subset of $\mathcal{M}$. More precisely, $\mathcal{M}$ can only be a source (first condition of Thm.~\ref{thm:existence_fpmonom}) to another vertex $v$ that also receives from the source~$s$. The vertex set for the subgraph admitting the symmetry generator $\mathcal{S}_{\eta}$ is thus $\mathcal{W}_\eta = \{v\}\cup\mathcal{M}$. Using conditions 2,3,4 of Thm.~\ref{thm:existence_fpmonom}, $\mathcal{S}_{\eta}[z^{\bm{\nu}}] = (i\sum_{j\in\mathcal{M}}\nu_j(\omega_j - \omega_1))\,z^{\bm{\nu}}$, i.e., the monomial is an eigenfunction of the symmetry, $\mathcal{S}_{\eta}[z^{\bm{\nu}}]$ is functionally dependent on $z^{\bm{\nu}}$ and $\mathcal{S}_{\eta}[z^{\bm{\nu}}]$ is not a new constant of motion.
\end{remark}

\subsection{Basic examples for Theorem~4}
\label{SIsubsec:basic_example_thm4}
The example that helped us obtain Thm.~\ref{thm:thm4} through the singular vectors of the determining matrix is the following one.
\begin{example}
\label{ex:star}
    Consider a directed star of 5 nodes with weight matrix
    \begin{align*}
        A = \begin{pmatrix}
            i\omega_1/2 & 0 & 0 & 0 & 0\\
            \mathcal{A}_1 & i\omega/2& 0 & 0 & 0 \\
            \mathcal{A}_1 & 0 & i\omega/2&  0 & 0 \\
            \mathcal{A}_1 & 0 & 0 &i\omega/2&  0 \\
            \mathcal{A}_1 & 0 & 0 & 0 &i\omega/2
        \end{pmatrix}\,,
    \end{align*}
    where $\mathcal{A}_1$ is any complex number, $\omega \in\mathbb{R}$ and we assume for now that $\omega_1 \neq \omega + 2\,\Imag(\mathcal{A}_1)$. The Koopman generator of the dynamics is $\mathcal{K} = i\omega_1z_1\partial_1 + \mathcal{K}_2 + \mathcal{K}_3 + \mathcal{K}_4 + \mathcal{K}_5$, where $\mathcal{K}_{\eta} = (i\omega z_\eta + \mathcal{A}_1z_1 - \bar{\mathcal{A}}_1\bar{z}_1z_\eta^2)\,\partial_\eta$ for $\eta \in\{2,3,4,5\}$. The most straightforward approach is to write the solution $z_1(t) = z_1(0)e^{i\omega_1t}$ and then substitute it in the four independent equations for $z_2$ to $z_5$. From there, the solutions for $z_2(t)$ to $z_5(t)$ are found by quadrature. Yet, using the results of the paper leads to conserved observables and we can avoid computing the latter quadratures. Indeed, Thm.~\ref{thm:cte_mvt_kuramoto_graphe} readily guarantees that there is one conserved cross-ratio
    \begin{align*}
        C_1(\bm{z}) := c_{2345}(\bm{z}) = \frac{(z_4 -z_2)(z_5 - z_3)}{(z_4 - z_3)(z_5 - z_2)}
    \end{align*}
    and four symmetries $\mathcal{S}_{\eta} = \mathcal{K}_{\eta} - i\omega_1 z_\eta \partial_\eta$. Using the derivatives of cross-ratios computed in Eq.~\eqref{eq:derivees_cr}, we thus find the four constants of motion
    \begin{align*}
        C_2(\bm{z}) &:= \mathcal{S}_2[C_1(\bm{z})] = [i(\omega - \omega_1)z_2 + \mathcal{A}_1z_1 - \bar{\mathcal{A}}_1\bar{z}_1z_2^2]\frac{(z_5 -z_3)(z_4 - z_5)}{(z_4 - z_3)(z_5 - z_2)^2}\\
        C_3(\bm{z}) &:= \mathcal{S}_3[C_1(\bm{z})] = [i(\omega - \omega_1)z_3 + \mathcal{A}_1z_1 - \bar{\mathcal{A}}_1\bar{z}_1z_3^2]\frac{(z_4 -z_2)(z_5 - z_4)}{(z_4 - z_3)^2(z_5 - z_2)}\\
        C_4(\bm{z}) &:= \mathcal{S}_4[C_1(\bm{z})] = [i(\omega - \omega_1)z_4 + \mathcal{A}_1z_1 - \bar{\mathcal{A}}_1\bar{z}_1z_4^2]\frac{(z_2 -z_3)(z_5 - z_3)}{(z_4 - z_3)^2(z_5 - z_2)}\\
        C_5(\bm{z}) &:= \mathcal{S}_5[C_1(\bm{z})] = [i(\omega - \omega_1)z_5 + \mathcal{A}_1z_1 - \bar{\mathcal{A}}_1\bar{z}_1z_5^2]\frac{(z_4 -z_2)(z_3 - z_2)}{(z_4 - z_3)(z_5 - z_2)^2}\,,
    \end{align*}
    which can be verified analytically with $\mathcal{K}[C_\eta] = 0$ or with symbolic calculations. Note also that $C_2 + C_3 + C_4 + C_5 = 0$, meaning that there is at least one functional dependency. In fact, it is easily verified symbolically that the rank of the Jacobian matrix of $(C_1, C_2,...,C_5)$ is 3, so we have three functionally independent constants of motion. 
    Note that there is also one more. Indeed, since there is a source, we have a monomial eigenfunction $z_1$ with eigenvalue $i\omega_1$ and $C_0(t, \bm{z}) = z_1e^{-i\omega_1 t}$ is a constant of motion. The dynamics can thus be reduced to two autonomous equations and three constants of motion (e.g., $C_1, C_2, C_3$) or one non-autonomous equation and four constants of motion  (e.g., $C_0, C_1, C_2, C_3$).

    If, moreover, $\omega_1 = \omega + 2\,\Imag(\mathcal{A}_1)$, then the cross-ratio
    \begin{align*}
        C_6(\bm{z}) := c_{1234}(\bm{z}) = \frac{(z_3 -z_1)(z_4 - z_2)}{(z_3 - z_2)(z_4 - z_1)}
    \end{align*}
    is also a constant of motion ($c_{1345}$ and others are functionally dependent with $c_{1234}, c_{2345}$ as shown in subsection~\ref{subsec:functional_independence}), along with
    \begin{align*}
        C_7(\bm{z}) &:= \mathcal{S}_2[C_6(\bm{z})] = [i(\omega - \omega_1)z_2 + \mathcal{A}_1z_1 - \bar{\mathcal{A}}_1\bar{z}_1z_2^2]\frac{(z_3 -z_1)(z_3 - z_4)}{(z_2 - z_3)^2(z_1 - z_4)}\\
        C_8(\bm{z}) &:= \mathcal{S}_3[C_6(\bm{z})] = [i(\omega - \omega_1)z_3 + \mathcal{A}_1z_1 - \bar{\mathcal{A}}_1\bar{z}_1z_3^2]\frac{(z_1 -z_2)(z_2 - z_4)}{(z_2 - z_3)^2(z_1 - z_4)}\\
        C_9(\bm{z}) &:= \mathcal{S}_4[C_6(\bm{z})] = [i(\omega - \omega_1)z_4 + \mathcal{A}_1z_1 - \bar{\mathcal{A}}_1\bar{z}_1z_4^2]\frac{(z_1 -z_2)(z_3 - z_1)}{(z_2 - z_3)(z_4 - z_1)^2}
        \,.
    \end{align*}
    In this case, there are 5 functionally independent constants of motion, say $C_0, C_1, C_2, C_3, C_6$. This completely integrates the system without having to perform quadratures.  See also the symbolic calculations for this example in the Mathematica script \textit{example\_cte\_mvt\_from\_symmetry\_1.wls} for symbolic validation and Fig.~\ref{fig:num_validation} for numerical validation.
\end{example}
\begin{figure}
    \centering
    \includegraphics[width=\linewidth]{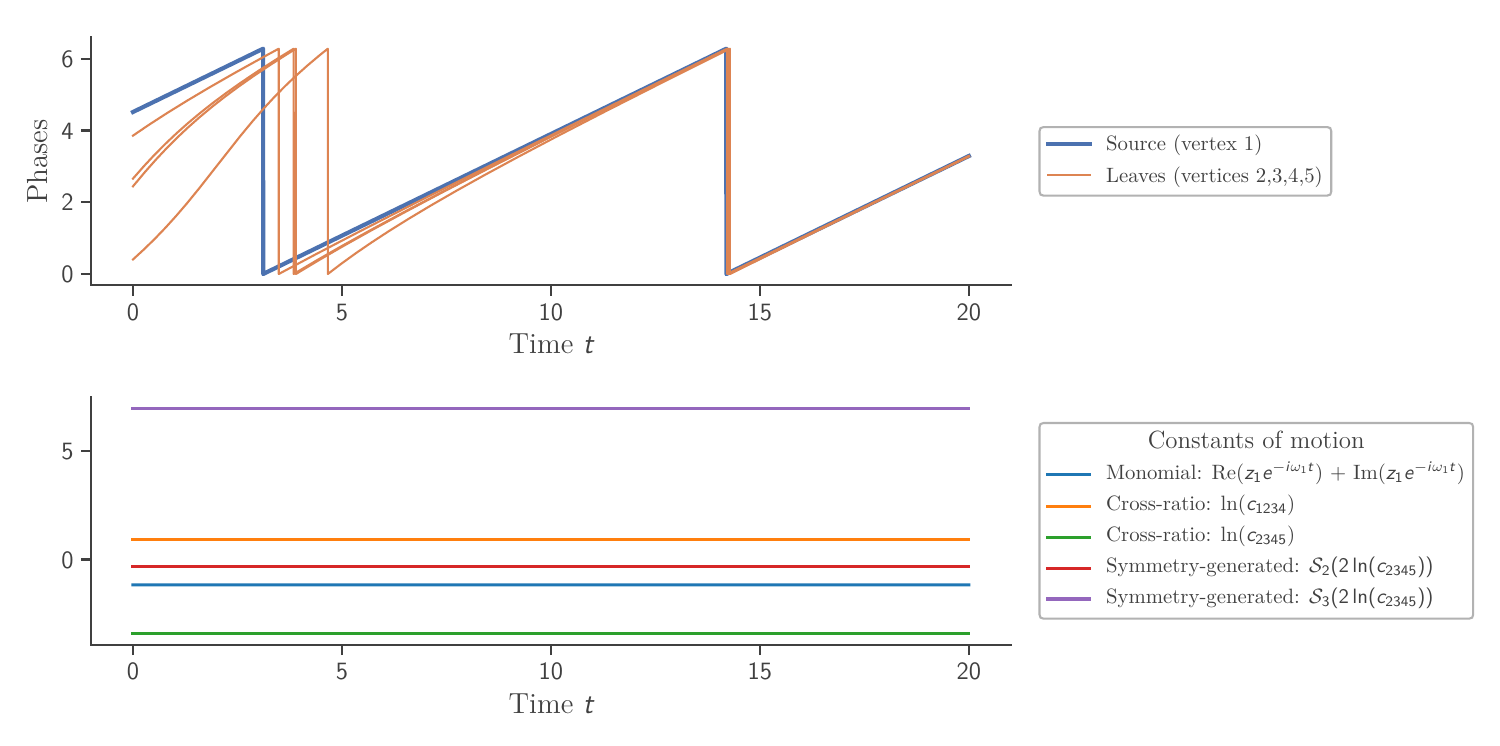}
    \vspace{-0.8cm}
    \caption[Basic example with coexisting conserved quantities of different types]{\textbf{Basic example with coexisting conserved quantities of different types.} The figure is related to Example~\ref{ex:star}. We evaluate the constants of motion at the phases for each time point to verify their conservation. The initial conditions $\bm{\theta}(0) \approx (4.51756368 \,\,3.85865453 \,\,2.66025984 \,\,0.4049007\,\,  2.44481427)^\top$ are drawn from a uniform distribution. Parameters: $\alpha = \pi/3$, $\sigma = 1$, $\mathcal{A}_1 = (\sigma/4)\exp(-i\alpha)$, $\omega = 1$, $\omega_1 = \omega + 2 \Imag(\mathcal{A}_1)$.}
    \label{fig:num_validation}
\end{figure}
\begin{remark}
    By applying the symmetry generators to the logarithm of the cross-ratios, the form of the new constants of motion is simplified in the above example. For instance, assuming that $c_{2345}$ is positive,
    \begin{align*}
        \mathcal{S}_2[\ln c_{2345}] = [i(\omega - \omega_1)z_2 + \mathcal{A}_1z_1 - \bar{\mathcal{A}}_1\bar{z}_1z_2^2]\frac{(z_4 - z_5)}{(z_4 - z_2)(z_5 - z_2)}\,.
    \end{align*}
    If $\mathcal{A}_1 = (\sigma/4)e^{-i\alpha}\in\mathbb{C}$ with $\sigma\in\mathbb{R}$ and $|\alpha| \leq \pi/2$, the real form for the constant of motion is
    \begin{align*}
        \mathcal{S}_2[2\ln c_{2345}] = (\omega - \omega_1 + (\sigma/2)\sin(\theta_1 - \theta_2 - \alpha))\frac{\sin\left(\frac{\theta_4 - \theta_5}{2}\right)}{\sin\left(\frac{\theta_4 - \theta_2}{2}\right)\sin\left(\frac{\theta_5 - \theta_2}{2}\right)}\,.
    \end{align*}
    Additionally, if $\alpha = 0$, $\omega_1 = \omega$ and the constant of motion is simplified to
    \begin{align*}
        \mathcal{S}_2[(2/\sigma)\ln c_{2345}] = \frac{C_{12}S_{12}S_{45}}{S_{42}S_{52}}\,,
    \end{align*}
    where $S_{jk} := \sin\left(\frac{\theta_j - \theta_k}{2}\right)$ and $C_{jk} := \cos\left(\frac{\theta_j - \theta_k}{2}\right)$.
\end{remark}
The leaves of the star can be sources within arbitrary subgraphs. One of the simplest cases is presented in the next example.
\begin{example}
 Consider the star from the previous example, but connect vertices 2 and 3 to a sixth vertex and vertices 4 and 5 to a seventh vertex. The weight matrix is thus
    \begin{align*}
        A = \begin{pmatrix}
            i\omega_1/2 & 0 & 0 & 0 & 0 & 0 & 0\\
            \mathcal{A}_1 & i\omega/2& 0 & 0 & 0 & 0 & 0\\
            \mathcal{A}_1 & 0 & i\omega/2&  0 & 0 & 0 & 0\\
            \mathcal{A}_1 & 0 & 0 &i\omega/2&  0 & 0 & 0\\
            \mathcal{A}_1 & 0 & 0 & 0 &i\omega/2 & 0 & 0\\
            A_{61} & A_{62} & A_{63} & 0 & 0 & i\omega_6/2 & 0\\
            A_{71} & 0 & 0 & A_{74} & A_{75} & 0 & i\omega_7/2
        \end{pmatrix}
    \end{align*}
    with $\omega_1 \neq \omega + 2\,\Imag(\mathcal{A}_1)$ and the Koopman generator is $\mathcal{K} = i\omega_1z_1\partial_1 + \sum_{\eta = 2}^5\mathcal{K}_{\eta} + \mathcal{K}_6 + \mathcal{K}_7$, where $\mathcal{K}_2$ to $\mathcal{K}_5$ are defined as in the previous example and $\mathcal{K}_6 = \sum_{k=1}^3(A_{6k}z_k - \bar{A}_{6k}\bar{z}_k z_6^2)\partial_6$, $\mathcal{K}_7 = \sum_{k\in\{1,4,5\}}(A_{7k}z_k - \bar{A}_{7k}\bar{z}_k z_7^2)\partial_7$. The monomial $z_1e^{-i\omega_1 t}$ and $c_{2345}$ are still conserved, but there remain only two symmetries:
    \begin{align*}
        \mathcal{S}_{1} = \mathcal{K}_{2} + \mathcal{K}_3 + \mathcal{K}_6 - i\omega_1\sum_{j\in\{2,3,6\}}z_j\partial_j \qquad \text{and} \qquad \mathcal{S}_{2} = \mathcal{K}_{4} + \mathcal{K}_5 + \mathcal{K}_7 - i\omega_1\sum_{j\in\{4,5,7\}}z_j\partial_j\,.
    \end{align*}
    We find that 
    {\footnotesize\begin{align*}
        &\mathcal{S}_{1}[c_{2345}(\bm{z})] = -\mathcal{S}_{2}[c_{2345}(\bm{z})]\\ &= [2i\Imag(\mathcal{A}_1) (z_2 z_3 - z_4 z_5) + \mathcal{A}_1 z_1 (z_4 + z_5 - z_2 - z_3) -  
     \bar{\mathcal{A}}_1\bar{z}_1 (z_4 z_5 (z_2 + z_3) - z_2z_3 (z_4 + z_5))]\frac{(z_2 - z_3) (z_4 - 
   z_5)}{(z_3 - z_4)^2 (z_2 - z_5)^2}
    \end{align*}}
    is another functionally independent constant of motion (see the symbolic calculations in the Mathematica scripts \textit{example\_cte\_mvt\_from\_symmetry\_2.wls} and \textit{example\_cte\_mvt\_from\_symmetry\_3.wls} when $\omega_1 = \omega + 2\Imag(\mathcal{A}_1)$).
\end{example}

\subsection{Explicit form for the symmetry transformation related to \texorpdfstring{$\mathcal{S}_\eta$}{Lg}}
\label{subsec:explicit_sym}

We have shown in Lemma~\ref{lem:sym_KL} that
\begin{align*}
    \mathcal{S}_{\eta} = \mathcal{K}_{\eta} - i\omega_{s} L_0^{\eta}\,,\quad\eta\in\{1,...,r\}\,,
\end{align*}
are symmetry generators representing the time evolution of peripheral oscillators in the frame of their source. Finding the explicit action of the symmetry $\exp(\varepsilon \mathcal{S}_\eta)$ thus implies having an explicit form for the flow of the peripheral oscillators, which is generally not available when there are more than three oscillators.

There are however special cases where we are able to find the symmetry transformation more explicitly. For instance, consider the symmetry generator
\begin{align*}
\mathcal{S}_\eta = \rho_\eta(\bm z) L_{-1}^\eta + i(\Omega_\eta - \omega_s) L_{0}^\eta - \overline{\rho_\eta(\bm z)} L_{1}^\eta
\end{align*}
for a group of peripheral oscillators~$\mathcal{W}_\eta$ with $\# \mathcal{W}_\eta \geq 3$ and
\begin{align*}
    L_n^\eta = \sum_{j\in\mathcal{W}_\eta} z_j^{n+1}\partial_j\,,\qquad \rho_\eta(\bm z) = \sum_{k\in\mathcal{W}_\eta\cup \{s\}} \mathcal{A}_{\eta k}z_k\,, \qquad \Omega_\eta = \omega_{\ell_\eta} - 2\Imag(\mathcal{A}_{\eta \ell_\eta})\,,
\end{align*}
where $\ell_\eta$ is the index of any oscillator within $\mathcal{W}_\eta$ and $\mathcal{A}_{\eta k} \in \mathbb{C}$. Note that the generator takes this form when the vertices in $\mathcal{W}_\eta$ with $\# \mathcal{W}_\eta \geq 4$ satisfy the conditions from Thm.~\ref{thm:cte_mvt_kuramoto_graphe}.  Without loss of generality, we can set $\omega_s = 0$ (frame of the source) and $\theta_s(0) = 0$ ($z_s(0) = 1$) to obtain
\begin{align}
    \mathcal{S}_\eta = \rho_\eta(\bm z) L_{-1}^\eta + i\Omega_\eta L_{0}^\eta - \overline{\rho_\eta(\bm z)} L_{1}^\eta\,,\quad\text{with}\quad
    \rho_\eta(\bm z) = \sum_{k\in\mathcal{W}_\eta} \mathcal{A}_{\eta k}z_k + \mathcal{A}_s
\end{align}
From there, Watanabe-Strogatz theory (following Ref.~\cite{Marvel2009}) almost tells us directly, up to minor modifications, the form of $\exp(\varepsilon \mathcal{S}_\eta)$. Indeed, given a solution $z_j(t)\in\mathbb{T}$ for all time $t$ and $j\in\mathcal{V}$, the symmetry acts as a Möbius transformation to generate the new solution
\begin{align}
    \tilde{z}_j(t) = \exp(\varepsilon \mathcal{S}_\eta)z_j(t) = \left\{
\begin{aligned}
    &\frac{e^{i\phi_\eta(\varepsilon)}z_j(t) + Z_\eta(\varepsilon)}{1 + e^{i\phi_\eta(\varepsilon)}\overline{Z_\eta(\varepsilon)}z_j(t)} \,, \quad &&j\in \mathcal{W}_\eta\,,\\
   &\qquad\quad z_j(t)\,, \quad &&j\in \mathcal{V}\setminus\mathcal{W}_\eta\,,
\end{aligned}
\right.
\label{eq:transformed_solution}
\end{align}
where $Z_\eta(\varepsilon)$ and $\phi_\eta(\varepsilon)$ satisfy
\begin{align}
    \od[]{Z_\eta}{\varepsilon} &= F_\eta(Z_\eta, \phi_\eta, t) + i\Omega_\eta Z_\eta - \overline{F_\eta(Z_\eta, \phi_\eta, t)}Z_\eta^2\,\,,\label{eq:ws1}\\
    \od[]{\phi_\eta}{\varepsilon} &= \Omega_\eta + i(\overline{F_\eta(Z_\eta, \phi_\eta, t)}Z_\eta - F_\eta(Z_\eta, \phi_\eta, t)\bar{Z}_\eta)\,,\label{eq:ws2}\\
    (Z_\eta(0), \phi_\eta(0)) &= (0, 0)\,,\phantom{\od[]{Z_\eta}{\varepsilon}}\label{eq:ws0}\\ 
    F_\eta(Z_\eta, \phi_\eta, t) &= \mathcal{A}_s + \sum_{k\in\mathcal{W}_\eta} \mathcal{A}_{\eta k}\frac{e^{i\phi_\eta(\varepsilon)}z_k(t) + Z_\eta(\varepsilon)}{1 + e^{i\phi_\eta(\varepsilon)}\overline{Z_\eta(\varepsilon)}z_k(t)}\,,
\end{align}
which is an autonomous dynamics (note the derivative over $\varepsilon$ and not $t$). To apply the symmetry transformation and get $\tilde{z}_j(t)$ numerically, one proceeds as follows:
\begin{enumerate}
    \item Obtain a solution $z_1(t),...,z_N(t)$ for $t$ in a discrete time-interval $t_0,...,t_{T-1}$ by integrating the equations of the Kuramoto model satisfying the conditions to have a symmetry from Lemma~\ref{lem:sym_KL};
    \item For $t = t_0$, integrate Eqs.~(\ref{eq:ws1}-\ref{eq:ws2}) with initial conditions~\eqref{eq:ws0} from 0 to $\varepsilon \in \mathbb{R}_{> 0}$;
    \item For $t = t_0$, apply the Möbius transformation such as in Eq.~\eqref{eq:transformed_solution} to obtain $\tilde{z}_j(t_0)$ for $j\in\mathcal{W}_\eta$ while $\tilde{z}_j(t_0)=z_j(t_0)$ for $j\in\mathcal{V}\setminus\mathcal{W}_\eta$;
    \item Repeat step 2 and step 3 for all $t \in \{t_1,...,t_{T-1}\}$ to obtain the new solution $\tilde{z}_1(t),...,\tilde{z}_N(t)$;
\end{enumerate}
One can always verify the validity of the transformed solution by integrating the related Kuramoto model at the new initial conditions generated by the symmetry transformation, provided in step 3.

If there are more than one part of peripheral oscillators satisfying the conditions of Thm.~\ref{thm:cte_mvt_kuramoto_graphe} (say, $r' \leq r$ parts), then the procedure is essentially the same. Indeed, in such case, one has a $r'$-parameter abelian symmetry group containing transformations of the form $\exp(\varepsilon_1 \mathcal{S}_1 +...+\varepsilon_{r'}\mathcal{S}_{r'}) = \exp(\varepsilon_1 \mathcal{S}_1)...\exp(\varepsilon_{r'}\mathcal{S}_{r'})$. Again, each transformation $\exp(\varepsilon_\eta \mathcal{S}_\eta)$ is a different Möbius transformation and one has to integrate the equations such as Eqs.~(\ref{eq:ws1}-\ref{eq:ws2}) from $\varepsilon_\eta = 0$ to $\varepsilon_\eta = \epsilon_\eta\in \mathbb{R}_{> 0}$ for all $\eta \in \{1,...,r'\}$.

Examples where a solution is transformed to another using the above formalism are provided in Ref.~\cite{Thibeault2026_koopman_kuramoto} under the name ``\textit{symmetry\_action\_kooku1\_one\_param\_step0.py}",``\textit{symmetry\_action\allowbreak\_kooku1\_two\_params.py}", ``\textit{symmetry\_action\_kooku1\_two\_params\_step1.py}" (step 2 and 3). Moreover, the scripts ``\textit{symmetry\_changes\_invariant\_sets.py}" and ``\textit{fig3\_plot\_invariant\_sets.py}" allow reproducing Fig.~3 (\textbf{B}, \textbf{C}, \textbf{D}).

As a final remark, taken altogether, the symmetry generators $i\,L_0$ (global dilatation or rotation of all the oscillators), $\mathcal{K} = \sum_{j,k \in\mathcal{V}}\left(A_{jk}z_k - \bar{A}_{jk}\bar{z}_kz_j^2\right)\partial_j$ (evolution of time-independent observables), the trivial generator $f(t)\,\mathcal{U}$ along with the time evolution of peripheral oscillators in the frame of the source $\mathcal{S}_1,...,\mathcal{S}_r$ form an abelian Lie algebra.

\subsection{Parametrized families of conserved quantities}
\label{SIsubsec:parametrized_families}
By considering the parameters $\Theta_1,...,\Theta_n$ of a dynamical system as state variables that do not evolve over time, there are other possibilities of symmetry generators. Indeed, one can always write a symmetry generator as
\begin{align*}
    v = \xi(t, \bm u, \bm\Theta)\partial_t + \sum_{j=1}^N\phi_j(t,\bm u, \bm\Theta)\partial_{u_j} + \sum_{k=1}^n\eta_k(t,\bm u, \bm\Theta)\partial_{\Theta_k}\,.
\end{align*}
In such context, one finds the following symmetry generator of the parameter-extended Kuramoto model.
\begin{lemma}[Heterogeneous translation of the phases and the phase lags is a symmetry]\phantom{.}\\
Let $\alpha_1,...,\alpha_N$ be real numbers. The infinitesimal generator 
\begin{align*}
\mathcal{S}_{\bm{\alpha}} &= i\sum_{j=1}^N \alpha_j z_j \partial_j-\sum_{p,q=1}^N (\alpha_p-\alpha_q)\,\partial_{\alpha_{pq}}
\end{align*}
generates a symmetry of the parameter-extended Kuramoto model acting such as
\begin{align*}
    e^{\varepsilon \mathcal{S}_{\bm{\alpha}}}f(\bm{z}, A, \bar{A}) = f(D_{\varepsilon}\bm z, D_{\varepsilon}AD_{-\varepsilon}, D_{-\varepsilon}\bar{A}D_{\varepsilon}) \,,
\end{align*}
where $D_{\varepsilon} = \mathrm{diag}(e^{i\varepsilon\alpha_1}, ...,e^{i\varepsilon\alpha_N})$. In real form, $\mathcal{S}_{\bm{\alpha}} = \sum_{j=1}^N \alpha_j \partial_{\theta_j}-\sum_{p,q=1}^N (\alpha_p-\alpha_q)\,\partial_{\alpha_{pq}}$ and
\begin{align*}
    e^{\varepsilon \mathcal{S}_{\bm{\alpha}}}g(\bm\theta, \bm\alpha) = g(\bm \theta + \varepsilon\bm\alpha\,, \,\,\alpha + \varepsilon(\bm{1}\bm{\alpha}^\top-\bm{\alpha}\bm{1}^\top))\,,
\end{align*}
where $\alpha = (\alpha_{jk})_{j,k=1}^N$ and $\bm{\alpha} = (\alpha_1,...,\alpha_N)$.
\end{lemma}
\begin{proof}
Using the bilinearity of the commutator yields
\begin{align}\label{eq:commutation_symalpha}
[\mathcal{K},\mathcal{S}_{\bm{\alpha}}]
=
i\sum_{j=1}^N \alpha_j[\mathcal{K},\ell_j^0]
-
\sum_{p,q=1}^N (\alpha_p-\alpha_q)\,[\mathcal{K},\partial_{\alpha_{pq}}].
\end{align}
Lemma~\ref{lem:commutations_kuramoto} directly gives
\begin{align*}
[\mathcal{K},\ell_j^0]
&=
\sum_{k=1}^N A_{jk}z_k\ell_j^{-1}
+
\sum_{k=1}^N \bar A_{jk}\,\bar z_k\ell_j^1
-
z_j\sum_{k=1}^N A_{kj}\ell_k^{-1}
-
\bar z_j\sum_{k=1}^N \bar A_{kj}\,\ell_k^1\,,
\end{align*}
where we recall that $A_{jk} = (1/2)(W_{jk}e^{-i\alpha_{jk}} + i\omega_j\delta_{jk})$ for all $j,k$. Hence, the first term in Eq.~\eqref{eq:commutation_symalpha} becomes
\begin{align*}
i\sum_{j=1}^N \alpha_j[\mathcal{K},\ell_j^0]
&=
i\sum_{j,k=1}^N(\alpha_j-\alpha_k)
\left(
A_{jk}z_k+\bar A_{jk}\,\bar z_k\,z_j^2
\right)\partial_j =
\frac{i}{2}\sum_{\substack{j,k=1\\j\neq k}}^N(\alpha_j-\alpha_k)
W_{jk}\left(
e^{-i\alpha_{jk}}z_k+e^{i\alpha_{jk}}\,\bar z_k\,z_j^2
\right)\partial_j
\end{align*}
after reindexing two of the four double sums. Then, for the second term in Eq.~\eqref{eq:commutation_symalpha} and some function $\psi(\bm{z},A)$,
\begin{align*}
[\mathcal{K},\partial_{\alpha_{pq}}]\psi(\bm{z},A)
&=
\mathcal{K}\!\left[\partial_{\alpha_{pq}}\psi(\bm{z},A)\right]
-
\partial_{\alpha_{pq}}
\Big[
\sum_{j,k=1}^N
\left(
A_{jk}z_k-\bar A_{jk}\,\bar z_k\,z_j^2
\right)\partial_j\psi(\bm{z},A)
\Big]
\\
&=-\frac{1}{2}\sum_{j,k=1}^N\partial_{\alpha_{pq}}\left((W_{jk}e^{-i\alpha_{jk}} + i\omega_j\delta_{jk})z_k-(W_{jk}e^{i\alpha_{jk}} -i\omega_j\delta_{jk})\bar z_k\,z_j^2\right)\partial_j\psi(\bm{z},A)
\\
&=
\frac{i}{2}W_{pq}\left(
e^{-i\alpha_{pq}}z_q+ e^{i\alpha_{pq}}\,\bar z_q\,z_p^2
\right)\partial_p\psi(\bm{z},A)\phantom{\sum_{j,k=1}^N}
\end{align*}
and
\begin{align*}
    \sum_{p,q=1}^N (\alpha_p-\alpha_q)\,[\mathcal{K},\partial_{\alpha_{pq}}] = \frac{i}{2}\sum_{\substack{p,q=1\\p\neq q}}^N (\alpha_p-\alpha_q)W_{pq}\left(
e^{-i\alpha_{pq}}z_q+e^{i\alpha_{pq}}\,\bar z_q\,z_p^2
\right)\partial_p\,.
\end{align*}
 Altogether, one gets $[\mathcal{K},\mathcal{S}_{\bm{\alpha}}] = 0$ which means that, by Lemma~\ref{lem:symkoo}, $\mathcal{S}_{\bm{\alpha}}$ generates a Lie symmetry. Let $\mathcal{S}_{\bm{\alpha}} = iL_0^{\bm\alpha} + T^{\bm\alpha}$ such that
\begin{align*} 
L_0^{\bm\alpha}:=\sum_{j=1}^N \alpha_j z_j\partial_j,
\qquad
T^{\bm\alpha}:=-\sum_{p,q=1}^N (\alpha_p-\alpha_q)\,\partial_{\alpha_{pq}}\,.
\end{align*}
The flow generated by the operators $iL_0^{\bm\alpha}$ and $T^{\bm\alpha}$ is respectively determined by the differential equations
\begin{align*}
    \od[]{\tilde{z}_j(\varepsilon)}{\varepsilon} = i\alpha_j \tilde z_j(\varepsilon)\,,\qquad \od[]{\tilde{\alpha}_{jk}(\varepsilon)}{\varepsilon} = -(\alpha_j - \alpha_k)
\end{align*}
with $\tilde{z}_j(0) = z_j$ and $\tilde\alpha_{jk}(0) = \alpha_{jk}$ for all $j,k\in\{1,...,N\}$. The solutions are respectively given by
\begin{align*}
    \tilde{z}_j(\varepsilon) = e^{i\varepsilon\alpha_{j}}z_j,\qquad \tilde\alpha_{jk}(\varepsilon) = -(\alpha_j - \alpha_k)\varepsilon + \alpha_{jk}\,.
\end{align*}
The real form is straightforward from there, noting that $\theta_j + \varepsilon \alpha_j$ for all $j$ and that the transformation of the phase lags in matrix form is $\alpha + \varepsilon(\bm{1}\bm{\alpha}^\top-\bm{\alpha}\bm{1}^\top)$. Moreover, one finds
\begin{align*}
    \frac{1}{2}(W_{jk}e^{-i\tilde\alpha_{jk}(\varepsilon)} + i\omega_j\delta_{jk}) = \frac{1}{2}(W_{jk}e^{i\varepsilon(\alpha_j - \alpha_k)}e^{-i\alpha_{jk}} + i\omega_j\delta_{jk})\,,
\end{align*}
or in matrix form, $\tilde{A}_\varepsilon = D_\varepsilon A D_{-\varepsilon}$ with $D_\varepsilon = \mathrm{diag}(e^{i\varepsilon\alpha_1},...,e^{i\varepsilon\alpha_N})$. Similar calculations give $\tilde{\bar{A}}_\varepsilon = D_{-\varepsilon} \bar A D_{\varepsilon}$. Since $[L_0^{\bm\alpha}, T^{\bm\alpha}]=0$, $e^{\varepsilon \mathcal{S}_{\bm{\alpha}}} = e^{i\varepsilon L_0^{\bm\alpha}}e^{\varepsilon T^{\bm\alpha}}$ and 
\begin{equation*}
e^{\varepsilon \mathcal{S}_{\bm{\alpha}}}f(z_1,...,z_N, A)
= f(e^{i\varepsilon\alpha_{1}}z_1,...,e^{i\varepsilon\alpha_{N}}z_N, D_\varepsilon A D_{-\varepsilon}, D_{-\varepsilon} \bar A D_{\varepsilon}) = f(D_\varepsilon \bm z, D_\varepsilon A D_{-\varepsilon}, D_{-\varepsilon} \bar A D_{\varepsilon})\,.
\end{equation*}
\end{proof}
\noindent The last lemma implies that we can generate families of Koopman eigenfunctions $\psi$. 
\begin{enumerate}
    \item \textit{Monomials}. If $z^{\bm{\mu}}$ is an eigenfunction of $\mathcal{K}$ (Thm~\ref{thm:monomials}), then $\mathcal{S}_{\bm{\alpha}}[z^{\bm{\mu}}] = i\left(\sum_{j=1}^N \mu_j\alpha_j\right)z^{\bm\mu}$. In this case, there is no new functionally independent eigenfunction generated by the transformation.
    \item \textit{Vandermonde-ratio}. If $V_{IJ}^{\bm\sigma\bm\tau}(\bm z) := (\prod_{p<q\in I}(\bar{z}_p-\bar{z}_q)^{\sigma_{pq}})/(\prod_{r<s\in J}(z_r-z_s)^{\tau_{rs}})$ is an eigenfunction of $\mathcal{K}$, then 
    \begin{align*}  e^{-\mathcal{S}_{\bm{\alpha}}}[V_{IJ}^{\bm\sigma\bm\tau}(\bm z)] 
        &= \frac{\prod_{p<q\in I}(e^{i\alpha_p}\bar{z}_p-e^{i\alpha_q}\bar{z}_q)^{\sigma_{pq}}}{\prod_{r<s\in J}(e^{-i\alpha_r}z_r-e^{-i\alpha_s}z_s)^{\tau_{rs}}}
        \\&= \frac{\prod_{p<q\in I}e^{i\alpha_p\sigma_{pq}}}{\prod_{r<s\in J}e^{-i\alpha_r\tau_{rs}}}\left(\frac{\prod_{p<q\in I}(\bar{z}_p-e^{-i(\alpha_p -\alpha_q)}\bar{z}_q)^{\frac{\sigma_{pq}}{\sigma_{p'q'}}}}{\prod_{r<s\in J}(z_r-e^{i(\alpha_r-\alpha_s)}z_s)^{\frac{\tau_{rs}}{\sigma_{p'q'}}}}\right)^{\sigma_{p'q'}}
    \end{align*}
    for some $p'<q' \in I$. Let $m := \# I$ and $n := \# J$. Extracting the factor in the parenthesis and absorbing $\sigma_{p'q'}$ into $\sigma_{pq}$ and $\tau_{rs}$ yields
    \begin{align*}
        V_{IJ}^{\bm\sigma\bm\tau\bm\alpha}(\bm z) := \frac{\prod_{p<q\in I}(\bar{z}_p-e^{-i(\alpha_p -\alpha_q)}\bar{z}_q)^{\sigma_{pq}}}{\prod_{r<s\in J}(z_r-e^{i(\alpha_r-\alpha_s)}z_s)^{\tau_{rs}}}
    \end{align*}
    for all $\bm\sigma$, $\bm\tau$, $\bm\alpha$, a family of Koopman eigenfunctions with ${{m+1}\choose{2}} + {{n+1}\choose{2}} - 3$ parameters (${{m}\choose{2}} + {{n}\choose{2}} - 1$ for the $\sigma$'s and $\tau$'s and $m+n-2$ for the $\alpha$'s) for the Kuramoto model satisfying the same first, second and fourth conditions of Thm.~\ref{thm:vandermonde_SI}, but the third condition becomes
    \begin{align*}
        A_{pr} = e^{i\alpha_p}\mathcal A\,T_re^{-i\alpha_r}\,,\qquad A_{rp} = e^{i\alpha_r}\bar{\mathcal{A}}\,S_pe^{-i\alpha_p} 
    \end{align*}
    for all $p\in I$, $r\in J$ and $\mathcal{A}\in\mathbb{C}$. Validations were done for $N = 5$ and $N = 9$ vertices in Mathematica notebooks called ``example\_N5\_extendedform.nb'' and ``example\_N9.nb'' in the folder ``symbolic/eigenfunctions/vandermonde\_eigenfunction'' in Ref.~\cite{Thibeault2026_koopman_kuramoto}.
    
    \item \textit{Cross-ratios}. If $c_{abcd}(\bm{z})$ is a constant of motion (Thm~\ref{thm:cte_mvt_kuramoto_graphe}), then 
    \begin{align*}
        e^{-\mathcal{S}_{\bm{\alpha}}}[c_{abcd}(\bm{z})] = \frac{(e^{-i\alpha_c}z_c - e^{-i\alpha_a}z_a)(e^{-i\alpha_d}z_d - e^{-i\alpha_b}z_b)}{(e^{-i\alpha_c}z_c - e^{-i\alpha_b}z_b)(e^{-i\alpha_d}z_d - e^{-i\alpha_a}z_a)} = \frac{(e^{i\gamma}z_c - z_a)(e^{i\delta}z_d - e^{i\beta}z_b)}{(e^{i\gamma}z_c - e^{i\beta}z_b)(e^{i\delta}z_d - z_a)} =: c_{abcd}^{\beta\gamma\delta}(\bm{z})\,,
    \end{align*}
    for all $\beta = \alpha_a - \alpha_b, \gamma = \alpha_a - \alpha_c, \delta = \alpha_a - \alpha_d$ form a three-parameter family of constants of motion for the Kuramoto model satisfying
    \begin{enumerate}
    \item 
\begin{align}\label{eq:identical_insiders_outcoming_edge_A_extended}
\begin{aligned}
    e^{i\beta}A_{ba} &= e^{i\gamma}A_{ca} = e^{i\delta}A_{da} =: \mathcal{A}_a\,,\\ 
    A_{ab} &= e^{i\gamma}A_{cb} = e^{i\delta}A_{db} =: \mathcal{A}_b\,, 
\end{aligned}
\qquad
\begin{aligned}
    A_{ac} &= e^{i\beta}A_{bc} = e^{i\delta}A_{dc} =: \mathcal{A}_c\,, \\
    A_{ad} &= e^{i\beta}A_{bd} = e^{i\gamma}A_{cd} =: \mathcal{A}_d\,;
\end{aligned}
\end{align}
\item 
\begin{equation}\label{eq:identical_outsiders_incoming_edge_A_extended}
    A_{ak} = e^{i\beta}A_{bk} = e^{i\gamma}A_{ck} = e^{i\delta}A_{dk},\quad \forall\, k\in\{1,...,N\}\setminus\{a, b, c, d\}\,;
\end{equation}
\item 
\begin{align}\label{eq:imag_cond_extended}
\omega_a - 2\Imag(\mathcal{A}_a) = \omega_b - 2\Imag(e^{-i\beta}\mathcal{A}_b) = \omega_c - 2\Imag(e^{-i\gamma}\mathcal{A}_c) = \omega_d - 2\Imag(e^{-i\delta}\mathcal{A}_d)\,.
\end{align}
\end{enumerate}
These conditions were verified for $N = 5$ vertices in a Mathematica notebook called ``symbolic/eigenfunctions/cross\_ratio\_cte\_motion/example\_N5\_extended\_crossratio.nb'' in Ref.~\cite{Thibeault2026_koopman_kuramoto}.

The extension for the cross-ratios can be generalized for the complex dynamics
\begin{align}\label{eq:complex_dynamics}
    \dot{Z}_j = f(t, \bm{Z})e^{-i\alpha_j} + g(t, \bm{Z})Z_j + h(t, \bm{Z})e^{i\alpha_j}Z_j^2
\end{align}
where $Z_j(t) \in \mathbb{C}$, $\alpha_j \in \mathbb{R}$ for all $j \in \{1,...,N\}$, and $f,g,h$ are some complex-valued functions. This is a phase-shifted version of the dynamics considered recently by Cestnik and Martens~\cite{Cestnik2024}. Its Koopman generator is
\begin{align*}
    \mathcal{K}_{fgh} = f(t, \bm{Z})\mathcal{L}_{-1}^{\bm\alpha} + g(t, \bm{Z})\mathcal{L}_{0}+ h(t, \bm{Z})\mathcal{L}_{1}^{\bm\alpha}\,
\end{align*}
where $\mathcal{L}_{-1}^{\bm\alpha} = \sum_{j=1}^N e^{-i\alpha_j}\partial_j$, $\mathcal{L}_{0} = \sum_{j=1}^NZ_j\partial_j$, and $\mathcal{L}_{1}^{\bm\alpha} = \sum_{j=1}^N e^{i\alpha_j}Z_j^2\partial_j$ with $\partial_j = \partial/\partial Z_j$. Making the change of variables $x_j = e^{i\alpha_j}Z_j$ implies that $\mathcal{L}_{-1}^{\bm\alpha} = \sum_{j=1}^N \partial/\partial x_j$, $\mathcal{L}_{0} = \sum_{j=1}^Nx_j\partial/\partial x_j$, $\mathcal{L}_{1}^{\bm\alpha} = \sum_{j=1}^N x_j^2\partial/\partial x_j$, and it is clear that the complex-valued cross-ratios $(x_c - x_a)(x_d - x_b)/((x_c - x_b)(x_d - x_a))$ are the joint invariants for these operators [recall Sec.~\ref{subsec:joint_invariants}], and thus, conserved quantities of Eq.~\eqref{eq:complex_dynamics}. In the original variables, one gets 
\begin{align*}
    C_{abcd}(\bm Z) = \frac{(e^{i\beta_c}Z_c - Z_a)(e^{i\beta_d}Z_d - e^{i\beta_b}Z_b)}{(e^{i\beta_c}Z_c - e^{i\beta_b}Z_b)(e^{i\beta_d}Z_d - Z_a)}\in\mathbb{C}\,.
\end{align*}
    \item \textit{WS integrals}. Consider the WS integral $C^{\text{ws}}_{p_1...p_\ell}(\bm\theta) = \prod_{u=1}^\ell \sin\left(\frac{\theta_{p_u} - \theta_{p_{u+1}}}{2}\right)$ with $p_0 = p_\ell$ and $p_{\ell + 1} = p_1$. In complex form, this integral is $I(\bm{z}) = \prod_{u=1}^\ell (1 - z_{p_u}\bar{z}_{p_{u+1}})$. Therefore, 
    \begin{align*}
        I_{\bm{\alpha}}(\bm z) = e^{-\mathcal{S}_{\bm{\alpha}}}[I(\bm z)] = \prod_{u=1}^\ell (1 - e^{i(\alpha_{p_{u+1}} - \alpha_{p_u})}z_{p_u}\bar{z}_{p_{u+1}})
    \end{align*}
    for all $\bm\alpha\in\mathbb{R}^\ell$ form a family of constants of motion with $\ell - 1$ degrees of freedom (all the differences between $\alpha$'s). Sufficient conditions to conserve the quantities of this family are stated in Lemma~\ref{lem:ws} for the sake of completeness.

\item \textit{Peripheral constants of motion}. Consider the symmetry generator $\mathcal{S}_\eta = \mathcal{K}_\eta - i\omega_sL_0^\eta$. The commutator $[\mathcal{K}_\eta, \mathcal{S}_{\mathcal{\alpha}}]$ is zero since $\mathcal{K}_\eta$ for all $\eta$ are independent Koopman generators. It is also easily shown that $[L_0^\eta, \mathcal{S}_{\bm{\alpha}}] = 0$. Hence, $[\mathcal{S}_\eta, \mathcal{S}_{\bm{\alpha}}] = 0$,
\begin{align*}
    \mathcal{S}_{\bm{\alpha}}[\mathcal{S}_\eta[c_{abcd}]] = \mathcal{S}_\eta[\mathcal{S}_{\bm{\alpha}}[c_{abcd}]] = \mathcal{S}_\eta[c_{abcd}^{\beta\gamma\delta}]
    \,,
\end{align*}
and the symmetry generator $\mathcal{S}_{\bm{\alpha}}$ also extends $\mathcal{S}_\eta[c_{abcd}]$ to a family of constants of motion. Let us consider one concrete example.
\end{enumerate}
\begin{example}
    Recall Example~\ref{ex:star} where a star of 5 vertices admitting conserved cross-ratios and peripheral constants of motion is considered. The system admitting these conserved quantities is extended to 
    \begin{align*}
        A = \begin{pmatrix}
            i\omega_1/2 & 0 & 0 & 0 & 0\\
            \mathcal{A}_1 & i\omega/2& 0 & 0 & 0 \\
            e^{-i\beta}\mathcal{A}_1 & 0 & i\omega/2&  0 & 0 \\
            e^{-i\gamma}\mathcal{A}_1 & 0 & 0 &i\omega/2&  0 \\
            e^{-i\delta}\mathcal{A}_1 & 0 & 0 & 0 &i\omega/2
        \end{pmatrix}\,,
    \end{align*}
    where $\beta = \alpha_2 - \alpha_3$, $\gamma = \alpha_2 - \alpha_4$, $\delta = \alpha_2 - \alpha_5$. Then, the extended conserved cross-ratio is 
    \begin{align*}
        e^{-\mathcal{S}_{\bm{\alpha}}}c_{2345}(\bm{z}) = \frac{(e^{i\gamma}z_4 - z_2)(e^{i\delta}z_5 - e^{i\beta}z_3)}{(e^{i\gamma}z_4 - e^{i\beta}z_3)(e^{i\delta}z_5 - z_2)}\,,
    \end{align*}
    and an example of extended peripheral constants of motion is 
    \begin{align*}
        e^{-\mathcal{S}_{\bm{\alpha}}}\mathcal{S}_2[c_{2345}(\bm{z})] &= [i(\omega - \omega_1)z_2 + \mathcal{A}_1z_1 - \bar{\mathcal{A}}_1\bar{z}_1z_2^2]\frac{(e^{i\delta}z_5 -e^{i\beta}z_3)(e^{i\gamma}z_4 - e^{i\delta}z_5)}{(e^{i\gamma}z_4 - e^{i\beta}z_3)(e^{i\delta}z_5 - z_2)^2}\,.
    \end{align*}
    A symbolic verification is provided in Ref.~\cite{Thibeault2026_koopman_kuramoto} (example\_N5\_star\_extended.nb in the folder symbolic/eigenfunctions/peripheral).
\end{example}


The following lemma states sufficient conditions to conserve a WS integral in the family generated by the symmetry $\mathcal{S}_{\bm{\alpha}}$. We present it for the sake of completeness~: it only corresponds to a slight adaptation of the conditions provided in Ref.~\cite{Watanabe1994}.
\begin{lemma}\label{lem:ws}
Let $\mathcal{W}\subset\mathcal{V}$ be a subset of vertices $\{p_1,\ldots,p_\ell\}$ with $2\leq\ell\leq N$ and $p_0:=p_\ell$, $p_{\ell+1}:=p_1$. If the Kuramoto model~\eqref{eq:kuramoto} satisfies
\begin{enumerate}
\item $\omega_{p_u}=\omega\in\mathbb{R}$ for all $u\in\{1,\ldots,\ell\}$;
\item $W_{p_uk}=0$ for all $u\in\{1,\ldots,\ell\}$ and $k\in\mathcal{V}\setminus\mathcal{W}$;
\item $W_{p_up_v}=w\in\mathbb{R}$ for all $u,v\in\{1,\ldots,\ell\}$;
\item $\alpha_{p_up_v}=\phi_{p_v}-\phi_{p_u}-\frac{\pi}{2}$ with $\phi_{p_u}\in\mathbb{R}$ for all $u,v\in\{1,\ldots,\ell\}$,
\end{enumerate}
then
\begin{align*}
I=\prod_{u=1}^{\ell}\sin\left(\frac{\vartheta_{p_u}-\vartheta_{p_{u+1}}}{2}\right)\quad\text{with}\quad\vartheta_{p_u}=\theta_{p_u}-\phi_{p_u}
\end{align*}
is a constant of motion.
\end{lemma}
\begin{proof}
The equations are arbitrary for $j\in\mathcal{V}\setminus\mathcal{W}$ in Eq.~\eqref{eq:kuramoto}. For $j\in\mathcal{W}$, applying the conditions to Eq.~\eqref{eq:kuramoto} yields 
\begin{align*}
    \dot{\theta}_{p_u} =\dot{\vartheta}_{p_u} = \omega + w\sum_{v=1}^\ell \cos(\vartheta_{p_v} - \vartheta_{p_u})\,,\qquad \forall\,u\in\{1,...,\ell\}\,,
\end{align*}
that is, the dynamics considered in Ref.~\cite{Watanabe1994}. The Koopman generator is $\mathcal{K}=\sum_{u=1}^\ell\dot{\theta}_{p_u}\partial_{p_u}+\mathcal{K}'$, where $\mathcal{K}'$ is the part acting on the oscillators labelled by indices in $\mathcal{V}\setminus\mathcal{W}$. Thus, $\mathcal{K}'[I]=0$ and $\mathcal{K}[I]=\sum_{u=1}^\ell\dot{\theta}_{p_u}\partial_{p_u}I$. From there, the argument of Watanabe and Strogatz~\cite[Appendix B.2]{Watanabe1994} applies directly and yields $\mathcal{K}[I]=0$.
\end{proof}

\begin{remark}
    Following the subsection~\ref{SIsubsec:WS_darboux}, the dynamics and the first condition of the last lemma is directly generalized by replacing $\omega$ by an arbitrary function $\beta$ of time and $\bm z$. This is because the WS function is an invariant of $L_0$. Note also that the diagonal elements of $W$ in the third condition could be chosen to be zero without loss of generality.
\end{remark}

\begin{remark}
In complex form the constant of motion takes the form
\begin{align*}
    \psi
    =
    \prod_{r=1}^{\ell}
    \left(1 - x_{r}\bar x_{r+1}\right) =(-2i)^\ell I,
\end{align*}
where $x_r = e^{-i\phi_{p_r}}z_{p_r} = e^{i(\theta_{p_r} - \phi_{p_r})} =: e^{i\vartheta_{p_r}}$. Indeed, the latter relation is obtained by noting that
\begin{align*}
    \psi
    &= \prod_{r=1}^{\ell}
    \left(1 - e^{i(\vartheta_{p_r}-\vartheta_{p_{r+1}})}\right) = (-1)^\ell\prod_{s=1}^{\ell}
    e^{i(\vartheta_{p_s}-\vartheta_{p_{s+1}})/2}\,\prod_{r=1}^{\ell}
    \left(e^{i(\vartheta_{p_r}-\vartheta_{p_{r+1}})/2} - e^{-i(\vartheta_{p_r}-\vartheta_{p_{r+1}})/2}\right),
\end{align*}
where $\prod_{s=1}^{\ell}e^{i(\vartheta_{p_s}-\vartheta_{p_{s+1}})/2} = 1$ because of the periodic condition. 
\end{remark}

\end{document}